\documentclass[a4paper,onecolumn,superscriptaddress,11pt,accepted=2023-09-05]{quantumarticle}
\pdfoutput=1
\usepackage[utf8]{inputenc}
\usepackage[english]{babel}
\usepackage{amsmath}
\usepackage{amssymb}
\usepackage{amsthm}
\usepackage[destlabel]{hyperref}
\usepackage{csquotes}
\usepackage[style=alphabetic]{biblatex}
\usepackage{algorithm2e}
\usepackage{caption}
\usepackage{subcaption}
\usepackage{multirow}
\usepackage[nameinlink, noabbrev, capitalise]{cleveref}
\usepackage{tocbasic}
\DeclareTOCStyleEntry[
  beforeskip=.5em plus 1pt,
  pagenumberformat=\textbf
]{tocline}{section}

\usepackage{tikz}
\usepackage{enumitem}
\setlist{nosep} 

\def\nrm#1{\mathopen{}\left\Vert #1\right\Vert\mathclose{}}
\def\floor#1{\mathopen{}\left\lfloor #1\right\rfloor\mathclose{}}
\def\ceil#1{\mathopen{}\left\lceil #1\right\rceil\mathclose{}}
\def\at#1{\mathopen{}\left(#1\right)\mathclose{}}

\def\set#1{\mathopen{}\left\{  #1\right\}\mathclose{}}
\def\ket#1{\mathopen{}\left|#1\right\rangle\mathclose{}}

\def\abs#1{\mathopen{}\left|#1\right|\mathclose{}}

\def\z{\mathbb{Z}}
\def\n{\mathbb{N}}
\def\c{\mathbb{C}}
\def\q{\mathbb{Q}}
\def\f{\mathbb{F}}
\def\r{\mathbb{R}}

\def\dst{\mathcal{D}_\diamond}

\def\tr{\mathrm{Tr}}

\newcommand{\eq}[1]{\hyperref[eq:#1]{(\ref*{eq:#1})}}
\renewcommand{\sec}[1]{\hyperref[sec:#1]{Section~\ref*{sec:#1}}}
\newcommand{\app}[1]{\hyperref[app:#1]{Appendix~\ref*{app:#1}}}
\newcommand{\theo}[1]{\hyperref[thm:#1]{Theorem~\ref*{thm:#1}}}
\newcommand{\algo}[1]{\hyperref[alg:#1]{Algorithm~\ref*{alg:#1}}}
\newcommand{\lemm}[1]{\hyperref[lem:#1]{Lemma~\ref*{lem:#1}}}
\newcommand{\defin}[1]{\hyperref[defn:#1]{Definition~\ref*{defn:#1}}}
\newcommand{\corr}[1]{\hyperref[cor:#1]{Corollary~\ref*{cor:#1}}}
\newcommand{\fig}[1]{\hyperref[fig:#1]{Figure~\ref*{fig:#1}}}
\newcommand{\tab}[1]{\hyperref[tab:#1]{Table~\ref*{tab:#1}}}
\newcommand{\propos}[1]{\hyperref[prop:#1]{Proposition~\ref*{prop:#1}}}
\newcommand{\problem}[1]{\hyperref[prob:#1]{Problem~\ref*{prob:#1}}}
\newcommand{\rema}[1]{\hyperref[rem:#1]{Remark~\ref*{rem:#1}}}

\newcommand{\exm}[1]{\hyperref[exm:#1]{Example~\ref*{exm:#1}}}

\newtheorem{thm}{Theorem}[section]
\newtheorem{lem}[thm]{Lemma}
\newtheorem{prop}[thm]{Proposition}

\newtheorem{cor}[thm]{Corollary}
\newtheorem{prob}[thm]{Problem}

\newtheorem{dfn}[thm]{Definition}
\theoremstyle{remark}
\newtheorem{rem}[thm]{Remark}
\newtheorem{ex}[thm]{Example}

\newcommand{\parnum}{[\arabic{parcount}]}

\newcounter{parcount}

\newcommand{\CP}[1]{}
\newcommand{\VK}[1]{}
\newcommand{\change}[1]{{#1}}
\begin{document}

\title{Shorter quantum circuits via single-qubit gate approximation}
%
\author{Vadym Kliuchnikov}
\affiliation{Microsoft Quantum, Redmond, WA, US}
\affiliation{Microsoft Quantum, Toronto, ON, CA}
\author{Kristin Lauter}
\affiliation{Facebook AI Research, Seattle, WA, US}
\author{Romy Minko}
\affiliation{University of Oxford, Oxford, UK}
\affiliation{Heilbronn Institute for Mathematical Research, University of Bristol, Bristol, UK}
\thanks{This work was supported by the CDT in Cyber Security at the University of Oxford (EP/P00881X/1) and the Additional Funding Programme for Mathematical Sciences, delivered by EPSRC (EP/V521917/1) and the Heilbronn Institute for Mathematical Research.}
\author{Adam Paetznick}
\affiliation{Microsoft Quantum, Redmond, WA, US}
\author{Christophe Petit}
\affiliation{University of Birmingham, Birmingham, UK}
\affiliation{Universit\'{e} Libre de Bruxelles, Brussels, Belgium}

\maketitle

\begin{abstract}
  We give a novel procedure for approximating general single-qubit unitaries from a finite universal gate set by reducing the problem to a novel magnitude approximation problem, achieving an immediate improvement in sequence length by a factor of 7/9. Extending the work of \cite{Hastings2017,Campbell2017}, we show that taking probabilistic mixtures of channels to solve fallback~\cite{BRS2015} and magnitude approximation problems saves factor of two in approximation costs.
  In particular, over the Clifford+$\sqrt{\textnormal{T}}$ gate set we achieve an average non-Clifford gate count of $0.23\log_2(1/\varepsilon)+2.13$ and T-count $0.56\log_2(1/\varepsilon)+5.3$ with mixed fallback approximations
  for diamond norm accuracy $\varepsilon$.
      
  This paper provides a holistic overview of gate approximation, in addition to these new insights. We give an end-to-end procedure for gate approximation for general gate sets related to some quaternion algebras, providing pedagogical examples using common fault-tolerant gate sets (V, Clifford+T and Clifford+$\sqrt{\textnormal{T}}$). We also provide detailed numerical results for Clifford+T and Clifford+$\sqrt{\textnormal{T}}$
  gate sets. In an effort to keep the paper self-contained, we include an overview of the relevant algorithms for integer point enumeration and relative norm equation solving. We provide a number of further applications of the magnitude approximation problems, as well as improved algorithms for exact synthesis, in the Appendices. 
\end{abstract}

\newpage
\setcounter{tocdepth}{2}
\tableofcontents
\newpage


\section{Introduction}

In the quantum circuit model, quantum algorithms are expressed as sequences of unitary operations and measurements. Any $n$-qubit unitary can be implemented by a circuit of elementary gates, comprising controlled-NOT (CNOT) gates and single-qubit gates~\cite{Barenco1995}. Fault tolerant quantum computers require that the single-qubit gates belong to a finite set. Such a set is called universal if it generates a dense covering of $SU(2)$. That is, if any unitary $U\in SU(2)$ can be approximated to any accuracy by a finite sequence of gates from the set. Of particular interest is the subject of approximating single-qubit diagonal unitaries, as a Euler decomposition guarantees that any single-qubit unitary can be decomposed into diagonal $R_z$- and $R_x$-rotations. In addition, diagonal $R_z$ rotations are directly used in many quantum algorithm.

The cost of an approximation is quantified by the gate complexity, or gate cost. Associating each gate $g_i$ in a sequence with a weight $w_i$, the gate cost of that sequence is $\sum_iw_i$. The gate cost of approximating $U$ to within $\varepsilon$ is then taken as the minimum gate cost of all possible approximating sequences. Note that select gates, such as the Pauli or Clifford gates, are considered cheap to implement and so take zero weight. Typically, expensive gates will be given a weight of 1, so that the gate cost of an approximation corresponds to the number of expensive gates in the sequence. Consequently, minimizing the length of an approximating sequence is a problem integral to the subject of gate synthesis. A fundamental and general result is the Solovay-Kitaev theorem, which states that a universal gate-set $G$ can approximate any unitary $U\in SU(2)$ to accuracy $\varepsilon$ by a finite sequence of gates from $G$ of length $O(\log^c(1/\varepsilon))$, where $c$ is a constant. Significant progress has been made since Solovay-Kitaev for specific gate-sets associated with fault-tolerant quantum computers. Bourgain and Gamburd~\cite{bourgain2011spectral} showed that universal gate-sets of unitaries with algebraic entries give approximating sequences with lengths $O(\log(1/\varepsilon))$.
\change{Such gate-sets are called efficiently universal.}
This result was quickly applied to find efficient constructive algorithms for the Clifford$+$T gate set~\cite{A3,Selinger} and, later, the V basis~\cite{BGS}. Constructive algorithms for optimal diagonal approximations for both gate sets followed soon thereafter~\cite{RossSelinger2014,Ross2015,BlassEtAl2015}. Many of these algorithms adopted a common framework of integer point enumeration followed by a solving a norm equation. 

\begin{figure}[h!]
    \centering
    \includegraphics[scale=1.2]{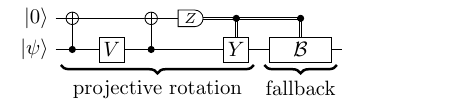}
    \caption{
        \label{fig:fallback-circuit}
        A fallback protocol circuit \cite{BRS2015}. 
        The circuit is composed of two steps, a projective rotation step and a fallback step. 
        The projective rotation step effects one of two diagonal rotations on the input $\ket{\psi}$ depending on the $Z$-basis measurement outcome of the top ancilla qubit.  
        The two rotation angles are determined by the matrix entries of the unitary $V$.
        If the measurement outcome is one, then a Pauli $Y$ is applied followed by the fallback step $\mathcal{B}$. 
    }
    %
\end{figure}

Measurements were also introduced to aid gate synthesis~\cite{paetznick2013repeat}, culminating in the fall-back circuit~\change{(\cref{fig:fallback-circuit})}, although this work was a departure from the common framework. Sarnak~\cite{sarnak2637letter} noted a connection between gate synthesis and quaternion algebras in his letter to Aaronson and Pollington, which has been used to build frameworks for both exact~\cite{Kliuchnikov2015a}\footnote{This work has been developed independently of \cite{sarnak2637letter}.} and approximate synthesis~\cite{Kliuchnikov2015b}. The letter also characterizes `golden gate sets', of which the Clifford$+$T and V gates are examples, that achieve optimal sequence lengths. For approximations of diagonal unitaries, this is shown to be $3\log_\ell(1/\varepsilon)$, \change{where $\ell$ is a parameter determined by the number theoretic structure of each gate-set.}
These results are further expanded and generalizations of two-step approach to diagonal approximations from \cite{RossSelinger2014} to other gate gets also discussed in~\cite{Parzanchevski}.
Recent research~\cite{Campbell2017,Hastings2017} shows that approximation with quantum channels, rather than unitaries, achieves quadratic improvement in $\varepsilon$ and reduces the length by factor of two.

The quality of an approximation $V$ for some desired unitary $U$ is captured by the accuracy parameter $\varepsilon$. 
\change{When we want to measure the distance between two unitary operators, we usually use the operator norm, which is the maximum singular value of the difference between the two operators. This allows us to bound how accurate a quantum algorithm is, based on how well we can approximate the unitary operations that make up the algorithm with other unitary operations~\cite{nielsen00}. However, in our work, we consider methods for approximating unitary operators using measurements, conditional gates, and probabilistic mixtures. This results in quantum channels, which are completely-positive trace-preserving linear maps on density matrices.
To measure the distance between two $n$-qubit quantum channels, we use the diamond norm, which is defined as
\begin{equation}
    \nrm{\mathcal{U}-\mathcal{V}}_\diamond=\max\limits_\rho\set{ \nrm{((\mathcal{U}-\mathcal{V})\otimes \mathcal{I}_{2^n})\at{\rho}}_1}, \text{ where } \nrm{A}_1 = \mathrm{Tr}\sqrt{A^\dagger A}
\end{equation}
where $\mathcal{I}_{2^n}$ is the identity map from $\mathbb{C}^{2^n\times 2^n}$ into $\mathbb{C}^{2^n\times 2^n}$ and the maximum is taken over all probability density matrices $\rho$. The diamond norm allows us to estimate the accuracy of a quantum algorithm based on the inaccuracies of the channels that make up the algorithm~(\app{diamond-norm-properties}). We also recall in~\app{diamond-norm-properties} that the diamond norm is bounded above by twice the spectral norm of the difference between the unitary operators that correspond to the channels. This makes it easy to compare our results, which are stated in terms of diamond norm accuracy, with previous results that are stated in terms of spectral norm accuracy.
} 

We consider three universal gate sets associated with fault-tolerant quantum computation: the V basis, the Clifford+T basis and the Clifford+$\sqrt{\text{T}}$ basis. The Clifford+T gate set is commonly used for fault tolerant quantum computation, and is known to be efficiently universal \cite{Selinger}, with gate cost depending solely on the number of $T$ gates used. 
The V basis was shown to be efficiently universal in~\cite{harrow2002efficient}, and provides a simple pedagogical example of approximation. The  Clifford$+\sqrt{\text{T}}$ gate set is an alternative to the Clifford$+$T gate set for fault tolerant computing. The merits of these gate sets with regard to fault tolerant computing are discussed in greater detail in \sec{fault-tolerant-gate-sets}.

The rest of this paper is organized as follows. \sec{main-results} summarizes our main results. In \sec{crypto-connection}, we briefly discuss connections between gate synthesis and cryptography. \sec{approximation-problems} defines the five approximation problems that are the focus of this paper. We detail a complete method for solving these problems in \sec{approx-solutions}, with examples for the V, Clifford$+$T and Clifford $+\sqrt{\text{T}}$ gate sets, in addition to a general solution. 
We provide extensive numerical results for our method in \sec{numerical-results} for Clifford$+$T and Clifford$+\sqrt{\text{T}}$. \sec{point-enum} and \sec{normeqsolv} recall algorithms for solving two problems that arise in our solution method: integer point enumeration and norm equation solving. 

\subsection{Fault tolerant gate sets} \label{sec:fault-tolerant-gate-sets}
Unitary synthesis translates the description of a quantum algorithm into a sequence of operations gates") that can be implemented on the target quantum computer.  The set of operations permitted by a particular quantum computing platform are limited by physical constraints and may not match the operations prescribed in the algorithm.  Moreover, even if the operations offered by the quantum computer match the operations in the algorithm, the accuracy to which the quantum computer can perform each operation is likely to be limited.  Loosely speaking, existing quantum computing systems offer single-qubit unitary operations with accuracy up to 
$10^{-4}$~(see~\href{https://static-content.springer.com/esm/art\%3A10.1038\%2Fs41586-019-1666-5/MediaObjects/41586_2019_1666_MOESM1_ESM.pdf\#page=19}{Fig. S.17} in \cite{Supremacy2019} for the accuracy of one and two qubit gates)
whereas useful quantum algorithms require accuracy of $10^{-10}$ or better~(see \href{https://journals.aps.org/prresearch/pdf/10.1103/PhysRevResearch.3.033055#t1}{Table I} in~\cite{Catalysis2021} for 
the typical number of gates in useful quantum algorithms).

Fault-tolerant quantum computation bridges the accuracy gap by encoding many physical qubits into a smaller number of logical qubits.  To guarantee accuracy, logical qubits must be encoded at all times. Operations of the quantum algorithm must be performed on the logical qubits \emph{while} they are encoded. Each logical operation must both preserve the code structure and carefully limit the spread of errors. Those requirements restrict the available set of logical quantum operations.

The cheapest form of logical operations involves executing the same physical operation to each of the physical qubits in the code.  For example, some codes admit the logical Hadamard operation by executing the physical Hadamard operation to each physical qubit. Unfortunately, these so-called "transversal" gates can yield only a sparse discrete set within a single quantum code 
\cite{EastinKnill,BravyiKonig,WebsterBartlett}.

Stabilizer codes, the most widely studied class of quantum error correcting codes, typically admit transversal implementation of some or all of the Clifford group---the group generated by \{$H$, $S$, CNOT\}. A broad and widely studied subset of stabilizer codes called CSS support transversal implementation of the CNOT operation, for example.\footnote{CSS is an initialism that comes from the three authors that first defined the codes: Calderbank, Shor and Steane.}  Circuits composed of Clifford group operations have the added benefit of being efficiently simulable, an essential feature for the study of fault tolerant quantum computing schemes.

At least one operation outside of the Clifford group is required for universal quantum computation (see Theorem 6.7.3 in \cite{Nebe2006}). In most fault-tolerant quantum computing proposals, that non-Clifford operation is the $T$ gate, \change{$T = \left(\begin{smallmatrix}e^{-i\pi/8}&0\\0&e^{i\pi/8}\end{smallmatrix}\right)$.}  The logical $T$ operation is typically implemented by "distillation" which combines many (noisy) physical $T$ gates with (less noisy) logical Clifford operations \cite{BravyiKitaevDistillation}.  Distillation is regarded as the most efficient known technique for non-Clifford gates, but remains roughly an order of magnitude more expensive than transversal operations despite much study \cite{BeverlandDistillation}.  Distillation can be used to construct other operations, but known distillation protocols are limited to operations that belong to the so-called "Clifford hierarchy" introduced in~\cite{Gottesman1999}.  Of the known distillation techniques the most cost competitive operations are $T$ and the three qubit double-controlled-Z (CCZ) \cite{Gidney2019}. \change{Controlled on the first two qubits, the CCZ gate sends $\ket{111}\rightarrow -\ket{111}$ and $\ket{x}\rightarrow\ket{x}$ for all other $x\in\{0,1\}^3$.}

To the best of our knowledge, there are no operations outside of the Clifford hierarchy that admit implementation through distillation of the corresponding resource states.
However, by including measurement it is possible to construct other kinds of circuits out of fault tolerant Clifford+$T$ gates.  For example, $V_Z = (I + 2iZ)/\sqrt{5}$ can be implemented with the circuit shown in Figure \ref{fig:vz-rus-circuit}. The idea then is to approximate a unitary $U$ with a sequence of Clifford+$V_Z$ gates then substitute Figure \ref{fig:vz-rus-circuit} for each $V_Z$ in the sequence. Unfortunately, that strategy yields higher number of $T$ gates than synthesis with Clifford+$T$ alone. The strategy could be redeemed with cheaper Clifford+$T$ implementations of $V_Z$ or similar non-Clifford gates.
But finding such circuits is difficult,and the best known circuits do not produce better resource requirements overall. 
Though we consider the $V$-basis gate set in this paper, it is largely for instructional purposes.

\CP{Has this search been systematized somehow? sounds like an interesting problem. Also, with a better clue on the respective cost of constructing various gate sets this way, we can perhaps combine these various gate sets and improve overall $T$-gate cost.}
\VK{ An attempt was made in \cite{Paetznick2013b}, I have also wrote some code to search for them in past (un-published). First, I would check how far are we from the lower-bounds described in 
Section 5 in \href{https://arxiv.org/pdf/1904.01124.pdf}{arXiv:1904.01124}. In particular, one should check the lower-bounds in post-selected measurements model that would be stronger than \href{https://arxiv.org/pdf/1904.01124.pdf\#thm.5.2}{Theorem 5.2 in arXiv:1904.01124}.
In \href{https://arxiv.org/pdf/1904.01124.pdf}{arXiv:1904.01124} things get looser when we go from post-selected model to a standard quantum computation model. It is plausible that we are quite close to those lower-bounds with fall-back strategy.
One can use  \href{https://arxiv.org/pdf/1904.01124.pdf}{arXiv:1904.01124} to show lower-bounds on the cost of particular gate-sets in terms of Clifford+T.
This will let one to rule-out a lot of potential candidates.}

Incorporation of measurements into Clifford+$T$ circuits can be used to emulate other gates \emph{within} the Clifford hierarchy, as well.
For example, the gate \change{$\sqrt{T}=\left(\begin{smallmatrix}e^{-i\pi/16}&0\\0&e^{i\pi/16}\end{smallmatrix}\right)$} can be implemented with Clifford+$T$ and measurements~\cite{LowerBounds2020}.
Approximations with the set Clifford+$\sqrt{T}$, through emulation, use the the same number of $T$ gates as direct Clifford+$T$ sequences.
Using Clifford+$\sqrt{T}$ gates, however, yields approximating sequences that are half the length of corresponding Clifford+$T$ sequences, offering an advantage when rotations must be executed as fast as possible.

\begin{figure}
    \centering
    \includegraphics[scale=1]{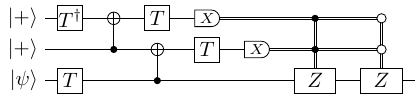}
    \caption{
        \label{fig:vz-rus-circuit}
        A Clifford+$T$ circuit implementation of $V_Z = (I+2iZ)/\sqrt{5}$ proposed in \cite{paetznick2013repeat}.  Conditioned on $X$-basis measurement outcomes of zero on the top two qubits, the circuit outputs $V_Z|\psi\rangle$.  For any other measurement outcome the circuit outputs $|\psi\rangle$.  Repeating the circuit until obtaining the 00 measurement outcome yields $V_Z|\psi\rangle$ using approximately $5.26$ $T$ gates, in expectation.
        \change{ The first conditional Pauli $Z$ gate is applied when measurement both outcomes are one, and the second conditional Pauli $Z$ gate is applied when both measurement outcomes are zero.}
    }
\end{figure}

\subsection{A roadmap to approximate quantum synthesis}

\change{Here we describe the standard approach to approximate unitary synthesis for a diagonal single-qubit unitary using the $V$ basis, then motivate generalisations to other gate sets and approximation techniques.}

\change{Suppose the approximation target is a single-qubit unitary $U\in SU(2)$, and target accuracy is some $\varepsilon>0$. The standard approach is to take the Euler decomposition of $U$, that is $U = e^{i\phi_1Z}e^{i\theta X}e^{i\phi_2 Z}$ for real $\phi_1, \theta, \phi_2$, and approximating each term independently. The matrices $e^{i\phi_j Z}$ are diagonal and we can diagonalize the central term by conjugation with Hadamard matrices. Now, for a diagonal target unitary, an approximation, $V=\left(\begin{smallmatrix} u & -v^\ast \\ v & u^\ast \end{smallmatrix}\right)$, is within $\varepsilon$ of the target if the top-left element resides within a pre-determined two-dimensional region. This region depends on the accuracy $\varepsilon$ and the rotation angle of the target unitary.}

\change{For approximation with the $V$ basis, unitaries that admit exact synthesis in $N$ gates are of the form $\frac{1}{\sqrt{5^N}}\left(\begin{smallmatrix} m_1 & -m_2^\ast \\ m_2 & m_1^\ast \end{smallmatrix}\right),$ where $m_1$ and $m_2$ are Gaussian integers, so in particular we have $m_1m_1^*+m_2m_2^*=5^N$. These matrices can additionally be efficiently factored into a sequence of basis elements. Moreover, given $N$ and $m_1$ there exists an efficient algorithm to determine whether an acceptable $m_2$ exists, and compute it if so. This establishes a general algorithm for approximate diagonal synthesis: repeatedly sample candidates for $m_1$ from a pre-determined approximation region and try to find a corresponding $m_2$, and when found apply the exact synthesis algorithm. Finding $m_1$ is a special case of integer point enumeration problem and finding $m_2$ is a special case of a relative norm equation problem, for both of which algorithms exist, discussed in \sec{point-enum} and \sec{normeqsolv}, respectively. A worked example of diagonal approximation with the $V$ basis is given in \sec{V-basis-example}.}

\change{This three step process generalises quite well to other gate sets. A `good' gate set generates unitaries of the form  $\frac{1}{\sqrt{\ell^N}}\left(\begin{smallmatrix} m_1 & -m_2^\ast \\ m_2 & m_1^\ast \end{smallmatrix}\right)$, where $\ell$ is a parameter specific to the gate set and $N$ corresponds to the minimum number of basis gates required to synthesise that unitary. It is known \cite{sarnak2637letter, KliuchnikovMaslovMosca2013, Kliuchnikov2015b} that the family of gate sets known as quaternion gate sets satisfy these conditions, and we rely on the framework of \cite{KliuchnikovMaslovMosca2013, Kliuchnikov2015b} in particular to solve the relative norm equation to compute $m_2$ and for exact synthesis.}

\change{Further generalisations also exist. We find that established techniques of fallback approximation and probabilistic mixing align with the outlined approach, albeit leading to various other approximation regions. In \sec{approximation-problems} of this paper, we will additionally show that an alternative to the Euler decomposition method also follows the three steps, again with a modified approximation region, this time in one-dimension. In fact, modifying the approximation region is the only affect of these generalisations; the rest of the algorithm proceeds as before.}

\subsection{Paper outline}
We summarize our main results in \sec{main-results}.  \sec{approximation-problems} defines the five approximation problems that are the focus of this paper. We detail a complete method for solving these problems in \sec{approx-solutions}, with examples for the V, Clifford$+$T and Clifford $+\sqrt{\text{T}}$ gate sets, in addition to a general solution. 
 \sec{point-enum} and \sec{normeqsolv} recall algorithms for solving two problems that arise in our solution method: integer point enumeration and norm equation solving. \sec{exact-synthesis} contains an algorithm for exact synthesis, which completes the gate approximation procedure. We provide extensive numerical results for the application of our method to single qubit gate approximation in \sec{numerical-results} for Clifford$+$T and Clifford$+\sqrt{\text{T}}$, and outline further applications in \sec{magnitude-approx-applications}. Finally, in \sec{crypto-connection}, we discuss a number of related problems from other fields. 

\section{Summary of main results}
\label{sec:main-results}

\begin{table}[h]
    \caption[Approximation cost scaling]{
    \label{tab:approximation-cost-scaling}
    Scaling of the approximation cost for random angles. 
    Approximation accuracy $\varepsilon$ is measured using diamond distance.
    Linear fit of the cost is based on the numerical results reported in~\fig{clifford-t-random} and~\fig{clifford-root-t-random}. 
    Mixed diagonal rows contain results for our new and improved version of a mixed diagonal approximation protocol first introduced in~\cite{Campbell2017,Hastings2017}.
    Results from \cite{Campbell2017,Hastings2017} apply to any gate set for which diagonal approximation is available, \change{this is why we mark mixed diagonal results with New$^\ast$}.
    Fallback rows correspond to our improved and generalized fallback synthesis method~(\cref{fig:fallback-circuit}) first introduced in~\cite{BRS2015}.
    For power cost, the cost of $T$ is two and the cost of $T^{1/2}, T^{3/2}$ is three. Our algorithms are optimal with respect to this cost. 
    For gate count cost, the cost of $T$, $T^{1/2}, T^{3/2}$ is one. 
    For $T$-count cost, the cost of $T$ is one and the cost of $T^{1/2}, T^{3/2}$ is four. 
    Clifford gate costs are always zero. 
    Different costs are discussed in \cref{sec:numerical-results}.
    Heuristic cost estimates are discussed in \cref{sec:cost-scaling}.
    Applications of magnitude approximation are illustrated in \cref{tab:magnitude-approx-applications}.
    }
\begin{center}
{
\scriptsize
\setlength{\tabcolsep}{0.5em}
\begin{tabular}{|c|c|l|l|c|c|}
\hline 
Gate set & Approximation & \multicolumn{2}{c|}{Linear fit of the cost ($\varepsilon < 10^{-4}$)} & Heuristic & Novelty \tabularnewline
\cline{3-4} \cline{4-4} 
(cost) & protocol & Mean & Max & cost estimate & \tabularnewline
\hline 
\hline 
                  & Diagonal \cite{RossSelinger2014} & $3.02\log_2(1/\varepsilon)+1.77 $ & $3.02\log_2(1/\varepsilon)+9.19 $ & $3.0\log_2(1/\varepsilon)+O(1)$ & Known \\
\cline{2-6}
                  & Fallback \cite{BRS2015} & $1.03\log_2(1/\varepsilon)+5.75 $ & $1.05\log_2(1/\varepsilon)+11.83$ & $1.0\log_2(1/\varepsilon)+O(1)$ & Improved \\
\cline{2-6}
Clifford+$T$      & Magnitude & -- & -- & $1.0\log_2(1/\varepsilon)+O(1)$ & New \\
\cline{2-6}
($T$-count)       & Mixed diagonal  & $1.52\log_2(1/\varepsilon)-0.01 $ & $1.54\log_2(1/\varepsilon)+6.85 $ & $1.5\log_2(1/\varepsilon)+O(1)$ & Improved \\
\cline{2-6}
\cref{fig:clifford-t-random}& Mixed fallback  & $0.53\log_2(1/\varepsilon)+4.86 $ & $0.57\log_2(1/\varepsilon)+8.83 $ & $0.5\log_2(1/\varepsilon)+O(1)$ & New \\
\cline{2-6}
                  & Mixed magnitude & -- & -- & $0.5\log_2(1/\varepsilon)+O(1)$ & New \\
\hline 
\hline 
                   & Diagonal        & $3.02\log_2(1/\varepsilon)+2.80 $ & $3.01\log_2(1/\varepsilon)+8.53 $ & $3.0\log_2(1/\varepsilon)+O(1)$ & New \\
\cline{2-6}
                   & Fallback        & $1.04\log_2(1/\varepsilon)+6.61 $ & $1.02\log_2(1/\varepsilon)+11.83$ & $1.0\log_2(1/\varepsilon)+O(1)$ & New \\
\cline{2-6}
Clifford+$\sqrt{T}$& Magnitude & -- & -- & $1.0\log_2(1/\varepsilon)+O(1)$ & New \\
\cline{2-6}
(power)            & Mixed diagonal  & $1.53\log_2(1/\varepsilon)+1.06 $ & $1.58\log_2(1/\varepsilon)+4.98 $ & $1.5\log_2(1/\varepsilon)+O(1)$ & New* \\
\cline{2-6}
\cref{fig:clifford-root-t-random-power}& Mixed fallback  & $0.56\log_2(1/\varepsilon)+5.32 $ & $0.62\log_2(1/\varepsilon)+7.66 $ & $0.5\log_2(1/\varepsilon)+O(1)$ & New \\
\cline{2-6}
                   & Mixed magnitude & -- & -- & $0.5\log_2(1/\varepsilon)+O(1)$ & New \\
\hline 
\hline 
                   & Diagonal        & $1.21\log_2(1/\varepsilon)+1.18 $ & $1.26\log_2(1/\varepsilon)+3.86 $ & $1.2\log_2(1/\varepsilon)+O(1)$ & New \\
\cline{2-6}
                   & Fallback        & $0.42\log_2(1/\varepsilon)+2.68 $ & $0.44\log_2(1/\varepsilon)+5.13 $ & $0.4\log_2(1/\varepsilon)+O(1)$ & New \\
\cline{2-6}
Clifford+$\sqrt{T}$& Magnitude & -- & -- & $0.4\log_2(1/\varepsilon)+O(1)$ & New \\
\cline{2-6}
(gate count)       & Mixed diagonal  & $0.61\log_2(1/\varepsilon)+0.43 $ & $0.64\log_2(1/\varepsilon)+2.52 $ & $0.6\log_2(1/\varepsilon)+O(1)$ & New* \\
\cline{2-6}
\cref{fig:clifford-root-t-random-rcount}& Mixed fallback  & $0.23\log_2(1/\varepsilon)+2.13 $ & $0.25\log_2(1/\varepsilon)+3.85 $ & $0.2\log_2(1/\varepsilon)+O(1)$ & New \\
\cline{2-6}
                   & Mixed magnitude & -- & -- & $0.2\log_2(1/\varepsilon)+O(1)$ & New \\
\hline 
\hline 
                   & Diagonal        & $3.03\log_2(1/\varepsilon)+2.48 $ & $3.25\log_2(1/\varepsilon)+14.40$ & $3.0\log_2(1/\varepsilon)+O(1)$ & New \\
\cline{2-6}
                   &Fallback        & $1.04\log_2(1/\varepsilon)+6.43 $ & $1.18\log_2(1/\varepsilon)+14.01$ & $1.0\log_2(1/\varepsilon)+O(1)$ & New \\
\cline{2-6}
Clifford+$\sqrt{T}$& Magnitude & -- & -- & $1.0\log_2(1/\varepsilon)+O(1)$ & New \\
\cline{2-6}
($T$ count)        & Mixed diagonal  & $1.53\log_2(1/\varepsilon)+1.02 $ & $1.68\log_2(1/\varepsilon)+7.30 $ & $1.5\log_2(1/\varepsilon)+O(1)$ &  New* \\
\cline{2-6}
\cref{fig:clifford-root-t-random-tcount}& Mixed fallback  & $0.56\log_2(1/\varepsilon)+5.30 $ & $0.67\log_2(1/\varepsilon)+9.85 $ & $0.5\log_2(1/\varepsilon)+O(1)$ & New \\
\cline{2-6}
                   & Mixed magnitude & -- & -- & $0.5\log_2(1/\varepsilon)+O(1)$ & New \\
\hline 
\end{tabular}
}
\end{center}
\end{table}

\begin{figure}
\caption[Cost of approximating random diagonal gates with Clifford+$T$ gates]{
\label{fig:clifford-t-random}
Cost of approximating a set of random diagonal rotation gates with Clifford+$T$ gates using four approximation protocols.
Diagonal rotation angles are random angles drawn from the uniform distribution on $[0,2\pi]$.
We fix a set of approximation accuracy values. For each value in the set we compute mean cost over 
all target angles. Vertical bars show the cost standard deviation for the given accuracy value. 
Shaded regions indicate range of costs from min to max over all angles for the given accuracy value. 
For the diagonal approximation protocol the cost for a given angle and given accuracy is equal to number of $T$ gates in the approximating sequence. 
For the other protocols the cost is the expectation of $T$-count. For example, if the first step of fall-back protocol requires 10 $T$ gates and 
fails with probability 0.01 and the step to correct failure requires 30 $T$ gates, the expected cost is 10.3. 
In all reported fallback protocols the probability of fallback is at most $0.01$.
}
\includegraphics{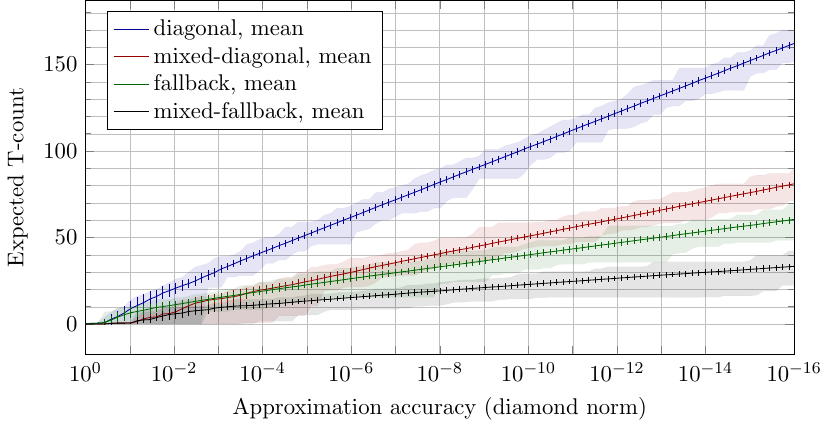}
\end{figure}

\begin{table}
\caption[Some problems that benefit from magnitude approximation]{
\label{tab:magnitude-approx-applications}
Examples of problems that benefit from magnitude approximation. 
The case of a general qubit unitary approximation is described in \cref{sec:magnitude-approximation} and \cref{sec:mixed-magnitude-approximation}.
Solutions to qubit approximation and general $\mathrm{SU}(4)$ approximation are outlined in \cref{sec:magnitude-approx-applications}.
More generally, the magnitude approximation problem can be used when compiling $\mathrm{CNOT}$ and rotation circuits for isometries produced by UniversalQCompiler~\cite{Iten2016,Iten2021,Malvetti2021} for 
fault-tolerant quantum computers.
}
\begin{center}
{
\scriptsize
\setlength{\tabcolsep}{0.5em}

\begin{tabular}{|c|c|c|c|c|}
\hline 
\multicolumn{2}{|c|}{Problems that benefit from} & Qubit state & General $\mathrm{SU}(2)$ & General $\mathrm{SU}(4)$\tabularnewline
\multicolumn{2}{|c|}{magnitude approximation} & preparation & approximation & approximation\tabularnewline
\hline 
\hline 
 & Number of real parameters & 2 & 3 & 15\tabularnewline
\cline{2-5} \cline{3-5} \cline{4-5} \cline{5-5} 
Problem & Number of magnitude  & 1 & 1 & 6\tabularnewline
properties & approximation instances &  &  & \tabularnewline
\cline{2-5} \cline{3-5} \cline{4-5} \cline{5-5} 
 & Number of diagonal & 1 & 2 & 9\tabularnewline
 & approximation instances &  &  & \tabularnewline
\hline 
\hline 
Gate set & Approximation method & \multicolumn{3}{c|}{Heuristic T-count scaling with diamond norm accuracy $\varepsilon$}\tabularnewline
\hline 
\hline 
 & Diagonal (known) & $6\log_{2}(1/\varepsilon)+O(1)$ & $9\log_{2}(1/\varepsilon)+O(1)$ & $45\log_{2}(1/\varepsilon)+O(1)$\tabularnewline
\cline{2-5} \cline{3-5} \cline{4-5} \cline{5-5} 
Clifford+T & Diagonal + Magnitude (new)& $4\log_{2}(1/\varepsilon)+O(1)$ & $7\log_{2}(1/\varepsilon)+O(1)$ & $33\log_{2}(1/\varepsilon)+O(1)$\tabularnewline
\cline{2-5} \cline{3-5} \cline{4-5} \cline{5-5} 
 & Mixed Diagonal (improved) & $3\log_{2}(1/\varepsilon)+O(1)$ & $4.5\log_{2}(1/\varepsilon)+O(1)$ & $22.5\log_{2}(1/\varepsilon)+O(1)$\tabularnewline
\cline{2-5} \cline{3-5} \cline{4-5} \cline{5-5} 
 & Mixed Diag. + Mag. (new)& $2\log_{2}(1/\varepsilon)+O(1)$ & $3.5\log_{2}(1/\varepsilon)+O(1)$ & $11.5\log_{2}(1/\varepsilon)+O(1)$\tabularnewline
\hline 
\end{tabular}

}
\end{center}
\end{table}


In \sec{approximation-problems} we define six approximation problems: diagonal unitary approximation, fallback approximation, magnitude approximation and their versions with probabilistic mixing. The first two of these problems have been the subject of research for some time, with many results pertaining to specific gate sets \cite{RossSelinger2014, KliuchnikovMaslovMosca2013, BGS, BRS2015,Kliuchnikov2015a,Kliuchnikov2015b}. The magnitude approximation problem is new \change{and appears in our new approach to approximating an arbitrary unitary, $U$. When we independently approximate the three elements of the Euler angle decomposition of $U$, approximating the central term $e^{i\theta X}$ results in a unitary $V$ with entries of similar magnitude to the entries of $U$. Hence, the magnitude approximation problem.} One of its applications is solving the general unitary approximation problem, which exploits the connection between unitary approximation and LPS graphs \change{as we explain in \sec{crypto-connection}}. Explicitly, we adapt the path-finding algorithm of Carvalho Pinto and Petit \cite{pinto2018better} to the quantum setting, requiring only two diagonal approximations and one more efficient magnitude approximation. The sequence costs obtained using our method improve on the standard Euler decomposition, which requires three diagonal approximations, by roughly one-third. \change{More precisely, we find that the costs from magnitude approximation are one-third those of diagonal approximations. This means our approach to the general unitary approximation problem, which requires two diagonal approximations and one magnitude approximation, has a cost that is $7/9$ that of the standard Euler decomposition method.}

 Stier \cite{stier} has concurrently and independently produced a similar result. We discuss additional applications of the magnitude approximation problem in \cref{sec:magnitude-approx-applications}.

The latter three problems are defined by applying the concept of channel mixing to diagonal, fallback and magnitude approximation, expanding on the ideas of \cite{Campbell2017,Hastings2017}.
Channel mixing employs a probabilistic combination of sequences of unitaries to approximate the target.
The key idea is to use probabilistic combination of under-rotated and over-rotated approximations of a given target.
We combine channel mixing with fallback and magnitude approximation to achieve a roughly two-fold improvement in cost compared to non-mixed problem variants.

To account for the fact that we are approximating with channels rather than unitaries, we use the diamond norm to measure the accuracy of approximation. 
We introduce the use of diagonal Clifford twirling, which ensures that the difference between the ideal and approximating channel is a Pauli channel. 
Because of this, we obtain a closed-form expression for the diamond distance which improves on analysis in \cite{Campbell2017,Hastings2017} and results in lower approximation costs.
The method in \cite{Campbell2017} uses diagonal approximation algorithms as black-boxes and finds under-rotations and over-rotations by modifying target rotation angle.
In contrast, we modify synthesis algorithms to directly find under and over-rotated approximations and further reduce approximations costs.

We provide a uniform approach to the six approximation problems and various gate sets.
For each problem, we show that a constraint on the diamond norm can be reduced to a constraint on a single complex number.
In contrast to \cite{BRS2015}, our approach to fallback approximation ensures desired success probability and approximation accuracy.
The set of feasible solutions is represented geometrically as a region in $\r^2$, illustrated in \sec{approximation-problems}.
\fig{diagonal-approximation-region-areas} compares the areas of the regions associated to the approximation problems with respect to varying approximation accuracy $\varepsilon$ and success probability $q$. 
The scaling of the region area with epsilon determines the scaling of the approximation sequence cost with accuracy~$\varepsilon$, 
as illustrated in \cref{tab:approximation-cost-scaling}. 
To establish dependence between region area scaling and cost scaling we use several heuristic assumptions as discussed in \cref{sec:cost-scaling}.
We also provide experimental justification of relation between the cost and region scaling. 

The results of our numerical experiments for Clifford+$T$  and Clifford+$\sqrt{T}$ gate sets are summarized in \cref{tab:approximation-cost-scaling}.
More detailed numerical results for Clifford+$T$ are provided in \cref{fig:clifford-t-random},
in particular they show that linear fits for cost scaling with accuracy are well justified.
Even more detailed results for Clifford+$T$ and for Clifford+$\sqrt{T}$ are in \sec{numerical-results}. We show numerical results for approximating uniformly random diagonal rotations and angles and rotations by Fourier angles $\pi/2^k$.
The study of uniformly random diagonal rotations is motivated by the use of diagonal rotations in quantum algorithm for
chemistry, material science applications~\cite{TheoryOfTrotterError}
and the rotations used in Quantum Signal Processing~\cite{QuantumSignalProcessing}; rotations by Fourier angles appear in the Quantum Fourier Transform~\cite{Nam2020} and 
preparation of Phase Gradient states~\cite{Gidney2018}. 

\change{To provide practical context for our results, let us consider their impact on two practically interesting quantum dynamics and quantum chemistry problems, which are discussed in detail in Appendix F of~\cite{beverland2022assessing}. The quantum dynamics problem instance involves $30100$ diagonal rotation gates and requires an accuracy of $\varepsilon_1 = 1.1\cdot 10^{-8}$ per rotation. In this problem, rotations are the only non-Clifford gates used. The quantum chemistry instance involves $2.06 \cdot 10^{8}$ diagonal rotations and requires an accuracy of $\varepsilon_2 =1.6 \cdot 10^{-11}$ per rotation. 
In addition to the diagonal rotation gates, this instance uses $5.4 \cdot 10^{11}$ $T$ gates. 
To implement diagonal rotations with an accuracy of $\varepsilon_1$ using the mixed fallback protocol, we need an average of $20$ $T$ states, while for an accuracy of $\varepsilon_2$, we need an average of $26$ T states. When using Clifford+$\sqrt{T}$, the number of $T$, $\sqrt{T}$, and $\sqrt{T}^3$ gates in the approximating sequence is $8$~(on average) for $\varepsilon_1$ and $10$~(on average) for $\varepsilon_2$.
}

We find that Clifford+$\sqrt{T}$ is a promising gate-set for approximation 
when executing rotations as fast impossible. In this case the execution speed is limited by the gate count, in particular when 
executing rotations using a circuit from \href{https://arxiv.org/pdf/1808.02892.pdf#figure.33}{Figure~33} in~\cite{GameOfSurfaceCodes}.
We achieve average gate-count scaling $0.23\log_2(1/\varepsilon)+2.13 $ when using mixed fallback protocol with Clifford+$\sqrt{T}$
and T-count similar to Clifford+$T$ approximations. We assume that each $\sqrt{T}$ gate requires four $T$ gates, which is justified in~\cref{sec:cost-scaling}.
These are the first numerical studies of approximation cost scaling for Clifford+$\sqrt{T}$ gates that include additive constants, which are practically important because the $\log_2(1/\varepsilon)$ prefactor is small. 
These are also the first numerical results for mixed fallback and mixed diagonal approximations for Clifford+${T}$.
We also provide more detail on approximate synthesis for general gate sets than the high-level approach outlined in~\cite{Parzanchevski}.



In \sec{approx-solutions} we describe a complete method for solving the six approximation problems, restricting the scope to considering gate sets that can be represented by quaternion algebras. The general solution method is described in \sec{approximation-problems:general}, and includes a process for constructing quaternion gate sets, as defined in \cite{Kliuchnikov2015b}. To summarize, a gate set is defined by a complex field $L$, its maximal totally real subfield $K$ and a fixed set of elements in $K$. A solution to an approximation problem involves finding a matrix $M=\left(\begin{smallmatrix}
m_1&-m_2^\ast\\m_2&m_1^\ast
\end{smallmatrix}\right)$ with entries in the integer ring of $L$. Our approach to finding $M$ can be summarized in two steps: point enumeration in a region defined by the approximation problem to find $m_1$, followed by solving a relative norm equation to recover $m_2$. To guide the reader, we work through three pedagogical examples: the V basis (\sec{approximation-problems:Vbasis}), the Clifford$+$T basis (\sec{approximation-solutions-clifford-t}), and the Clifford$+\sqrt{\textnormal{T}}$ basis (\sec{approximation-solutions-clifford-root-t}). 


\section{Approximation problems}\label{sec:approximation-problems}
In this section we introduce six problems that address the approximation of qubit
unitaries. Recall that in this paper we follow a two-step approach 
to solving approximation problems. First, in this section, we relate each problem to one-dimensional or two-dimensional 
regions, which define `good' approximations. Second, in \cref{sec:approx-solutions}, we find sequences of gates $g_1,\ldots,g_n$ over gate-set $G$ such that 
for a unitary computed by the sequence
$g_1\ldots g_n = \at{\begin{smallmatrix} u & -v^\ast \\ v & u^\ast \end{smallmatrix}}$
the top-left entry $u$ belongs to a two-dimensional region of the complex plain,
or the absolute values $|u|$ belongs to an interval, that is one-dimensional region. We show that this membership condition is sufficient for the unitary to be a good approximation.
In this section we also show that these sequences $g_1,\ldots,g_n$ then can be used to construct solutions to the approximation problems.

We begin by establishing some notation and key definitions~\cite{Kaye2007,nielsen00,watrous2018}.
An arbitrary two-by-two unitary matrix with determinant one~(i.e., a special unitary matrix) can be written as:
\[
    U = \at{\begin{array}{cc}u & -v^{\ast} \\ v & u^{\ast}\end{array}}, \text{ for } u,v \in \c \text{ such that } \abs{u}^2 + \abs{v}^2 = 1
\]
Using the polar form of complex numbers we can write $u = r_1 \exp\at{i \psi_1}$ and $v = i r_2 \exp\at{i \psi_2}$.
Let us introduce $r_1 = \cos(\theta)$ and $r_2 = \sin(\theta)$ for some $\theta \in [0, \pi/2]$ because $r_1 ^2 + r_2 ^2 = 1$.
The unitary $U$ can then be expressed as
\begin{equation}\label{eq:polar}
    U = \cos(\theta) e^{i \psi_1 Z} + \sin(\theta) i X e^{ i \psi_2 Z}  \text{ for } \psi_1, \psi_2, \theta \in \mathbb{R} 
\end{equation}
We will interchangeably use both parameterizations of a special unitary $U$. We denote the special unitary group, that is the group of all two by two unitary matrices with determinant equal to 1, by $\mathrm{SU}(2).$ We will also frequently use the fact that Pauli matrices are Hermitian and equal to their inverses. 

A probabilistic ensemble of pure quantum states is represented as a trace-one positive semidefinite operator called a \emph{density matrix}.
The most general type of transformation on quantum state is a \emph{channel}, a linear completely positive trace-preserving map on the space of density matrices.
The action of a unitary $U$ on density matrix $\rho$ is given by
\begin{equation}
  \mathcal{U}(\rho) = U\rho U^\dagger
\end{equation}
and we refer to $\mathcal{U}$ as the ``channel induced by $U$''. 
For diagonal unitaries of the form $e^{i\phi Z}, e^{i\theta X}$ we denote the induced channel by
\begin{equation}
  \mathcal{Z}_\phi(\rho) = e^{i\phi Z}\rho e^{-i\phi Z},\,\mathcal{X}_\theta(\rho) = e^{i\theta X}\rho e^{-i\theta X}
\end{equation}
To measure the distance between channels we use the diamond norm
\begin{equation}
  \nrm{\Phi}_\diamond := \max_\rho \nrm{\left(\Phi\otimes\mathcal{I}\right)(\rho)}_1
\end{equation}
where $\mathcal{I}$ is the channel induced by the identity matrix.
For additional discussion and facts about the diamond norm see~\app{diamond-norm-properties}.

Our main goal in this paper is to solve single-qubit unitary approximation problems. The most general form is:

\begin{prob}[General qubit unitary approximation]\label{prob:general-unitary-approximation}
    Given:
    \begin{itemize}
        \item target unitary $U\in \mathrm{SU}(2)$,
        \item gate set $G$, a finite set of unitary matrices with determinant one
        \item accuracy $\varepsilon$,\footnote{The parameter $\varepsilon$ is commonly referred to as the \textit{precision} in the literature. Since we use $\varepsilon$ as a measure of approximation error from a target, we believe the term accuracy is more appropriate.} 
        a positive real number
    \end{itemize}
    Find a channel $\mathcal{V}$ implemented using elements of $G$ and computational basis measurements such that 
    \[
        \nrm{ \mathcal{U} - \mathcal{V}}_\diamond \le \varepsilon,
    \]
    where $\mathcal{U}$ is the channel induced by $U$.
\end{prob}

In a simpler case, when channel $\mathcal{V}$ corresponds to unitary $V$ equal to the product $g_1 \ldots g_n $ of two-by-two matrices from gate set $G$,
the diamond norm is tightly bounded by twice the minimum spectral norm distance between $\pm U$ and $V$~(see \cref{cor:diamond-distance-between-general}).
To avoid frequent explicit references to channels $\mathcal{U}$ and $\mathcal{V}$ induced by unitaries $U$ and $V$ we introduce distance between 
unitaries $U$ and $V$ as: 
\begin{equation}
\dst\at{U,V} = \nrm{ \mathcal{U}  - \mathcal{V} }_\diamond.
\end{equation}

Of particular interest is the case where $U$ is a diagonal unitary, namely $U= e^{i\phi Z}$ for real $\phi$.
This case is very common in many quantum algorithms.
In addition, the state-of-the-art way of solving the general unitary approximation problem is to use \emph{Euler angle decomposition} to reduce the problem
to three diagonal unitary approximation problems.
Recall that $e^{i \theta X } = \cos\at{\theta} I + i \sin\at{\theta} X$.
Euler decomposition is performed by solving for $\phi_1,\phi_2$ in the equation below:
\begin{equation}\label{eq:edecomp}
    U = \cos(\theta) e^{i \psi_1 Z} + \sin(\theta) i X e^{ i \psi_2 Z}
    =
    e^{i \phi_1 Z} e^{i \theta X} e^{i \phi_2 Z}
    =
    \change{\cos\at{\theta}} e^{i\at{\phi_1+\phi_2} Z} + \sin\at{\theta} i X e^{ i \at{\phi_2 - \phi_1} Z}
\end{equation}
To obtain the last equality we used the fact that $e^{i \phi Z} X = X e^{ -i \phi Z}$, since $XZX = -Z$ and for any invertible matrix $A$ and any matrix $B$, it is the case that $A e^B A^{-1} = e^{A B A^{-1}}$.

In this section, 
we demonstrate how the general unitary approximation problem reduces to two diagonal approximations and a search for elements in a one-dimensional (1D) region,
that we call magnitude approximation problem,
improving on the traditional Euler angle decomposition approach. 
We then introduce a series of four problems for approximating diagonal unitaries, corresponding to the combinations of using probabilistic mixing (or not) and using fallback protocols~\cite{RUS1} (or not).
For each problem we give a reduction to the search for elements in two-dimensional (2D) regions. 
We conclude the section with applying mixing to the magnitude approximation.
\cref{tab:approximation-problems-summary} summarizes this section.

\begin{table}[h]
\caption[Approximation problems summary]{
\label{tab:approximation-problems-summary}
Summary of the qubit unitary approximation protocols.
Each protocol corresponds to a "key problem" for which
the top-left entry of matrix
$\at{\begin{smallmatrix} u & -v^\ast \\ v & u^\ast \end{smallmatrix}}$
belongs to a two-dimensional region of complex plain,
or the absolute value $|u|$ belongs to a one-dimensional interval.
Some approximation protocols use others as sub-protocols.
Combining the solution to the key problem(s) with the sub-protocols is described by the statements 
in "Full protocol analysis" column. 
Comparisons of the geometric regions are given in \cref{sec:geometric-interpretations}, \cref{tab:approx-region-areas} and \cref{fig:diagonal-approximation-regions}.
For cost scaling see \cref{tab:approximation-cost-scaling} and \cref{tab:magnitude-approx-applications}.
}
\begin{center}
{
\scriptsize
\setlength{\tabcolsep}{0.5em}
\begin{tabular}{|c|c|c|c|c|c|c|}
\hline 
Approximation &  & Target & Key & Key & Full & Sub-\tabularnewline
protocol & Section & unitary & problem & problem & protocol & protocols\tabularnewline
 &  &  & definition & region & analysis & \tabularnewline
\hline 
\hline 
General  & \cref{sec:magnitude-approximation}  & $\mathrm{SU}(2)$ & \cref{prop:magnitude-approximation-condition} & \cref{prop:magnitude-approximation-condition}  & \cref{prop:general-unitary-approximation} & Diagonal or\tabularnewline
unitary &  &  & (magnitude) & \cref{fig:amplitude-interval} &  & Fallback\tabularnewline
\hline 
Mixed general & \cref{sec:mixed-magnitude-approximation} & $\mathrm{SU}(2)$ & \cref{prob:magnitude-mixing} & \cref{prop:magnitude-mixing-condition} & \cref{prop:general-mixing-condition} & Mixed diagonal\tabularnewline
unitary &  &  & (magnitude) & \cref{fig:magnitude-mixing-condition} &  & or mixed fallback\tabularnewline
\hline 
Diagonal & \cref{sec:approx:diag} & $e^{i\varphi Z}$ & \cref{prob:diagonal-approximation} & \cref{prop:segment-condition} &  \cref{prop:segment-condition} & -\tabularnewline
unitary &  &  & (diagonal) & \cref{fig:diagonal-condition} &  & \tabularnewline
\hline 
Mixed &  \cref{sec:unitary-mixing} & $e^{i\varphi Z}$ & \cref{prob:diagonal-approximation-by-unitary-mixing} & \cref{prop:diagonal-mixing-condition} & \cref{prop:diagonal-mixing-condition}  & -\tabularnewline
diagonal unitary &  &  & (diagonal) & \cref{fig:diagonal-mixing-condition} &  & \tabularnewline
\hline 
Fallback & \cref{sec:fallback-approximation}  & $e^{i\varphi Z}$ & \cref{prob:projective-approximation} & \cref{prop:segment} & \cref{prop:fallback-approximation} & Diagonal\tabularnewline
(\cref{fig:fallback-circuit}) &  &  & (projective) & \cref{fig:fallback-condition} &  & \tabularnewline
\hline 
Mixed & \cref{sec:fallback-mixing}  & $e^{i\varphi Z}$ & \cref{prob:projective-rotation-mixing} & \cref{prop:fallback-mixing-condition} & \cref{thm:fallback-mixture-diamond-distance} & Mixed diagonal\tabularnewline
fallback &  &  & (projective) & \cref{fig:fallback-mixing-condition} &  & \tabularnewline
\hline 
\end{tabular}
}
\end{center}
\end{table}

\subsection{Magnitude approximation}
\label{sec:magnitude-approximation}

In the general unitary approximation problem (\problem{general-unitary-approximation}), the task is to approximate an arbitrary unitary $U$.
The standard approximation strategy is to use the Euler angle decomposition $U = e^{i\phi_1 Z} e^{i \theta X} e^{i\phi_2 Z}$ and independently
approximate the three elements of the product.
Our new approach is to first approximate $e^{i \theta X}$ up to phases, that is finding a unitary $V$ equal to
$e^{i\phi'_1 Z} e^{i \theta' X} e^{i\phi'_2 Z}$ such that $\theta$ and $\theta'$ are close. In other words only magnitudes of entries of $U$ and $V$ are close, and phases $\phi'_1$ and $\phi'_2$ are arbitrary.
We then re-express $U$ as:
\begin{equation}
U =  e^{i(\phi_1-\phi'_1) Z}  \underline{ e^{i\phi'_1 Z} e^{i \theta X} e^{i\phi'_2 Z} } e^{i(\phi_2-\phi'_2)Z}.
\end{equation}
The underlined  middle part of the product is approximated by $V$, so it remains to approximate two diagonal $Z$ rotations.

The first main insight behind this strategy is that magnitude approximations have lower cost (see \cref{tab:approximation-cost-scaling} and \cref{sec:cost-scaling}) and are easier to find than diagonal approximations.
The second insight is that, for a random angle $\theta$, the approximation cost of a diagonal unitary $e^{i\theta Z}$ is independent of $\theta$ (see~\cref{fig:clifford-t-random} and \cref{sec:numerical-results}).
Therefore, we may freely adjust the angles of the $Z$-axis rotations in the Euler decomposition in order to compensate for phase inaccuracy of the $X$-axis rotation.
This results in a circuit that is approximately $7/9$ times the length, in terms of gate-count, of the solution resulting directly from Euler decomposition. An analogous strategy, discussed in \sec{crypto-connection}, was developed by Carvalho Pinto and Petit in \cite{pinto2018better} for path finding in LPS graphs, and they noted that their method could be adapted to the quantum setting. This was also confirmed by Stier \cite{stier}, concurrent to the work done in this paper. To construct $V$, we use the following proposition, which determines the approximate synthesis of any unitary by imposing the condition that the norm of its upper left element lies in a given interval.
\newpage

\begin{prop}[Magnitude approximation condition]
    \label{prop:magnitude-approximation-condition}
    \label{prop:absolute-value-constraint}
    Let $\theta$ be from $[0,\pi/2]$ and \change{$\varepsilon\leq2$} be a positive real number. 
    Suppose that we have found a special unitary $V$ 
    $$
    V = g_1 \ldots g_n =  \at{\begin{array}{cc}u & -v^{\ast} \\ v & u^{\ast}\end{array}}
    $$
    over gate set $G$ 
    such that $|u|$ belongs to the interval $\left\{ \cos(\theta'') : \theta'' \in [0,\pi/2], |\theta'' - \theta| \le \delta \right\}$ for $\delta = \arcsin(\varepsilon/2)$.

    Then unitary $V$ satisfies the inequality $\dst\at{V, e^{i \phi_1' Z} \underline{e^{i \theta X}} e^{i \phi_1' Z}} \le \varepsilon$, 
    for $\phi_1'$ and $\phi'_2$ defined by the equality $V = e^{i \phi_1' Z} e^{i \theta' X} e^{i \phi_1' Z}$ with $\theta' \in [0,\pi/2]$.
    We call such $V$ a \textbf{magnitude $\varepsilon$-approximation} of $e^{i \theta X}$. 
    For a geometric interpretation of the constraint see~\fig{amplitude-interval}.

\end{prop}
\begin{proof}
    By unitary invariance property of the diamond norm (see \cref{prop:diamond-norm-unitary-invariance}),
    distance $\dst\at{V, e^{i \phi_1' Z} \underline{e^{i \theta X}} e^{i \phi_1' Z}}$ is equal to $\dst\at{e^{i \theta' X}, e^{i \theta X}}$.
    Now we use the the fact that diamond norm distance between unitaries $U$ and $V$ is equal to the diameter of the smallest disc containing eigenvalues of  $U^\dagger V$~(see \cref{thm:diamond-distance-between-unitary-channels}).
    The eigenvalues of $e^{i (\theta'-\theta) X}$ are $e^{\pm i(\theta - \theta')}$ because Hadamard diagonalizes $e^{i\phi X} = H e^{i\phi Z} H$.
    Because we only have two eigenvalues, the diameter of the disc containing them is equal to distance $|e^{i(\theta - \theta')} - e^{-i(\theta - \theta')}|$,
    which is equal to $2|\sin(\theta - \theta')|$. It remain to upper-bound this quantity.
    By definition of $\theta'$, $|u| = \cos\at{\theta'}$. This implies that $|\theta' - \theta | \le \delta$ because cosine is a bijection from $[0,\pi/2]$ onto $[0,1]$.
    Using $\sin(\delta) =\varepsilon/2$ we get the required bound.
\end{proof}

\begin{figure}
    \centering
    \includegraphics[width=0.5\textwidth]{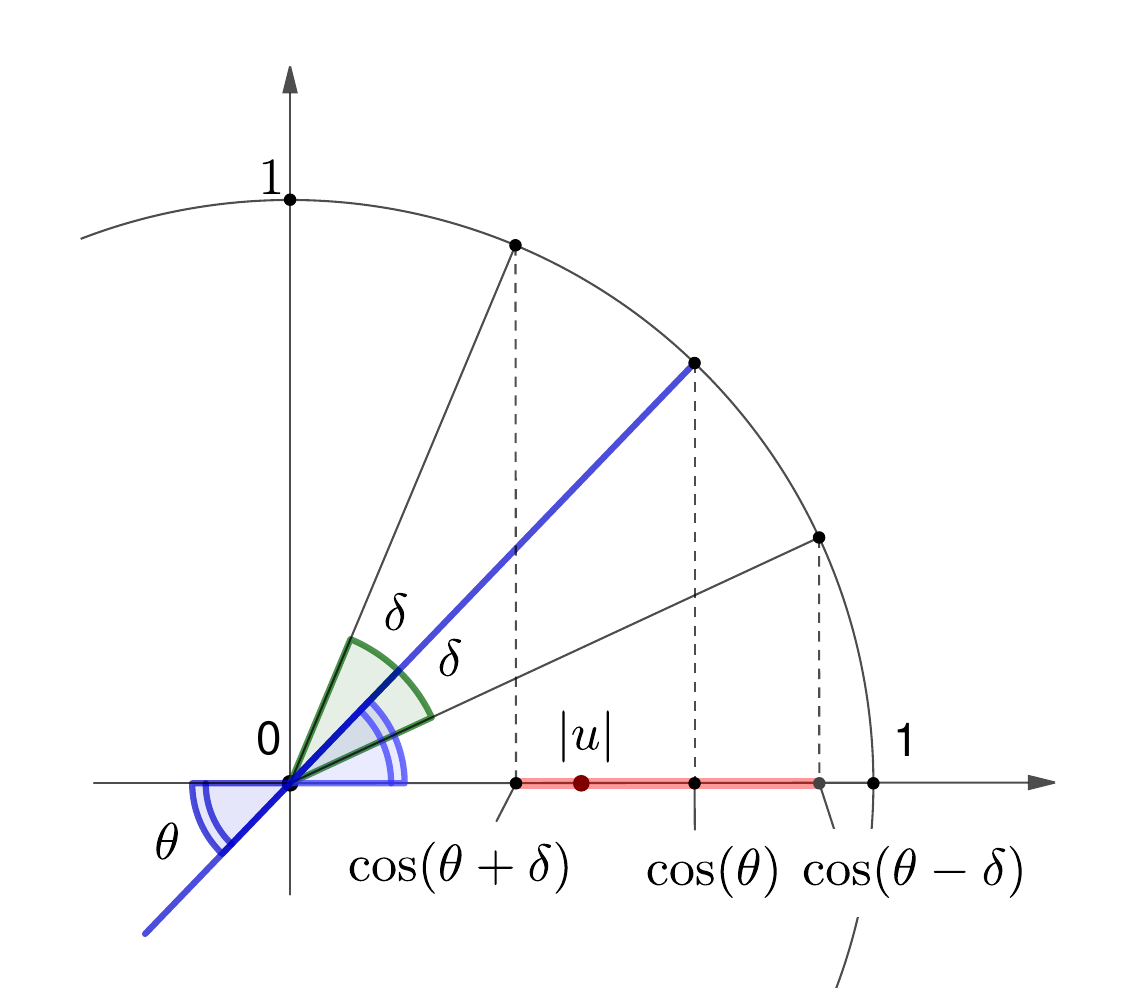}
    \caption{
        \label{fig:amplitude-interval}
        Geometric interpretation of the constraint on complex number $u$ appearing in~\propos{absolute-value-constraint}. 
        Possible absolute values $|u|$ belong to the interval $\left\{ \cos(\theta'') : \theta'' \in [0,\pi/2], |\theta'' - \theta| \le \delta \right\}$ for $\delta = \arcsin(\varepsilon/2)$ and are shown as blue and dashed blue 
        segments on the horizontal axis. 
    }
\end{figure}

Note that when $ \arcsin(\varepsilon/2) \le \theta \le \pi/2 - \arcsin(\varepsilon/2)$, the absolute value $|u|$ must simply belong to interval $[\change{\cos(\theta-\arcsin(\varepsilon/2))},\cos(\theta+\arcsin(\varepsilon/2))]$.
\change{\begin{rem}Since in practice $\varepsilon$ is chosen to be very small, we are able to take the approximation $\arcsin(x)\approx x$ in \propos{absolute-value-constraint}, hence, $\delta\approx \varepsilon/2$.
\end{rem}}

A simple way to leverage~\propos{absolute-value-constraint} for general unitary approximation is to split the accuracy $\varepsilon$ evenly among the three factors of the Euler decomposition.  Then, use~\propos{absolute-value-constraint} to find a magnitude $\epsilon/3$-approximation of the $X$-axis rotation.  Finally, find $\varepsilon/3$-approximations of the two remaining $Z$-axis rotations, adjusting the angles to compensate for the phase inaccuracy of the $X$-axis rotation. 
This strategy is captured formally in the following proposition.

\begin{prop}[General unitary approximation] \label{prop:general-unitary-approximation}
  Suppose we are given a target unitary $U = e^{i\phi_1 Z} e^{i \theta X} e^{i\phi_2 Z}$ and target accuracy $\varepsilon$. 
  Let $V$ be a magnitude $\varepsilon_0$-approximation of $e^{i \theta X}$ (see~\cref{prop:absolute-value-constraint}) and let $V=e^{i\phi'_1 Z} e^{i \theta' X} e^{i\phi'_2 Z}$.
  Let channels $\Psi_k$ be within diamond norm distance $\varepsilon_k$ from unitary $e^{i(\phi_k-\phi'_k)Z}$, for $k=1,2$, and let  $\varepsilon \ge \varepsilon_0 + \varepsilon_1 + \varepsilon_2$.

  Then channel $\mathcal{U}$ induced by $U$ and composition $\Psi_1 \mathcal{V} \Psi_2$ satisfy 
$$
  \nrm{ \mathcal{U} - \Psi_1 \mathcal{V} \Psi_2 }_\diamond \le \varepsilon \text{, where } \mathcal{V} \text{ is the channel induced by } V.
$$
\end{prop}

\begin{proof}
  Let us write $U$ as a product $U_1 U_0 U_2$, where $U_k = e^{i(\phi_k-\phi'_k)Z}$ and $U_0 = e^{i\phi'_1 Z} e^{i \theta X} e^{i\phi'_2 Z}$.
  We then write channel $\mathcal{U}$ as composition of channels $\mathcal{U}_1 \mathcal{U}_0 \mathcal{U}_2$, where $\mathcal{U}_k$ is 
  the channel induced by $U_k$.
  Using the chain rule for diamond norm we have: 
  \begin{equation} \label{eq:su-2-chain-rule}
  \nrm{ \mathcal{U} - \Psi_1 \mathcal{V} \Psi_2  }_\diamond =
  \nrm{ \mathcal{U}_1\mathcal{U}_0\mathcal{U}_2 - \Psi_1 \mathcal{V} \Psi_2  }_\diamond \le 
  \nrm{ \mathcal{U}_1 - \Psi_1}_\diamond + \nrm{ \mathcal{U}_0 - \mathcal{V} }_\diamond + \nrm{\mathcal{U}_2 - \Psi_2}_\diamond
  \end{equation}
  By \cref{prop:magnitude-approximation-condition} $\nrm{ \mathcal{U}_0 - \mathcal{V} }_\diamond \le \varepsilon_0$.
  Combining this bound with \cref{eq:su-2-chain-rule} completes the proof.
\end{proof}

One can optimize the choices of $\varepsilon_k$ in the above \cref{prop:general-unitary-approximation}.
For random diagonal approximation, the cost scales as $3\log_2(1/\varepsilon) + O(1)$
and for random magnitude approximation the cost scales as $\log_2(1/\varepsilon) + O(1)$~(see~\cref{tab:approximation-cost-scaling}).
To minimize overall sequence length one can choose $\varepsilon_1 = \varepsilon_2 = 0.43 \varepsilon$ and $\varepsilon_0 = 0.14 \varepsilon$,
however this improves the sequence length only by a small additive constant $0.95$ in comparison to distributing errors equally.

\subsection{Diagonal unitary approximation \label{sec:approx:diag}}
The Euler angle decomposition~\cref{eq:edecomp} describes a qubit unitary as a product of two diagonal unitaries of the form $e^{i\theta Z}$ and one $X$ rotation of the form $e^{i\theta X}$. 
\propos{absolute-value-constraint} further reduces the $X$ rotation to a one-dimensional search problem, leaving just the diagonal unitaries.
Therefore, the special case of diagonal unitary approximation is relevant to the general unitary approximation problem. In this section we recall some of the known results regarding the diagonal approximation problem.

\begin{prob}[Diagonal unitary approximation]
    \label{prob:diagonal-approximation}
    Given:
    \begin{itemize}
        \item target angle $\theta$, a real number,
        \item gate set $G$, a finite set of two by two unitary matrices with determinant one,
        \item accuracy $\varepsilon$, a positive real number,
    \end{itemize}
    Find a sequence $g_1, \ldots, g_n$ of elements of $G$ such that
    \[
        \mathcal{D}_\diamond\left(\exp(i\theta Z), g_1\dots g_n\right) \le \varepsilon.
    \]
\end{prob}

Observe that \problem{diagonal-approximation} is a special case of the general unitary approximation problem, where the target unitary is diagonal
and approximating channel $\mathcal{V}$ is induced by unitary $g_1\dots g_n$.
The diagonal unitary approximation problem is easier to solve because it admits the following bound on the diamond norm
that depends only on the top left entry of $V=g_1\ldots g_n.$

\begin{lem}[Diamond difference from a diagonal unitary]
    \label{lem:diagonal-diamond-difference}
    Given an angle $\theta$ and a unitary 
    \[
      V = \at{\begin{array}{cc} u & -v^\ast \\ v & u^\ast \end{array}},
    \]
    the distance 
    \[
      \mathcal{D}_\diamond(e^{i\theta Z},V) = 2 \sqrt{1 - \left(\mathrm{Re}(ue^{-i\change{\theta}})\right)^2} \leq  2\sqrt{2 - 2|\mathrm{Re}\at{u e^{-i\theta}}|}.
    \]
\end{lem}
Proof of this bound is in \cref{cor:diamond-distance-between-general-and-diaognal} in~\app{diamond-norm-properties}.
\lemm{diagonal-diamond-difference} immediately suggests a simple condition for solutions of the diagonal approximation problem.
\begin{prop}[Diagonal approximation condition]
    \label{prop:segment-condition}
    Let $g_1,\dots,g_n$ be a sequence of gates from a gate set $G$ and let
    \[
        g_1\dots g_n = \at{\begin{array}{cc} u & -v^\ast \\ v & u^\ast \end{array}}.
    \] 
    Then $g_1,\dots,g_n$ is a solution to the diagonal approximation problem for target angle $\theta$, gate set $G$ and accuracy $\varepsilon$ if 
    \begin{equation}
      \abs{\mathrm{Re}\left(ue^{-i\theta}\right)} \ge \sqrt{ 1 - \varepsilon^2/4}. 
    \end{equation}
    For a geometric interpretation of the constraints see~\fig{diagonal-condition}.
\end{prop}
\begin{proof}
    Let $\mathcal{V}$ be the channel induced by unitary $V = g_1\dots g_n$.
    Then by~\lemm{diagonal-diamond-difference}
    \begin{equation}
      \nrm{\mathcal{Z}_\theta - \mathcal{V}}_\diamond
      = 2 \sqrt{1-|\mathrm{Re}(u e^{-i\theta})|^2}
      \leq 2 \sqrt{\varepsilon^2/4} 
      = \varepsilon.
    \end{equation}
\end{proof}

\begin{figure}
    \centering
    \includegraphics[scale=0.4]{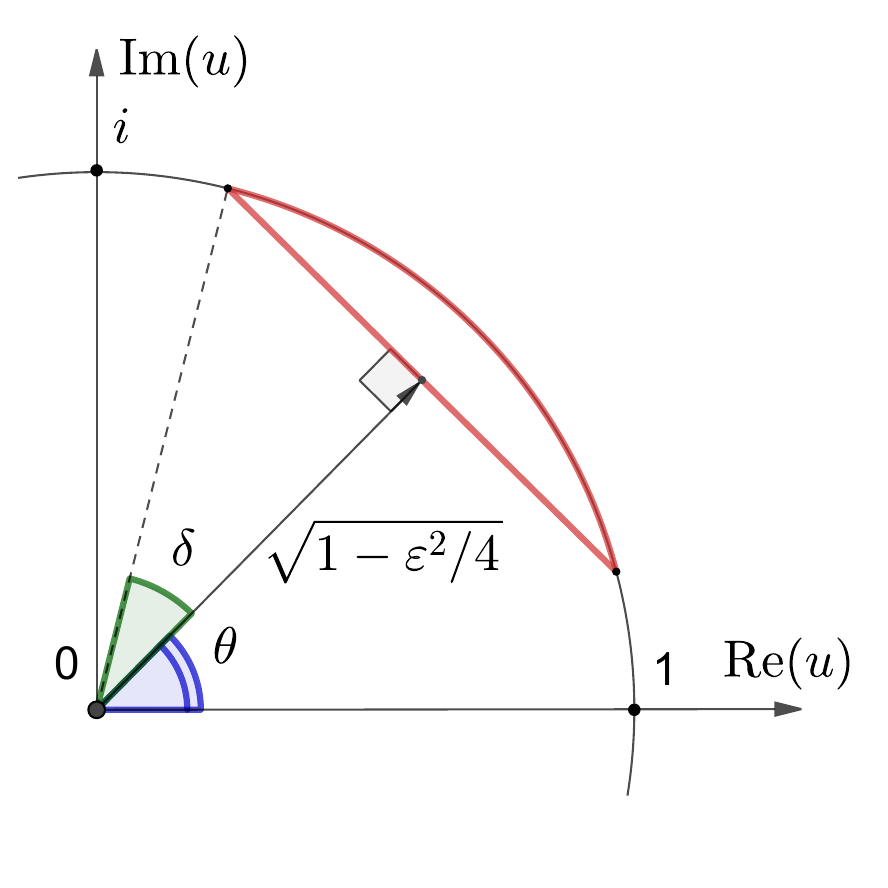}
    \caption{
        \label{fig:diagonal-condition}
        Geometric interpretation of constraints on complex number $u$ appearing in ~\propos{segment-condition}.
        The region with red boundary contains complex numbers $u$ that satisfy constraints $\mathrm{Re}\at{u e^{-i \theta}} \ge \sqrt{1-\varepsilon^2/4}$ and $\abs{u} \le 1$.
        Note that the segment spans points with angular coordinates $[\theta - \delta, \theta + \delta]$ for $\delta = \arcsin(\varepsilon/2)$.
        Constraints in \propos{segment-condition} lead to two regions: one with segment spanning points with angular coordinates $[\theta - \delta, \theta + \delta]$ and another one with $[\pi + \theta - \delta, \pi+\theta + \delta]$.
    }
\end{figure}

\subsection{Fallback approximation \label{sec:fallback-approximation}}

Fallback protocols~\cite{BRS2015} offer a more efficient way to approximate diagonal unitaries by incorporating measurements.
A fallback protocol is a non-deterministic single-qubit quantum channel consisting of two steps: a projective rotation and a conditional fallback. 
The projective rotation and fallback steps may be implemented in a variety of ways. We limit our discussion to fallback protocols with the form illustrated in~\fig{fallback-circuit}.

\newpage
For a fixed single-qubit unitary
\begin{equation}
  V = \left(\begin{array}{cc} u & -v^\ast \\ v & u^\ast \end{array} \right)
\end{equation}
the corresponding projective rotation effects one of two diagonal rotations on $\ket{\psi}$ depending on the measurement outcome.
With probability $|u|^2$ the measurement outcome is zero, the projective rotation is said to have ``succeeded'' and 
the input state undergoes the transformation
\begin{equation}
  \ket{\psi} \mapsto e^{i\theta_0 Z}\ket{\psi} = e^{i\mathrm{Arg}(u)Z}\ket{\psi}.
\end{equation}
otherwise, the measurement outcome is one, the projective rotation is said to have ``failed'' and
\begin{equation}
  \ket{\psi} \mapsto e^{i\theta_1 Z}\ket{\psi} = e^{i\mathrm{Arg}(v) Z}\ket{\psi}.
\end{equation}

The projective rotation is intended to approximate a target diagonal unitary $e^{i\theta Z}$ so that $\theta_0 \approx \theta$. 
The constraints necessary to achieve that goal are captured by the following problem.
\begin{prob}[Projective approximation]
    \label{prob:projective-approximation}
    Given:
    \begin{itemize}
        \item target angle $\theta$, a real number,
        \item success probability $q$, a positive real number between $0$ and $1$,
        \item gate set $G$, a finite set of two by two unitary matrices with determinant one,
        \item accuracy
        $\varepsilon$, a positive real number,
    \end{itemize}
    find a sequence $g_1,\ldots, g_n$ of elements in $G$, such that for $u,v$ defined via $g_1 \ldots g_n = \left(\begin{array}{cc} u & -v^\ast \\ v & u^\ast \end{array} \right)$
    the following holds:
    \begin{itemize}
        \item $|u|^2 \ge q$, and
        \item $\mathcal{D}_\diamond\left(e^{i\theta Z},e^{i \mathrm{Arg}(u) Z}\right) \le \varepsilon$. 
    \end{itemize}
\end{prob}

\begin{figure}
    \centering
    \includegraphics[scale=0.4]{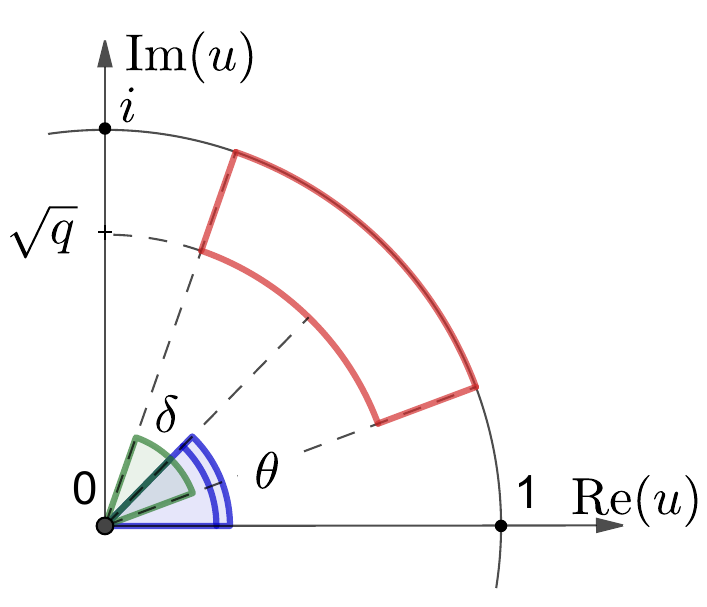}
    \caption{
        \label{fig:fallback-condition}
        Geometric interpretation of constraint on complex number $u$ appearing in \propos{segment}.
        The region with red boundary contains complex numbers $u$ that satisfy constraints $\mathrm{Arg}\at{u} \in  [\theta - \delta, \theta + \delta]$ and $\abs{u} \ge q$, where $\delta = \arcsin(\varepsilon/2)$.
        Constraints in \propos{segment} lead to two regions: one with $\mathrm{Arg}\at{u} \in  [\theta - \delta, \theta + \delta]$ and another one with $\mathrm{Arg}\at{u} \in  [\pi + \theta - \delta, \pi+\theta + \delta]$
    }
\end{figure}

Much like the case of the diagonal approximation problem~(\ref{prob:diagonal-approximation}), solutions to \problem{projective-approximation} can be characterized entirely by conditions on the complex value $u$ at the top-left entry of the circuit unitary $g_1 \ldots g_n$. These conditions are, however, less restrictive than those prescribed by~\propos{segment-condition}. 
\change{We compare the conditions in detail in \cref{sec:geometric-interpretations} with the geometric representations shown in \cref{fig:diagonal-approximation-regions}.}
\begin{prop}[Projective approximation condition]\label{prop:segment}  
  Let $g_1,\dots,g_n$ be a sequence of gates from a gate set $G$ and
    $g_1\ldots g_n = \at{\begin{smallmatrix} u & -v^\ast \\ v & u^\ast \end{smallmatrix}}$.
    Then $g_1,\dots,g_n$
    is a solution to the projective approximation problem~(\problem{projective-approximation}) if $u$ 
    satisfies
    \[
      \abs{u} \ge \sqrt{q} \mathrm{~and~}
      \left( \sin\abs{\mathrm{Arg}(u)-\theta} \le \varepsilon/2 \mathrm{~or~} \sin\abs{\mathrm{Arg}(u)-(\theta+\pi)} \le \varepsilon/2 \right)
    \]
    For a geometric interpretation of these constraints see~\fig{fallback-condition}.
\end{prop}

\begin{proof}
    The condition $|u| \geq \sqrt{q}$ is equivalent to $|u|^2 \geq q$ from~\problem{projective-approximation}.

    It remains to show that $\nrm{\mathcal{Z}_\theta - \mathcal{Z}_{\mathrm{Arg}(u)}}_\diamond \le \varepsilon$.
    According to \cref{cor:diamond-distance-between-diaognal}, we have 
    $\nrm{\mathcal{Z}_\theta - \mathcal{Z}_{\mathrm{Arg}(u)}}_\diamond \le 2\sin|\mathrm{Arg(u)}-\theta|$.
    This immediately implies that the channels are $\varepsilon$-close when $\sin\abs{\mathrm{Arg}(u)-\theta} \le \varepsilon/2$.
    Because $\mathcal{Z}_\theta = \mathcal{Z}_{\theta+\pi}$ inequality $\sin\abs{\mathrm{Arg}(u)-(\theta+\pi)} \le \varepsilon/2$
    also ensures $\nrm{\mathcal{Z}_\theta - \mathcal{Z}_{\mathrm{Arg}(u)}}_\diamond \le \varepsilon$.
\end{proof}

\cite{BRS2015} constructs a solution to~\problem{projective-approximation} by first approximating the target phase factor $e^{i\theta}$ with a cyclotomic rational of the form $z^\ast/z$, then searching for a real-valued modifier to achieve the desired success probability $q$. 
The characterization of the fallback approximation problem given by \propos{segment} differs by addressing accuracy ($\varepsilon$) and success probability ($q$) conditions simultaneously, resulting in an intuitive geometric description as illustrated in~\fig{fallback-condition}.

Any solution to the diagonal approximation problem is also a solution to the corresponding projective approximation problem. The projective problem admits additional and possibly cheaper solutions.

\problem{projective-approximation} constrains the action of a successful projective rotation but ignores the failure case.
The difference $\theta-\theta_1$ between the target and failure angles may be large, in general.
Therefore, in the case of failure (measurement outcome one), the fallback step is applied in order to recover and approximate the target rotation.   

The problem of constructing a fallback step can be treated independently of the projective rotation.
In \cite{BRS2015}, the fallback step is a unitary $B \approx e^{i(\theta - \theta_1)Z}$ chosen so that the net effect of the failure case is
\begin{equation}
  \ket{\psi} \mapsto B e^{i\theta_1Z}\ket{\psi} \approx e^{i\theta Z}\ket{\psi}.
\end{equation}
This choice corresponds directly to the diagonal approximation \cref{prob:diagonal-approximation} defined earlier.
A complete fallback protocol of this form may be constructed by first solving~\problem{projective-approximation} and then solving~\problem{diagonal-approximation} for appropriate values of $\varepsilon$.
This is captured by the following proposition that follows from standard properties of the diamond norm (see~\cref{app:diamond-norm-properties}).

\begin{prop}[Fallback approximation]
    \label{prop:fallback-approximation}
    Suppose we are given:
    \begin{itemize}
        \item target angle $\theta$, a real number,
        \item success probability $q$, a positive real number between $0$ and $1$,
        \item gate set $G$, a finite set of two by two unitary matrices with determinant one
    \end{itemize}
    and
    \begin{itemize}
      \item real numbers $\varepsilon_1, \varepsilon_2$
      \item $g_1 \ldots g_n = \left(\begin{array}{cc} u & -v^\ast \\ v & u^\ast \end{array} \right)$, a solution to~\problem{projective-approximation} for $\{\theta, q, G, \varepsilon_1\}$, and
      \item $b_1 \ldots b_m = B$, a solution to~\problem{diagonal-approximation} for $\{\theta - \mathrm{Arg}(v), G, \varepsilon_2\}$
    \end{itemize}
    then  overall fallback protocol accuracy is
    \begin{equation}
        \nrm{\mathcal{Z}_\theta - |u|^2 \mathcal{Z}_{\mathrm{Arg}(u)} - |v|^2 \mathcal{B} \mathcal{Z}_{\mathrm{Arg}(v)}}_\diamond \le \varepsilon_1 + |v|^2 \varepsilon_2 
    \end{equation}
    where $\mathcal{B}(\rho) := B\rho B^\dagger$. 
\end{prop}

A simple approach to solving the above problem is to choose $\varepsilon_1 = \varepsilon/2$, solve \cref{prob:projective-approximation}, then choose $\varepsilon_2 = \varepsilon/2/|v|^2$
and then solve the corresponding instance of \cref{prob:diagonal-approximation}.

\CP{This looks slightly  suboptimal to me (when $|v|$ is small most errors should be allocated to $\epsilon_1$ to save on overall expected length)}
\VK{Given out scaling pre-factors cost of $\varepsilon/2$ approximation is very close to the cost of  $\varepsilon$ approximation. That is why we call this "a simple" approach. It is slightly sub-optimal.}

\cref{prop:fallback-approximation} can be generalized to admit an arbitrary quantum channel (denoted by $\mathcal{B}$ in~\fig{fallback-circuit}) as the fallback step.  
For example, the fallback may be simply to repeat the projective rotation until success is achieved \cite{paetznick2013repeat,RUS1}.  
In Section \ref{sec:fallback-mixing} we consider fallbacks that are probabilistic mixtures of unitaries.

The cost of of fallback protocol is a random variable. When success probability is $q = 1 - p$ for small $p$, the average cost is equal to
the cost of the projective rotation step plus $p$ times the cost of the fallback step, which is the cost of a diagonal approximation.
The worst case cost is the sum of the costs of the projective and the diagonal approximation. 
High worst case cost might become a problem when using $N$ fallback approximations in parallel,
however we can always ensure that the probability of at least one of them requiring a fallback step is $p$
by choosing the success probability of each of them $q = 1 - p/N$.

\CP{Implicitly assuming that $qN$ is small?} 
\VK{ $q$ is a success probability, to it is not small. We only assume overall failure probability to be $p < 1$.}

\subsection{Mixed diagonal unitary approximation \label{sec:unitary-mixing}}
\cref{prob:diagonal-approximation}
describes synthesis of a qubit unitary by construction and application of a deterministic sequence of elementary gates. An alternative approach, proposed by \cite{Campbell2017} and \cite{Hastings2017},
is to construct several sequences of elementary gates and apply one of them according to a probability distribution. Given the correct probabilistic mixture of unitaries the overall error of the approximation is reduced quadratically, cutting the approximation cost roughly in half. 
In this paper, we introduce the use of diagonal Clifford twirling to construct these sequences, which ensures that the difference between the ideal and approximating channel is a Pauli channel. 
Because of this, we obtain a closed-form expression for the diamond distance (see~\cref{thm:diamond-distance-beween-pauli-channels}) which improves on analysis in \cite{Campbell2017,Hastings2017} and results in lower approximation costs.
Method in \cite{Campbell2017} uses diagonal approximation algorithms as black-boxes and finds under-rotations and over-rotations by modifying target rotation angle.
In contrast, we modify synthesis algorithms to directly find under and over-rotated approximations and further reduce approximations costs.

\begin{prob}[Diagonal unitary approximation by unitary mixing]
    \label{prob:diagonal-approximation-by-unitary-mixing}
    Given:
    \begin{itemize}
        \item target angle $\theta$, a real number,
        \item gate set $G$, a finite set of two by two unitary matrices with determinant one,
        \item accuracy $\varepsilon$, a positive real number,
    \end{itemize}
    Find 
    \begin{itemize}
        \item $G_1,\ldots, G_n$, a sequence of sequences $G_k$ of elements of $G$  and
        \item $p_1,\ldots, p_n$, a probability distribution
    \end{itemize}
    such that
    \[
        \left\Vert \mathcal{Z}_\theta - \sum_{k=1}^n p_k \mathcal{G}_k \right\Vert_\diamond \le \varepsilon
    \]
    where $\mathcal{G}_k$ is the channel obtained by applying the sequence $G_k$.
\end{prob}
This problem generalizes~\problem{diagonal-approximation} by allowing a random choice among multiple gate sequences.

\cite{Campbell2017} gives an algorithm for constructing the mixture by ``Z twirling'' two unitary approximations: an under-rotation and an over-rotation.
The \emph{twirl} of a unitary $U$ over generators $\mathcal{G}$ is a channel obtained by uniformly selecting a random element $V$ over the set generated by $\mathcal{G}$ and then applying $VUV^\dagger$.
For example, the twirl of $U$ over $\{Z, S=e^{-i\pi Z/4}\}$ which we denote by $\mathcal{T}_U$ is given by
\begin{equation}
  \mathcal{T}_{U}(\rho) =
  \frac{1}{4} 
  \sum_{V\in\{I,Z,S,S^\dagger\}} \left(VUV^\dagger\right) \rho \left(V^\dagger U^\dagger V\right).
\end{equation}


We show that by twirling over the set $\{Z, S\}$, instead of $Z$ alone, the approximation error of the unitary mixture is a probabilistic mixture of Pauli operators---i.e., a Pauli channel. This allows for an alternative proof of \cite{Campbell2017} and \cite{Hastings2017} and yields a simple expression for the approximation error in terms of diamond distance. 

The procedure is as follows.  Find two unitaries (defined formally below): $U_1$ an under-rotation and $U_2$ an over-rotation.  Calculate a probability $p$ (also defined below) that depends on $U_1$ and $U_2$.  With probability $p$ select $U_1$ and otherwise select $U_2$. Then apply the $\{Z,S\}$ twirl to that selection.  The resulting channel $p\mathcal{T}_{U_1} + (1-p)\mathcal{T}_{U_2}$ approximates a diagonal unitary $e^{i\theta Z}$ with an accuracy given by the following theorem.

\begin{thm}[Diamond difference of a twirled mixture]
  \label{thm:unitary-mixture-diamond-distance}
  Let $\theta$ be an angle and let unitaries
  \begin{equation}
    \label{eq:over-under-unitary-approximations}
    U_1 =
    \left(
    \begin{array}{cc}
        r_1e^{i(\theta + \delta_1)} & v_1^*                        \\
        v_1                           & r_1e^{-i(\theta + \delta_1)}
      \end{array}
    \right)
  \end{equation}
  \begin{equation}
    U_2 =
    \left(
    \begin{array}{cc}
        r_2e^{i(\theta + \delta_2)} & v_2^*                        \\
        v_2                           & r_2e^{-i(\theta + \delta_2)}
      \end{array}
    \right)
  \end{equation}
  for real values $r_1, r_2$ and $\sin(\delta_1) < 0 < \sin(\delta_2)$.  Define probability
  \begin{equation}
    \label{eq:mixing-probability}
    p = \frac{r_2^2\sin(2\delta_2)}{r_2^2\sin(2\delta_2) - r_1^2\sin(2\delta_1)}.
  \end{equation}
  Then
  \begin{equation}
    ||p\mathcal{T}_{U_1} + (1-p)\mathcal{T}_{U_2} - \mathcal{Z}_\theta||_\diamond = 2\left(1 - p r_1^2\cos^2(\delta_1) - (1-p) r_2^2\cos^2(\delta_2)\right).
  \end{equation}
\end{thm}

The proof of the theorem in given in \cref{sec:diamond-difference-twirled-mixture}.
A simple way to leverage \cref{thm:unitary-mixture-diamond-distance} is by splitting an approximation error $\varepsilon$ evenly between $U_1$ and $U_2$ so that
\begin{equation}\begin{aligned}
    \label{eq:evenly-split-approximation-error}
    1 - r_1^2\cos^2(\delta_1) & \le \varepsilon/2  \\
    1 - r_2^2\cos^2(\delta_2) & \le \varepsilon/2.
  \end{aligned}\end{equation}
The synthesis task then is to find two unitary approximations, an ``under rotation'' $U_1$ and ``over rotation'' $U_2$ each such that
\begin{equation}
  |r_k \cos(\delta_k)| \ge \sqrt{1- \varepsilon/2} =
  1 - \varepsilon/4 - \varepsilon^2/32 + o(\varepsilon^2), \text{ for }k=1,2.
\end{equation}
This strategy is captured in \propos{diagonal-mixing-condition}, below.
\newpage
\begin{prop}[Diagonal mixing approximation condition]\label{prop:diagonal-mixing-condition}  
  Suppose we are given sequences $g_1,\dots,g_n$ and $h_1,\dots,h_m$ of gates from a gate set $G$.
  Define $u_k,v_k$ from the equations below
  \[
        g_1\dots g_n = \at{\begin{array}{cc} u_1 & -v_1^\ast \\ v_1 & u_1^\ast \end{array}},
        h_1\dots h_m = \at{\begin{array}{cc} u_2 & -v_2^\ast \\ v_2 & u_2^\ast \end{array}}.
  \]
    Then 
    \begin{itemize}
      \item sequence $G_1,\ldots,G_4,H_1,\ldots,H_4$, where 
        $G_k = \sigma_k,g_1,\dots,g_m,\sigma_k^\dagger$, 
        $H_k = \sigma_k,h_1,\dots,h_m,\sigma_k^\dagger$, 
        and $\sigma_1,\sigma_2,\sigma_3,\sigma_4 = I,S,Z,S^\dagger$.
      \item probability distribution $p/4,p/4,p/4,p/4,(1-p)/4,(1-p)/4,(1-p)/4,(1-p)/4$ 
    \end{itemize}
    is a solution to the diagonal unitary approximation~\problem{diagonal-approximation-by-unitary-mixing} with accuracy $\varepsilon$ and target angle $\theta$ if 
    \begin{itemize}
      \item $u_1$ satisfies $\abs{\mathrm{Re}\at{u_1 e^{-i\theta}}} \geq \sqrt{1-\varepsilon/2}$, $\mathrm{Im}\at{u_1 e^{-i\theta}} < 0$, and
      \item $u_2$ satisfies $\abs{\mathrm{Re}\at{u_2 e^{-i\theta}}} \geq \sqrt{1-\varepsilon/2}$, $\mathrm{Im}\at{u_2 e^{-i\theta}} > 0$.
    \end{itemize}
    For a geometric interpretation of these constraints see \fig{diagonal-mixing-condition}.
\end{prop}
\begin{proof}
  First, note that the sequences $G_1,G_2,G_3,G_4$ along with probabilities
  $\{p/4, p/4, p/4, p/4\}$ corresponds to the $\{Z,S\}$ twirl of $U_1 = g_1\dots g_n$ with probability $p$,
  \begin{equation}
    p\mathcal{T}_{U_1}(\rho) 
    = \frac{p}{4}\sum_{\sigma \in \{I,Z,S,S^\dagger\}} \at{\sigma U_1 \sigma^\dagger} \rho \at{\sigma^\dagger U_1^\dagger \sigma}.
  \end{equation}
  Similarly, the sequences $H_1,H_2,H_3,H_4$ along with probabilities
  $\{(1-p)/4, (1-p)/4, (1-p)/4, (1-p)/4\}$ 
  corresponds to the $\{Z,S\}$ twirl of $U_2 = h_1\dots h_m$ with probability $(1-p)$,
  \begin{equation}
    (1-p)\mathcal{T}_{U_2}(\rho) 
    = \frac{1-p}{4}\sum_{\sigma \in \{I,Z,S,S^\dagger\}} \at{\sigma U_2 \sigma^\dagger} \rho \at{\sigma^\dagger U_2^\dagger \sigma}.
  \end{equation}
  We therefore seek to show that
  $
    \nrm{p\mathcal{T}_{U_1} + (1-p)\mathcal{T}_{U_2} - \mathcal{Z}_\theta}_\diamond \leq \varepsilon.
  $
  Let $r_1 = |u_1|, \delta_1 = \mathrm{Arg}(u_1) - \theta$ and similarly $r_2 = \abs{u_2}, \delta_2 = \mathrm{Arg}(u_2) - \theta$.
  Then $\mathrm{Im}\at{u_1 e^{-i\theta}} = \sin(\delta_1) < 0$ and $\mathrm{Im}\at{u_2 e^{-i\theta}} = \sin(\delta_2) > 0$.
  Substituting
  $U_1 = \at{\begin{smallmatrix} u_1 & -v_1^\ast \\ v_1 & u_1^\ast \end{smallmatrix}}$,
  $U_2 = \at{\begin{smallmatrix} u_2 & -v_2^\ast \\ v_2 & u_2^\ast \end{smallmatrix}}$ 
  into~\theo{unitary-mixture-diamond-distance} and using 
  $\abs{\mathrm{Re}(u_1 e^{-i\theta})} \geq \sqrt{1-\varepsilon/2}$, $\abs{\mathrm{Re}(u_2 e^{-i\theta})} \geq \sqrt{1-\varepsilon/2}$, 
  we obtain
  \begin{equation}\begin{aligned}
    \nrm{p\mathcal{T}_{U_1} + (1-p)\mathcal{T}_{U_2} - \mathcal{Z}_\theta}_\diamond 
    &= 2\at{1-pr_1^2\cos^2(\delta_1) - (1-p)r_2^2\cos^2(\delta_2)} \\
    &= 2\at{1-p \mathrm{Re}\at{u_1 e^{-i\theta}}^2 - (1-p)\mathrm{Re}\at{u_2 e^{-i\theta}}^2} \\
    &\leq 2\at{1-p(1-\varepsilon/2) - (1-p)(1-\varepsilon/2)} \\
    &=\varepsilon.
  \end{aligned}\end{equation}
  
\end{proof}

\begin{figure}
    \centering
    \includegraphics[scale=0.38]{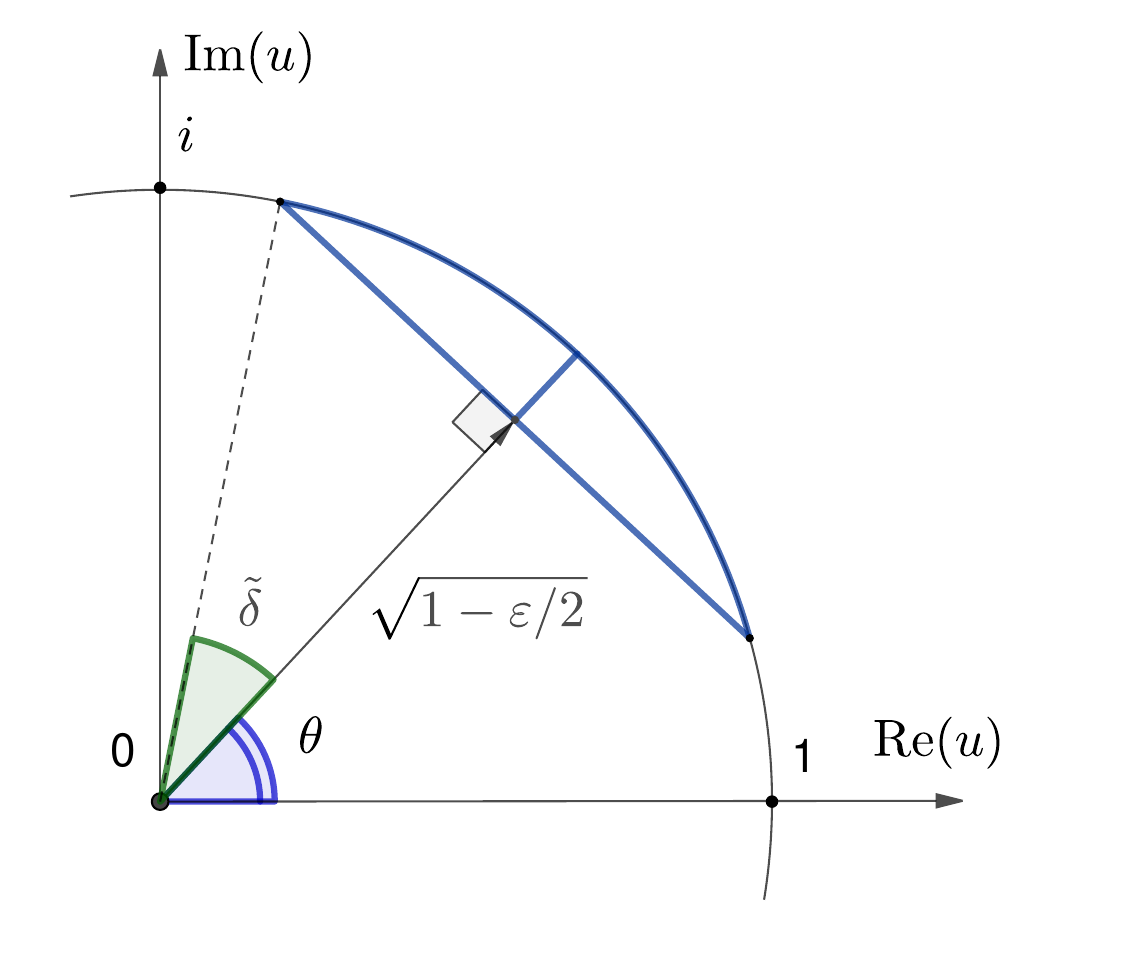}
    \caption{
        \label{fig:diagonal-mixing-condition}
        Geometric interpretation of constraints on complex numbers $u_1$ and $u_2$ appearing in~\propos{diagonal-mixing-condition}.
        The region with blue boundary contains complex numbers $u$ that satisfy constraints $\mathrm{Re}\at{u e^{-i \theta}} \ge \sqrt{1-\epsilon/2}$ and $\abs{u} \le 1$.
        This region is split into two parts, an under-rotation region for $u_1$ with angular coordinates spanning $[\theta- \tilde\delta,\theta]$ and an over-rotation region for $u_2$
        with angular coordinates spanning $[\theta,\theta+\tilde\delta]$ for $\tilde\delta = \arcsin(\sqrt{\varepsilon/2})$.
    }
\end{figure}

As observed by \cite{Campbell2017,Hastings2017}, the constraints imposed by~\propos{diagonal-mixing-condition} admit quadratically better scaling in $\varepsilon$ as compared to approximation without mixing (\propos{segment-condition}), which would require $|r \cos(\delta)| \ge \sqrt{1 - \varepsilon^2/4}$.

Evenly splitting the error as in~\propos{diagonal-mixing-condition}
does not yield optimal solutions in general. A better approach is to first find a cheap (but possibly poor) approximation of the target.  With the first approximation fixed, a search region for the second approximation can be defined. 
In particularly, this is useful when identity is a sufficiently good under-rotated or over-rotated approximation to the target rotation. This happens in practice when approximating Fourier angles.
\change{We give a detailed treatment in \cref{appendix:optimal-mixing-solutions}.}

The main technical component of \cref{thm:unitary-mixture-diamond-distance} is to show that the twirled mixture $p\mathcal{T}_{U_1}(\rho) + (1-p)\mathcal{T}_{U_2}(\rho)$ is equal to the target rotation $e^{i\theta Z}$ followed by a Pauli channel error. \CP{I haven't checked calculations in the remaining of this section and linked appendices}
\begin{lem}[Twirled mixture yields Pauli channel error]
  \label{lem:pauli-channel-error}
  Let $\theta$ be an angle, $U_1, U_2$ be unitaries as in
  \cref{eq:over-under-unitary-approximations} and $p$ be a probability as in \cref{eq:mixing-probability}.
  Then
  \begin{equation}
    \label{eq:pauli-channel-error-claim}
    p\mathcal{T}_{U_1}(\rho) + (1-p)\mathcal{T}_{U_2}(\rho) = \mathcal{E}(\mathcal{Z}_\theta(\rho))
  \end{equation}
  where $\mathcal{E}$ is a qubit Pauli channel
  \begin{equation}
    \mathcal{E}(\rho) = p\mathcal{P}_{r_1, \delta_1}(\rho) + (1-p)\mathcal{P}_{r_2,\delta_2}(\rho)
  \end{equation}
  and
  \begin{equation}
    \mathcal{P}_{r,\delta}(\rho) = r^2\cos^2(\delta)\rho + \frac{1-r^2}{2}(X\rho X + Y\rho Y) + r^2\sin^2(\delta)Z\rho Z.
  \end{equation}
\end{lem}

This lemma can be proved by calculating the four by four process matrices for channels induced by $U_1$, $U_2$, channels $\mathcal{T}_{U_k}$ and $\mathcal{E}$.
Recall that for a qubit channel $\Psi$ the process matrix $\chi$ is given by
$$
    \Psi(\rho)  = \sum_{P,Q \in \{I,X,Y,Z\}} \chi_{P,Q} P \rho Q.
$$
The $\{Z,S\}$ twirl eliminates all but two off-diagonal elements and the mixture with probability $p$ eliminates the remaining off-diagonal elements.
We then see that the process matrix of $\mathcal{E}$ is a diagonal matrix and therefore $\mathcal{E}$ is a Pauli channel.
Finally we apply the following result from \href{https://arxiv.org/pdf/1109.6887.pdf#page=14}{Section V.A} in \cite{Magesan2012} to get a closed form expression for the 
diamond norm distance between $\mathcal{E}$ and the identity channel:
\begin{thm}[Diamond norm distance between Pauli channels]
\label{thm:diamond-distance-beween-pauli-channels}
Suppose $\mathcal{E}_1$,  $\mathcal{E}_2$ are $n$-qubit Pauli channels, that is 
$$
  \mathcal{E}_1(\rho) = \sum_{P \in \{I,X,Y,Z\}^{\otimes n}} q_P P \rho P^\dagger,\,  \mathcal{E}_2(\rho) = \sum_{P \in \{I,X,Y,Z\}^{\otimes n}} r_P P \rho P^\dagger, 
$$
then $\nrm{\mathcal{E}_1-\mathcal{E}_2}_\diamond = \sum_{P \in \{I,X,Y,Z\}^{\otimes n}} |q_P - r_P|$.
\end{thm}

Detailed proofs of \cref{lem:pauli-channel-error} and \cref{thm:unitary-mixture-diamond-distance} are given in \cref{sec:diamond-difference-twirled-mixture}.



\subsection{Mixed fallback approximation \label{sec:fallback-mixing}}
The results of \cite{Campbell2017,Hastings2017} halve the cost of qubit unitary approximation by taking probabilistic mixtures of unitaries.  We show that an additional factor of two improvement in cost can be obtained by taking probabilistic mixtures of \emph{channels}.  In particular, we demonstrate a procedure for mixing fallback protocols (see~\sec{fallback-approximation}).
The basic idea is to apply at random one of two projective rotations, one that approximates $e^{i\theta Z}$ by over-rotation and one that approximates $e^{i\theta Z}$ by under-rotation. If the projective rotation fails, then a corresponding fallback channel is applied.  

\begin{prob}[Diagonal unitary approximation by projective rotation mixing]
    \label{prob:projective-rotation-mixing}
    Given:
    \begin{itemize}
        \item target angle $\theta$, a real number,
        \item success probability $q$, a positive real number between $0$ and $1$,
        \item gate set $G$, a finite set of two by two unitary matrices with determinant one,
        \item accuracy $\varepsilon$, a positive real number.
    \end{itemize}
    Find 
    \begin{itemize}
        \item $G_1,\ldots,G_n$, a sequence of sequences of elements of $G$ and
        \item $p_1,\ldots,p_n$, a probability distribution
    \end{itemize}
    such that 
    \begin{itemize}
        \item $|u_k|^2 \ge q$ for all $k\in [n]$, and
        \item the diamond norm of the channel below is at most $\varepsilon$
        \[
           \sum_{k=1}^n p_k |u_k|^2 \left( \mathcal{Z}_{\mathrm{Arg}(u_k)} - \mathcal{Z}_\theta \right)
        \]
    \end{itemize}
    where $u_k$ is the top-left entry of the unitary $\at{\begin{smallmatrix} u_k & -v_k^\ast \\ v_k & u_k^\ast \end{smallmatrix}}$ corresponding to sequence $G_k$.
\end{prob}

In analogy to~\problem{diagonal-approximation-by-unitary-mixing}, this problem generalizes the projective rotation approximation problem (\problem{projective-approximation}) by allowing multiple projective rotation circuits in convex combination.  
Note that the elements of the probability distribution $p_1,\ldots,p_n$ are distinct from the success probabilities of the projective rotations.

In the analysis of the fall-back protocol in \cref{prop:fallback-approximation}  we considered only unitary fallback steps.  
We now consider fallbacks that are probabilistic mixtures of unitaries. The channel $\mathcal{F}$ for a fallback protocol has the form
\begin{equation}
  \mathcal{F}(\rho) = q\mathcal{Z}_{\theta_0}(\rho) + (1-q)\mathcal{B'}(\rho)
\end{equation}
where $\mathcal{B'}$ is the composition of the failure rotation $e^{i\theta_1 Z}$ and the fallback step $\mathcal{B}$ and $q$ is the probability of success (measurement outcome of zero).

The following theorem provides a simple closed form bound for the diamond distance of a mixture of fallback protocols, similar to the expression obtained from \cref{thm:unitary-mixture-diamond-distance} for unitary mixtures.  The proof is provided in \cref{sec:diamond-distance-fallback-mixture}.
\begin{thm}[Diamond distance of a fallback mixture]
  \label{thm:fallback-mixture-diamond-distance}
  Let $\theta$ be an angle and let fallback channels
  \begin{equation}
    \mathcal{F}_1(\rho) = q_1\mathcal{Z}_{\theta + \delta_1}(\rho) + (1-q_1)\mathcal{B}_1(\rho)
  \end{equation}
  \begin{equation}
    \mathcal{F}_2(\rho) = q_2\mathcal{Z}_{\theta + \delta_2}(\rho) + (1-q_2)\mathcal{B}_2(\rho)
  \end{equation}
  where $\sin(\delta_1) \leq 0 \leq \sin(\delta_2)$.
  Define probability
  \begin{equation}
    p = \frac{q_2\sin(2\delta_2)}{q_2\sin(2\delta_2) - q_1\sin(2\delta_1)}.
  \end{equation}
  Then
  \begin{equation} \label{eq:mixed-projective-rotation-diamond-distance}
        \nrm{  p q_1 \left( \mathcal{Z}_{\theta + \delta_1} - \mathcal{Z}_\theta \right) + (1-p) q_2 \left( \mathcal{Z}_{\theta + \delta_2} - \mathcal{Z}_\theta \right) }_\diamond = 
        2\left(pq_1\sin^2(\delta_1) + (1-p)q_2\sin^2(\delta_2)\right) 
  \end{equation}
  and the total accuracy of the mixed fall-back approximation protocol is
  \begin{equation}\begin{aligned}
      \label{eq:fallback-mixture-bound}
      ||p\mathcal{F}_1 + (1-p)\mathcal{F}_2 - \mathcal{Z}_\theta||_\diamond
      \leq & \, 2\left(pq_1\sin^2(\delta_1) + (1-p)q_2\sin^2(\delta_2)\right) \\
           & + p(1-q_1)||\mathcal{B}_1 - \mathcal{Z}_\theta||_\diamond + (1-p)(1-q_2)||\mathcal{B}_2 - \mathcal{Z}_\theta||_\diamond.
    \end{aligned}\end{equation}
\end{thm}

The goal of a synthesis algorithm then is to bound \cref{eq:fallback-mixture-bound}%
\footnote{When the composition $\mathcal{B}_k \mathcal{Z}_{-\theta}$ is a Pauli channel, we can replace inequality in \cref{eq:fallback-mixture-bound} by an exact value}
by an approximation error $\varepsilon$. We have some flexibility in bounding the accuracy of the components of the two fallback protocols.
We may set the accuracy of each term separately
\begin{equation}\begin{aligned}
  \label{eq:fallback-mixing-terms}
    2\left( pq_1\sin^2(\delta_1) + (1-p)q_2\sin^2(\delta_2) \right) & = \varepsilon_1   \\
    p(1-q_1)||\mathcal{B}_1 - \mathcal{Z}_{\theta}||_\diamond       & = \varepsilon_2   \\
    (1-p)(1-q_2)||\mathcal{B}_2-\mathcal{Z}_{\theta}||_\diamond       & = \varepsilon_3   \\
    \varepsilon_1 + \varepsilon_2 + \varepsilon_3 & \leq \varepsilon.
  \end{aligned}\end{equation}

The first condition is ensured by solving \cref{prob:projective-rotation-mixing}. 
The second two conditions are ensured by solving the mixed diagonal approximation problems.
Note that for the two fallback terms $||\mathcal{B}_1-\mathcal{Z}_\theta||_\diamond$ and $||\mathcal{B}_2-\mathcal{Z}_\theta||_\diamond$, the accuracy is scaled by $1-q_1$ and $1-q_2$ thereby reducing the fallback step approximation T-count on average by $1.5 \log_2(1/(p(1-q_1)))$ and $1.5 \log_2(1/(1-p)/(1-q_2))$ when using mixed diagonal approximation with Clifford+$T$ gate set.
This is in comparison to the cost of mixed diagonal approximation with Clifford+$T$ gate set of accuracy $\varepsilon$.
As in previous sections, \cref{prob:projective-rotation-mixing} is solved by finding gate sequences with certain constraints on the top-left entry of the unitaries 
they compute.

\begin{prop}[Projective rotation mixing approximation condition]\label{prop:fallback-mixing-condition}  
  Suppose that we are given two sequences $G = g_1,\dots,g_n$ and $ H = h_1,\dots,h_m$ from a gate set $G$.
  Define
  \[
        g_1\dots g_n = \at{\begin{array}{cc} u_1 & -v_1^\ast \\ v_1 & u_1^\ast \end{array}},
        h_1\dots h_n = \at{\begin{array}{cc} u_2 & -v_2^\ast \\ v_2 & u_2^\ast \end{array}}.
  \]
  Suppose that for angle $\theta$ and accuracy $\varepsilon$
  \begin{itemize}
      \item $u_1$ satisfies $\abs{u_1} \ge \sqrt{q}$, $-\sqrt{\varepsilon/2} \leq \sin\at{\mathrm{Arg}(u_1) - \theta} \leq 0$, and
      \item $u_2$ satisfies $\abs{u_2} \ge \sqrt{q}$, $0 \leq \sin\at{\mathrm{Arg}(u_2) - \theta} \leq \sqrt{\varepsilon/2}$.
  \end{itemize}
  Then sequence $G,H$ and probability distribution $p, 1-p$ for 
  \begin{equation}
    p = \frac{q_2\sin(2\delta_2)}{q_2\sin(2\delta_2) - q_1\sin(2\delta_1)}, \text{ where } q_k = |u_k|^2, \delta_k = \mathrm{Arg}(u_k)-\theta
  \end{equation}
    is a solution to the projective rotation mixing approximation~\problem{projective-rotation-mixing}.
    For a geometric interpretation of the constraints on $u_1, u_2$ see \fig{fallback-mixing-condition}.
\end{prop}
\newpage
\begin{proof}
The conditions $|u_1| \geq \sqrt{q}$, $|u_2| \geq \sqrt{q}$ trivially ensure success probability conditions of~\problem{projective-rotation-mixing}.
We need to bound diamond distance of the channel
$$
p |u_1|^2 \left( \mathcal{Z}_{\mathrm{Arg}(u_1)} - \mathcal{Z}_\theta \right) + (1-p) |u_2|^2 \left( \mathcal{Z}_{\mathrm{Arg}(u_2)} - \mathcal{Z}_\theta \right).
$$
By~\theo{fallback-mixture-diamond-distance} it is equal to
\begin{equation}\begin{aligned}
  & 2\left(p|u_1|^2 \sin^2(\delta_1) + (1-p)|u_2|^2 \sin^2(\delta_2) \right) \leq 2\left(p\sin^2(\delta_1) + (1-p)\sin^2(\delta_2) \right)  \\
  &\leq \varepsilon
\end{aligned}\end{equation}
\end{proof}

\begin{figure}
    \centering
    \includegraphics[scale=0.4]{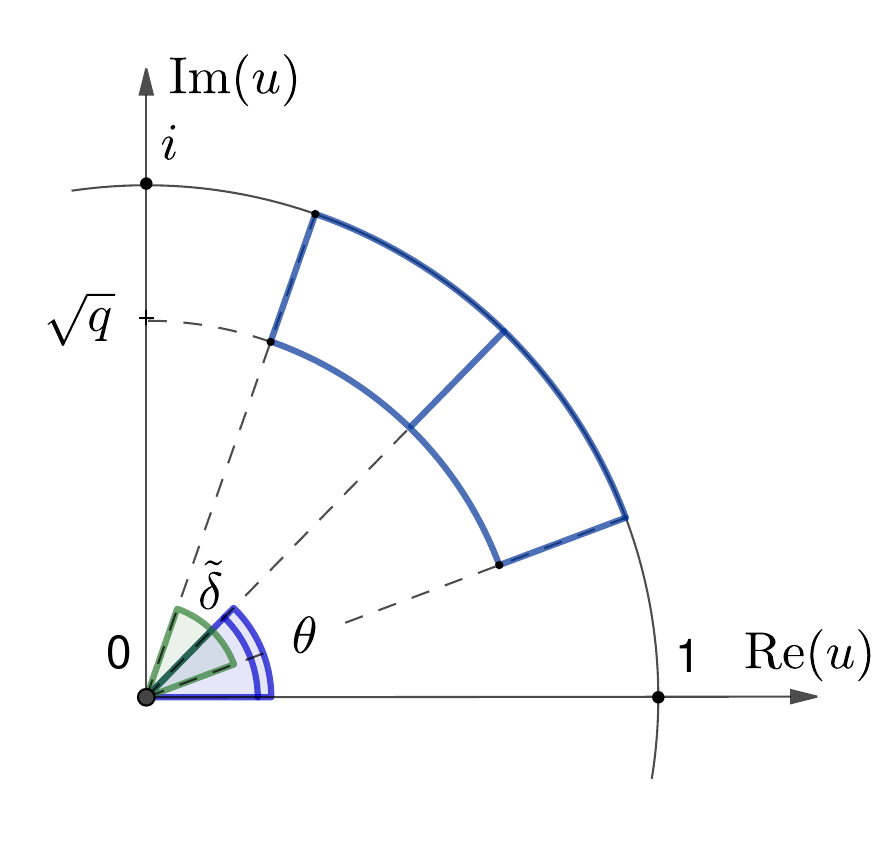}
    \caption{
        \label{fig:fallback-mixing-condition}
        A geometric interpretation of constraints on the projective rotations given by~\propos{fallback-mixing-condition}.
        The sector with blue boundary contains complex numbers $u$ that satisfy constraints $\mathrm{Arg}\at{u} \in  [\theta - \tilde\delta, \theta + \tilde\delta]$ and $\abs{u} \ge \sqrt{q}$, where $\tilde\delta = \arcsin(\sqrt{\varepsilon/2})$.
        This sector is split into two parts, an over-rotation sector for one projective rotation and an under-rotation sector for the other projective rotation.
    }
\end{figure}

The conditions 
$|\sin(\delta_k)| \leq \sqrt{\varepsilon/2}$
imposed by~\propos{fallback-mixing-condition} are quadratically looser than the equivalent condition for projective rotations without mixing (\cref{prop:segment}) which requires
$|\sin(\delta)| \leq \varepsilon/2$.
When combined with unitary mixing for the fallback step, this yields a quadratic improvement in $\varepsilon$ for the entire fallback protocol.
The gate cost of a fallback protocol scales as $C\log_2(1/\varepsilon) + O(1)$.
Thus the quadratic improvement in $\varepsilon$ translates to a roughly two times savings in expected gate cost over conventional fallback protocols.
For the more detailed cost comparisons see \cref{tab:approximation-cost-scaling}.

For reasonable ranges of $\varepsilon$ the $\log(1/\varepsilon)$ term is well below $100$, making additive constants and higher order terms an important consideration.  In that sense, solutions obtained by~\propos{fallback-mixing-condition} are sub-optimal.  
The overall cost can be optimized by a more careful assignment of $\varepsilon_1, \varepsilon_2$ and $\varepsilon_3$ in \cref{eq:fallback-mixing-terms}. This is discussed further in \cref{appendix:optimal-mixing-solutions}.

\subsection{Mixed magnitude approximation \label{sec:magnitude-mixing}}
\label{sec:mixed-magnitude-approximation}

We have shown that taking a probabilistic mixture of channels leads to improvement in accuracy for diagonal approximations with and without the fallback protocol. It is natural to then question whether a similar improvement can be achieved for general unitary approximation. We show that this is indeed the case, following the same strategy of finding approximations corresponding to under- and over- rotations of a target angle. 

We begin by defining the following problem for approximating an arbitrary X-rotation up to phases.

\begin{prob}[Magnitude approximation by mixing]
    \label{prob:magnitude-mixing}
    Given:
    \begin{itemize}
        \item target angle $\theta$,
        \item gate set $G$, a finite set of two by two unitary matrices with determinant one,
        \item accuracy $\varepsilon$, a positive real number.
    \end{itemize}
    Find 
    \begin{itemize}
        \item $G_1,\ldots,G_n$, a sequence of sequences of elements of $G$ and
        \item $p_1,\ldots,p_n$, a probability distribution on these sequences
    \end{itemize}
    such that
        \[
          \left\Vert \sum_{k=1}^n p_k \mathcal{X}_{\arccos(\abs{u_k})}-\mathcal{X}_\theta \right\Vert_\diamond \leq \varepsilon.
        \]

    where $u_k$ is the top-left entry of the matrix corresponding to sequence $G_k$ and $\mathcal{X}_\theta$ is the channel induced by $e^{i \theta X}$.
\end{prob}

As in \problem{diagonal-approximation-by-unitary-mixing} and \problem{projective-rotation-mixing}, \problem{magnitude-mixing} allows for a solution comprising a probabilistic mixture of channels. We can further assume without loss of generality that $\abs{\cos(\theta)}\ge1/\sqrt{2}$, by the following remark.

\begin{rem}
Let $g_1\dots g_n$ be a solution to \problem{magnitude-mixing} for target angle $\theta$. Then $iX \cdot g_n^\dagger \dots g_1^\dagger$ is a solution for target angle $\frac{\pi}{2}-\theta$, since $iX\cdot g_n^\dagger \dots g_1^\dagger =  iX e^{-i\phi_2 Z} (-iX) iX e^{-i\theta X} e^{-i\phi_1 Z} = e^{i\phi_2 Z} e^{i(\pi/2-\theta)X} e^{i\phi_1 Z}.$
\end{rem}

Hence, if $\abs{\cos(\theta)}<1/\sqrt{2}$ we can simply apply magnitude approximation to $\pi/2 - \theta$, noting that $\cos(\pi/2-\theta)=\sin(\theta) \ge 1/\sqrt{2}$.

 The following proposition shows that the error bound on the diamond norm in \problem{magnitude-mixing} induces a constraint on the top-left entries of the matrices corresponding to sequences $G_k$. Our approach is to again find X-approximations corresponding to under- and over-rotations of the angle $\delta$.

\begin{prop}[Mixed magnitude approximation condition]\label{prop:magnitude-mixing-condition}  
  Suppose we are given sequences $g_1,\dots,g_n$ and $h_1,\dots,h_m$ of gates from a gate set.
  Define complex numbers $u_k, v_k$
  \[
        g_1\dots g_n = \at{\begin{array}{cc} u_1 & -v_1^\ast \\ v_1 & u_1^\ast \end{array}},
        h_1\dots h_m = \at{\begin{array}{cc} u_2 & -v_2^\ast \\ v_2 & u_2^\ast \end{array}}
  \]
  Suppose that for accuracy $\varepsilon$ and target angle $\theta$        \begin{itemize}
      \item $\abs{u_1}\in \{\cos(\theta''):\theta''\in[0,\pi/2], 0\le\theta-\theta''\le\arcsin\sqrt{\varepsilon/2}\}$, and
      \item $\abs{u_2}\in \{\cos(\theta''):\theta''\in[0,\pi/2], 0\le\theta''-\theta\le\arcsin\sqrt{\varepsilon/2}\}$.
    \end{itemize}
  Define 
  $$
  p = \frac{\sin(2\delta_2)}{\sin(2\delta_2) - \sin(2\delta_1)}, \text{ where } \delta_k = \mathrm{arccos}(u_k) - \theta
  $$
  The the sequences  $g_1,\dots,g_n$ and $h_1,\dots,h_m$ and probability distribution $p, 1-p$ is a
  a solution to the magnitude mixing approximation~\problem{magnitude-mixing}  with $n=2$, accuracy $\varepsilon$ and target angle $\theta$.
  For a geometric interpretation of the constraints on $|u_k|$ see \fig{magnitude-mixing-condition}.
\end{prop}

\begin{proof}
Let $g_1\dots g_n= e^{i\phi_1Z}e^{i\theta_u X}e^{i\phi_2Z}$ with $\theta_u = \theta +\delta_1$ and
$h_1\dots h_m= e^{i\psi_1Z}e^{i\theta_oX}e^{i\psi_2Z}$ with $\theta_u = \theta +\delta_2$. Then 
\begin{equation}\begin{aligned}
\nrm{p\mathcal{X}_{\theta_u} +(1-p)\mathcal{X}_{\theta_o} - \mathcal{X}_\theta}_\diamond = \nrm{p\mathcal{Z}_{\theta_{u}}+(1-p)\mathcal{Z}_{\theta_{o}}-\mathcal{Z}_{\theta}}_\diamond.
\end{aligned}
\end{equation}
using the identity $He^{i\theta Z}H = e^{i\theta X}$ and unitary invariance of the diamond norm~(\cref{prop:diamond-norm-unitary-invariance}). 
The norm on the right hand side and the expression for $p$ are of the form required for \theo{unitary-mixture-diamond-distance}, so we can conclude
$\nrm{p\mathcal{X}_{\theta_u} +(1-p)\mathcal{X}_{\theta_o} - \mathcal{X}_\theta}_\diamond \le \varepsilon$
when $2p(\sin^2(\delta_1))+2(1-p)(\sin^2(\delta_2)) \le \varepsilon$. 
It remain to show that $\sin^2(\delta_k) \le \varepsilon/2$.

Consider the under-rotated case. By definition $\abs{u_1}=\cos(\theta_u)$ and so $\theta-\theta_u\le\arcsin\sqrt{\varepsilon/2}$. So we have
$
\sin^2(\delta_1) \le \varepsilon/2
$
as required, and analogously for $\abs{u_2}$.
\end{proof}

\begin{figure}
    \centering
    
    \includegraphics[scale=0.4]{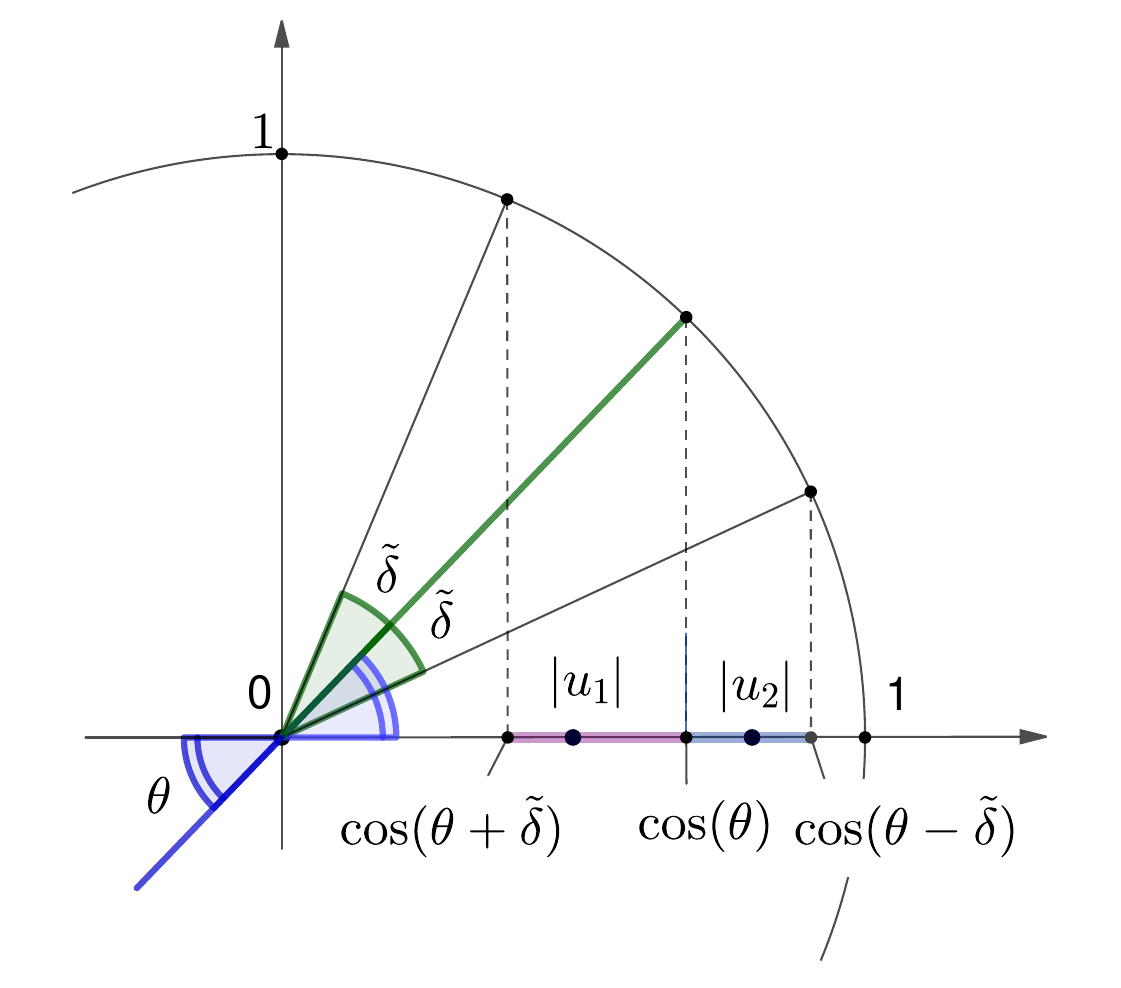}
    \caption{
        \label{fig:magnitude-mixing-condition}
        A geometric interpretation of constraints given by~\propos{magnitude-mixing-condition}.
       Absolute values $\abs{u}$ belonging to the interval $\{\cos(\delta''):\delta'' \in [0,\frac{\pi}{2}], \abs{\delta'' - \theta}\le \sqrt{\tilde\delta}\}$, for $\tilde\delta = \arcsin\sqrt{\varepsilon/2}$ are shown on the vertical axis.  This interval is split into two parts, an over-rotation interval (blue) and an under-rotation interval (purple).
    }
\end{figure}

In comparison to \propos{absolute-value-constraint}, the constraint on the top-left matrix entries in \propos{magnitude-mixing-condition} is quadratically looser, thus admitting a greater possible number of candidate values. Analogously to the unmixed case, mixed magnitude approximations can be extended to approximations of arbitrary unitaries.

\begin{prop}[General unitary approximation with mixing]\label{prop:general-mixing-condition}  
Let $U=e^{i\alpha Z}e^{i\theta X}e^{i\beta Z}$ be an arbitrary unitary in $SU(2)$. 
Let  $g_1,\dots,g_n$ and $h_1,\dots,h_m$ be sequences from a gate set $G$ that satisfy the constraints of \propos{magnitude-mixing-condition}.
Let probability $p$ be defined as in \propos{magnitude-mixing-condition}.
Define angles $\phi_k, \psi_k$ from equations  $g_1\dots g_n= e^{i\phi_1Z}e^{i\theta_u X}e^{i\phi_2Z}$ and
$h_1\dots h_m= e^{i\psi_1Z}e^{i\theta_oX}e^{i\psi_2Z}$.
Suppose $\Phi_1,\Phi_2,\Psi_1$ and $\Psi_2$ are $\varepsilon$-approximations of $\mathcal{Z}_{\alpha-\phi_1},\mathcal{Z}_{\beta-\phi_2}, \mathcal{Z}_{\alpha-\psi_1}$ and $\mathcal{Z}_{\beta-\psi_2}$, respectively.
  Then
  \[ \nrm{p\Phi_1\mathcal{G}\Phi_2 + (1-p)\Psi_1\mathcal{H}\Psi_2 - \mathcal{U}}_\diamond \le 3\varepsilon,
    \]
    where $\mathcal{G}$ and $\mathcal{H}$ are channels induced by products $g_1\cdots g_n$ and $h_1\cdots h_m$.
\end{prop}
\begin{proof}
 Using the Euler decomposition of $U$ we have $\mathcal{U} = \mathcal{Z}_{\alpha}\mathcal{X}_\theta\mathcal{Z}_{\beta}$ and so	 
\[ 
\nrm{p\Phi_1\mathcal{G}\Phi_2 + (1-p)\Psi_1\mathcal{H}\Psi_2 - \mathcal{U}}_\diamond = 
\nrm{p\Phi_1\mathcal{G}\Phi_2 + (1-p)\Psi_1\mathcal{H}\Psi_2 - \mathcal{Z}_{\alpha}\mathcal{X}_\theta\mathcal{Z}_{\beta}}_\diamond.
\] 
Applying the triangle inequality, we bound this norm from above by 
\[
p\nrm{\Phi_1\mathcal{G}\Phi_2-\mathcal{Z}_{\alpha}\mathcal{X}_{\theta_u}\mathcal{Z}_{\beta}}_\diamond
+ 
(1-p)\nrm{\Psi_1\mathcal{H}\Psi_2-\mathcal{Z}_{\alpha}\mathcal{X}_{\theta_o}\mathcal{Z}_{\beta}}_\diamond 
+ \nrm{\mathcal{Z}_\alpha(p\mathcal{X}_{\theta_u}+(1-p)\mathcal{X}_{\theta_o} - \mathcal{X}_\theta)\mathcal{Z}_\beta}_\diamond.\]

Now, using unitary invariance of the diamond norm with $e^{-i\alpha Z}$ and $e^{-i\beta Z}$, we can simplify the three terms in the expression above and bound each by an accuracy measure. Concretely, we obtain the following set of equations
\begin{equation*}
\begin{aligned}
p\nrm{\Phi_1\mathcal{G}\Phi_2-\mathcal{Z}_{\alpha}\mathcal{X}_{\theta_u}\mathcal{Z}_{\beta}}_\diamond 
=
p\nrm{\Phi_1\underline{\mathcal{G}}\Phi_2-\mathcal{Z}_{\alpha}\mathcal{Z}_{-\phi_1}\underline{\mathcal{Z}_{\phi_1}\mathcal{X}_{\theta_u}\mathcal{Z}_{\phi_2}}\mathcal{Z}_{-\phi_2}\mathcal{Z}_{\beta}}_\diamond
= \varepsilon_1
\\
(1-p)\nrm{\Psi_1\mathcal{H}\Psi_2-\mathcal{Z}_{\alpha}\mathcal{X}_{\theta_o}\mathcal{Z}_{\beta}}_\diamond
=
(1-p)\nrm{\Psi_1\underline{\mathcal{H}}\Psi_2-\mathcal{Z}_{\alpha}\mathcal{Z}_{-\psi_1}\underline{\mathcal{Z}_{\psi_1}\mathcal{X}_{\theta_o}\mathcal{Z}_{\psi_2}}\mathcal{Z}_{-\psi_2}\mathcal{Z}_{\beta}}_\diamond 
= \varepsilon_2
\\
\nrm{\mathcal{Z}_\alpha(p\mathcal{X}_{\theta_u}+(1-p)\mathcal{X}_{\theta_o} - \mathcal{X}_\theta)\mathcal{Z}_\beta}_\diamond 
= \varepsilon_3.
\\
\end{aligned}
\end{equation*}
such that the claim holds if $\varepsilon_1+\varepsilon_2+\varepsilon_3\le 3\varepsilon.$
Consider the first norm $\nrm{\Phi_1\mathcal{G}\Phi_2-\mathcal{Z}_{\alpha}\mathcal{X}_{\theta_u}\mathcal{Z}_{\beta}}_\diamond$ and apply the chain rule
\begin{equation}
\begin{aligned}
\nrm{\Phi_1\underline{\mathcal{G}}\Phi_2-\mathcal{Z}_{\alpha}\mathcal{Z}_{-\phi_1}\underline{\mathcal{Z}_{\phi_1}\mathcal{X}_{\theta_u}\mathcal{Z}_{\phi_2}}\mathcal{Z}_{-\phi_2}\mathcal{Z}_{\beta}}_\diamond
\le 
\nrm{\Phi_1 - \mathcal{Z}_{\alpha-\phi_1}}_\diamond 
+ \\
\nrm{\mathcal{G} - \mathcal{Z}_{\phi_1}\mathcal{X}_{\theta_u}\mathcal{Z}_{\phi_2}}_\diamond 
+ 
\nrm{\Phi_2 - \mathcal{Z}_{\beta-\phi_2}}_\diamond
\le
2\varepsilon
\end{aligned}
\end{equation}
The same argument applies \emph{mutatis mutandis} to $\nrm{\Psi_1\mathcal{H}\Psi_2-\mathcal{Z}_{\alpha}\mathcal{X}_{\theta_o}\mathcal{Z}_{\beta}}_\diamond$.
Since the sequences $g_1,\dots,g_n$ and $h_1,\dots,h_m$ and $p$ satisfy \propos{magnitude-mixing-condition} we also have 
\[
\nrm{\mathcal{Z}_\alpha(p\mathcal{X}_{\theta_u}+(1-p)\mathcal{X}_{\theta_o} - \mathcal{X}_\theta)\mathcal{Z}_\beta}_\diamond = \nrm{p\mathcal{X}_{\theta_u}+(1-p)\mathcal{X}_{\theta_o} - \mathcal{X}_\theta)} \le \varepsilon.
\]
Therefore $\varepsilon_1 + \varepsilon_2 + \varepsilon_3 \leq 2\varepsilon p + 2\varepsilon(1-p) + \varepsilon = 3\varepsilon$.

\end{proof}

\subsection{Geometric interpretations}
\label{sec:geometric-interpretations}
In the sections above, we defined two methods for approximating diagonal unitaries: by direct unitary sequences (\problem{diagonal-approximation}), or by fallback protocols (\cref{prob:projective-approximation}).  Both of these methods can be extended by using probabilistic mixtures (\problem{diagonal-approximation-by-unitary-mixing} and \problem{projective-rotation-mixing}).  
Each of these problems involves finding one or more sequences of gates that induce  two-by-two unitary matrices of the form 
$\left(\begin{smallmatrix}u &-v^*\\ v & u^* \end{smallmatrix}\right)$.
In each case, solutions can be described by conditions on the top-left entry $u$.  See~\propos{segment-condition},~\propos{segment},~\propos{diagonal-mixing-condition}, and~\propos{fallback-mixing-condition}.
Those conditions can be illustrated geometrically by regions on the complex plane:~\fig{diagonal-condition} and~\fig{fallback-condition}.

\fig{diagonal-approximation-regions} shows these regions overlayed one on top of another.
This geometric illustration makes clear the progressive increase in solution space going from unitary approximation, to fallback approximation and then to probabilistic mixtures. The region areas for each approximation problem can be quickly computed using basic formulas for the areas of sectors and triangles. For instance, for diagonal approximation without mixing the region area is given by $$\delta - \frac{1}{2}\sin(2\delta)$$ where $\delta$ is the angle subtending the minimal sector containing the region. The region areas are given to leading order of $\varepsilon$ in \tab{approx-region-areas}.

The projective approximation region encloses the unitary approximation region, provided that the probability of success $q$ satisfies 
a modest $q \leq 1 - \varepsilon^2/4$.  
Loosely speaking, the condition $q = 1-\varepsilon^2/4$ can be interpreted as the point at which the projection failure may be treated deterministically as an approximation error and no longer needs a fallback step.

Except for large values of $\varepsilon$, the unitary mixture region also encloses the (non-mixing) unitary approximation region.  Finally, the projective mixing approximation region encloses all of the other regions, provided that $q \leq 1-\varepsilon/2$.

Indeed, for the chosen value of $\varepsilon=0.1$, the illustration in~\fig{diagonal-approximation-regions} under-represents the relative difference in region sizes.  Practical values of $\varepsilon$ are typically several orders of magnitude smaller, for which the relative difference in region sizes is dramatically larger.  
\fig{diagonal-approximation-region-areas} shows the areas of each of the approximation regions as a function of $\varepsilon$ and $q$.

\begin{figure}
  \begin{subfigure}{0.49\textwidth}
    \includegraphics[scale=0.25]{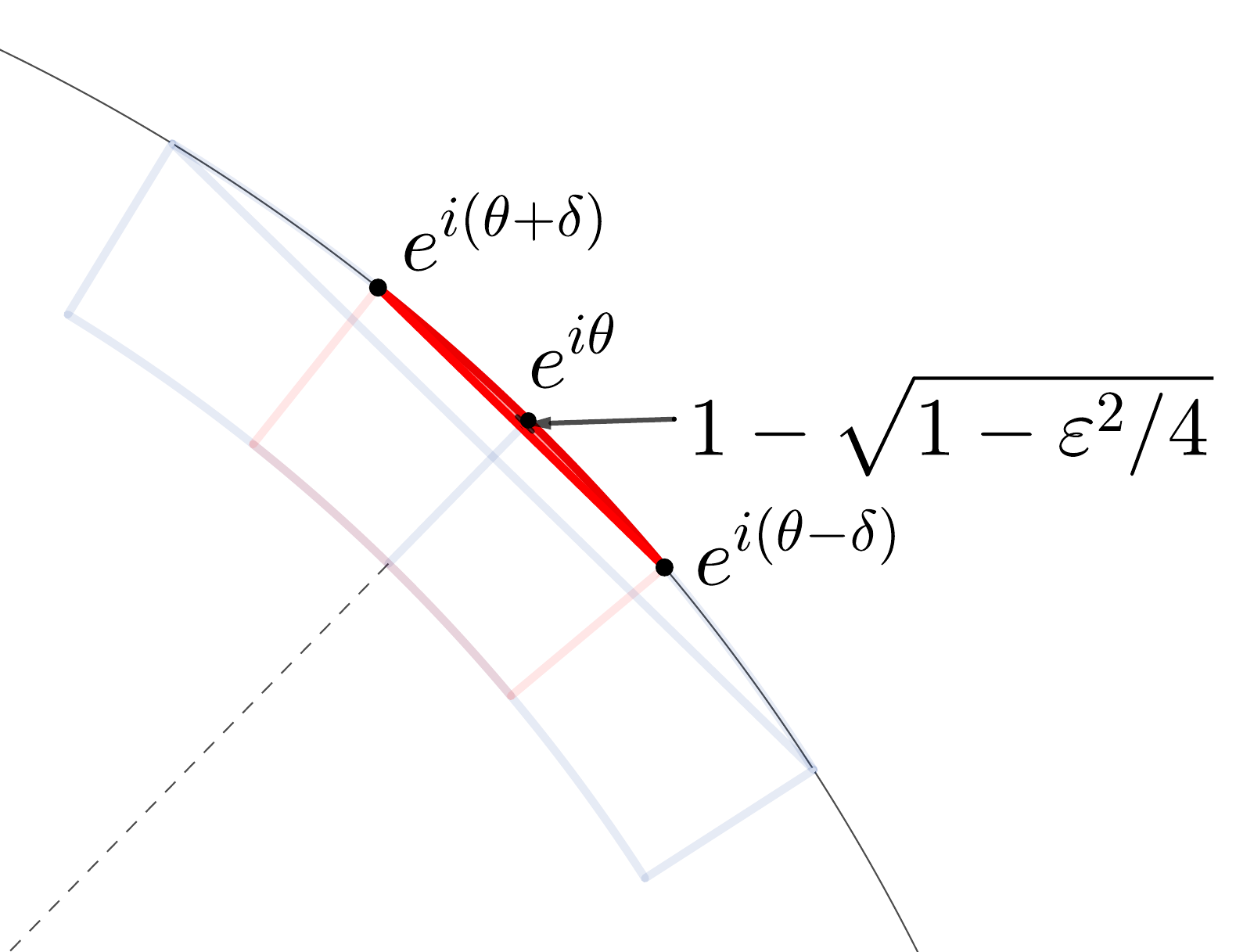}
    \caption{Unitary region (\propos{segment-condition}) with area in $O(\varepsilon^3)$}
  \end{subfigure}
  \hfill
  \begin{subfigure}{0.49\textwidth}
    \includegraphics[scale=0.25]{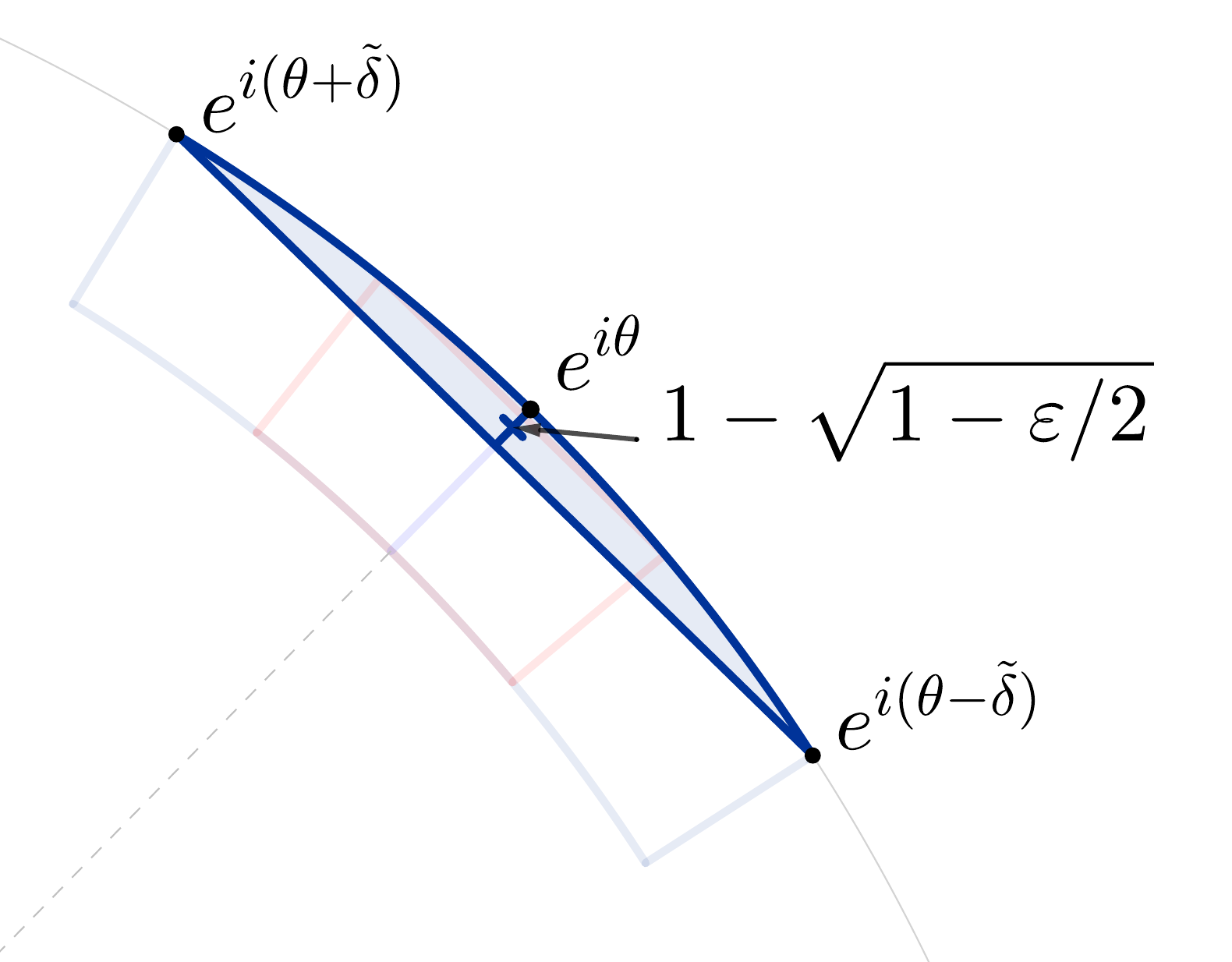}
    \caption{Unitary mixing regions (\propos{diagonal-mixing-condition}) with area in $O(\varepsilon^{3/2})$}
  \end{subfigure}
  \begin{subfigure}{0.49\textwidth}
    \includegraphics[scale=0.25]{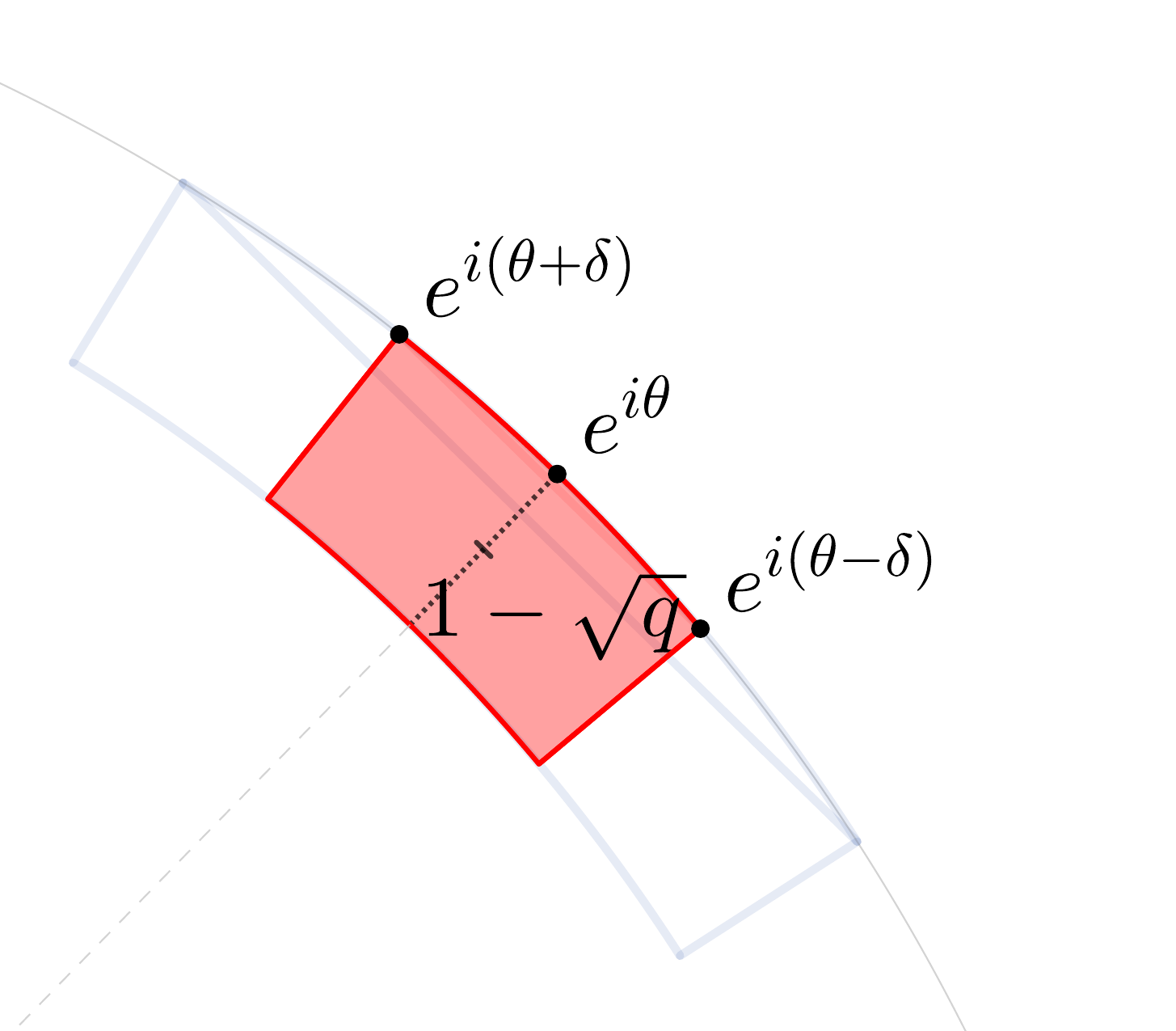}
    \caption{Projective rotation region (\propos{segment}) with area in $O(\varepsilon)$}
  \end{subfigure}
  \hfill
  \begin{subfigure}{0.49\textwidth}
    \includegraphics[scale=0.25]{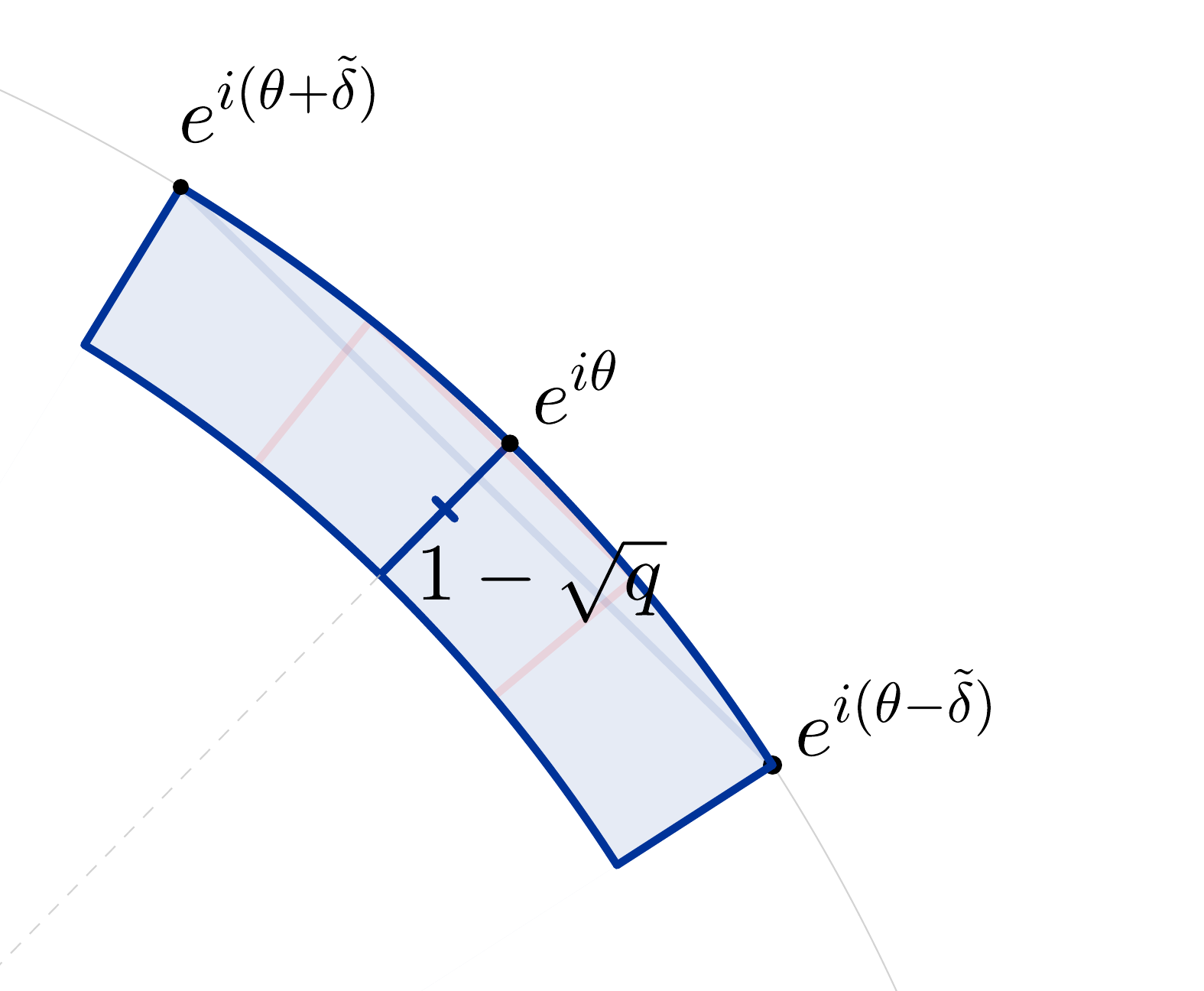}
    \caption{Projective rotation mixing regions (\propos{fallback-mixing-condition}) with area in $O(\varepsilon^{1/2})$}
  \end{subfigure}
  \caption{
    \label{fig:diagonal-approximation-regions}
    Approximation regions for target diagonal unitary $e^{i\theta Z}$. The figures above show close-ups of a section of the unit circle on the complex plane. Each colored highlight indicates a region of valid solutions for the corresponding approximation problem.  For illustration, we have used (an impractical) approximation accuracy $\varepsilon=0.1$ and projective rotation success probability $q=0.9$.  
    The unitary (a) and projective rotation (c) regions (without mixing) shown in red each subtend an angle of $2\delta = 2\arcsin(\varepsilon/2)$.
    The unitary mixing (b) and projective rotation mixing (d) regions shown in blue each subtend an angle of $2\tilde\delta = 2\arcsin(\sqrt{\varepsilon/2})$.  
    The unitary mixing regions (b) fully encompasses the unitary region (a).  Likewise, the projective rotation mixing region (d) fully encompasses the projective rotation region (c).
    For $q \leq 1 - \varepsilon^2/4$, the projective rotation region (c) encompasses the unitary region (a).  For $q \leq 1-\varepsilon/2$ the mixed projective regions encompasses all other regions.
  }
\end{figure}

\begin{figure}
  \centering
  \includegraphics[width=\textwidth]{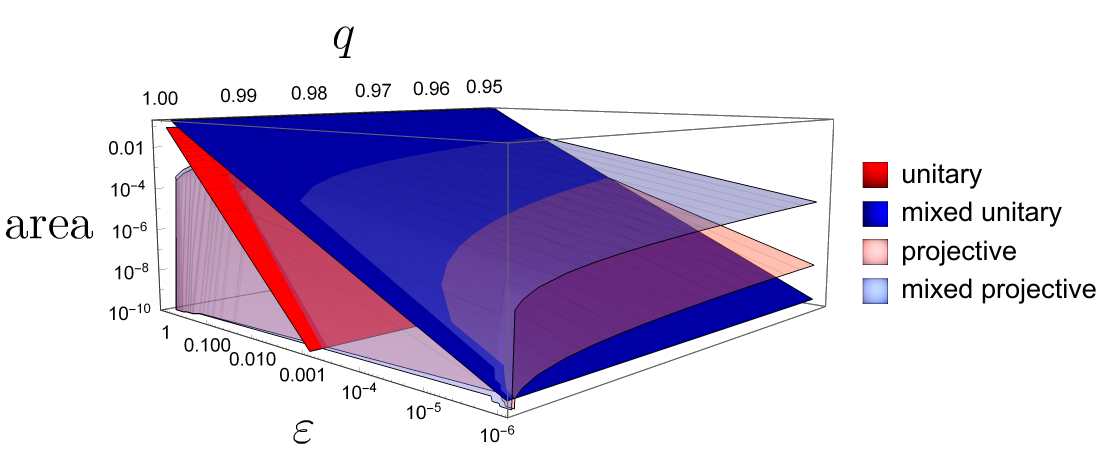}
  \caption{\label{fig:diagonal-approximation-region-areas}
    Areas of the solution regions prescribed by~\propos{segment-condition} (unitary),~\propos{diagonal-mixing-condition} (mixed unitary),~\propos{segment} (projective) and~\propos{fallback-mixing-condition} (mixed projective).  Each curve shows the area of the region in the complex plane as a function of approximation accuracy $\varepsilon$ and success probability $q$. 
    The unitary and projective areas, highlighted in red, scale quadratically with $\varepsilon$. The mixed unitary and mixed projective areas, highlighted in blue, scale linearly with $\varepsilon$.  
    For all values of $\varepsilon$ and $q$ the mixed unitary (resp. projective) region has larger area than the corresponding non-mixed unitary (resp. projective) region, thereby admitting more candidate solutions.
    Except for large values of $\varepsilon$ and $q$, the projective and mixed projective regions have larger area then their respective unitary and mixed unitary regions.
  }
\end{figure}

\begin{table}
\caption[Approximation region areas with and without mixing]{
\label{tab:approx-region-areas}
Region areas for unitary approximation problems with and without mixing. 
The geometric regions associated with each problem are illustrated in \cref{fig:diagonal-approximation-regions}. The diagonal and projective approximation problems result in two-dimensional regions, while the magnitude approximation problem results in a one-dimensional interval. In the latter case, \emph{Region area} refers to the length of the interval. In all cases, it can be seen that the mixed version of the problem corresponds to a larger approximation region.
}
\begin{center}
{
\scriptsize
\setlength{\tabcolsep}{0.5em}

\begin{tabular}{|c|c|c|}
\hline 
Approximation Problem & Region area (big-O) & Region area (with pre-factors) \tabularnewline

\hline 
\hline 
Diagonal unitary approximation&$O(\varepsilon^3)$&$\varepsilon^3/12+O(\varepsilon^5)$\tabularnewline
Mixed diagonal approximation&$O(\varepsilon^{3/2})$&$(2/3)(\varepsilon/2)^{3/2} + O(\varepsilon^{5/2})$\tabularnewline
Projective rotation approximation&$O(\varepsilon)$&$(1-q)\varepsilon/2 +O(\varepsilon^3)$\tabularnewline
Mixed projective approximation&$O(\varepsilon^{1/2})$&$(1-q)(\varepsilon/2)^{1/2} +O(\varepsilon^{3/2})$\tabularnewline
Magnitude approximation &$O(\varepsilon)$ &$ \varepsilon/\sqrt{2}+O(\varepsilon^3)$
\tabularnewline
Mixed magnitude approximation&$O(\varepsilon^{1/2})$&$\varepsilon^{1/2}+O(\varepsilon^{3/2})$\tabularnewline\hline

\end{tabular}

}
\end{center}
\end{table}

\section{Solutions to approximation problems for common gate sets}\label{sec:approx-solutions}

In Section~\ref{sec:approximation-problems}, we related solutions to various approximation problems to specific geometric regions.  
In this section, we specialize these results to a few specific gate sets. Unitaries that can be synthesized exactly with these gate sets define discrete subsets within the above convex bodies, and this naturally leads to an enumeration strategy to solve approximation problems.
We first describe this approach for the V basis, the Clifford$+T$ basis and the Clifford$+\sqrt{T}$ basis, and we then generalize it to a family of gate sets introduced by \cite{Kliuchnikov2015b} for diagonal approximation problems. 

Throughout this section, we introduce the following notation to highlight common methodology across our three illustrative examples. We use $L$ to denote the field in which the entries of the unitaries defining a gate set lie, and $O_L$ to denote the integer ring of $L$.
We associate a \emph{gate set determinant} $\ell \in K$ 
to each gate set, such that any element generated by the gate set can be written as $\frac{1}{\sqrt{\ell^N}}\left(\begin{smallmatrix}u &-v^*\\
     v & u^*
 \end{smallmatrix}\right)$ with $u,v\in O_L$ and $N\in\z$. 
The determinant of a unitary with elements in $L$ lies in the maximal totally real subfield, which we denote by $K$. We have $L=K(i)$, where $i^2=-1$. The norm of an element in $L$ is a mapping from $L$ to $K$ defined by taking the product of an element of $L$ with its complex conjugate.  We denote the ring of integers of $K$ by $O_K$. \change{Finally, we make reference to objects known as orders. An order in a finite-dimensional algebra $M$ over $\mathbb{Q}$ is a subring of $M$ that is also a $\z$-lattice, and which spans $M$ over $\mathbb{Q}$.}

\subsection{V basis\label{sec:approximation-problems:Vbasis}}

\subsubsection{Quaternion maximal order}

We recall that the V basis consists of the following six  matrices:
\begin{align*}
V_{\pm Z} &= \frac{1}{\sqrt{\ell}}\begin{pmatrix}
1\pm2i & 0\\
0 & 1\mp 2i
\end{pmatrix},
&V_{\pm Y} &= \frac{1}{\sqrt{\ell}}\begin{pmatrix}
1 & \mp 2\\
\pm 2 & 1
\end{pmatrix},
&V_{\pm X} &= \frac{1}{\sqrt{\ell}}\begin{pmatrix}
1 & \pm 2i\\
\pm 2i & 1
\end{pmatrix},
\end{align*}
where $\ell= 5$. 
Let $K = \q$ and let $L=\q(i)=\{a_0+ia_1:a_0,a_1\in\q\}$
, where $i^2 = -1$. 
Let $O_K = \z$ and $O_L=\z[i]=\{a_0+ia_1:a_0,a_1\in\z\}$ be the rings of integers of $K$ and $L$ respectively. Any element $t=a_0+ia_1\in O_L$ can be written as a 2-dimensional vector over $O_K$, namely $(a_0,a_1)$. There are two homomorphisms of $L$ into $\c$ related by complex conjugation. Denote by $\sigma$ the homomorphism such that $\sigma(i) = i$.  

Let $M_2(L)$ be the algebra of all $2\times2$ matrices with entries in $L$, and let $\mathcal{O}$ be an order in $M_2(L)$ that contains all the $V$ basis elements scaled by $\sqrt{\ell}$. For concreteness we will set 
\begin{equation}
\mathcal{O} := \mathbb{Z}\cdot I+\mathbb{Z}\cdot iX+\mathbb{Z}\cdot iY+\mathbb{Z}\cdot iZ.
\end{equation}
We extend $\sigma$ over $\mathcal{O}$ in a natural way, namely for $M\in \mathcal{O}$ we define $\sigma(M)$ as the matrix whose elements are the images of the elements of $M$ under $\sigma$. 
As observed in~\cite{BGS, Kliuchnikov2015b}, elements of $\mathcal{O}$
with determinant $\ell^N$ correspond to unitaries that can be expressed as a product of $N$ matrices from the V gate set via the map $\sigma'(M) = \frac{1}{\sqrt{\ell^N}}\sigma(M)$.

\ex{Let $V = V_Z\cdot V_{X}=\frac{1}{5}\left(\begin{smallmatrix}
	1 + 2i & 2i-4\\2i+4&1-2i
	\end{smallmatrix}\right)$. Then, $M_V=\left(\begin{smallmatrix}
	1 + 2i & 2i-4\\2i+4&1-2i
	\end{smallmatrix}\right)= I + 2\cdot iX - 4\cdot iY + 2\cdot iZ \in\mathcal{O}$ and $\sigma'(M_V)=V.$ Since $\det{(M_V)} = 5^2,$ we have $N=2$ as expected, as $V$ is the product of two $V$ basis matrices. Note that the sequence $V_ZV_{X}$ cannot be simplified (over the V basis) so $N$ is minimal.}

\ex{Let $V = V_Z V_{X}V_{-X}V_YV_{-Z}=\frac{1}{\sqrt{3125}}\left(\begin{smallmatrix}
	25 & 30-40i\\-30-40i&25
	\end{smallmatrix}\right)$. Then, 
	$$
	M_V=\left(
	\begin{smallmatrix}
	25 & 30-40i\\-30-40i&25
	\end{smallmatrix}
	\right)
	= 25\cdot I -40\cdot iX + 30\cdot iY \in\mathcal{O}
	$$ 
	and $\sigma'(M_V)=V.$ Then $\det{(M_V)}=3125=5^5$ so $V$ can be expressed as the product of five $V$ basis elements.
	However, $M'_V=\left(\begin{smallmatrix}
	5 & 6-8i\\-6-8i&5
	\end{smallmatrix}\right)= 5\cdot I -8\cdot iX + 6\cdot iY \in\mathcal{O}$, is also such that $\sigma'(M'_V)=V.$ Here, $\det{(M'_V)} = 125 = 5^3,$ giving $N=3$. Since $V_{P}V_{-P} = V_{-P}V_P=I$, for $P \in \{X,Y,Z\}$, the sequence $V_Z V_{X}V_{-X}V_YV_{-Z}$ simplifies to $V_ZV_YV_{-Z}$, so $V$ can in fact be expressed as a product of \emph{three} $V$ basis elements. The sequence cannot be simplified further, so this $N$ is minimal.}\label{exm:minN}


\subsubsection{Solving approximation problems}
 Finding a solution to any approximation problem over the V basis involves finding a matrix $M=\left(\begin{smallmatrix}m_1 & -m_2^\ast\\m_2&m_1^\ast\end{smallmatrix}\right)$ with additional constraints on $m_1$ depending on the approximation problem,  such that $\det(M)=\ell^N$. In essence, we seek matrices which can be achieved by the V basis, with elements falling in a particular region. Our approach is to first determine candidate values for $m_1$ via a specific enumeration problem, then to deduce the $m_2$ values  that satisfy the determinant constraint  by solving a norm equation. These two steps are repeated while iterating over $N$, beginning with $N=1$, until a valid $M$ is found. In the following, the point enumeration and norm equation steps are described for fixed $N.$

For the diagonal~(\problem{diagonal-approximation}, \problem{diagonal-approximation-by-unitary-mixing}) and fallback~(\cref{prop:fallback-approximation}, \problem{projective-rotation-mixing}) approximation problems, $M$ is such that
$\sigma_1(m_1)/\sqrt{\sigma_1(\ell^N)} \in R_{\mathrm{approx}} $,
where $R_{\mathrm{approx}}$ is a specific region of $\c$ depending on the problem. Namely, we consider $R_{\mathrm{approx}}$ as one of the regions defined in \propos{segment-condition}, \propos{segment}, \propos{diagonal-mixing-condition} and \propos{fallback-mixing-condition}.
For magnitude approximation (\cref{prop:magnitude-approximation-condition}, \cref{prob:magnitude-mixing}) with our new decomposition, $M$ must be such that 
$\sigma_1(m_1m_1^\ast)/\sigma_1(\ell^N) \in I_{\mathrm{approx}}$,
where $I_{\mathrm{approx}}\subset [0,1]$ where $I_{\mathrm{approx}}$ is an interval of $\mathbb{R}$ as defined in \propos{absolute-value-constraint}, \cref{prop:magnitude-mixing-condition}.  
%
%
Formally, we solve the  following point enumeration problems.
\begin{prob}[2D point enumeration (V basis)]
	Let $R_{\mathrm{approx}}$ be a 2D region corresponding to a particular approximation problem and fix $N\in\n$. 
	$$\textrm{Find all } (a_0,a_1) \in \mathbb{Z}^{2} \textrm{ such that } \frac{1}{\sqrt{\ell^N}}(a_0,a_1)
	\in R_{\mathrm{approx}}.$$
\end{prob}

\begin{prob}[1D point enumeration (V basis)]
	Let $I_{\mathrm{approx}} \subset [0,1]$ be a real interval corresponding to a particular approximation problem and fix $N\in\n.$
	$$
	\text{Find all } n \in \z \text{ such that }
	\frac{n}{{\ell^N}}
	\in I_{\mathrm{approx}}.
	$$
\end{prob}
\noindent In the first case we set $m_1=a_0+ia_1$ for every solution $(a_0,a_1)$. In the second case we first solve the {norm equation}  $n=a_0^2+a_1^2$, and for every solution we obtain a candidate value $m_1=a_0+ia_1$. 

To satisfy the determinant condition, solving the approximation problems requires that we keep only those $m_1$ for which the following problem is solvable. 
\begin{prob}[Norm equation (V basis)\label{prob:V-norm}]
	Given $m_1 \in \z[i]$ and integer $N$, find $m_2\in\z[i]$ such that $$m_2m_2^\ast = \ell^N - m_1m_1^\ast\in\z.$$
\end{prob}

For every pair of solutions $(m_1,m_2)$ we then deduce a matrix $M=\left(\begin{smallmatrix}m_1 & -m_2^\ast\\
m_2 & m_1^\ast\end{smallmatrix}\right)$. Since $m_2$ is a solution to \problem{V-norm} we have $\det(M)=\ell^N$ and the matrix $\sigma'(M)=\frac{1}{\sqrt{\ell^N}}\sigma_1(M)=\frac{1}{\sqrt{\ell^N}}\left(\begin{smallmatrix}m_1 & -m_2^\ast\\
m_2 & m_1^\ast\end{smallmatrix}\right)$ is unitary.

In summary, given a target unitary and associated region or interval, the following procedure finds an approximation over the $V$ basis. For a fixed value of $N$, an element $m_1\in\z[i]$ is obtained by solving an integer point enumeration problem defined by the target region. Together with $N$, $m_1$ defines a norm equation, which is solved to obtain an element $m_2\in\z[i].$ If no solution to either problem is found, the value of $N$ is increased. The point enumeration and norm equation steps are repeated for each value of $N$ until a valid pair $(m_1,m_2)$ is obtained. Each pair defines a matrix $M\in\mathcal{O}$ as above with determinant $\ell^N$. Then, the unitary $\sigma^\prime(M)$ is factorized over the V basis using an existing exact synthesis algorithm (\change{for example, the algorithm outlined in \sec{exact-synthesis-v-basis}}) to obtain a solution to the approximation problem. 

\subsection{Clifford+\texorpdfstring{$T$}{T} basis}\label{sec:approximation-solutions-clifford-t}
\subsubsection{Gate set}

The single-qubit Clifford group is defined as the set of unitaries that preserve the Pauli matrices under conjugation. That is, $\mathcal{C}$ is in the single-qubit Clifford group if and only if for any Pauli matrix $P$, the matrix $\mathcal{C}^\ast P \mathcal{C} $
is also a Pauli matrix.

We recall that the $S$, $H$ and $T$ gates are defined as follows:

$$ S=e^{-i\pi/4 Z} = 
\begin{pmatrix}e^{-i\pi/4}&0\\0&e^{i\pi/4}\end{pmatrix},
\qquad 
H  =\frac{1}{\sqrt{2}}\begin{pmatrix}1&1\\
1& -1
\end{pmatrix},
\qquad
T=e^{-i\pi/8Z} = \begin{pmatrix}e^{-i\pi/8}&0\\0&e^{i\pi/8}\end{pmatrix}.$$
The single-qubit Clifford group is generated by the $H$ and $S$ gates,  and the Clifford$+T$ group is generated by the single-qubit Clifford group and the $T$ gate.
%
%
Moreover we have $T^2 = S$, so the Clifford$+T$ group is generated by $H$ and $T$. 
We also recall the matrices $T_x,T_y$ defining rotations by $\frac{\pi}{4}$ about the $x$ and $y$ axes, namely
\begin{align*}
T_x&:=\begin{pmatrix}\cos(\frac{\pi}{8})&-i\sin(\frac{\pi}{8})\\-i\sin(\frac{\pi}{8})&\cos(\frac{\pi}{8})
\end{pmatrix}=\frac{1}{\sqrt{\ell}}\left(I + \frac{I-iX}{\sqrt{2}}\right),\\
T_y&:=\begin{pmatrix}\cos(\frac{\pi}{8})&-\sin(\frac{\pi}{8})\\\sin(\frac{\pi}{8})&\cos(\frac{\pi}{8})
\end{pmatrix}=\frac{1}{\sqrt{\ell}}\left(I + \frac{I-iY}{\sqrt{2}}\right)
\end{align*}
where $\ell = 2 + \sqrt{2}$.
Note that $T$ similarly defines the rotation of $\frac{\pi}{4}$ about the $z$ axis and we can write $T =\frac{1}{\sqrt{\ell}}\left(I + \frac{I-iZ}{\sqrt{2}}\right)$. We can obtain $T_x$ and $T_y$ from $T$, and vice versa, by conjugation with single-qubit Clifford unitaries. Synthesis via a circuit of $T_x,T_y, T$ and Hadamard gates therefore corresponds to synthesis in the Clifford$+T$ basis, up to a global phase.

In evaluating the cost of approximate synthesis with Clifford$+T$ gates, we assume that Clifford gates are low cost, and only count $T$ gates, or equivalently the total number of $T_x$, $T_y$ and $T$ matrices. 
See~\sec{fault-tolerant-gate-sets} for a justification of this assumption.

\subsubsection{Quaternion maximal order}\label{sec:CTqmo}

Let $K = \q(\sqrt{2})$ and let $L=\q(\zeta_8)$, where $\zeta_8 = e^{2\pi i/8}$.
The ring of integers of $L$ is $$O_L = \z[\zeta_8] = \set{a_0 + a_1 \zeta_8 + a_2 \zeta_8^2 + a_3 \zeta_8^3 : a_k \in \z } = \z[\sqrt{2}] + \frac{1+i}{\sqrt{2}} \cdot \z[\sqrt{2}].$$ The ring of integers of $K$ is the real subring $O_K=\z[\sqrt{2}]=\{ b_0 + b_1 \sqrt{2} : b_0,b_1 \in\z \}\subset O_L$. 
We can identify any element $m$ in $O_L$ with a 4-dimensional vector $\boldsymbol{m}=(a_0,a_1,a_2,a_3)\in\z^4$ using the integral basis above.
There are four distinct injective field homomorphisms that embed $L$ into $\c$, related to one another by complex conjugation and $\sqrt{2}$-conjugation. Define two such homomorphisms $\sigma_1$, $\sigma_2$ by
\begin{equation}
\sigma_1(\zeta_8) = \frac{1+i}{\sqrt{2}},\quad \sigma_2(\zeta_8) = \frac{-(1+i)}{\sqrt{2}}.
\end{equation}

We represent $\sigma_1,\sigma_2$ by the matrix 

$$\Sigma := \left(
\begin{array}{cccc}
1 & 1/\sqrt{2}  & 0 & -1/\sqrt{2}  \\
0 & 1/\sqrt{2} & 1& 1/\sqrt{2} \\
1 & -1/\sqrt{2}  & 0 & 1/\sqrt{2}  \\
0 & -1/\sqrt{2} & 1& -1/\sqrt{2} 
\end{array}
\right) $$
where
$(\mathrm{Re}\,\sigma_1(m),\mathrm{Im}\,\sigma_1(m),\mathrm{Re}\,\sigma_2(m),\mathrm{Im}\,\sigma_2(m))^T=\Sigma\boldsymbol{m}^T$.

Let $n=mm^\ast$ and write $n = b_0+b_1\sqrt{2}$, $b_0,b_1\in\z$. We can identify $n$ with the 2-dimensional vector $\boldsymbol{n}=(b,b_1)$ or with  $(\sigma_1(n),\sigma_2(n))^T= \left(
\begin{smallmatrix}
1 & \sqrt{2}\\
1 & -\sqrt{2}
\end{smallmatrix}
\right)\boldsymbol{n}^T$ through the above homomorphisms. We choose one homomorphism arbitrarily, say $\sigma_1$, to embed elements into Euclidean space. Both $\sigma_1$ and $\sigma_2$ are necessary to express the solvability constraints imposed by the norm equation for elements in $L$.
%
Let $M_2(L)$ be the algebra of $2\times2$ matrices with entries in $L$, and let $\mathcal{O}$ be a maximal order in $M_2(L)$ which contains $T_x$, $T_y$ and $T$.
For concreteness we will set  $\mathcal{O}=\sum_{i=1}^4O_K\cdot\omega_i$ in what follows, where
$$\omega_1=I,\qquad
\omega_2=\frac{I+iX}{\sqrt{2}},\qquad
\omega_3=\frac{I+iY}{\sqrt{2}},\qquad
\omega_4=\omega_3\omega_2=\frac{I+iX+iY+iZ}{2}.$$ 

The homomorphisms $\sigma_1,\sigma_2$ extend  over $\mathcal{O}$ in a natural way. Elements of $\mathcal{O}$ correspond to $2\times 2$ unitaries via the map $\sigma'(M)=\frac{1}{\sqrt{\sigma_1(\det(M))}}\sigma_1(M)$. Elements of $\mathcal{O}$ with determinant equal to 1 correspond to Clifford gates, and elements of $\mathcal{O}$ with determinant $\ell^N$ correspond to unitaries that can be expressed as a product of $N$ gates $T_x$, $T_y$ and $T$ \cite{GossetKliuchnikovMoscaRusso2014}). \change{We give a method for exact synthesis of Clifford gates in \sec{exact-synthesis}.}

\subsubsection{Solving approximation problems}
Finding a solution to any approximation problem (as defined in \sec{approximation-problems})
over the Clifford$+T$ gate set 
involves finding a matrix
\begin{equation}\label{eq:targetM}M=\left(\begin{smallmatrix}m_1 & -m_2^\ast\\m_2&m_1^\ast\end{smallmatrix}\right)=X_1\omega_1+X_2\omega_2+X_3\omega_3+X_4\omega_4,\end{equation}
or equivalently finding $X_i\in O_K$, with additional constraints on $m_1$ depending on the approximation problem,  such that $\det(M)=\ell^N$. Recall that these matrices will correspond to unitaries which are products of gates from the Clifford$+T$ basis.

Let us first examine the sets $M_{\mathrm{diag}}$ and $M_{\mathrm{off-diag}}$, in which we will look for elements $m_1$ and $m_2$, respectively. From Equation \eq{targetM} we have
\begin{eqnarray*}
M_{\mathrm{diag}}&=& \left\{ m_1 : \left(\begin{smallmatrix}m_1 & -m_2^\ast\\m_2&m_1^\ast\end{smallmatrix}\right) \in \mathcal{O} \right\}= \left\{X_1+\frac{X_2+X_3}{\sqrt{2}}+\frac{X_4}{2}+\frac{X_4}{2}i:X_i\in O_K\right\}\\
&=&\frac{1}{\sqrt{2}}O_K+\left(\frac{1+i}{2}\right)O_K\\
&=&\frac{1}{\sqrt{2}}O_L.
\end{eqnarray*}

Let $M_\mathcal{O}$ denote the elements of $L$ corresponding to diagonal elements of $\mathcal{O}$. That is elements $m_1$ such that $\left(\begin{smallmatrix} m_1 & 0 \\0 & m_1^*\end{smallmatrix}\right)\in\mathcal{O}$. By Equation \eq{targetM}, we can see $M_\mathcal{O} = O_L$.\\

Similarly, we have \begin{eqnarray*}
M_{\mathrm{off-diag}}&=& \left\{ m_2 : \left(\begin{smallmatrix}m_1 & -m_2^\ast\\m_2&m_1^\ast\end{smallmatrix}\right) \in \mathcal{O} \right\} = \left\{\frac{\sqrt{2}X_1-X_3}{2}+\frac{\sqrt{2}X_2+X_3}{2}i:X_i\in O_K\right\}\\
&=&\frac{1}{\sqrt{2}}O_K+\left(\frac{1+i}{2}\right)O_K\\
&=&\frac{1}{\sqrt{2}}O_L.
\end{eqnarray*} Hence, for all $m_1\in M_{\mathrm{diag}}, m_2\in M_{\mathrm{off-diag}}$ there exist $\hat{m_1},\hat{m_2}\in O_L$, such that $m_1=\frac{\hat{m_1}}{\sqrt{2}}$ and $m_2=\frac{\hat{m_2}}{\sqrt{2}}.$ For fixed $m_1$, $M_\mathrm{off-diag}$ is restricted to the subset $$M^{m_1}_\mathrm{off-diag}=\set{m_2\in M_\mathrm{off-diag}:\left(\begin{smallmatrix}
m_1&-m_2^\ast\\m_2&m_1^\ast
\end{smallmatrix}\right)\in\mathcal{O}}.$$ 
Noticing that $iY$, $(iY)^{-1}\in\mathcal{O}$, we see that $m_2\in M^{0}_\mathrm{off-diag}\iff m_2\in M_\mathcal{O}$.

Our approach is to first determine candidate values for $m_1$ via a specific enumeration problem, then to deduce corresponding values for $m_2$ by solving a norm equation. These two steps are repeated while iterating over $N$, beginning with $N=1$, until a valid $M$ is found. This approach is analogous to that used in \sec{approximation-problems:Vbasis} for the $V$ basis. In the following sections, the point enumeration and norm equation steps are described for fixed $N.$
For every pair of solutions $(m_1,m_2)$ we deduce a matrix $M=\left(\begin{smallmatrix}m_1 & -m_2^\ast\\
m_2 & m_1^\ast\end{smallmatrix}\right)$. The unitary $\sigma'(M)$ is factorized over the Clifford$+T$ basis to obtain a solution to the approximation problem.

\subsubsection{Finding \texorpdfstring{$m_1$}{m1}: an enumeration problem}\label{sec:CTm1}

For the diagonal~(\problem{diagonal-approximation}, \problem{diagonal-approximation-by-unitary-mixing}) and fallback~(\cref{prop:fallback-approximation}, \problem{projective-rotation-mixing}) approximation problems, we need 
$\sigma_1(m_1)/\sqrt{\sigma_1(\ell^N)} \in R_{\mathrm{approx}} $,
where $R_{\mathrm{approx}}\subset D_1$ is a specific region of $\c$ depending on the problem, and $D_1$ denotes the disk of radius 1 about the origin. 
For magnitude approximation (\cref{prop:magnitude-approximation-condition}, \cref{prob:magnitude-mixing}), $m_1$ must be such that %
$\sigma_1(m_1m_1^\ast)/\sigma_1(\ell^N) \in I_{\mathrm{approx}}$,
where $I_{\mathrm{approx}}\subset [0,1].$

%
In order to satisfy the determinant condition we then naturally consider the following norm equation,
  \begin{equation}m_2m_2^\ast = \ell^N-m_1m_1^\ast,\label{eq:gennorm}\end{equation}
 which we would a priori need to solve for every candidate value of $m_1$ satisfying the previous constraints.
\change{We observe, however, that this problem can only have solutions if the right-hand side of the equation is totally positive, that is, $\sigma_k(\ell^N-m_1m_1^\ast)>0$ for all $k$}. This means  that we only need to consider values of $m_1$ which additionally satisfy
$\sigma_2(m_1)/\sqrt{\sigma_2(\ell^N)}\in D_1$ or, equivalently, $\sigma_2(m_1m_1^\ast)/\sigma_2(\ell^N)\in [0,1]$. 
Since $m_1=\hat{m_1}/\sqrt{2}, m_2=\hat{m_2}/\sqrt{2}$, there is an equivalent norm equation for a given $\hat{m_1}$:

\begin{equation}
	\hat{m_2}\hat{m_2}^\ast = 2\ell^N - \hat{m_1}\hat{m_1}^\ast.
\end{equation}
The conditions on $\hat{m_1}$ are scaled accordingly: $\hat{m_1}$ must satisfy $\sigma_2(\hat{m}_1)/\sqrt{\sigma_2(2\ell^N)}\in D_1$ or, equivalently, $\sigma_2(\hat{m}_1\hat{m}_1^\ast)/\sigma_2(2\ell^N)\in [0,1]$. 

We write $\hat{m_1} = a_0 + a_1\zeta_8 + a_2\zeta_8^2+a_3\zeta_8^3$ and $\hat{n} = \hat{m_1}\hat{m_1}^\ast = b_0 + b_1\sqrt{2}$, with all coefficients in $\z.$ 
Let $\Sigma$ be as defined in \sec{CTqmo} and let $\Sigma'=\left(
\begin{smallmatrix}
1 & \sqrt{2}\\
1 & -\sqrt{2}
\end{smallmatrix}
\right)$. The operation $\Sigma$ (respectively $\Sigma'$) embeds $\hat{m_1}$ (respectively $\hat{n}$) into the Euclidean space of the approximation regions. In order to satisfy the constraints imposed by both the approximation regions and the norm equation, we define normalization matrices $\Lambda$ and $\Lambda'$ for $\Sigma$ and $\Sigma'$, respectively. Let $\Lambda$ and $\Lambda'$ be the diagonal matrices with $\left(\sqrt{\sigma_1(2\ell^N)},\sqrt{\sigma_1(2\ell^N)},\sqrt{\sigma_2(2\ell^N)},\sqrt{\sigma_2(2\ell^N)}\right)$ and $\left(\sigma_1(2\ell^N),\sigma_2(2\ell^N)\right)$ on their respective diagonals.
Candidate values for $\hat{m_1}$ are obtained by solving the point enumeration problems below.

\begin{prob}[2D point enumeration (Clifford$+T$ basis)]
	Let $R_{\mathrm{approx}}$ be a two-dimensional region corresponding to a particular approximation problem. Find $(a_0,a_1,a_2,a_3) \in \mathbb{Z}^{4}$ such that $\Lambda^{-1}\Sigma\cdot(a_0,a_1,a_2,a_3)^T
	\in R_{\mathrm{approx}}\times D_1.$ 
\end{prob}

\begin{prob}[1D point enumeration (Clifford$+T$ basis)]
	Let $I_{\mathrm{approx}} \subset [0,1]$ be a real interval corresponding to a particular approximation problem. Find $(b_0,b_1) \in \z^2$ such that $\Lambda'^{-1}\Sigma'\cdot(b_0,b_1)^T
	\in I_{\mathrm{approx}}\times [0,1]$.
\end{prob}


In the first case, we immediately recover a candidate value for $\hat{m_1}$. In the second case, we recover a candidate value for $\hat{n}$, then solve the norm equation $\hat{m_1}\hat{m_1}^\ast= \hat{n}$ and for every solution we obtain a candidate value $\hat{m_1}$. Then we set $m_1=\frac{\hat{m_1}}{\sqrt{2}}$.

\subsubsection{Finding \texorpdfstring{$m_2$}{m2}: solving a norm equation}
Given a candidate value for $m_1$, we proceed to solve a norm equation problem (or determine there is no solution), restricting $m_2$ to $M^{m_1}_\mathrm{off-diag}$:
\begin{prob}\label{prob:T-norm}
	Given $m_1\in\frac{1}{\sqrt{2}}O_L$ and integer $N$, find $m_2\in M^{m_1}_\mathrm{off-diag}$ such that $$m_2m_2^\ast = \ell^N - m_1m_1^\ast\in\frac{1}{2}O_K.$$

\end{prob}
Fixing an arbitrary $m\in M^{m_1}_\mathrm{off-diag}$, we have $M^{m_1}_\mathrm{off-diag}=m + O_L$, since for any two $m,m'\in M^{m_1}_\mathrm{off-diag}$ we have $m-m'\in M^0_\mathrm{off-diag}=O_L.$
Since $M_\mathrm{off-diag}=M_\mathrm{diag}=\frac{1}{\sqrt{2}}O_L$, \problem{T-norm} can then be reformulated as

\begin{prob}\label{prob:T-norm-quotient}
	Given $\hat{m_1}\in\z[\zeta_8]$, integer $N$, and $m\in \sqrt{2}M^{m_1}_\mathrm{off-diag}$ find $\hat{m_2}\in m+\sqrt{2}\z[\zeta_8]$ such that $$\hat{m_2}\hat{m_2}^* = 2\ell^N - \hat{m_1}\hat{m_1}^\ast\in\z[\sqrt{2}].$$
\end{prob}
Solving \problem{T-norm-quotient} for $\hat{m}_2$ then yields a solution to \problem{T-norm}: $m_2=\hat{m_2}/\sqrt{2}$.

\subsection{Clifford\texorpdfstring{$+\sqrt{T}$}{+Sqrt(T)} basis} \label{sec:approximation-solutions-clifford-root-t}

We now demonstrate how the solution framework applies to the Clifford$+\sqrt{T}$ basis. Note that the Clifford$+T$ group is contained within the Clifford$+\sqrt{T}$ group and unitaries in the latter are defined over the complex field $\q(\zeta_{16}) \supseteq \q(\zeta_8)$. Clearly, the fields $L$ and $K$ defined for Clifford$+\sqrt{T}$ are of higher degree over $\q$ than the respective Clifford$+T$ fields. Accordingly, in this section we work with larger matrices for $\Sigma, \Lambda$ and $\Sigma', \Lambda'$. The point enumeration problems are also higher dimensional.
The framework otherwise proceeds as for the Clifford$+T$ basis.


\subsubsection{Gate set}
Let $\ell = 2 + 2\cos(\frac{\pi}{8})=2+(\zeta_{16}+\zeta_{16}^{-1})$, where  $\zeta_{16} = e^{2\pi i/16}$. Let also \change{$\eta=2\cos(\frac{\pi}{8})$, $\beta = \eta^3+3\eta$ and $\mu = \eta^2-3.$} We recall that the $\sqrt{T}$ gate is defined as follows:
$$\sqrt{T} = \left(\begin{array}{cc}
e^{-i\pi/16}&0\\0&e^{i\pi/16}
\end{array}
\right) = \frac{1}{\sqrt{2 + 2 \cos(\frac{\pi}{8})}}
\left(\begin{array}{cc}
1 + e^{-i\pi/8}&0\\0&1 + e^{i\pi/8}
\end{array}
\right).
$$
The $\sqrt{T}$ gate defines a rotation about the $z$ axis by $\frac{\pi}{8}$. The Clifford$+\sqrt{T}$ group is generated by the single qubit Clifford group and the $\sqrt{T}$ gate. Note that we will use the notation $T^{1/2}$ interchangeably with $\sqrt{T}$ in the following discussion.
We also recall the matrices $T_x^{1/2},T_y^{1/2}$ defining rotations by $\frac{\pi}{8}$ about the $x$ and $y$ axes, namely
\begin{align*}
T_x^{1/2} &=\begin{pmatrix}\cos(\frac{\pi}{16})&-i\sin(\frac{\pi}{16})\\-i\sin(\frac{\pi}{16})&\cos(\frac{\pi}{16}),
\end{pmatrix}= \frac{1}{\sqrt{\ell}}\left(I + \frac{\eta(I-i\mu X)}{2}\right)\\
T_y^{1/2} &=\begin{pmatrix}\cos(\frac{\pi}{16})&-\sin(\frac{\pi}{16})\\\sin(\frac{\pi}{16})&\cos(\frac{\pi}{16})
\end{pmatrix}= \frac{1}{\sqrt{\ell}}\left(I + \frac{\eta(I-i\mu Y)}{2}\right).
\end{align*}
We can additionally write $\sqrt{T} =\frac{1}{\sqrt{\ell}}\left(I + \frac{\eta(I-\mu iZ)}{2}\right)$.
Observe that $\sqrt{T}^2 = T$ and $\left(T_a^{1/2}\right)^2=T_a$ with $a=x,y$, as suggested by the notation. We can obtain the unitaries $\textrm{T}_x^{k/2}$ and $\textrm{T}_y^{k/2}$ from $T^{k/2}$, for $k=1,2,3$, and vice versa, by conjugation with single-qubit Clifford unitaries. Here $T_a^{3/2}=\left(T_a^{1/2}\right)^3.$ Synthesis via a circuit of unitaries in $\{T^{k/2},\textrm{T}_a^{k/2} : a=x,y \quad k=1,2,3\}$ and Clifford gates therefore corresponds to synthesis in the Clifford $+\sqrt{T}$ basis, up to a global phase.

\subsubsection{Quaternion maximal order}\label{sec:CSTqmo}
Let $K$ be the totally real number field $K = \q(\zeta_{16}+\zeta_{16}^{-1})$, and let $L$ be the field $L = \q(\zeta_{16})$. The ring of integers of $L$ is $$O_L = \z[\zeta_{16}] =\left\{ \sum\limits_{i=0}^7a_k\zeta_{16}^k: a_k\in\z\right\} =  \z\left[2\cos\left(\frac{\pi}{8}\right)\right] + \zeta_{16}\z\left[2\cos\left(\frac{\pi}{8}\right)\right]$$
and the ring of integers of $K$ is the real subring $$O_K = \z\left[2\cos\left(\frac{\pi}{8}\right)\right]=\left\{b_0+b_1\cdot2\cos\left(\frac{\pi}{8}\right) + b_2\sqrt{2} + b_3\cdot 2\cos\left(\frac{3\pi}{8}\right) : b_k\in\z\right\}\subset O_L.$$ 
We can identify any element $m$ in $O_L$ with an 8-dimensional vector $\boldsymbol{m}=(a_0,a_1,\dots,a_7)\in\z^8$ using the integral basis above.
There are 8 distinct injective field homomorphisms that embed $L$ into $\c$, which can be grouped into pairs depending on their images when restricted to $K$. Define four such homomorphisms $\sigma_1,\sigma_2, \sigma_3,\sigma_4$ by
\begin{eqnarray*}
\sigma_1(\zeta_{16}) = \cos(\frac{\pi}{8}) +i\cos(\frac{3\pi}{8}),\quad 
 \sigma_2(\zeta_{16}) = \cos(\frac{3\pi}{8})+i\cos(\frac{\pi}{8}),\\
\sigma_3(\zeta_{16}) = -\cos(\frac{3\pi}{8})+i\cos(\frac{\pi}{8}),\quad
\sigma_4(\zeta_{16}) = \cos(\frac{\pi}{8})+i\cos(\frac{3\pi}{8}).
\end{eqnarray*}

We represent $\sigma_1,\sigma_2, \sigma_3,\sigma_4$ by the matrix 

$$\Sigma :=  \left(
\begin{array}{cccccccc}
1 & \cos(\frac{\pi}{8})&\frac{1}{\sqrt{2}}&\cos(\frac{3\pi}{8})&0&-\cos(\frac{3\pi}{8})&-\frac{1}{\sqrt{2}}&-\cos(\frac{\pi}{8}) \\
0 &\cos(\frac{3\pi}{8})&\frac{1}{\sqrt{2}}&\cos(\frac{\pi}{8})&1&\cos(\frac{\pi}{8})&\frac{1}{\sqrt{2}}& \cos(\frac{3\pi}{8})\\
1&\cos(\frac{3\pi}{8})&-\frac{1}{\sqrt{2}}&-\cos(\frac{\pi}{8})&0&\cos(\frac{\pi}{8})&\frac{1}{\sqrt{2}}&-\cos(\frac{3\pi}{8})\\
0&\cos(\frac{\pi}{8})&\frac{1}{\sqrt{2}}&-\cos(\frac{3\pi}{8})&-1&-\cos(\frac{3\pi}{8})&\frac{1}{\sqrt{2}}&\cos(\frac{\pi}{8})\\
1 &-\cos(\frac{3\pi}{8})&-\frac{1}{\sqrt{2}}&\cos(\frac{\pi}{8})&0&-\cos(\frac{\pi}{8})&\frac{1}{\sqrt{2}}&\cos(\frac{3\pi}{8})\\
0&\cos(\frac{\pi}{8})&-\frac{1}{\sqrt{2}}&-\cos(\frac{3\pi}{8})&1&-\cos(\frac{3\pi}{8})&-\frac{1}{\sqrt{2}}&\cos(\frac{\pi}{8})\\
1 &-\cos(\frac{\pi}{8}) &\frac{1}{\sqrt{2}}&-\cos(\frac{3\pi}{8})&0&\cos(\frac{3\pi}{8})&-\frac{1}{\sqrt{2}}&  \cos(\frac{\pi}{8}) \\
0 &\cos(\frac{3\pi}{8}) &-\frac{1}{\sqrt{2}}&\cos(\frac{\pi}{8})&-1&\cos(\frac{\pi}{8})&-\frac{1}{\sqrt{2}}& \cos(\frac{3\pi}{8})\\
\end{array}
\right)  $$
where
$$ (\mathrm{Re}\,\sigma_1(m),\mathrm{Im}\,\sigma_1(m),\mathrm{Re}\,\sigma_2(m),\mathrm{Im}\,\sigma_2(m),\mathrm{Re}\,\sigma_3(m),\mathrm{Im}\,\sigma_3(m),\mathrm{Re}\,\sigma_4(m),\mathrm{Im}\,\sigma_4(m))^T=\Sigma\boldsymbol{m}^T.$$

Let $n = mm^\ast$ and write $n = b_0 + b_1\cdot 2\cos(\frac{\pi}{8})+b_2\sqrt{2}+b_3\cdot2\cos(\frac{3\pi}{8}).$ We can identify $n$ with the 4-dimensional vector $\boldsymbol{n}=(b_0,b_1,b_2,b_3)$, or with $(\sigma_1(n),\sigma_2(n),\sigma_3(n),\sigma_4(n))^T= \Sigma'\boldsymbol{n}^T$ where
$$\Sigma' := \left(
\begin{array}{cccc}
1 & 2\cos(\frac{\pi}{8})&\sqrt{2} & 2\cos(\frac{3\pi}{8})\\
1 &-2\cos(\frac{3\pi}{8})&-\sqrt{2}&-2\cos(\frac{\pi}{8})\\
1&-2\cos(\frac{3\pi}{8})&\sqrt{2}&2\cos(\frac{\pi}{8})\\
1 &-2\cos(\frac{\pi}{8})&\sqrt{2}& -2\cos(\frac{3\pi}{8})
\end{array}
\right)$$ through the above homomorphisms. As for the Clifford$+T$ basis, we choose a homomorphism arbitrarily, for example $\sigma_1$, to embed elements into Euclidean space. 

%
Let $M_2(L)$ be the algebra of all $2\times2$ matrices with entries in $L$. 
Let $\mathcal{O}$ be a maximal order in $M_2(L)$ which contains $T_x^{1/2}$, $T_y^{1/2}$ and $T^{1/2}$, namely $\mathcal{O}=\sum_{i=1}^4O_K\cdot\omega_i$, where
$$\omega_1=I,\qquad
\omega_2=\frac{I+iX}{\sqrt{2}},\qquad
\omega_3=\frac{I+iY}{\sqrt{2}},\qquad
\omega_4=\omega_3\omega_2=\frac{I+iX+iY+iZ}{2}.$$
The homomorphisms $\sigma_1,\sigma_2,\sigma_3,\sigma_4$ extend  over $\mathcal{O}$ in a natural way. Elements of $\mathcal{O}$ correspond to $2\times 2$ unitaries via the map $\sigma'(M)=\frac{1}{\sqrt{\sigma_1(\det(M))}}\sigma_1(M)$. Elements of $\mathcal{O}$ with determinant $\ell^N$ correspond to unitaries that can be expressed as a product of $N$ gates T$^{k/2}_x$, T$^{k/2}_y$ and T$^{k/2}$ with $k=1,2,3$ (see~\sec{exact-synthesis}, \cite{GossetKliuchnikovMoscaRusso2014}), hence in the Clifford + $\sqrt{T}$ gates.

\subsubsection{Solving approximation problems}
Finding a solution to any approximation problem over the Clifford+$\sqrt{T}$ gate set involves finding a matrix
\begin{equation}\label{eq:sqtargetT}M=\left(\begin{smallmatrix}m_1 & -m_2^\ast\\m_2&m_1^\ast\end{smallmatrix}\right)=X_1\omega_1+X_2\omega_2+X_3\omega_3+X_4\omega_4\in\mathcal{O},\end{equation}
or equivalently finding $X_i\in O_K$, with additional constraints on $m_1$ depending on the approximation problem, such that $\det(M)=\ell^N$. 

Let us first examine the sets $M_{\mathrm{diag}}$ and $M_{\mathrm{off-diag}}$, in which we will look for elements $m_1$ and $m_2$, respectively. From Equation \eq{sqtargetT} we have
\begin{eqnarray*}
M_{\mathrm{diag}}&=&\left\{ m_1 : \left(\begin{smallmatrix}m_1 & -m_2^\ast\\m_2&m_1^\ast\end{smallmatrix}\right) \in \mathcal{O} \right\}=\left\{X_1+\frac{X_2+X_3}{\sqrt{2}}+\frac{X_4}{2}+\frac{X_4}{2}i:X_i\in O_K\right\}\\
&=&\frac{1}{\sqrt{2}}O_K+\frac{1+i}{2}O_K.
\end{eqnarray*}
As before, let $M_\mathcal{O}$ denote the set of elements $m_1\in L$ such that $\left(\begin{smallmatrix} m_1&0\\0&m_1^*\end{smallmatrix}\right)\in\mathcal{O}$. From Equation \eq{sqtargetT}, we have $M_\mathcal{O} = O_K + \frac{1+i}{\sqrt{2}}O_K$ and so clearly $M_\mathrm{diag} = \frac{1}{\sqrt{2}}M_\mathcal{O}$. Similarly, we have $M_{\mathrm{off-diag}}=\frac{1}{\sqrt{2}}M_\mathcal{O}$. Note that $O_L\subsetneq M_{\mathcal{O}}$, since $\zeta_{16}$ is in $O_L$ but not in $M_{\mathcal{O}}$.

As with the Clifford$+T$ basis, our approach is to iterate over $N$, beginning with $N=1$, and for each $N$ to first determine candidate values for $m_1$ via a specific enumeration problem, then to deduce corresponding values for $m_2$ by solving a norm equation, until a valid $M$ is found. In the following sections, the point enumeration and norm equation steps are described for fixed $N.$
For every pair of solutions $(m_1,m_2)$ we deduce a matrix $M=\left(\begin{smallmatrix}m_1 & -m_2^\ast\\
m_2 & m_1^\ast\end{smallmatrix}\right)$. The unitary $\sigma'(M)$ is factorized over the Clifford$+\sqrt{\textrm{T}}$ basis. 

\subsubsection{Finding \texorpdfstring{$m_1$}{m1}: an enumeration problem}\label{sec:CSTm1}
For both the diagonal and fallback approximation problems, we need $\sigma_1(m_1)/\sqrt{\sigma_1(\ell^N)} \in R_{\mathrm{approx}}$, where $R_{\mathrm{approx}}\subset D_1$ is a specific region of $\c$ defined by the problem. 

For the general approximation problem, $m_1$ must be such that $\sigma_1(m_1m_1^\ast)/\sigma_1(\ell^N) \in I_{\mathrm{approx}}$, where $I_{\mathrm{approx}}\subset [0,1].$
In order to satisfy the determinant condition we then naturally consider the following norm equation,
  \begin{equation}m_2m_2^\ast = \ell^N-m_1m_1^\ast,\label{eq:gennorm2}\end{equation}
 which we would a priori need to solve for every candidate value of $m_1$ satisfying the previous constraints.
Again, we observe that this norm equation only has solutions if its right-hand side is totally positive. This means that we only need to consider those candidates $m_1$ that additionally satisfy $\sigma_k(m_1)/\sqrt{\sigma_k(\ell^N)}\in D_1$or, equivalently, $\sigma_k(m_1m_1^\ast)/\sigma_k(\ell^N)\in [0,1]$, for $k=2,3,4$. 

Writing any $m_1=a_0+a_1i$ with $a_0,a_1\in K$, we see that $M_{\mathrm{diag}}$ can be considered as a full rank $O_K$ lattice in $K^2$. We therefore have a $\z$-basis, $\{y_0,\dots,y_{7}\}$, for $M_{\mathrm{diag}}$ and can write any element $m_1\in M_{\mathrm{diag}}$ as $m_1=\sum\limits_{i=0}^{7}a_0y_0, a_0\in\z$.

Since $M_\mathrm{diag}=\frac{1}{\sqrt{2}}O_K+\frac{1+i}{2}O_K$, we also have $n := m_1m_1^\ast \in \frac{1}{2}O_K$. Since $m_1\in\frac{1}{\sqrt{2}}M_\mathcal{O}$, there exists $\hat{m_1}\in M_\mathcal{O}$ such that $m_1=\frac{\hat{m_1}}{\sqrt{2}}$ and furthermore, $\hat{m_1}\hat{m_1}^\ast=2n:=\hat{n}\in O_K$. We write $\hat{n}=b_0+b_1\cdot2\cos(\frac{\pi}{8})+b_2\sqrt{2}+b_3\cdot2\cos(\frac{3\pi}{8})$ with all coefficients in $\z$.

Let $\Sigma_\mathcal{O}$ be defined as the matrix with rows:
\begin{eqnarray*}
\Sigma_\mathcal{O}^{(2j)} &=& (\mathrm{Re}(\sigma_j(y_0)),\dots,\mathrm{Re}(\sigma_j(y_{7}))\\
\Sigma_\mathcal{O}^{(2j+1)} &=& (\mathrm{Im}(\sigma_j(y_0)),\dots,\mathrm{Im}(\sigma_j(y_{7})),
\end{eqnarray*}
for $1\le j\le 7$, where the $\sigma_j$ are defined in \sec{CSTqmo}. Additionally, take $\Sigma^\prime$ as defined in \sec{CSTqmo}, and define normalization matrices $\Lambda$ and $\Lambda^\prime$. That is, $\Lambda$ and $\Lambda^\prime$ are diagonal matrices with entries $\big(\sqrt{\sigma_1(\ell^N)},\sqrt{\sigma_1(\ell^N)},\dots,\sqrt{\sigma_4(\ell^N)},\sqrt{\sigma_4(\ell^N)}\big)$ and  $(\sigma_1(2\ell^N),(\sigma_2(2\ell^N),(\sigma_3(2\ell^N),\sigma_4(2\ell^N))$ on the main diagonal, respectively. Hence the operations $\Lambda\Sigma_\mathcal{O}$ and $\Lambda'\Sigma'$ first embed an element $m_1$ or $\hat{n}$ into the Euclidean space of our approximation regions, then normalizes it to satisfy the constraints.

Candidate values for $m_1$ are then obtained by solving point enumeration problems below.

\begin{prob}[2D point enumeration (Clifford$+\sqrt{T}$ basis)]
	Let $R_{\mathrm{approx}}$ be a 2D region corresponding to a particular approximation problem. Find $(a_0,a_1,a_2,a_3,a_4,a_5,a_6,a_7) \in \mathbb{Z}^{8}$ such that 
	$$\Lambda^{-1}\Sigma_\mathcal{O}\cdot(a_0,a_1,a_2,a_3,a_4,a_5,a_6,a_7)^T
	\in R_{\mathrm{approx}}\times D_1\times D_1 \times D_1.$$ 
\end{prob}

\begin{prob}[1D point enumeration (Clifford$+\sqrt{T}$ basis)]
	Let $I_{\mathrm{approx}} \subset [0,1]$ be a real interval corresponding to a particular approximation problem.  Find $(a^\prime_0,a^\prime_1,a^\prime_2,a^\prime_3) \in \z^4$ such that 
	
	$$\Lambda^{\prime^{-1}}\Sigma^\prime\cdot(b_0,b_1,b_2,b_3)^T
	\in I_{\mathrm{approx}}\times[0,1]\times[0,1]\times[0,1].$$ 
\end{prob}

In the first case, we immediately recover a candidate value for $m_1$. In the second case, we recover a candidate value for $\hat{n}$, solve the norm equation $\hat{m_1}\hat{m_1}^\ast= \hat{n}$ for $\hat{m_1} \in M_{\mathcal{O}}$  and for every solution $\hat{m_1}$  we obtain a candidate value $m_1$ by setting $m_1=\frac{\hat{m_1}}{\sqrt{2}}$.

\subsubsection{Finding \texorpdfstring{$m_2$}{m2}: solving a norm equation}

Given a candidate value for $m_1$, we proceed to solve a norm equation problem (or determine there is no solution), restricting $m_2$ to $M^{m_1}_\mathrm{off-diag}$:
\begin{prob}\label{prob:sqrtT-norm}
	Given $m_1\in\frac{1}{\sqrt{2}}M_\mathcal{O}$ and integer $N$, find $m_2\in M^{m_1}_\mathrm{off-diag}$ such that $$m_2m_2^\ast = \ell^N - m_1m_1^\ast\in\frac{1}{2}O_K.$$

\end{prob}
Fixing an arbitrary $m\in M^{m_1}_\mathrm{off-diag}$, we have $M^{m_1}_\mathrm{off-diag}=m + M_\mathcal{O}$, since for any two $m,m'\in M^{m_1}_\mathrm{off-diag}$ we have $m-m'\in M^0_\mathrm{off-diag}=M_\mathcal{O}.$
Since $M_\mathrm{off-diag}=M_\mathrm{diag}=\frac{1}{\sqrt{2}}M_\mathcal{O}$, \problem{sqrtT-norm} can then be reformulated as

\begin{prob}\label{prob:sqrtT-norm-quotient}
	Given $\hat{m_1}\in M_\mathcal{O}$, integer $N$, and $m/\sqrt{2}\in M^{m_1}_\mathrm{off-diag}$ find $\hat{m_2}\in m+\sqrt{2}M_\mathcal{O}$ such that $$\hat{m_2}\hat{m_2}^* = 2\ell^N - \hat{m_1}\hat{m_1}^\ast\in O_K.$$

\end{prob}
Solving \problem{sqrtT-norm-quotient} for $\hat{m}_2$ then yields a solution to \problem{sqrtT-norm}: $m_2=\hat{m_2}/\sqrt{2}$.
\subsection{General case}\label{sec:approximation-problems:general}

In this section, we extrapolate from  the three preceding examples to outline a general method for solving approximate synthesis properties, and describe the properties required by gate sets to which this method applies.

\subsubsection{Gate sets}
We consider quaternion gate sets as defined by Kliuchnikov \emph{et al.} in \cite{Kliuchnikov2015b}. Informally, these are gate sets which are described by \emph{totally definite quaternion algebras}. 

Let $K$ be a totally real number field and take totally positive elements $a,b\in K$. Define $L$ to be the extension $L:=K(\sqrt{-a})$ and let $i\in L$ be such that $i^2=-a$. There are $2d$ distinct injective field homomorphisms embedding $L$ into $\c$, where $d = [K:\q]$. Fix $\sigma_1,\ldots,\sigma_d$ as any $d$ homomorphisms from $L$ that are pairwise distinct when restricted to $K$. 

A quaternion algebra $(\frac{-a,-b}{K}):=Q$ over the field $K$ is an algebra of the form $K + K\textbf{i} + K\textbf{j}+K\textbf{k}$ where $ \textbf{i}^2 = -a,\textbf{j}^2=-b$ and $\textbf{ij}=-\textbf{ji}=\textbf{k}$. A totally definite quaternion algebra has $a,b>$ totally positive. An element in $Q$ is written $q = q_0+q_1\textbf{i}+q_2\textbf{j}+q_3\textbf{k}$, $q_0,q_1,q_2,q_3\in K$, with conjugate $\bar{q}=q_0-q_1\textbf{i}-q_2\textbf{j}-q_3\textbf{k}$. The reduced norm of $q$ is $\mathrm{nrd}(q) = q\bar{q}$.

Let $M_2(L)$ be the set of $2\times 2$ matrices with elements in $L$. Define the $K$-linear map $\kappa: Q\rightarrow M_2(L)$ by
\begin{eqnarray}
\kappa(1) = I,\quad \kappa(\textbf{i}) = \sqrt{-a}Z,\quad \kappa(\textbf{j}) = -\sqrt{-b}Y,\quad \kappa(\textbf{k}) = \sqrt{-ab}X,
\end{eqnarray}
where $X,Y,Z$ are the Pauli matrices. Notice that $\kappa$ defines an isomorphism of quaternion algebras, with $\kappa(\textbf{k})=\kappa(\textbf{i})\kappa(\textbf{j})$. Concretely, we have a correspondence between elements in $Q$ and matrices in $M_2(L)$ of the form $M = \left( \begin{smallmatrix}q_0 +q_1\sqrt{-a} & -q_2\sqrt{b}+q_3\sqrt{-ab}\\q_2\sqrt{b}+q_3\sqrt{-ab}&q_0-q_1\sqrt{-a}\end{smallmatrix}\right)$, where the corresponding quaternion is $q:=q_0+q_1\textbf{i}+q_2\textbf{j}+q_3\textbf{k}$, such that $\kappa(q)=M$. Observe that $\det(M) = \mathrm{nrd}(q)=q_0-aq_1^2-bq_2^2+abq_3^2.$
The set of matrices of this form corresponds to $\mathrm{SU}(2)$ via the map 
\begin{equation}
\sigma'(M) = \frac{1}{\sqrt{\sigma_1(\det(M))}}\cdot\sigma_1(M),\end{equation}
where $\sigma_1$ is the natural extension over matrices of the homomorphism from $L$ into $\c$. Let $S$ be a set of elements from $K$. Consider  the gate set to be those matrices with determinant in $S.$

For the V, Clifford$+\textnormal{T}$ and Clifford$+\sqrt{\textnormal{T}}$ bases, the corresponding fields and integer rings are given in \tab{gates}.
\begin{table}[h!]
	\begin{center}
		{ 
			\caption{Number field correspondences for the V, Clifford$+T$ and Clifford$+\sqrt{T}$ gate sets.}\label{tab:gates}
			\begin{tabular}{|c|c|c|c|c|}
				\hline 
				Gate set & $K$ & $L$ & $O_{K}$ & $O_{L}$\tabularnewline
				\hline 
				\hline 
				V basis  & $\mathbb{Q}$ & $\mathbb{Q}(i)$ & $\mathbb{Z}$ & $\mathbb{Z}[i]$\tabularnewline
				\hline 
				Clifford+$\mathrm{T}$ & $\mathbb{Q}(\sqrt{2})$ & $\mathbb{Q}(\zeta_{8})$ & $\mathbb{Z}[\sqrt{2}]$ & $\mathbb{Z}[\zeta_{8}]$\tabularnewline
				\hline 
				Clifford+$\sqrt{\mathrm{T}}$ & $\mathbb{Q}(\zeta_{16}+\zeta_{16}^{-1})$ & $\mathbb{Q}(\zeta_{16})$ & $\mathbb{Z}[\zeta_{16}+\zeta_{16}^{-1}]$ & $\mathbb{Z}[\zeta_{16}]$\tabularnewline
				\hline 
			\end{tabular}
		}
	\end{center}
\end{table}

\subsubsection{Quaternion maximal order}

For a given gate set, $K$ and $L$, there exists $\mathcal{O}$, an order of $M_2(L)$, containing the preimages of the gate set unitaries under $\sigma'$. We note here that while this order does not  need to be maximal, maximal orders have several properties which allow for efficient factorization of elements \cite{Kliuchnikov2015b}. For a thorough background on quaternion orders, we direct the reader to \cite{voight2005quadratic}. 

The order $\mathcal{O}$ is constructed as follows. The gate set elements are mapped to matrices in $M_2(L)$. Let $\mathcal{L_K}$ be the $O_K$-lattice obtained by taking an $O_K$ linear combination of the elements of the ring generated by these matrices. Then, $\mathcal{O}$ can be taken as any order containing this lattice. Note that, due to the multiplicative properties of the determinant, elements in $\mathcal{O}$ with determinant equal to $\ell$ for some $\ell\in\langle S\rangle$ will correspond to gate set elements. Then $\ell=\prod s_i^{N_i}$ for some set of elements $s_i\in S$. Let $N:=\sum N_i$, which gives the length of a sequence of basis elements that produces the corresponding gate set element (when the class number of the quaternion algebra is one). Observe that for the gate sets we consider, we have $S = {\ell}$ and so $\det(M)=\ell^N$ for some $N\in\z^\geq$. Clifford+$\sqrt{T}$ is an example of a gate set for which the corresponding quaternion algebra has class number two. Recall that in \exm{minN}, two distinct elements in $\mathcal{O}$ corresponded to the same gate set element, each with a distinct $N$ value. We look for minimal $N$, as this will correspond to the shortest possible basis sequence. This will be the $N$ for which the entries of $M\in\mathcal{O}$ are integral and not all divisible by $s_i\in S$, for all $i$. Since the approximation method outlined here iterates over increasing $N$, the sequence obtained will be optimal.

\begin{rem}In addition, we look for orders $\mathcal{O}$ in which gates that are considered `low-cost' in the gate set behave as units. 
This forces the determinant of matrices corresponding to low-cost gates to be 1, ensuring that $N$ is a count of `expensive' gates in a sequence. In essence, this makes the determinant a useful cost measure for approximation. For the $V$-basis, these low-cost gates are the Pauli matrices; for Clifford$+\mathrm{T}$ and Clifford$+\sqrt{\mathrm{T}}$ these are the Clifford unitaries.\end{rem}

The definitions for $\mathcal{O}$ and $\ell$ corresponding to the V, Clifford$+\textnormal{T}$ and Clifford$+\sqrt{\textnormal{T}}$ bases are given in the Table \ref{tab2}.
\begin{table}[h]
	\caption{Maximal orders for V, Clifford$+T$ and Clifford$+\sqrt{T}$ gate sets.}\label{tab2}
	\begin{center}
		{ 
			\begin{tabular}{|c|c|c|}
				\hline 
				Gate set & $\ell$ & $\mathcal{O}$\tabularnewline
				\hline 
				\hline 
				V basis  &  $5$ & $O_K\cdot I+O_K\cdot iX+O_K\cdot iY+O_K\cdot iZ$\tabularnewline
				\hline 
				Clifford+$\mathrm{T}$ &  $2+\sqrt{2}$ & $O_{K}\cdot I+O_{K}\cdot\frac{I+iX}{\sqrt{2}}+O_{K}\cdot\frac{I+iY}{\sqrt{2}}+O_{K}\cdot\frac{I+iZ+iX+iY}{2}$\tabularnewline
				\hline 
				Clifford+$\sqrt{\mathrm{T}}$ &  $2+2\cos\frac{\pi}{8}$ & $O_{K}\cdot I+O_{K}\cdot\frac{I+iX}{\sqrt{2}}+O_{K}\cdot\frac{I+iY}{\sqrt{2}}+O_{K}\cdot\frac{I+iZ+iX+iY}{2}$\tabularnewline
				\hline 
			\end{tabular}
		}
	\end{center}
\end{table}

\subsubsection{Solving approximation problems}

For fixed $N\in\n$, finding a solution to any approximation problem over a gate set involves finding a matrix \begin{equation*}M=\left(\begin{smallmatrix}m_1 & -m_2^\ast\\m_2&m_1^\ast\end{smallmatrix}\right)\in\mathcal{O},\label{eq:general-m-basis}\end{equation*}
 with additional constraints on $m_1$ depending on the approximation problem, such that $\det(m)=\ell^N$.
Our approach to finding $M$ can be summarized in two steps:\\

\begin{enumerate}
	\item point enumeration in a target region to find $m_1$ (\sec{general-point-enumeration}), followed by
	\item solving a relative norm equation to recover $m_2$ (\sec{general-norm-equation-solution}).\\
\end{enumerate}

For the diagonal and fallback approximation problems, with and without mixing, we look for elements ${M=\left(\begin{smallmatrix}m_1 & - m_2^\ast\\m_2&m_1^\ast\end{smallmatrix}\right)}$
of $\mathcal{O}$, such that
\begin{equation*}
\sigma_1(m_1)/\sqrt{\sigma_1(\ell^N)} \in R_{\mathrm{approx}} \subset D_1,\end{equation*}
where $R_{\mathrm{approx}}$ is the region defined by the problem.
For the general unitary approximation problem, $m_1$ is required to satisfy
\begin{equation*}\sigma_1(m_1m_1^\ast)/\sigma_1(\ell^N) \in I_{\mathrm{approx}}  \subset [0,1],\end{equation*}
where $I_{\mathrm{approx}}$ is the real interval defined by the parameters of the problem. We observe that for the relative norm equation \begin{equation*}m_2m_2^\ast = \ell^N-m_1m_1^\ast.\end{equation*}  to have a solution, it is necessary that, for all $k$, $\sigma_k(\ell^N-m_1m_1^\ast)>0$. This means we only need to consider those candidates $m_1$ that satisfy either
\begin{equation*}
\sigma_k(m_1)/\sqrt{\sigma_k(\ell^N)} \in  D_1
\textnormal{ or, equivalently, }
\sigma_k(m_1m_1^\ast)/\sigma_k(\ell^N) \in  [0,1]\end{equation*}
for all $k>1$.

From each pair $(m_1,m_2)$ we can deduce a matrix $M= \left(\begin{smallmatrix}m_1 & - m_2^\ast\\m_2&m_1^\ast\end{smallmatrix}\right)$. The unitary $\sigma'(M)$ is factorized over the desired gate set to obtain a solution to the approximation problem. If no solution exists for the given $N$, set $N:= N+1$ and repeat the process. Thus, iterating over $N$, initialized at 1, will give the solution corresponding to the shortest gate sequence.

\subsubsection{Finding \texorpdfstring{$m_1$}{m1}: an enumeration problem}\label{sec:general-point-enumeration}
The problem of finding candidates $m_1\in L$ satisfying the conditions of an approximation problem can be reduced to an integer point enumeration problem. To better understand the set to which the main diagonal entries, $m_1$, belong, we define the map $h$ from $L$ to $Q$ by
\begin{equation*}
    h(a_0+ia_1) = a_0+a_1\textbf{i}.
\end{equation*}

 We see that $\kappa(h(m))$ sends an element from $L$ to the diagonal matrix $\left(\begin{smallmatrix}
     m & 0\\0&m^\ast
 \end{smallmatrix}\right)\in M_2(L).$

 Hence we are enumerating elements $m_1\in L$ from the set
 
\begin{equation*}
    M_\mathrm{diag}=\set{m_1:\exists m_2\in L\textnormal{ s.t. } \kappa(h(m_1)) + h(m_2)\textbf{j})\in \mathcal{O}}.
\end{equation*}

We additionally define the set of diagonal matrices in $\mathcal{O}$:
\begin{equation*}
M_\mathcal{O}:=\set{m\in L: \kappa(h(m))\in \mathcal{O}}.
\end{equation*}

Observe that enumerating $m_1$ from $M_\mathrm{diag}$ is equivalent to enumerating $a_0,a_1\in K$  from the set
\begin{equation*}
   \mathcal{L}_\mathcal{O}=\{(a_0,a_1): \exists a_2,a_3\in K \textnormal{ s.t. } a_0I+a_1\sqrt{-a}Z - a_2\sqrt{-b}Y + a_3\sqrt{-ab}X\in\mathcal{O}\}.
\end{equation*}
We make use of the following lemma to find a $\z$-basis for $M_\mathrm{diag}$.
\begin{lem}
  $\mathcal{L}_\mathcal{O}$ is a full rank $O_K$-lattice in $K^2$.
\end{lem}
\begin{proof}Since $\mathcal{O}$ is closed under addition and scalar multiplication over $O_K$, so is $\mathcal{L}_\mathcal{O}$. Consider an $O_K$-linearly independent generating set $G$ of $\mathcal{L}_\mathcal{O}$ and let $g_1,\dots,g_r$ be the subset of these that are $K$-linearly independent. Then $r\le 2.$ We have $I\in\mathcal{O}$, so $(1,0)\in\mathcal{L}_\mathcal{O}$. Suppose for a contradiction that $\mathcal{L}_\mathcal{O}$ contains no elements of the form $(a_0,a_1),a_1\neq 0$ in $\mathcal{L}_\mathcal{O}$. Let $\{\omega_i\}_{i=1,\dots,4}$ be a basis for $\mathcal{O}$, with corresponding elements in $\mathcal{L}_\mathcal{O}$ denoted by $(\omega_{i,0},\omega_{i,1})$. By assumption, $\omega_{i,1}=0 \forall i$. Since $\kappa$ is an isomorphism of quaternion algebras, we can write each basis element in the form $\omega_{i,0}I-\omega_{i,2}\sqrt{-b}Y+\omega_{i,3}\sqrt{-ab}X$. Then, we can see that at least two of the basis elements must be $K$-linearly dependent. Hence, we have a contradiction and so $r=2$. So $\mathcal{L}_\mathcal{O}$ spans $K^2$ as a $K$ vector space and clearly $\mathrm{rank}(\mathcal{L}_\mathcal{O})=2d$.
\end{proof}
Hence, we can conclude that there exists a $\z$-basis for $\mathcal{L}_\mathcal{O}$ and so also for $M_\mathrm{diag}$, which we denote $\{y_i\}$, for $i=0,\dots,2d-1$.

For the remainder of this section, let us consider orders of the form $\mathcal{O}=\sum\limits_{i=1}^4 O_K\omega_i$, with
\begin{equation}\label{eq:O-basis}
 \omega_1=I,\quad
\omega_2=\frac{I+iZ}{\sqrt{2}},\quad
\omega_3=\frac{I+iY}{\sqrt{2}},\quad
\omega_4=\omega_3\omega_2=\frac{I+iX+iY+iZ}{2},
\end{equation}
   %

 as for the Clifford$+$T and Clifford$+\sqrt{\textrm{T}}$ bases. In this case, we have \begin{equation*}
    M_\mathcal{O} = O_K +\frac{1+i}{\sqrt{2}}O_K,
\end{equation*}

which allows us to establish $M_\mathrm{diag}$ as a fractional $M_\mathcal{O}$ ideal. Moreover, for the bases considered in this paper, we have that $M_\mathrm{diag}$ is also principal, 
so
\begin{equation}\label{eq:fractional-ideal}
    M_\mathrm{diag} = \frac{1}{\xi}M_\mathcal{O},\quad \xi\in L.
\end{equation}

The definitions for $M_\mathcal{O}$ and $\xi$ corresponding to the V, Clifford$+\textnormal{T}$ and Clifford$+\sqrt{\textnormal{T}}$ bases are given in the Table \ref{tab3}.
\begin{table}[!h]
	\caption{$\xi, M_\mathcal{O}$ for V, Clifford$+T$ and Clifford$+\sqrt{T}$ gate sets.}\label{tab3}
	\begin{center}
		{ 
			\begin{tabular}{|c|c|c|}
				\hline 
				Gate set & $\xi$ & $M_\mathcal{O}$\tabularnewline
				\hline 
				\hline 
				V basis  &  1 & $O_L$\tabularnewline
				\hline 
				Clifford+$\mathrm{T}$ &  $\sqrt{2}$ & $O_{L}$\tabularnewline
				\hline 
				Clifford+$\sqrt{\mathrm{T}}$ &  $\sqrt{2}$ & $O_{K}+\frac{1+i}{\sqrt{2}}O_{K}$\tabularnewline
				\hline 
			\end{tabular}
		}
	\end{center}
\end{table}

\paragraph{Case 1: Diagonal Approximation}
For diagonal approximation (with and without fallback and mixing) the first normalized embedding $\sigma_1(m_1)/\sigma_1(\ell^N)$ falls in a two dimensional region, $R_\mathrm{approx}.$
Define the $2d\times 2d$ matrix $\Sigma_\mathcal{O}$ with rows:
\begin{eqnarray*}
\Sigma_\mathcal{O}^{(2j)} &=& \left(\mathrm{Re}(\sigma_j(y_0)),\dots,\mathrm{Re}(\sigma_j(y_{2d-1}))\right)\\
\Sigma_\mathcal{O}^{(2j+1)} &=& \left(\mathrm{Im}(\sigma_j(y_0)),\dots,\mathrm{Im}(\sigma_j(y_{2d-1})\right).
\end{eqnarray*}

So $\Sigma_\mathcal{O}$ is the matrix with entries corresponding to the real and imaginary components of the images of the $l_i$ under each of the $d$ homomorphisms. Let $\Lambda$ be the diagonal matrix with $\left(\sqrt{\sigma_1(\ell^N)},\sqrt{\sigma_1(\ell^N)} \dots,\sqrt{\sigma_d(\ell^N)},\sqrt{\sigma_d(\ell^N)} \right)$ on the diagonal. Then the operation $\Lambda\Sigma_\mathcal{O}z$ first embeds $z$ into the Euclidean space corresponding to $M_\mathrm{diag}$, then normalizes the result with respect to the norm $\ell^N$. Finding $m_1$ is now an integer point enumeration problem:

\begin{prob}
    Find $z\in\z^{2d}$ such that $\Lambda^{-1}\Sigma_\mathcal{O} z \in R_\mathrm{approx}\times D_1^{d-1}$.
\end{prob}
Each solution $z=(z_0,\dots,z_{2d-1})$ yields a candidate for $m_1$ by setting $m_1=z_0y_0+\dots+z_{2d-1}y_{2d-1}.$

\paragraph{Case 2: Magnitude Approximation}
For general unitary approximation the first normalized embedding $\sigma_1(m_1m_1^\ast)/\sigma_1(\ell^N)$ belongs to the interval $I_\mathrm{approx}$ and the remaining $d-1$ embeddings satisfy $\sigma_k(m_1m_1^\ast)/\sigma_k(\ell^N) \in  [0,1].$ 

We are looking for values $n=m_1m_1^\ast$ satisfying the above conditions, such that $m_1\in M_\mathrm{diag}$. Consider the set
\begin{equation*}
\{n:\exists m_1\in M_\mathrm{diag}\textnormal{ such that }m_1m_1^\ast=n\}
\end{equation*}
 and let $M_\mathrm{norm}$ be the set generated multiplicatively by the above set. 
  From Equation \eq{fractional-ideal}, we see that
\begin{equation*}
    M_\mathrm{norm} \subseteq \frac{1}{\xi\xi^\ast}O_K,
\end{equation*}
a fractional $O_K$ ideal. For this reason we can enumerate points $\hat{n}=\xi\xi^\ast n\in O_K$. 
Let $k_0,\dots,k_{d-1}$ be an integral basis for $K$ and define $\Sigma'$ as the $d\times d$ matrix with rows:
$$\Sigma'_j=\left(\sigma_j(k_0),\dots,\sigma_j(k_{d-1})\right).$$ Define $\Lambda'$ as the diagonal normalization matrix with $\left(\sigma_1(\xi\xi^\ast)\cdot\sigma_1(\ell^N),\dots,\sigma_d(\xi\xi^\ast)\cdot\sigma_d(\ell^N)\right)$ on the diagonal. Finding $\hat{n}$ is now an integer point enumeration problem in a parallepiped:
\begin{prob}
    Find $z\in\z^{d}$ such that $\Lambda'^{-1}\Sigma' z \in I_\mathrm{approx}\times [0,1]^{d-1}$.
\end{prob}

Each solution $z=(z_0,\dots,z_{d-1})$ yields a candidate for $\hat{n}$ by setting $\hat{n}=z_0k_0+\dots+z_{d-1}k_{d-1}.$ Recovery of $m_1$ requires a solution to the norm equation
\begin{equation*}
    \hat{m_1}\hat{m_1}^\ast = \hat{n},\quad \hat{m_1}\in M_\mathcal{O}.
\end{equation*}
Finally the candidate $m_1$ is defined as $m_1=\hat{m_1}/\xi$.

\subsubsection{Finding \texorpdfstring{$m_2$}{m2}: solving a norm equation}\label{sec:general-norm-equation-solution}
Finding a candidate for $m_2$ amounts to solving a norm equation, with the added constraint that the pair $(m_1,m_2)$ corresponds to a matrix in the order $\mathcal{O}.$
Given a candidate $m_1\in M_\mathrm{diag}$, we define the set containing valid candidates for $m_2$ as
\begin{equation*}M_{\mathrm{off-diag}}^{m_1} = \{m_2: \kappa(h(m_1)+h(m_2)\textbf{j})\in\mathcal{O}\}.
\end{equation*}
 To satisfy the determinant condition, we require 
 \begin{equation}\label{eq:initial-norm-equation} m_2m_2^\ast=\ell^N-m_1m_1^\ast,\quad m_2\in M^{m_1}_\mathrm{off-diag}.\end{equation}

 Any element $m_2\in M^{m_1}_\mathrm{off-diag}$ belongs to the larger set \begin{equation*}
    M_\mathrm{off-diag} = \set{m_2:\exists m_1\textnormal{ s.t. }\kappa(h(m_1)+h(m_2)\textbf{j})\in\mathcal{O}},
\end{equation*}
which, as shown for $M_\mathrm{diag}$, is a fractional $M_\mathcal{O}$ ideal.  In the following discussion, we show that a solution for $m_2$ can be recovered from a related norm equation, in which we solve for elements in $M_\mathcal{O}$, under the assumption that $M_\mathrm{off-diag}$ is moreover a principal fractional ideal. That is, 
\begin{equation*}M_\mathrm{off-diag} = \frac{1}{\xi'}M_\mathcal{O},\quad\xi'\in L.\end{equation*}

Fix $m\in M_{\mathrm{off-diag}}^{m_1}$. For any other $m'\in M^{m_1}_\mathrm{diag}$ we have $
    \kappa(h(m)\textbf{j}-h(m')\textbf{j})\in\mathcal{O}
$. Therefore, we write $M_{\mathrm{off-diag}}^{m_1}=m + M^{0}_\mathrm{off-diag}$, where $M^{0}_\mathrm{off-diag}$ is the principal fractional $M_\mathcal{O}$ ideal $M^{0}_\mathrm{off-diag}=\set{m':\kappa(h(m')\textbf{j})\in\mathcal{O}}.$
 We take
\begin{equation*}
   M^{0}_\mathrm{off-diag} = \frac{1}{\chi}M_\mathcal{O},\quad\chi\in L.
\end{equation*}
Note that a representative $m$ is found by considering the quotient lattice $M_\mathrm{off-diag}/M^0_\mathrm{off-diag}$.
The definitions for $\xi,\xi'$ and $\chi$ corresponding to the V, Clifford$+\textnormal{T}$ and Clifford$+\sqrt{\textnormal{T}}$ bases are given in the Table \ref{tab4}.
\begin{table}[!h]
	\caption{Fractional ideal representatives for V, Clifford$+T$ and Clifford$+\sqrt{T}$ gate sets.}\label{tab4}
	\begin{center}
		{ 
			\begin{tabular}{|c|c|c|c|c|}
				\hline 
				Gate set & $\xi$ & $\xi'$ &$\chi$&$M_\mathcal{O}$ \tabularnewline
				\hline 
				\hline 
				V basis  &  1 & 1&1&$O_L$\tabularnewline
				\hline 
				Clifford+$\mathrm{T}$ &  $\sqrt{2}$ & $\sqrt{2}$&1&$O_L$\tabularnewline
				\hline 
				Clifford+$\sqrt{\mathrm{T}}$ &  $\sqrt{2}$ & $\sqrt{2}$&1&$O_K+\frac{1+i}{\sqrt{2}}O_K$\tabularnewline
				\hline 
			\end{tabular}
		}
	\end{center}
\end{table}

The norm equation in Equation \eq{initial-norm-equation} can now be reformulated to look for a solution in $M_\mathcal{O}$.

\begin{prob}\label{prob:general-norm-quotient}
    Given $\hat{z}/\xi'\in M_\mathrm{off-diag}, m_1 \in M_\mathrm{diag}$, find $z\in M_\mathcal{O}$ such that $$\abs{\frac{\hat{z}}{\xi'}+\frac{z}{\chi}}^2=\ell^N-m_1m_1^\ast.$$
\end{prob}
A solution $z$ yields a candidate for $m_2$ by setting $m_2=\hat{z}/\xi'+z/\chi$.
Since $m_1=\hat{m_1}/\xi$ for some $m_1\in M_\mathcal{O}$, if $\xi=\xi'$ and $\chi=1$, then \problem{general-norm-quotient} is simplified to:
\begin{prob}\label{prob:general-norm-quotient-simplified}
    Find $z\in M_\mathcal{O}$ such that $\abs{\hat{z}+\xi z}^2=\xi\xi^\ast\ell^N-\hat{m_1}\hat{m_1}^\ast$, where $\hat{z},\hat{m_1}\in M_\mathcal{O}.$
\end{prob}
Clearly the V, Clifford$+$T and Clifford$+\sqrt{\textnormal{T}}$ bases admit this simplified case. Of course, solving the norm equation for the V basis is already straightforward, but is included here for completeness.
\begin{rem}\label{rem:norm-eq-quotient}
By applying the variable substitution $z'=\hat{z}+\xi z$, we see that \problem{general-norm-quotient} is equivalent to solving \begin{equation}\label{eq:norm-eq-quotient}\abs{z'}^2=r\in O_K,\quad z'\in \hat{z} + \xi M_\mathcal{O},\end{equation} where $r = \xi\xi^\ast\ell^N-\hat{m_1}\hat{m_1}^\ast$. In other words, $z'$ must lie in the same quotient in $M_\mathcal{O}/\xi M_\mathcal{O}$ as $\hat{z}.$ 
\end{rem}

We discuss solving these norm equations in \sec{normeqsolv}. In particular, for the special case of fields with class number equal to 1, we suggest a simplified solution for Equation \eq{norm-eq-quotient}.

\subsection{Heuristic approximation cost scaling with accuracy}\label{sec:cost-scaling}

\change{Having presented a method for synthesising arbitrary unitaries, we now discuss the cost of our algorithm. We establish a heuristic scaling of the power cost function with the area of the 2D or 1D regions related to the six approximation problems considered in \cref{sec:approximation-problems}.}
All gate sets we consider are related to integral quaternion with norm $\ell^N$ for some $\ell$ fixed by the gate set and $N$ being a power cost of the given approximating quaternion.
Let $R_{\varepsilon,q}$ be a 2D or 1D region with $\varepsilon$ being diamond norm accuracy and $q$ being success probability, 
During the point enumeration step of our algorithms for 2D problems, we are looking for integer points of dimension $2d$ in 
the bounded subset of $\mathbb{R}^{2d}$ given by equation below: 
$$
 (a_0,\ldots,a_{2d-1}) \in \Lambda \Sigma^{-1}_\mathcal{O} \left( R_{\varepsilon,q} \times D_1 \times \ldots D_{d-1} \right)
$$
Now, applying the Gaussian heuristic, we assume that there exist an integer point in the subset of $\mathbb{R}^{2d}$ when the volume of the subset is $1$.
Taking into account that $\det{\Lambda} = \mathrm{Nrm}(\ell)^N$ we get the following condition for the existence of integer points:
$$
N \log(\mathrm{Nrm}(\ell)) + \log(\mathrm{Area}(R_{\varepsilon,q})) + \log(\pi^{d-1}/\det(\Sigma_\mathcal{O})) = 0
$$ 
Define $b = \mathrm{Nrm}(\ell)$, then the Gaussian heuristic implies power cost scaling: 

$$
N = -\log_b(\mathrm{Area}(R_{\varepsilon,q})) + \log_b(\det(\Sigma_\mathcal{O})/\pi^{d-1})
$$
The relation between power cost and area for 1D problems is similarly
$$
N = -\log_b(\mathrm{Length}(I_{\varepsilon})) + O(1).
$$

Now we specialize above calculation to Clifford+$T$ and Clifford+$\sqrt{T}$ gate sets and specific approximation problems.
Recalling that for Clifford+$T$ and Clifford+$\sqrt{T}$ logarithm base $b =2$ and using expression for the 
regions areas in \cref{tab:approx-region-areas}, we derive heuristic power cost scaling expression in the top half of \cref{tab:approximation-cost-scaling}.
We note that for projective (fallback) rotation approximation $N$ is $\log_b(1/(1-q)\cdot1/\varepsilon) + O(1)$,
and for magnitude approximation $N$ is $\log_b(1/\varepsilon) + O(1)$.
We expect that magnitude approximations are shorter by the additive constant $\log_b(1/(1-q))$, where $1-q$ is the fall-back step probability of the fall-back protocol.
Heuristically, we assume that increasing volume by a constant factor or $\log(1/\varepsilon)$ factor ensures that we can find integer points 
for which the corresponding norm equations are solvable. This does not affect constant in front of $\log_b(1/\varepsilon)$.

In applications we are interested in two other cost metrics for our gate sequences: non-Clifford gate count ( that we simply call gate count) and 
T-count, that is the number of T states needed to executed given sequence.
The gate count is a good proxy for how fast we can execute the sequence, where each non-Clifford 
gate is executed using circuit from \href{https://arxiv.org/pdf/1808.02892.pdf#figure.33}{Figure~33} in \cite{GameOfSurfaceCodes}.
The T-count is a good proxy for space-time volume needed to execute the sequence on a fault-tolerant quantum computer,
because the space-time volume required is typically dominated by the space-time volume needed to distill $T$ states.
For Clifford+$T$ approximations these two other cost metrics are equal to the power cost.
It remains to estimate them for Clifford+$\sqrt{T}$ approximations.
We assume that number of $\sqrt{T},\sqrt{T}^3$ gates denoted by $N_{\sqrt T}$ in our sequences is the same 
as number of $T$ gates. This is justified by our numerical results.
Recall that $\sqrt{T},\sqrt{T}^3$ contribute three to the power cost and $T$ contributes $2$. 
For this reason we have $N_{\sqrt T} = 0.2 N$ and gate count is $0.4 N$.
To estimate $T$-count we assume that every $\sqrt{T}$ and $\sqrt{T}^3$ gate can be execute using four $T$ states. 
This is because the circuit from \href{https://arxiv.org/pdf/1808.02892.pdf#figure.33}{Figure~33} in \cite{GameOfSurfaceCodes}.
consumes one $\sqrt{T}$ state and one $T$ state.
Producing one $\sqrt{T}$ state requires 3 $T$ states in the worst case using catalysis protocol described in \href{https://arxiv.org/pdf/1904.01124.pdf#figure.caption.7}{Figure 6a} in \cite{LowerBounds2020}.
We see that T-count for Clifford+$\sqrt{T}$ approximations heuristically scales the same way as power cost.
Above implies heuristic cost scaling expressions in \cref{tab:approximation-cost-scaling}.

\section{Integer point enumeration problems}\label{sec:point-enum}

In \sec{approx-solutions}, we described a general method for solving approximate synthesis problems on quaternion gate sets, with three examples from commonly used gate sets. In this section, we focus on the first step in that method: integer point enumeration in a convex region. Where relevant, we re-use the notations introduced in previous sections.

\change{To provide context for this section, we will review integer point enumeration methods previously used in circuit synthesis literature. The simplest approach to finding an integer point inside a convex body is to fit a box of size one into the convex body and round each of the coordinates of the center of the box to the nearest integers. This approach only works for approximating with $V$ basis and is similar to the randomized algorithm in \cite{BGS}. However, even for $V$ basis, this approach produces sub-optimal scaling of the sequence length with the accuracy $\varepsilon$. For Clifford+${T}$, the convex bodies that appear in the approximation problem do not typically contain a size one box.}

\change{In~\cite{Selinger}, Selinger found a way to use the unit group of $\mathbb{Z}[\sqrt{2}]$ to reshape the convex bodies so that they contain a size one box. Later, this approach was generalized to work for other gate sets by relying on the approximate closest vector problem in a unit group lattice of the rings of integers related to the gate-set in \cite{Kliuchnikov2015b}. When applied to the diagonal approximation problem, these approaches produce sub-optimal scaling of the sequence length with the accuracy $\varepsilon$ (for the most target angles) because they do not find all the points inside the convex body. It turns out that this approach is more useful for magnitude approximation problems for any gate-set, and we discuss it in detail in \cref{sec:parallelotope-enumeration}.}

\change{Ross and Selinger~\cite{RossSelinger2014} designed an efficient integer point enumeration algorithm for problems related to the Clifford+${T}$ gate-set that finds all the integer points. In~\cite{Parzanchevski}, the authors point out that one should use the polynomial integer point enumeration algorithm \cite{lenstra1983integer} for all gate-sets. In \cref{sec:general-enum}, we review the algorithm from \cite{lenstra1983integer} and highlight the modification to include convex bodies with quadratic constraints. We use this modification in our implementation and are the first to report on the use of generic integer point enumeration for circuit synthesis problems in practice.}

Recall that the conditions on $m_1$, as defined in \sec{approx-solutions}, define a target region in which we want to enumerate integer vectors. \sec{general-enum} outlines an algorithm for integer point enumeration in convex bodies of a particular form and shows how this can be applied to the target regions prescribed by approximate synthesis. 
In the case of magnitude approximation, the target region is a parallelotope.  \sec{parallelotope-enumeration} gives an alternative method in that case, making use of the number-theoretic structure arising from the quaternion gate sets. 

\subsection{General point enumeration}\label{sec:general-enum}
In this section we describe an approach to solving integer point enumeration problems in a subset of convex bodies, and we apply this method to the regions arising from approximate synthesis. The algorithm is from Lenstra \cite{lenstra1983integer}, and it applies to  problems with the following general form. 

\begin{prob}[Integer point enumeration in a convex body]\label{prob:convex-enum}
Let $R$ be a bounded convex body of positive volume satisfying $R=\{x\in\r^d: Ax \le b\}$,
where $A$ is an invertible $d\times d$ matrix and $b$ is a $d$-dimensional column vector. Find all $x\in  R\cap\z^d$.
\end{prob}
The inequality in \problem{convex-enum} denotes an element-wise comparison between two vectors.

\begin{rem}[Target regions defined by approximate synthesis problems]
The target regions defined by the approximation problems are not necessarily of the form in \problem{convex-enum}. For instance, the target region defined by \cref{prop:fallback-approximation}, illustrated in \cref{fig:fallback-condition}, is \textit{not} convex. In these cases, we can take the convex hull of $R$ and apply Lenstra's algorithm, discarding any solutions not in $R$.
\end{rem}

For the remainder of this section, we assume that $R$ is of the form required by~\problem{convex-enum}.
In theory, we could find maximum and minimum bounds for $R$ in each dimension, thus defining a $d$-dimensional box (\href{https:\/\/en.wikipedia.org\/wiki\/Hyperrectangle}{hyperrectangle}) $C_R$ such that $R\subset C_R$. A solution to \problem{convex-enum} is then found by enumerating all integer points in the box and checking each point to see if it lies in $R$. In practice, this strategy is less than optimal, as there may be numerous points in the box, but the set $R$ might contain no integer points. For an example, see~\cref{fig:flat-parallelogram}. Lenstra's algorithm circumvents this problem by using a constructive version of Khinchin's Flatness Theorem described below.

To state the Khinchin's Flatness Theorem we define the width of a convex body.
For a non-empty convex body $R$, the \emph{width} of $R$ along a vector $r$ is defined as 
$w_r(R)=\max\limits_{x\in R}\{r^Tx\} - \min\limits_{x\in R}\{r^Tx\}$. 
The width of $R$ is  $w(R):=\min\limits_{r\in\z^d, r\ne 0}\{w_r(R)\}$, the vector $r_{\min}(R)$ where the minimum is achieved is the flat direction of $R$.
The Flatness Theorem (\cref{flatness}) guarantees that the solution to \problem{convex-enum} is non-empty if $w(R)$ is above some constant $\omega(d)$. Banaszczyk proved that $\omega(d) =O(d)$ \cite{banaszczyk1995inequalities}.

\begin{thm}[Khinchin's Flatness Theorem, attributed to \cite{khinchin1948quantitative}]\label{flatness}
Let \change{$R \subseteq \r^d$} be a full-dimensional non-empty convex body. 
Either $R$ contains an integer point, or $w(R) \le \omega(d)$, where $\omega(d)$ is a constant depending on the dimension only.
    \end{thm}

Above theorem shows that in the case of no integer points in $R$, convex body $R$ must have small width. Suppose now that the flat direction is known $r_{\min}(R)$ and we width is small.
Next we show how $r_{\min}(R)$ is used to find an integer in $R$.

Integer point enumeration in a d-dimensional convex body $R$ reduces to several instances of integer point enumeration in a $(d-1)$-dimensional convex body.
It is convenient to represent $\z^d$ as disjoint union of $d-1$ dimensional lattices in $\r^d$
$$
 \z^d = \bigcup_{k \in \z} \{ z : z^T r_{\min}(R')=k,\, z \in \z^d \} 
$$
with each lattice $L_k = \{ z : z^T r_{\min}(R')=k,\, z \in \z^d \}$ contained in the hyperplane $H_k = \{x\in\r^d: r_{\min}(R')^Tx = k\}$.
The proposition below bounds the number of such hyperplanes intersecting $R$.

\begin{prop}\label{flat-cor}
Let \change{$R\subseteq\r^d$} be a full-dimensional non-empty convex body. The number of hyperplanes of the form $H_k = \{x\in\r^d: r^Tx = k\},\, k\in\z$ intersecting  $R$ is bounded by $w_r(R)+1$.
\end{prop}

Using the flat direction of $R$, we have reduced finding an integer point in $d$-dimensional convex body $R$ to finding an integer point in at most $w(R)+1$ $(d-1)$-dimensional convex bodies $R \cap H_k$. Using this approach the algorithm for finding an integer point in convex body $R$ (or determining that $R$ has no integer points) terminates 
in polynomial time when $d$ is fixed. It remains to discuss an algorithm for finding the flat direction of $R$.

We discuss an efficient approach finding an approximation to the flat direction of $R$ instead of the flat direction of $R$. More precisely, we will compute $r$ from $\z^d$ such that ratio $w_r(R) / w(R)$ is not too big. Such $r$ can be used instead of $r_{\min}(R)$ for the reduction to $(d-1)$-dimensional integer point enumeration problems discussed above.
Let us transform $R$ so that it is roughly spherical in shape, via an invertible $d\times d$ matrix $\tau$. Each point $y$ in $\tau R\cap\tau\z^d$ corresponds to a point \change{$x\in R\cap\z^d$} by $x = \tau^{-1}y$. More precisely, we can find $\tau$ such that $\tau R$ satisfies 
\begin{equation}
\label{eq:spherical-condition}B(p,\delta)\subset \tau R \subset B(p, \Delta)
\end{equation}
 for balls with center point $p$ and radii $\delta$ and $\Delta$, such that the ratio $\frac{\Delta}{\delta}$ is bounded by a constant dependent only on $d$, $c_1(d)$; Lenstra takes $c_1(d)= 2d^{3/2}$. See \fig{balls} for an illustration. 
 For the regions considered in this work it is easy to find $\tau$ with a better ratio $\Delta/\delta$. 
Using above we see that $R$ contains an ellipsoid $R^\prime = \tau^{-1}B(p,\delta)$, for which width can be calculated more efficiently:
$$
 w(R')=w(\tau^{-1}B(0,\delta)) = \min_{r \in \z^d, r\ne 0} 2\delta \cdot \max_{\nrm{x} \le 1} r^T \tau^{-1} x = 2\delta \cdot \min_{r \in \z^d, r\ne 0} \nrm{ (\tau^{-1})^T r} 
$$
In other words, finding the flat direction $r'_{\min}$ of $\tau^{-1}B(p,\delta)$ is equivalent to finding the shortest vector of 
lattice $(\tau^{-1})^T \z^d$ which is the dual lattice of $\tau \z^d$.
Then applying \cref{flatness} to $R^\prime$ determines whether an integer point in $R$, if it exists, lies in $R^\prime$ or the width of $R'$ is bounded by $\omega(d)$.  In the latter case,
using \cref{eq:spherical-condition}, we see that $w(R)$ is bounded by $\Delta/\delta\cdot w(R') \le c_1(d)\omega(d)$, that is a constant dependent on the dimension only. 
Lenstra's algorithm (\algo{gen-enum}) uses approximation to the shortest vector of $(\tau^{-1})^T \z^d$.
More precisely, let $b_1,\ldots ,b_d$ being an LLL-reduced basis\footnote{\change{So named for the authors Lenstra, Lenstra and Lov{\'a}sz for their basis reduction algorithm \cite{lenstra1982factoring}.}} of lattice $\tau \z^d$,
and let $b^\ast_1,\ldots ,b^\ast_d$ be Gram-Schmidt orthogonalization of $b_1,\ldots ,b_d$, 
then vector $b^\ast_d / \nrm{b^\ast_d}^2$ belongs to the dual lattice  $(\tau^{-1})^T \z^d$
and is an approximation to the shortest vector of the dual lattice.

\begin{figure}[h!]
  \begin{subfigure}{0.49\textwidth}
    \centering
    \includegraphics[height=0.7\textwidth]{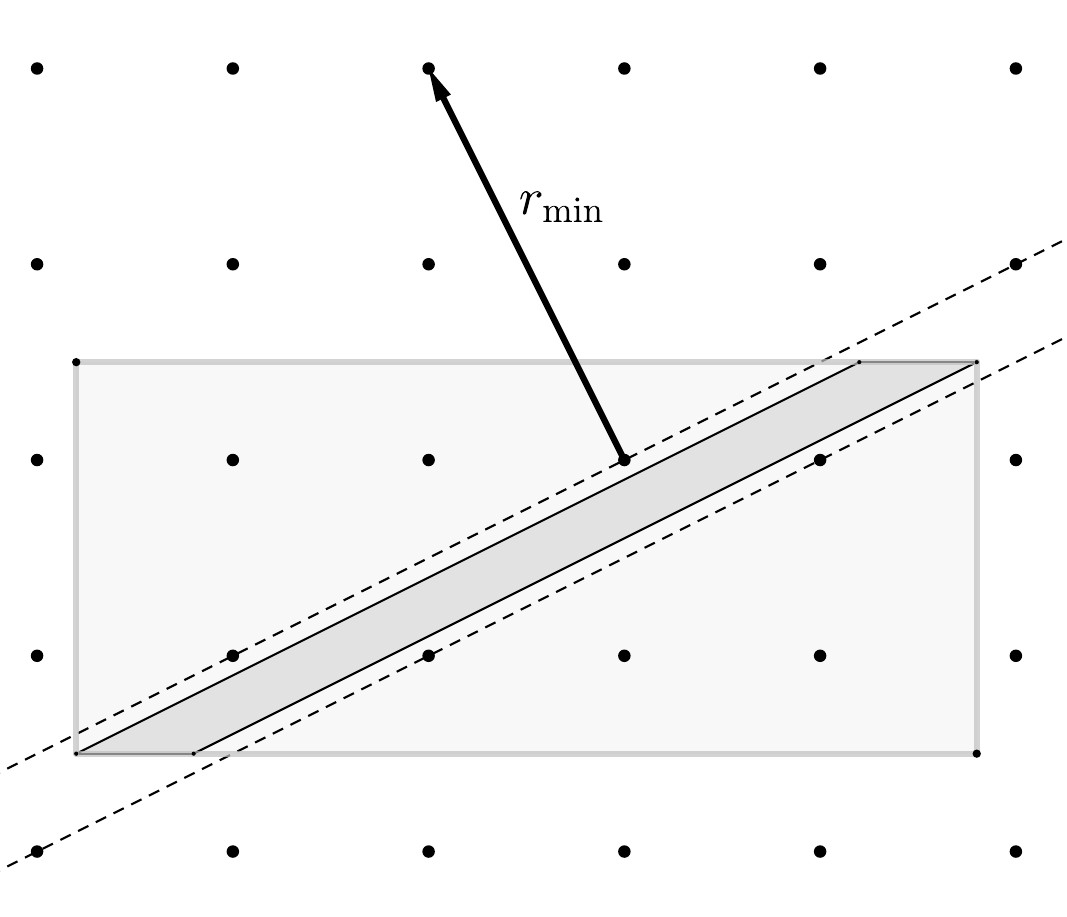}
    \caption{\label{fig:flat-parallelogram}Parallelogram $R$ with no integer points and bounding box $C_R$ with many integer points. Vector $r_{\min}(R)$ is the flat direction of the parallelogram.}
  \end{subfigure}
  \hfill
  \begin{subfigure}{0.49\textwidth}
    \centering
    \includegraphics[height=0.7\textwidth]{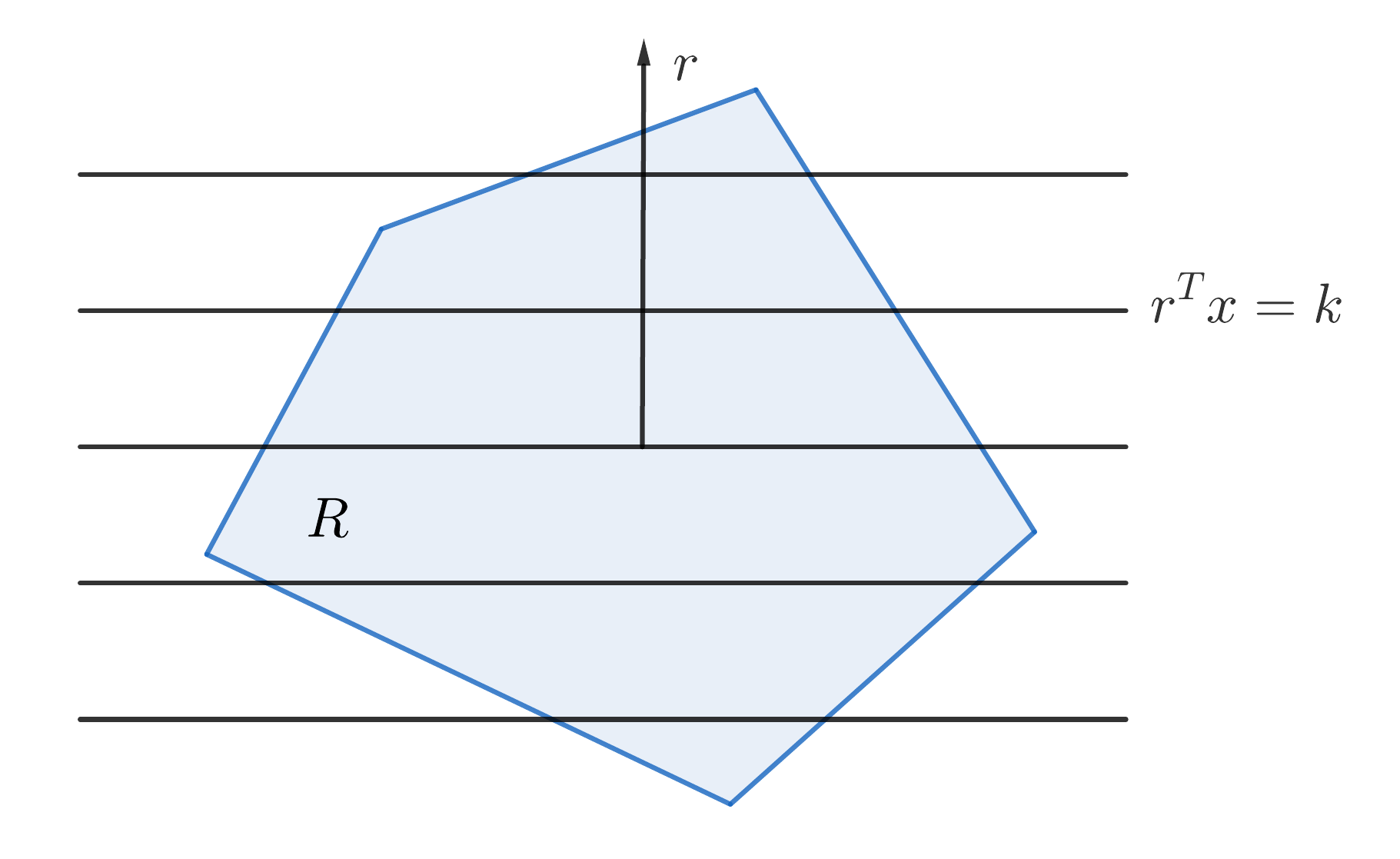}
    \caption{\label{fig:hyperplanes} The hyperplanes $\{x\in R:r^Tx=k\}$ for $k \in \z$ intersecting $R$.}
  \end{subfigure}
  \caption{Integer point enumeration in convex bodies and the flat direction.}
\end{figure}

 %

\begin{algorithm}[!h]
\SetKwData{Left}{left}\SetKwData{This}{this}\SetKwData{Up}{up}
\SetKwFunction{Union}{Union}\SetKwFunction{FindCompress}{FindCompress}
\SetKwInOut{Input}{Input}\SetKwInOut{Output}{Output}
\LinesNumbered
\Input{A $n\times d$ real-valued matrix $A$ and a $d$-dimensional column vector $b$ defining the convex body $R: = \{x\in\r^d: Ax\le b\}$.}
\Output{A subset $X \subset \{x:x\in R\cap\z^d\}$}
\BlankLine
$X\leftarrow\emptyset$\;
Compute $\tau$ such that $\tau R$ is ``roughly spherical'' with center $p$ as in \cref{eq:spherical-condition}\label{step:tau}\;
$\mathcal{L}\leftarrow \tau\z^d$, with LLL reduced basis $b_1,\dots,b_d$ such that $\abs{b_1}\le\dots\le\abs{b_d}$\label{step:basis}\;
$b_1^\ast, \ldots, b_d^\ast$ is a Gram Schmidt Orthogonalised basis  $b_1,\dots,b_d$ \;
$r_{min} \leftarrow$ coordinates of $b^\ast_d / \abs{b^\ast_d}^2$ in basis $(\tau^{-1})^T$, an approximation to $r_{\min}(R)$ \;
$k_{min}\leftarrow\ceil{\min\limits_{y\in R}\{(r_{\min})^T y\}}$, using linear programming\;
$k_{max}\leftarrow\floor{\max\limits_{y\in R}\{(r_{\min})^Ty \}}$, using linear programming\;
Let $T$ be invertible integer matrix such that $\tau T$ is matrix with columns $b_1,\ldots,b_d$, let $\tilde{A} = A T$ \;
\For{$k_{min}\le k \le k_{max}$\label{step:hyper}}{
$\tilde{x}_{d} \leftarrow k$\;
$A'\leftarrow \tilde{A}_{[1:d-1]}$, the $n\times (d-1)$ matrix consisting of the first $d-1$ columns of $\tilde{A}$\;
$b'\leftarrow b-\tilde{A}_{[d]}k$, where $\tilde{A}_{[d]}$ denotes the $d^\text{th}$ column of $\tilde{A}$\;
Run \algo{gen-enum} on inputs $A'$ and $b'$ to obtain $X' \subset \{\tilde{x}:=(\tilde{x}_1,\dots, \tilde{x}_{d-1})\in \z^{d-1}: A'x'\le b'\}$\;
$X\leftarrow X\cup T \{(\tilde{x}_1,\dots,\tilde{x}_d): (\tilde{x}_0,\dots,\tilde{x}_{d-1})\in X'\}$\;}
Output $X$\label{step:outenum}\;
\caption{Integer point enumeration in a bounded, positive-volume convex body satisfying $R=\{x\in\r^d: Ax \le b\}$. \label{alg:gen-enum}}
\end{algorithm}

\begin{figure}[h!]
    \centering
    \includegraphics[width=0.35\textwidth]{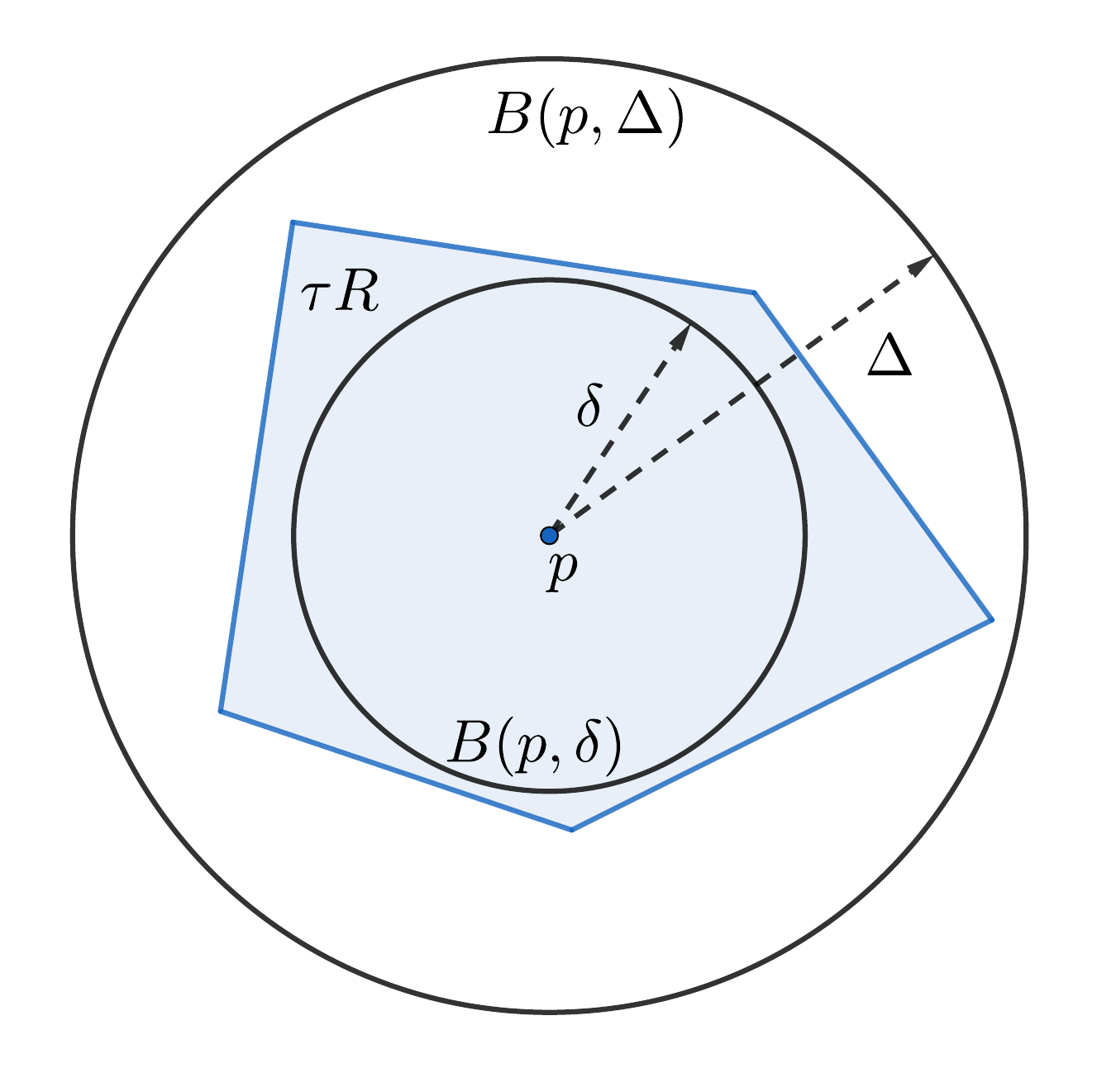}
    \caption{The concentric balls $B(p,\Delta)$ and $B(p,\delta)$, centered at $p$ with radii $\Delta$ and $\delta$, respectively, such that $B(p,\delta)\subset \tau R \subset B(p, \Delta)$ and $\Delta/\delta\leq c_1(d)$.}
    \label{fig:balls}
\end{figure}



For fixed dimension $d$, Lenstra's algorithm (\algo{gen-enum}) runs in polynomial time in the length of the input \cite{lenstra1983integer}. This algorithm can be modified 
to enumerate integer points in $R$ and require polynomial time per point.
In practice, we use quadratic constraints $(x-p)^T Q(x-p) \le 1$ for point $p$ and symmetric matrix $Q$ in \cref{prob:convex-enum} in addition to the linear constraints. 
The approach in this section can be modified to handle quadratic constraints. 
In this case, in \algo{gen-enum}, $k_{\min}, k_{\max}$ are found by solving quadratic optimization 
problems as opposed to using linear programming.



\subsection{Point enumeration in a parallelotope}
\label{sec:parallelotope-enumeration}
In this section, we will show how to exploit the number-theoretic structure present in the integer ring $O_K$ for the instances of point enumeration relating to the magnitude approximation problems. 
These have a special shape: as described in \sec{approx-solutions}, point enumeration occurs in a $d$-dimensional box $I\times[0,1]^{d-1}$ where $d=[K:\q]$. 
%
Specifically, we are interested in the following problem.

\begin{prob}[Box Enumeration Problem]\label{prob:box-enumeration}
Let $K$ be a totally real number field of degree $d$ and let $O_K$ be its ring of integers. Given real numbers $\{g_j,h_j \ |\ g_j< h_j, j =1,\ldots,d\}$,
find $n$ in $O_K$ contained in the box, that is $n$ such that $\sigma_j\left(n\right) \in [g_j,h_j]$ or determine that no such $n$ exists.
\end{prob}

We rely on the same point enumeration approach as in the previous section, however in this special case it is easier to predict the number of integer points in the box.
Interestingly the number of elements $n$ from $O_K$ in the box is proportional to the volume of the box $\prod_j (h_j-g_j)$.
This observation was first made in \cite{Selinger} and later generalized in \cite{Kliuchnikov2015b}.
The lower-bound on the number of integer points is summarized by the following proposition.

\begin{prop}
Let $K$ be a totally real number field, 
there exists a constant $V_0$ dependent only on $K$ such that every box enumeration problem~(\cref{prob:box-enumeration}) with box volume at least $V_0$ has 
at least one solution.
\end{prop}

A corollary of the above proposition, is that any box contains at least $\lfloor \prod (h_j-g_j)/V_0 \rfloor $ elements of $O_K$.
In the rest of this section we sketch the proof of the proposition. 
The proof proceeds in two steps.
First we observe that if the box contains a certain parallelotope $P$ centered at zero shifted by any vector $t$, then it must contain an element of $O_K$.
Second we show that when the volume of the box is at least $V_0$, then the point enumeration problem is equivalent 
to point enumeration in another box $[g'_1,h'_1]\times \ldots \times [g'_d,h'_d]$ such that this box contains the parallelotope $P$ translated by the center of the box.

Recall that the ring of integers $O_K$ corresponds a $d$-dimensional lattice $L$ in $\r^d$.
Let $k_i$ be an integral basis for $O_K$ and $\sigma_i$ be embedding of $K$ into $\r$, the lattice basis is given by 
$$
 b_j = (\sigma_1(k_j),\ldots,\sigma_d(k_j))^T, \text{ for } j=1,\ldots,d
$$
and the corresponding basis matrix $B$ has columns $b_j$. Recall that the parallelotope 
$$
 C(B) = \left\{ Bx : x_k \in [-1/2,1/2) \right\}
$$
translated by lattice points from $L$ defines partition of $\r^d$, that is for any point $x$ in $\r^d$, 
there is a unique lattice point $n$, such that $x \in n + C(B)$. 
Equivalently, there is always a lattice point in $x - C(B)$.
For this reason, if a box contains $P = -C(B)$, it must 
contain at least one lattice point. It is easy to find such a point by computing $B^{-1} x$ 
and rounding all the coordinates to the nearest integer.
For the box to contain the shifted parallelotope $P= -C(B)$ the following constraints should hold:
$$
 h_j - g_j \ge \max_{x \in C(B)} x_j - \min_{x \in C(B)} x_j \text{ for all } j = 1,\ldots,d
$$
This completes the first step of the proof.

For the second step of the proof we use units and the unit group of $O_K$.
Let $u$ be a unit of $O_K$, that is $u^{-1}$ is also in $O_K$.
First note that $z$ from $O_K$ is contained in the box $[g_1,h_1]\times \ldots \times [g_d,h_d]$, 
if and only if $u z$ is contained in the transformed box $(\sigma_1(u) [g_1, h_1])\times \ldots \times (\sigma_d(u)[g_d,h_d])$. 
Integers $u z$ belong to the box with dimensions $|\sigma_j(u)|(h_j - g_j)$.
Note that the overall volume of the box is the same because $\prod_j |\sigma_j(u)| = 1$.

The next step is to show that if the box has sufficiently big volume, we can always 
find a unit $u$, such that the transformed box contains shifted copy of $C(B)$, that is:
\begin{equation} \label{eq:unit-conditions}
|\sigma_j(u)|(h_j - g_j) \ge \max_{x \in C(B)} x_j - \min_{x \in C(B)} x_j \text{ for all } j = 1,\ldots,d
\end{equation}
For this we use the unit group of $O_K$.
Recall that according to Dirichlet's Unit Group theorem, any unit of $O_K$ can be written as: 
$$
u = \pm u_1^{m_1} \ldots u_{m_1}^{m_{d-1}}, \text{ for } m_j \in \z
$$
Substituting this expression for $u$ in terms of $u_k$ above into \cref{eq:unit-conditions}
and taking $\log$ of both sides of the inequality gives us:
\begin{equation} \label{eq:shifted-lattice}
\sum_{i=1}^{d-1} m_i \log|\sigma_j(u_i)| + \log(h_j - g_j) \ge \log(\max_{x \in C(B)} x_j - \min_{x \in C(B)} x_j) \text{ for all } j = 1,\ldots,d
\end{equation}
Let us discuss the geometric interpretation of the above inequality.
Vectors 
$$
b'_i = (\log|\sigma_1(u_i)|,\ldots,\log|\sigma_d(u_i)|)^T, \text{ for } i = 1,\ldots,d-1
$$
define a $d-1$-dimensional lattice $L'$ in $\r^d$ with basis matrix $B'$ (known a the unit lattice of $O_K$), contained in sub-space $x_1 + \ldots + x_d = 0$.
For the inequality to be true, the intersection between the shifted lattice 
$$
(\log(h_1 - g_1),\ldots,\log(h_1 - g_1))^T + L'
$$
and the direct product of half-open intervals 
$$
 R =  [\log(\max_{x \in C(B)} x_1 - \min_{x \in C(B)} x_1),+\infty)\times \ldots \times  [\log(\max_{x \in C(B)} x_d - \min_{x \in C(B)} x_d),+\infty)
$$
must be non-empty.
Note that the shifted lattice is contained in the subspace $x_1 + \ldots + x_d = \log \prod_j(h_j - g_j) $ which is determined by the box volume $ V = \prod_j(h_j - g_j)$.
To ensure that the inequalities in \cref{eq:shifted-lattice} have a solution, 
it is sufficient to ensure that the intersection of $R$ and the subspace $x_1 + \ldots + x_d = \log  V$ contains a shifted parallelotope $t + C(B')$ for some shift vector $t$.
The bigger the volume $V = \prod_j(h_j - g_j)$, the bigger intersection between $R$ and subspace $x_1 + \ldots + x_d = \log V$. 
In other words, there exists $V_0$, such for all $V \ge V_0$ intersection between $R$ and subspace $x_1 + \ldots + x_d = \log V$ contains $t + C(B')$ for some shift vector $t$.
This completes the proof.

In the above proof, the value $V_0$ depends on the choice of integral basis of $O_K$ and fundamental units $u_1,\ldots,u_{d-1}$.
Using reduced bases for $O_K$ and the unit lattice can improve $V_0$. 
Further improvements can be achieve by using different fundamental domains for the lattices.
One can replace $C(B)$ with $C(B^*)$, where $B^*$ is the matrix corresponding to Gram-Schmidt orthogonalisation of basis $b_1,\ldots,b_d$. 
In this case rounding coordinates of $B^{-1} x$ is replaced by the Nearest-Plane algorithm. 
For further improvement, one can replace $C(B)$ with lattice's Voronoi cell and rounding with solving the Closest Vector Problem.
Using an approach similar to the one described above one can also show the following:

\begin{prop}
Let $K$ be a totally real number field, 
there exists a constant $V'_0$ dependent only on $K$ such that every box enumeration problem~(\cref{prob:box-enumeration}) with box volume at most $V'_0$ has 
at most one solution.
\end{prop}

\section{Relative norm equations}\label{sec:normeqsolv}

\change{In Section \ref{sec:approx-solutions}, we explain how each solution to point enumeration gives rise to a relative norm equation, which must be solved to complete a solution to some approximation problem. In this section, we provide a general approach for solving such relative norm equations.}

\change{We begin by giving a brief overview of the approaches to the relative norm equation problem stated below, which has been studied in the circuit synthesis literature:}

\begin{prob}\label{prob:rel-norm}
Let $K$ be a totally real number field with extension $L=K(i)$, such that $i^2=-1$. Let $O_K$ and $O_L$ be the respective integer rings. Given $r$ from $O_K$, compute $m$ from $O_L$ such that $mm^\ast = r$, or determine that no such $m$ exists.
\end{prob}

\change{In \cite{Selinger,RossSelinger2014}, the above problem is solved for $K=\mathbb{Q}(\sqrt{2})$ by relying on the Euclidean algorithm in $\mathbb{Z}[e^{i\pi/4}]$. The algorithm is polynomial time, assuming access to a factoring oracle. In \cite{KliuchnikovMaslovMosca2013}, the authors rely on a more general approach to solving the relative norm equation problem, which is available in the PARI/GP software package and is based on S-units. In \cite{Kliuchnikov2015b}, an algorithm that works for a wide range of number fields is described. The algorithm is shown to be polynomial time, assuming access to a factoring oracle, and works for number fields with a non-trivial class group. Although the exposition in \cite{Kliuchnikov2015b} is very technical, we provide here a simplified version of a similar algorithm that works for the special cases common in practice.}

\change{In our analysis of Clifford+T and Clifford+$\mathrm{\sqrt{T}}$ gate sets, we find that we need to solve a more general problem than Problem \ref{prob:rel-norm}. Otherwise, we would miss some of the solutions to the circuit synthesis problems we consider. This subtlety is addressed in a gate-set specific way for Clifford+T in \cite{Ross2015,KliuchnikovMaslovMosca2013}. Here, we provide an approach that generalizes to many gate sets and also discuss the reduction of the relevant more general problem to Problem \ref{prob:rel-norm} in Section \ref{sec:rel-norm-shortcut}. The reduction holds for Clifford+T and Clifford+$\mathrm{\sqrt{T}}$.}

Next we review the solution to \cref{prob:rel-norm} 
and some simplifying assumptions.
Note that for this problem to have a solution, it is necessary (but not sufficient) that $r$ is totally positive in $O_K$ i.e. $\sigma_k(r)>0$ for all homomorphisms from $K$ into $\r$ \cite{cohen1993course}. From now on we assume this is the case.

The algorithm for solving norm equations given here (\algo{rel-norm}) takes advantage of a number of properties specific to the examples of $L$ and $K$ that we consider. In particular, we look at fields $L$ and $K$ which are Galois fields, and whose rings of integers are principle ideal domains. 
For more general gate sets, Algorithm 7.5.15 of Cohen \cite{Cohen2012} gives a method for solving relative norm equations, covering both Galois and non-Galois extensions. 
%
We use the method from \cite{Kliuchnikov2015a} and justify that for the fields we consider a solution to a relative norm equation can be found in polynomial time given factoring oracle.
%
Examples of fields $K$ and $L$ of interest to us are provided in~\tab{gates} (\sec{approx-solutions}).
Our approach to solve \problem{rel-norm} uses common  properties of these fields, which are captured in the following definition.
\begin{dfn}[Common properties of relevant fields]\label{defn:comprops}
 Let $i$ be such that $i^2 = -1$. We define $K$ and $L=K(i)$ as fields having the following properties:
\begin{enumerate}
    \item $L/K$ is a cyclic Galois extension,
    \item $O_K$ and $O_L$ are principal ideal domains,
    \item $[O_L^\times :  WO^{\times}_K]=1$ or $2$, where $W$ is the group of roots of unity in $L$,
    \item $O_L$ is Euclidean with respect to the field norm of $L$ ($O_L$ is \emph{norm-Euclidean}), 
\item The generators of the unit group of $O_L$ are known,
    \item Given $u$ a totally positive unit in $O_K$, there exists a unit $w\in O_L$ such that $ww^\ast = u$.
\end{enumerate}
\end{dfn}


The fields which we consider in \sec{approx-solutions}, (see~\tab{gates}), are cyclotomic fields, and satisfy \defin{comprops} \cite{washington1997introduction}. For the interested reader, Lemmermeyer provides a survey of cyclotomic fields known to be norm-Euclidean in \cite{lemmermeyer1995euclidean}. 

%

\subsection{Solution overview}

We will describe how to determine the existence of a solution and how to compute one when it exists. We require the following definition.
\begin{dfn}\change{ Let $K/\mathbb{Q}$ be an extension of the rational numbers. Let $p$ be an ideal and let $p = \prod_iq_i^{e_i}$ be the prime decomposition of $p$ in $K$. We say
	\begin{itemize}
	\item $p$ \emph{splits} in $K$ if $i>1$
	\item $p$ is ramified in $K$ if $e_i >1$ for some $i$, and
	\item $p$ is \emph{inert} in $K$, otherwise.
	\end{itemize}}
\end{dfn}

\algo{rel-norm} summarizes the method for solving relative norm equations in $O_K$. First, we compute the prime factorization of the absolute norm $R = N(r) \in \mathbb{Z}$. Suppose $R = \prod_kp_k^{v_k}.$ Each prime $p_k$ is factored into prime $O_L$ ideals $\mathfrak{p}^{(k)}_i$, for which we compute the ideal generators $\xi^{(k)}_i$ such that $\xi^{(k)}_i O_L = \mathfrak{p}^{(k)}_i$. The $\xi^{(k)}_i$ have corresponding primes $\eta_i^{(k)}$ in $O_K$. Through trial division of $r$ by each prime $\eta_i^{(k)}\in O_K$, we compute a representation of $r$ as a product of primes in $O_K$, $r = u \prod{_{j}} \eta_j^{e_j}$ (up to multiplication by a unit $u$). If $e_j$ is even for all relatively inert $\eta_j$, the relative norm equation is solvable (\lemm{even powers}). Supposing a solution $m$ exists, we compute $w$, a unit of $O_L$, such that $ww^\ast=u$ then write $m$ in the form above.

\begin{algorithm}[!hp]
\SetKwData{Left}{left}\SetKwData{This}{this}\SetKwData{Up}{up}
\SetKwFunction{Union}{Union}\SetKwFunction{FindCompress}{FindCompress}
\SetKwInOut{Input}{Input}\SetKwInOut{Output}{Output}
\LinesNumbered
\Input{$r$ in $O_K$}
\Output{$m$ in $O_L$ such that $mm^\ast = r$, or \textit{No solutions.}}
\BlankLine
 Compute $R\leftarrow N(r)\in\z$\label{step:Rnorm}\;
Compute prime factorization $R=\prod_kp_k^{v_k}$\label{step:rfac}\;
For each $p_k$, apply \algo{ok-primes} to find all prime $O_L$ ideals $\mathfrak{p}_i$ such that $p_k O_L \subset \mathfrak{p}_i$, corresponding ideal generators $\xi_i$ and  primes in $O_K$, $\eta_i$.\label{step:forins}\;
Find integers $\{e_i\}$ and $u$, a unit of $O_K$, such that $r= u\prod_j\eta_i^{e_i}$\label{step:rinprime}\; 
\eIf{all $e_i$\textnormal{ even for all relatively inert }$\eta_i$\label{step:existcheck}}{
Compute $w\in O_L$ such that $ww^\ast = u$\label{step:w}\;
Output solution $m=w\prod\limits_{\substack{i:\,\eta_i\\\textnormal{ inert}}}\eta_i^{e_i/2}\prod\limits_{\substack{i:\,\eta_i\textnormal{ split/}\\\textnormal{ramified}}}\xi_i^{e_i}$\label{step:output} \;}{\textnormal{Output }\textit{No solution.}}
\caption{Algorithm for solving relative norm equations.\label{alg:rel-norm}}
\end{algorithm}\DecMargin{1em}

\begin{algorithm}[!hp]
\SetKwData{Left}{left}\SetKwData{This}{this}\SetKwData{Up}{up}
\SetKwFunction{Union}{Union}\SetKwFunction{FindCompress}{FindCompress}
\SetKwInOut{Input}{Input}\SetKwInOut{Output}{Output}
\LinesNumbered
\Input{$p$, prime}
\Output{The set $\{(\eta_i,\xi_i,\mathfrak{p}_i)\}$, where $p=\prod\mathfrak{p}_i$, $\xi_iO_L=\mathfrak{p}_i$ and $\eta_i\in O_K$ prime above $p$. }
\BlankLine
Factor $p$ into prime $O_L$ ideals $\mathfrak{p}_i$\label{step:primefac}\;
Compute ideal generators $\xi_i$ such that $\xi_iO_L=\mathfrak{p}_i$\label{step:primegens}\;
\eIf{\textnormal{there exists }$v$\textnormal{, a unit in }$O_L^\times/O_K^\times$\textnormal{, such that }$v\xi_i$\textnormal{ prime in }$O_K$ \label{step:smallprime}}{Set $\eta_i:=v\xi_i$\;}{Set $\eta_i:=\xi_i\xi_i^\ast$\;}
Output $(\eta_i,\xi_i,\mathfrak{p}_i)$ for each $i$.

\caption{Subroutine of \algo{rel-norm} for computing primes above $p$ in $O_K$.\label{alg:ok-primes}}
\end{algorithm}

 The correctness of \algo{rel-norm} is proved in \propos{alg-sol}. Moreover, \propos{alg-sol} shows that when a solution to the relative norm equation problem (\problem{rel-norm}) exists, we can write it as $$m = w\prod\limits_{\substack{j:\,\eta_j\\
\textnormal{ inert}}}\eta_j^{e_j/2}\prod\limits_{\substack{j:\,\eta_j\textnormal{ split/}\\\textnormal{ramified}}}\xi_j^{e_j},$$
where $w$ is a unit of $O_L$, $\xi_j$ are generators of prime ideals in $O_L$ and $\eta_j$ are primes in $O_K$ and $u$ is a unit in $O_K$. \change{The desired norm, $r$, can hence be written as $\prod_ju\eta_j^{e_j}.$}

\begin{prop}\label{prop:alg-sol}
\algo{rel-norm} returns, in polynomial time, a solution to \problem{rel-norm}, if one exists, and returns \emph{No solution} otherwise.
\end{prop}
\begin{rem}
\algo{rel-norm} implicitly gives a description of all solutions, if more than one is needed. If $m_2 =  w\prod\limits_{\substack{i:\,\eta_i\\\textnormal{ inert}}}\eta_i^{r_i/2}\prod\limits_{\substack{i:\,\eta_i\textnormal{ split/}\\\textnormal{ramified}}}\xi_i^{e_i}$ is a solution, then so is $$w\prod\limits_{\substack{i:\,\eta_i\\\textnormal{ inert}}}\eta_i^{r_i/2}\prod\limits_{\substack{i:\,\eta_i\textnormal{ split/}\\\textnormal{ramified}}}(\xi_i^{e_i - e} (\xi^{\ast}_i)^e)$$ for each $e$ between $0$ and $e_i$. 
\end{rem}

\subsection{Subroutines of \algo{rel-norm}}
We now look at the subroutines of \algo{rel-norm} in more detail and prove \propos{alg-sol}. 
\paragraph{Step \ref{step:Rnorm}: Computing \texorpdfstring{$R=N(r)$}{R=N(r)}}
Representing $r$ in integer coordinates with respect to an integral basis of $O_K$ gives a closed form multivariate polynomial expression for the absolute norm function.

\paragraph{Step \ref{step:rfac}: Computing the prime factorization of \texorpdfstring{$R=N(r)$}{R=N(r)}}

Recall that in the application to approximate synthesis, the possible values for $R=N(r)$ are bounded above by $\ell^N$. Using the Prime Number Theorem, we heuristically expect $R$ to be prime with probability approximately $1/\log(\ell^N)$. Primality testing can be done with any polynomial time algorithm, such as Miller-Rabin or AKS. For composite $R$, variants of Miller-Rabin can be used to return factors under some conditions, for example when $R$ is coprime to a Miller-Rabin witness.

If a candidate value for $m_1$ results in an $R$ that is inefficient to factorize, it can be discarded, although one should note that this may result in a non-optimal unitary approximation. We certainly expect $R$ to be easier to factorize than RSA integers, in general, recalling the relation between $N$ and area given in \sec{cost-scaling}.

\paragraph{Step \ref{step:forins}: Computing a set of primes in \texorpdfstring{$O_K$}{OK} dividing \texorpdfstring{$r$}{r}}
For each prime factor $p_k$ of $R$, \algo{ok-primes} is used to factorise $p_k$ into a product of prime ideals $\mathfrak{p}_i$, compute corresponding ideal generators $\xi_i$, and, ultimately, find primes in $O_K$, $\eta_i$, which divide $r$.

 A detailed description of \algo{ok-primes} is given in \sec{ok-prime-subroutine}.

\paragraph{Step \ref{step:rinprime}: Representing \texorpdfstring{$r$}{r} as a product of primes in $O_K$}
\algo{ok-primes} produces the set $\{\eta_i,\xi_i,\mathfrak{p}_i\}$ for each prime factor $p_k$ of $R$. We have $N(\mathfrak{p}_i)=p_k^{f_i}$, where $f_i=[O_L/\mathfrak{p}_i:\z/p_k]$ (Thm. 4.8.5, \cite{cohen1993course}). Trial division of $r$ by each $\eta_i$ will determine the $e_i$'s, using the exponents $v_k$ in the prime factorization of $R$ to determine when all prime ideals above each $p_k$ are covered. The representation of $r$ as a product of primes in $O_K$ is then $u\prod_i\eta_i^{e_i}$, where $u$ is a unit in $O_K$.

\paragraph{Step \ref{step:existcheck}: Determining the existence of a solution}
Step \ref{step:existcheck} in the algorithm determines whether a solution to the relative norm equation exists. The existence criterion is captured in the following Lemma.

\begin{lem}\label{lem:even powers}
Given $r\in O_K$, where $r=u\prod \eta_i^{e_i}$ with $\eta_i$ prime in $O_K$ and $u$ a unit in $O_K$, the relative norm equation $m_2m_2^\ast = r, m_2\in O_L$ is solvable if and only if $e_i$ is even, for all relatively inert $\eta_i$ \cite{Kliuchnikov2015b}.
\end{lem}

Clearly, Step \ref{step:existcheck} ensures that \algo{rel-norm} correctly returns \textit{No solution} if no solution exists for input $r.$ 

\paragraph{Step \ref{step:w}: Finding a unit $w$}


The existence of $w$, a unit in $O_L$, such that $ww^\ast = u$ is guaranteed by \defin{comprops}. Let us describe how to compute $w$ in Step \ref{step:w}. We can take advantage of the fact that we are in the unit group of $O_L$ with the theorem below.

\begin{thm}[Dirichlet's Unit Theorem, 1846]\label{dirichlet}
Let $K$ be a number field with $r_1$ real homomorphisms and $2r_2$ pairs of conjugate homomorphisms. Let $r = r_1+r_2-1.$ Each order $\mathcal{O}$ of $K$ contains multiplicatively independent units $u_1,\dots,u_r$ of infinite order such that every unit in $\mathcal{O}$ can be written explicitly in the form $$\omega^k u_1^{k_1}\cdots u_r^{k_r},$$
where $\omega$ is a root of unity in $\mathcal{O}.$
\end{thm}

By Theorem \ref{dirichlet}, $w$ can be written as $\omega^k u_1 ^{k_1} \ldots u_{d-1}^{k_{d-1}}$ for some $k,k_1,\dots,k_{d-1}\in\z$.  It follows that $u=ww^\ast = (\omega\omega^\ast)^k (u_1u_1^\ast) ^{k_1} \ldots (u_{d-1}u_{d-1}^\ast)^{k_{d-1}}$. Since the unit group generators are known, by \defin{comprops}, the following lemma asserts that we can generate the entire group of totally positive units of $O_K.$

\begin{lem}
Let $d$ be the degree of $K$. Let $u_1,\dots, u_{d-1}$ be infinite order units of $O_L$, the ring of integers of $L$. Then the multiplicative group generated by $u_1u_1^\ast,\dots, u_{d-1}u_{d-1}^\ast$ is equal to the group of totally positive units of $O_K.$
\end{lem}
\begin{proof}
Clearly, $u_iu_i^\ast$ are totally positive units of $O_K.$

Let $u$ be some totally positive unit of $O_K$ not equal to $u_iu_i^\ast$ for all $i$. By \defin{comprops}, there exists a unit $w\in O_L$ such that $u=ww^\ast$. Then, by Dirichlet we have $w=\omega^k u_1 ^{k_1} \ldots u_d^{k_{d-1}}$, with $k_i\in\z$ and $\omega$ a root of unity. Clearly, $u$ is a product of integer powers of $u_1u_1^\ast,\dots, u_{d-1}u_{d-1}^\ast$.
\end{proof}

This representation of $u$ as a product of unit group generators motivates the reduction of finding $w$ to an instance of finding an integer vector in the lattice induced by the unit group $O_L^\times$. Let the lattice be denoted $\mathcal{L}_S$, generated by basis vectors $$(\log \sigma_1(u_j u_j^\ast), \ldots, \log \sigma_d(u_j u_j^\ast)),\quad j = 1,\ldots,{d-1}.$$ Let $v$ be the vector $(\log \sigma_1(u), \ldots, \log \sigma_d(u))^T$.

We can find an integer vector $(x_1, \ldots x_{d-1})$ such that 
$$
 u = (u_1 u_1^\ast)^{x_1} \cdots (u_{d-1} u_{d-1}^\ast)^{x_{d-1}}
$$
by solving

$$A(x_1,\dots,x_{d-1})^T = v,$$
where $A$ is the $d\times (d-1)$ matrix whose rows correspond to the basis vectors of $\mathcal{L}_S$. Setting $w = u_1^{x_1}\cdots u_{d-1}^{x_{d-1}}$ completes Step \ref{step:w}.

\paragraph{Step \ref{step:output}: Computing a solution \texorpdfstring{$m$}{m} to the relative norm equation problem}
Setting \begin{equation*}m=w\prod\limits_{\substack{j:\,\eta_i\\
\textnormal{ inert}}}\eta_i^{e_i/2}\prod\limits_{\substack{i:\,\eta_i\textnormal{ split/}\\\textnormal{ramified}}}(\xi_i)^{e_i}\end{equation*}
yields
\begin{equation}\label{eq:output}mm^\ast=ww^\ast\prod\limits_{\substack{i:\,\eta_i\\
\textnormal{ inert}}}\eta_i^{e_i}\prod\limits_{\substack{i:\,\eta_i\textnormal{ split/}\\\textnormal{ramified}}}(\xi_i\xi_i^\ast)^{e_i} = u\prod\eta_i^{e_i}=r,\end{equation}

as required.

Equation \eq{output} shows that the output of \algo{rel-norm} is a solution to \problem{rel-norm}. Since each subroutine of the algorithm runs in polynomial time, \algo{rel-norm} runs in polynomial time. This proves \propos{alg-sol}.

\subsection{Subroutines of \algo{ok-primes}}\label{sec:ok-prime-subroutine}
The subroutine, \algo{ok-primes}, called at Step \ref{step:forins} of \algo{rel-norm} describes a process for lifting primes $p_k$ to prime ideals in $O_L$ and finding corresponding primes in $O_K$.
\paragraph{Step \ref{step:primefac}: Factorizing primes into prime ideals}
By Theorem \ref{ideal-factors} and Theorem 14.14 of \cite{von2013modern} there exists a polynomial time algorithm to factor each $p_k$ into ideals $\mathfrak{p}_i$. The following theorem provides a reduction of factoring rational primes $p$ into prime ideals  to factoring a polynomial $\mathrm{mod}\,p$.
\begin{thm}[Cohen, Thm.4.8.13 \cite{cohen1993course}]\label{ideal-factors}
Let $L = \q(\theta)$, where $\theta$ is an algebraic integer with minimal polynomial $T(X)$. Let $f = [O_L:\z[\theta]]$ and let $p$ be a prime not dividing $f$. Suppose
$$
T(X) = \sum\limits T_i(X)^{e_i} \mod p.
$$
Then, the prime decomposition of $pO_L$ is given by
$$pO_L = \prod\limits\mathfrak{p}_i^{e_i},$$
where $\mathfrak{p}_i = pO_L + T_i(\theta)O_L$.
\end{thm}
When $L$ is a cyclotomic field, $L=\q(\zeta_n)$ for some $n$ where $\zeta_n$ is a primitive $n^{th}$ root of unity and $O_L=\z[\zeta_n]$ \cite{ireland1990classical}. Hence, $f=[O_L:\z[\zeta_n]]=1$ and there is no prime $p$ such that $p\mid f$. So $\mathfrak{p}_i = p_kO_L +\alpha_i O_L$, for each prime in the prime factorization of $R$, where $\alpha_i\in O_L$. Factoring $T(X)$, the minimal polynomial of $\zeta_n$, over the finite field $\f_p$ can be done in polynomial time in the degree of $T(X)$ and $\log(p)$ (Thm. 14.14, \cite{von2013modern}). This completes Step \ref{step:primefac} of the algorithm.  

\paragraph{Step \ref{step:primegens}: Computing ideal generators}
Since $O_L$ is a principal ideal domain, computing the generator $\xi_i$ of a prime ideal, as in Step \ref{step:primegens}, is an instance of the Principal Ideal Problem. Recall that $O_L$ is norm-Euclidean, so the generator of $p_k O_L + \alpha_i O_L$ is computed using the Euclidean algorithm. Set $\xi_i=\mathrm{GCD}(p_k, \alpha_i)$ where $\mathrm{GCD}(a, b)\in O_L$ is the greatest common divisor computed by the Euclidean algorithm, using the absolute norm. Then, since $L$ is Galois, the generators of all prime ideals above $p_k$ can be recovered using Galois automorphisms of $L$, as the Galois group acts transitively on prime ideals $\mathfrak{p}_i$.

\paragraph{Step \ref{step:smallprime}: Finding \texorpdfstring{$\eta$}{eta}, prime in $O_K$}
For the computation at Step \ref{step:smallprime}, recall that the unit group $O^{\times}_K$
of $O_K$ is a finite index subgroup of the unit group of $O^{\times}_L$. Our aim is to find a prime factorization of $p\in O_K$, up to multiplication by a unit. To that end, the following lemma is used to find primes $\eta_i$ in $O_K$ corresponding to certain prime ideal generators $\xi_i$.

\begin{lem}\label{lem:inert}
Let $\xi\in O_L$ be the generator of a prime ideal. If there exists $v$ a unit from the finite quotient $O_L^\times/O_K^\times$ such that $v\xi\in O_K$ then $v\xi$ is relatively inert in $O_L$.
\end{lem}
\begin{proof}
Let $v$ be such a unit from the finite quotient $O_L^\times/O_K^\times$ and suppose $v\xi$ not relatively inert in $O_L$. Then there exist $a,b\in O_L$ such that $v\xi=ab$. Then $\xi=v^{-1}(v\xi)=v^{-1}ab$, a contradiction.
\end{proof}

Each generator $\xi_i$ is multiplied by a representative $v$ of each element in the quotient $O^{\times}_L / O^{\times}_K$. If $v\xi_i\in O_K$, \lemm{inert} asserts that $\eta_i := v\xi_i$ is prime in $O_K$. By Property (3) of \defin{comprops}, iterating through each element in the quotient to find a valid $v$ is efficient. If this process fails, $\eta_i$ is set to $\xi_i\xi_i^{\ast}$. Then $\eta_i$ is prime in $O_K$ and relatively split or ramified in $O_L.$

\subsection{A shortcut property for solving the norm equation}
\label{sec:rel-norm-shortcut}

Recall that in the general solution outline of \sec{approximation-problems:general}, the norm equation problem was simplified to the following problem.

\begin{prob}\label{prob:quotient-norm-simplified}
Given, $\hat{z}$ in $M_\mathcal{O}$, find $z'\in \hat{z} + \xi M_\mathcal{O}$ such that \begin{equation*}\abs{z'}^2=r\in O_K,\end{equation*} where $r = \xi\xi^\ast\ell^N-\hat{m_1}\hat{m_1}^\ast$.
\end{prob}
We consider fields with class number equal to 1, and demonstrate a `shortcut' for solving this problem. The following lemma identifies a property of fields $K,L$ and ideals $I$ sufficient to guarantee this simplification.

\begin{lem}\label{lem:shortcut-property} Let $I$ be an integral ideal of $O_L$ fixed by conjugation, so $x\in I\implies x^\ast\in I$.  Let $U$ be the group of torsion units of $O_L$ modulo $I$. If \begin{equation*}
    \forall q \in O_L/I,\quad \exists u \in U\text{ such that } q^\ast = uq
\end{equation*}
then
\begin{equation*}
    \forall z\in O_L \text{ such that }\abs{z}^2=r\in O_K,\quad \exists u' \text{ such that } \abs{u'z}^2=r\text{ and }u'z-\hat{z}\in I.
\end{equation*}
\end{lem}
\begin{proof}
Suppose $z'$ is a solution to the norm equation with quotient constraint,
\begin{equation}\label{eq:quotient-norm-equation}
    \abs{z'}^2=r,\quad z'\in \hat{z} + I, \hat{z}\in O_L. 
\end{equation} Then $z'$ is also a solution to the general norm equation
\begin{equation}\label{eq:general-norm-equation}
    \abs{z'}^2=r,\quad z'\in O_L
\end{equation}
and hence can be written as $z'=uz_0^2z_1^{e_1}(z_1^\ast)^{n_1-e_1}\cdots z_m^{e_m}(z_m^\ast)^{n_m-e_m}$ for integers $e_i,n_i$, where $u$ is a torsion unit of $O_L$, $z_0O_L$ is a product of relatively inert prime ideals, and the $z_i$ are such that $z_iO_L$ is a prime $O_L$ ideal and $z_iz_i^\ast O_K$ is a prime $O_K$ ideal.  We can similarly write $z$ as $wz_0^2z_1^{c_1}(z_1^\ast)^{n_1-c_1}\cdots z_m^{c_m}(z_m^\ast)^{n_m-c_m}$, for integers $n_i,c_i$, where $w$ is a torsion unit of $O_L$.

Now consider a ring homomorphism $\gamma$ defined by $\gamma(z) = z + I$. By assumption on $O_L/I$, there exist $x_i\in U$ such that $x_i\gamma(z_i)^\ast = \gamma(z_i)$.
In other words, there exists a torsion unit $x'_i$ such that $z_i^\ast + I = x'_i z_i + I$, using that $I$ is fixed by conjugation. Observe that $\gamma(z')=\gamma(\hat{z})$ since $z'$ is a solution to Equation \eq{quotient-norm-equation}. However, we also have
\begin{equation}
\begin{aligned}
\gamma(z') &= \gamma(u)\cdot \gamma(z_0) \gamma(z_1)^{n_1} \gamma((x'_1)^{n_1 - e_1}) \ldots \gamma(z_m)^{n_m} \gamma((x'_m)^{n_m - e_m}) \\
\gamma(z) &= \gamma(w) \cdot \gamma(z_0) \gamma(z_1)^{n_1} \gamma((x'_1)^{n_1 - c_1}) \ldots \gamma(z_1)^{n_m} \gamma((x'_m)^{n_m - c_m}) \\
\end{aligned}
\end{equation}
Based on the above we set $u' = uw^{-1}(x'_1)^{c_1 - e_1} \ldots (x'_m)^{c_m - e_m}$ and see that $u'z$ is such that $\abs{u'z}^2=r$ and $\gamma(u'z)=\gamma(z')=\gamma(\hat{z})$, as required.

\end{proof}

In essence, for any solution $z$ to Equation \eq{general-norm-equation}, there exists a torsion unit such that $u'z$ is a solution to Equation \eq{quotient-norm-equation}. We call this property the `Shortcut Property'. There exist ideals in which the shortcut property holds for $K,L$ corresponding to the Clifford$+$T and Clifford$+\sqrt{\text{T}}$ bases.

Let $I$ be an integral ideal of $O_L$ fixed by conjugation such that $I\subseteq \xi M_\mathcal{O}\subseteq O_L$.
Then, $
\xi M_\mathcal{O}$ is equal to the finite disjoint union $\xi M_\mathcal{O}=\bigsqcup_k(z_k+I)$. Then, \lemm{shortcut-property} shows that a solution to \problem{quotient-norm-simplified} can be found by solving the general norm equation $\abs{z}^2=r$, $z\in O_L$, then checking whether $u'z-(\hat{z}+z_k)\in I$ for some $z_k$ and unit $u$. This requires only finitely many checks. 
Finally, a solution $u'z$ yields a candidate for $m_2$ by setting $m_2 = u'z/\xi$.

\section{Exact synthesis}\label{sec:exact-synthesis}


This section covers the details of exact synthesis algorithms for the three example gate sets considered in~\sec{approx-solutions}: $V$ basis, Clifford+$T$ and Clifford+$\sqrt{T}$.
Given a gate set $G$ and a matrix $U$ from a certain set $\mathcal{U}$ uniquely determined by $G$, the goal of exact synthesis is to produce a sequence $g_1,\dots, g_n$ of gates from $G$
such that $U$ is equal to the product of those gates, $U = g_1\dots g_n$.
The set $\mathcal{U}$ is closed under left and right multiplication by elements of $g$.

The algorithm is similar for each gate set.  Roughly, given a matrix $U$ 
\begin{enumerate}
    \item select a gate $g$ from the gate set such that $g^\dagger U$ has a lower cost than $U$,
    \item set $U \leftarrow g^\dagger M$
    \item repeat until cost of $U$ is zero.
\end{enumerate}
The gate sequence is recovered by collecting the gate $g$ selected at each iteration.
Intuitively, the algorithm works by picking off each gate of the product $g_1\dots g_n$ one at a time. Multiplication by $g^\dagger$ cancels the left most gate in $M$.

Importantly, it is possible to efficiently select a gate $g$ so that the cost decreases monotonically.  When the cost is zero, this \change{means} the remaining gate is a Clifford gate.

Algorithms for exact synthesis has been proposed previously by~\cite{BGS} for the $V$ basis and by~\cite{KliuchnikovMaslovMosca2013,Forest2015} for Clifford+$T$ and Clifford+$\sqrt{T}$.  We present the algorithms here for completeness.  We provide a modified version of the algorithm with several optimizations for computational performance.

We begin with the simplest case, the $V$ basis, and then address the progressively more complicated Clifford+$T$ and Clifford+$\sqrt{T}$ gate sets.

\subsection{V basis}\label{sec:exact-synthesis-v-basis}
Recall from~\sec{approximation-problems:Vbasis} the six $V$ basis matrices
$V_{\pm X}$, $V_{\pm Y}$, $V_{\pm Z}$ and
order 
\begin{equation}
    \label{eq:vbasis-order}
    \mathcal{O} := \mathbb{Z}\cdot I+\mathbb{Z}\cdot iX+\mathbb{Z}\cdot iY+\mathbb{Z}\cdot iZ.
\end{equation}
This order contains the $V$ basis matrices each scaled by $\sqrt{5}$.
For notational convenience we use $V_x = \sqrt{5}V_{+X}$, $V_y = \sqrt{5}V_{+Y}$ and $V_z = \sqrt{5} V_{+Z}$ to refer to the scaled $V$ basis matrices. 
Note that $V_{-P} = V_{+P}^\dagger$ for $P\in \{X,Y,Z\}$ and $V_P V^\dagger_P = \mathrm{Det}(V_P) I = 5 I$.

Define $M_V$ as the function that, according to~\eq{vbasis-order}, maps integers $a,b,c,d$ to a matrix in the natural way:
\begin{equation}
    M_V(a,b,c,d) := aI + ibX + icY + idZ.
\end{equation}
Any matrix $M_V(a,b,c,d)$ with determinant $5^n$ can be decomposed (exactly) into a length-$n$ sequence of (scaled) $V$ gates~\cite{BGS,Kliuchnikov2015b}. 
Note that $\mathrm{Det}(M_V(a,b,c,d))=1$ if and only if $M_V(a,b,c,d)$ is a Pauli matrix.
\begin{thm}[V basis exact decomposition]\label{thm:v-basis-exact-decomposition}
Let $a,b,c,d\in\mathbb{Z}$ such that $\mathrm{Det}(M_V(a,b,c,d)) = 5^n$ for integer $n \ge 1$, and such that at least one of $a,b,c,d$ is not divisible by $5$.  Then there exists a sequence $V_1, V_2, \dots, V_n$, $V_k\in\{V_x, V_y, V_z, V^\dagger_x, V^\dagger_y, V^\dagger_z\}$ and Pauli matrix $V_0$ such that
\begin{equation}
    M_V(a,b,c,d) = V_0 \prod_{k=1}^n V_k .
\end{equation}
\end{thm}
The requirement that one of $a,b,c,d$ is not divisible by $5$ avoids artificially scaled inputs (e.g., $5I$).  Scalars can be removed by simply dividing out the factors of $5$. The proof follows by induction on the following Lemma.

\begin{lem}[V basis factorization] \label{lem:vbasis-factorization}
  Let $a,b,c,d\in\mathbb{Z}$ such that $\mathrm{Det}(M_V(a,b,c,d)) = 5^n$ for integer $n \ge 1$, and such that at least one of $a,b,c,d$ is not divisible by $5$.  Then there exists $V\in\{V_x, V_y, V_z, V^\dagger_x, V^\dagger_y, V^\dagger_z\}$ and $a',b',c',d'\in\mathbb{Z}$ such that
  \begin{equation} \label{eq:vbasis-factorization}
      M_V(a,b,c,d) = V M_V(a',b',c',d'),
  \end{equation}
and $\mathrm{Det}(M_V(a',b',c',d')) = 5^{n-1}$.
\end{lem}

In other words, the matrix $M_V(a,b,c,d)$ can be factored into two parts: a $V$ matrix and another matrix of the form $M_V$.
Multiplication on the left by $V^\dagger$ yields
\begin{equation} \label{eq:vbasis-alternate-factorization}
  V^\dagger M_V(a,b,c,d) = \mathrm{Det}(V) M_V(a',b',c',d') = M_V(5a', 5b', 5c', 5d').
\end{equation}
and therefore
\begin{equation}
    \mathrm{Det}(M_V(a',b',c',d')) 
    = \mathrm{Det}(V^\dagger M_V(a',b',c',d') / 5) 
    = \mathrm{Det}(M_V(a,b,c,d)) / 5
    = 5^{n-1}.
\end{equation}

If we define the entrywise modulus
\begin{equation}
    M_V(a,b,c,d) \mathrm{\,mod\,} 5 := M_V(a\mathrm{\,mod\,} 5, b\mathrm{\,mod\,} 5, c\mathrm{\,mod\,} 5, d\mathrm{\,mod\,} 5),
\end{equation}
then
\begin{equation}
    V^\dagger M_V(a,b,c,d)\mathrm{\,mod\,} 5 = V^\dagger \at{M_V(a,b,c,d)\mathrm{\,mod\,}5}\mathrm{\,mod\,}5,
\end{equation}
by linearity.  Equation \ref{eq:vbasis-alternate-factorization} then implies that
\begin{equation} \label{eq:vdagger-solution-mod5}
    V^\dagger \at{M_V(a,b,c,d)\mathrm{\,mod\,}5}\mathrm{\,mod\,}5 = M_V(0,0,0,0).
\end{equation}
Solutions for $V$ therefore depend only on values of $a,b,c,d$ modulo $5$. The proof proceeds by exhaustive numeric calculation over all tuples $a,b,c,d \mod 5$, $(a,b,c,d) \ne (0,0,0,0)$, such that $\mathrm{Det}(M_V(a,b,c,d)) = 0 \mod 5$. 

Exhausting over all tuples $a,b,c,d\mod 5$ produces a $12$-bit indexed table 
that can be used to lookup an appropriate $V$ for any $M_V(a,b,c,d)$ with determinant $5^n$.  This lookup can be used to construct an efficient algorithm for exact synthesis.

\begin{algorithm}
\SetAlgoLined
\KwIn{Elements $a,b,c,d$ from $\mathbb{Z}$ such that $\mathrm{Det}(M_V(a,b,c,d)) = 5^n$ for integer $n \ge 0$}
\KwOut{Sequence of matrices in $\{V_x, V_y, V_z,V^\dagger_x, V^\dagger_y, V^\dagger_z\}$ and a Pauli gate}
 $gates \leftarrow$ empty list\;
 \While{$\mathrm{Det}(M_V(a,b,c,d)) = 0\mod 5$}{
  $V \leftarrow$ Lookup$_V$($a\,\mathrm{mod}\,5,b\,\mathrm{mod}\,5,c\,\mathrm{mod}\,5,d\,\mathrm{mod}\,5$)\;
  $M_V(a,b,c,d) \leftarrow V^\dagger M_V(a,b,c,d)/\mathrm{Det}(V)$\;
  prepend $V$ to $gates$\;
 }
 \KwRet{gates, $M_V(a,b,c,d)$}
 \caption{V basis exact synthesis.}
\end{algorithm}

The algorithm "picks off" each $V$ gate sequentially.  At each step, the leading factor of $V$ is removed from $M_V(a,b,c,d)$ by multiplying on the left by $V^\dagger$.  The resulting tuple $a'',b'',c'',d''$ is then divided by $5$ yielding a new $M_V(a',b',c',d')$ with determinant $5^{n-1}$.  The output is the sequence of picked-off $V$ gates, in reverse order.

\subsection{Clifford + \texorpdfstring{$T$}{T}}
Recall from~\sec{approximation-solutions-clifford-t} the $T$ matrices
$T_P := \frac{1}{\sqrt{2+\sqrt{2}}}\at{I + \frac{I-iP}{\sqrt{2}}}$ for $P\in\{X,Y,Z\}$ and corresponding quaternion order
\begin{equation}
    \label{eq:t-basis-order}
    \mathcal{O} 
    = \mathbb{Z}[\sqrt{2}]\cdot I 
    + \mathbb{Z}[\sqrt{2}]\cdot\frac{I+iX}{\sqrt{2}}
    + \mathbb{Z}[\sqrt{2}]\cdot\frac{I+iY}{\sqrt{2}}
    + \mathbb{Z}[\sqrt{2}]\cdot\frac{I+iX+iY+iZ}{2}.
\end{equation}
\begin{equation}
\mathbb{Z}[\sqrt{2}] = \{ a + b \sqrt 2 : a,b \in \mathbb{Z} \}
\end{equation}
This order contains the $T$ matrices each scaled by $\sqrt{2+\sqrt{2}}$.
For notational convenience we use
$T_x = \left(\sqrt{2+\sqrt{2}}\right)T_X$,
$T_y = \left(\sqrt{2+\sqrt{2}}\right)T_Y$, and
$T_z = \left(\sqrt{2+\sqrt{2}}\right)T_Z$
to refer to the scaled $T$ matrices. 
Note that $T_p T_p^\dagger = \mathrm{Det}(T_p) I = (2+\sqrt 2)I$ for $p = x,y,z$.

Define $M_T$ as a function that, according to~\eq{t-basis-order}, maps elements $a,b,c,d$ of $\mathbb{Z}[\sqrt{2}]$ to a matrix:
\begin{equation}
    M_T(a,b,c,d) := a\cdot I + b\cdot \frac{I + iX}{\sqrt{2}} + c\cdot \frac{I+iY}{\sqrt{2}} + d\cdot \frac{I+iX+iY+iZ}{2}.
\end{equation}

Any matrix $M_T(a,b,c,d)$ with determinant $(2+\sqrt{2})^n$ can be decomposed into a length-$n$ sequence of $T$ gates~\cite{Kliuchnikov2015b}. 
We omit the formal theorem because the situation is analogous to that of~\theo{v-basis-exact-decomposition}.
Note that $\mathrm{Det}(M_T(a,b,c,d))=1$ if and only if $M_T(a,b,c,d)$ is a Clifford unitary.
For reference, we provide a standard $T$ factorization.

\begin{lem}[Clifford+$T$ factorization]\label{lem:t-factorization}
  Let $a,b,c,d\in\mathbb{Z}[\sqrt{2}]$ such that $\mathrm{Det}(M_T(a,b,c,d)) = (2+\sqrt{2})^n$ for integer $n \ge 1$,
  and at least one of $a,b,c,d$ is not divisible by $2+\sqrt{2}$. 
  Then there exists $g\in\{T_x, T_y, T_z\}$ and $a',b',c',d'\in\mathbb{Z}[\sqrt{2}]$ such that
  \begin{equation}\label{eq:t-factorization}
      M_T(a,b,c,d) =  g M_T(a',b',c',d').
  \end{equation}
and $\mathrm{Det}(M_T(a',b',c',d')) = (2+\sqrt{2})^{n-1}$.
\end{lem}
The proof of this Lemma is analogous to that of~\lemm{vbasis-factorization}, except that we exhaust over tuples $a,b,c,d$ modulo $2+\sqrt{2}$ instead of tuples modulo $5$. 
More precisely, we consider elements of $\mathbb{Z}[\sqrt{2}]$ modulo prime ideal $(2+\sqrt{2})\mathbb{Z}[\sqrt{2}] = \sqrt{2}\mathbb{Z}[\sqrt{2}]$.
Considering values modulo $2 + \sqrt 2$ is the same as considering values modulo $\sqrt 2$.
Every element of $\mathbb{Z}[\sqrt{2}]$ is of the from $\sqrt 2 z$ or $\sqrt 2 z + 1$ for some $z$ from $\mathbb{Z}[\sqrt{2}]$,
that is there are two possible values modulo $\sqrt{2}\mathbb{Z}[\sqrt{2}]$.

The Clifford+$T$ exact synthesis algorithm is shown below.
\begin{algorithm}
\SetAlgoLined
\KwIn{Elements $a,b,c,d$ from $\mathbb{Z}[\sqrt 2]$ such that $\mathrm{Det}(M(a,b,c,d)) = \ell^n$ for $\ell = 2+\sqrt{2}, n \in \mathbb{Z}, n \ge 0$}
\KwOut{Sequence of matrices in $\{T_x, T_y, T_z\}$ and a Clifford unitary.}
 $gates \leftarrow$ empty list\;
 \While{$\mathrm{Det}(M_T(a,b,c,d)) = 0 \mathrm{ mod } \ell$}{
  $g \leftarrow$ Lookup$_T$($a \text{ mod }\ell,b \text{ mod } \ell, c \text{ mod }\ell, d \text{ mod }\ell$)\;
  $M_T(a,b,c,d) \leftarrow g^\dagger M_T(a,b,c,d) / \mathrm{Det}(g)$\;
  prepend $g$ to $gates$\;
 }
 \KwRet{gates, $M_T(a,b,c,d)$}
 \caption{Clifford+$T$ exact synthesis.}
\end{algorithm}
The table Lookup$_T$($a,b,c,d$) is pre-calculated by enumerating over all tuples $a,b,c,d$ modulo $2+\sqrt{2}$, in analogy to the $V$ basis case.
There are only $2^4$ different options to consider in this case.
Once factorizaton is complete, the \change{remaining} matrix that determinant one and therefore is a Clifford unitary.

\subsection{Clifford + \texorpdfstring{$\sqrt{T}$}{RootT}}
Recall from~\sec{approximation-solutions-clifford-root-t} the $\sqrt{T}^k$ matrices

$$\sqrt{T}_Z^k = \frac{1}{\sqrt{2 + 2 \cos\frac{\pi k}{8}}}
\left(\begin{array}{cc}
1 + e^{-i\pi k/8}&0\\0&1 + e^{i\pi k/8}
\end{array}
\right) = \frac{I\cdot(1+\cos\frac{\pi k}{8}) - iZ\cdot\sin\frac{\pi k}{8}}{\sqrt{2 + 2 \cos\frac{\pi k}{8}}}
$$
For $\sqrt{T}^k_X, \sqrt{T}^k_Y$ we replace $Z$ by $X,Y$ in the equation above.

The corresponding quaternion order is
\begin{equation}
    \label{eq:sqrtt-basis-order}
    \mathcal{O} 
    = \mathbb{Z}[2\cos\frac{\pi}{8}]\cdot I 
    + \mathbb{Z}[2\cos\frac{\pi}{8}]\cdot\frac{I+iX}{\sqrt{2}}
    + \mathbb{Z}[2\cos\frac{\pi}{8}]\cdot\frac{I+iY}{\sqrt{2}}
    + \mathbb{Z}[2\cos\frac{\pi}{8}]\cdot\frac{I+iX+iY+iZ}{2}.
\end{equation}
\begin{equation}
\mathbb{Z}[2\cos\frac{\pi}{8}] = \{a + b\cdot 2\cos\frac{\pi}{8} + c\sqrt{2} + d\cdot2\cos\frac{3\pi}{8}: a,b,c,d \in\mathbb{Z}\}
\end{equation}
This order contains the $\sqrt{T}_P, T_P, \sqrt{T}^3_P$ matrices 
when rescaled appropriately. We define such rescaled versions below: 
$$
\sqrt{T}_p =  \ell ( I\cdot(1 + \cos\frac{\pi}{8})-iP\cdot\cos\frac{3\pi}{8} ), \, \ell = (2 + 2\cos\frac{\pi}{8})
$$
$$
\sqrt{T}^3_p = u_1 \ell ( I\cdot(1 + \cos\frac{3\pi}{8})-iP\cdot\cos\frac{\pi}{8} ), u_1 = 1 - 2\cos\frac{3\pi}{8}-\sqrt 2
$$
$$
T_p = u_2 ( I\cdot(1 + \cos\frac{\pi}{4})-iP\cdot\cos\frac{\pi}{4} ), u_2 = 1 + 2\cos\frac{\pi}{8} - 2\cos\frac{3\pi}{8}
$$
With the above rescaling we have 
$$
\sqrt{T}_p(\sqrt{T}_p)^\dagger = \mathrm{Det}(\sqrt{T}^3_p) \cdot I =\ell^3 \cdot I
$$
$$
\sqrt{T}^3_p(\sqrt{T}^3_p)^\dagger = \mathrm{Det}(\sqrt{T}^3_p) \cdot I =\ell^3 \cdot I
$$
$$
T_p(T_p)^\dagger = \mathrm{Det}(T_p)\cdot I =\ell^2 \cdot I
$$

Define $M_{\sqrt{T}}$ as a function that, according to~\eq{sqrtt-basis-order}, maps elements $a,b,c,d$ of $\mathbb{Z}[2\cos\frac{\pi}{8}]$ to a matrix:
\begin{equation}
    M_{\sqrt{T}}(a,b,c,d) := a\cdot I + b\cdot \frac{I + iX}{\sqrt{2}} + c\cdot \frac{I+iY}{\sqrt{2}} + d\cdot \frac{I+iX+iY+IZ}{2}.
\end{equation}

Any matrix $M_{\sqrt{T}}(a,b,c,d)$ with determinant $\ell^n = (2+2\cos(\pi/8))^n$ can be decomposed into a sequence of $\sqrt{T}$ gates~\cite{Forest2015, Kliuchnikov2015b}.  
Unlike the $V$ and $T$ gate sets, the length of the sequence cannot be deduced exactly from the integer power $n$.

We again omit a formal decomposition theorem, which would be analogous to~\theo{v-basis-exact-decomposition}.
The remainder of this section describes $\sqrt{T}$ factorization algorithm.

\begin{lem}[Clifford + $\sqrt{T}$ factorization]
  Let $a,b,c,d\in\mathbb{Z}[2\cos\frac{\pi}{8}]$ such that $\mathrm{Det}(M_{\sqrt{T}}(a,b,c,d)) = \ell^n$ for integer $n \ge 3$, and at least one of $a,b,c,d$ is not divisible by $\ell$.  Then there exists $g\in\{T_p^{k/2}: p\in\{x,y,z\}, k\in\{1,2,3\}\}$ and $a',b',c',d'\in\mathbb{Z}[2\cos\frac{\pi}{8}]$ such that
  \begin{equation}
      g^\dagger M_{\sqrt{T}}(a,b,c,d) = \mathrm{Det}(g) M_{\sqrt{T}}(a',b',c',d').
  \end{equation}
  and $\mathrm{Det}(M_{\sqrt{T}}(a',b',c',d')) = \ell^{n-k}$ where $ \mathrm{Det}(g) = \ell^k$
\end{lem}

One can check via brute-force search that the only matrices with determinants $1,\ell,\ell^2$ are Clifford gates and a $T_p$ times a Clifford gate correspondingly. 

The asymmetry between determinants for $T_p$ and $T_p^{1/2}$ and $T_p^{3/2}$ means that the gate sequence length cannot be calculated from $n$ alone. 
Given a sequence with $N_1$ elements $T_p^{1/2}$, $N_2$ elements $T_p$ and $N_3$ elements $T_p^{3/2}$ the total power $n$ must be $3N_1+2N_2+3N_1$.

\begin{algorithm}
\SetAlgoLined
\KwIn{Elements $a,b,c,d$ from $\mathbb{Z}[2\cos\frac{\pi}{8}]$ such that $\mathrm{Det}(M_{\sqrt{T}}(a,b,c,d)) = \ell^n$ for $ \ell = 2+2\cos\frac{\pi}{8}, n \in \mathbb{Z}, n \ge 0$}
\KwOut{Sequence of matrices in $\{T_x^{1/2}, T_y^{1/2}, T_z^{1/2}, T_x, T_y, T_z, T_x^{3/2}, T_y^{3/2}, T_z^{3/2}\}$ and a Clifford unitary}
 $gates \leftarrow$ empty list\;
 \While{$\mathrm{Det}(M_{\sqrt{T}}(a,b,c,d))\,\mathrm{ mod }\,\ell^3 = 0$}{
  $g \leftarrow$ Lookup$_{\sqrt{T}}$($a\text{ mod }\ell^3,b\text{ mod }\ell^3,c\text{ mod }\ell^3,d\text{ mod }\ell^3$)\;
  $M_{\sqrt{T}}(a,b,c,d) \leftarrow  g^\dagger M_{\sqrt{T}}(a,b,c,d)/\mathrm{Det}(g)$\;
  prepend $g$ to $gates$\;
 }
 Rescale $M_{\sqrt{T}}(a,b,c,d)$ so it has determinant $(2 + \sqrt 2)^n$ for $n=0,1$\;
 Apply Clifford+$T$ synthesis to $M_{\sqrt{T}}(a,b,c,d)$\;
 \KwRet{gates, $M_{\sqrt{T}}(a,b,c,d)$}
 \caption{Clifford + $\sqrt{T}$ exact synthesis.}
\end{algorithm}

Like the Clifford+$T$ case, the table Lookup$_{\sqrt{T}}$($a,b,c,d$) is pre-calculated by enumerating over all tuples $a,b,c,d$ modulo $\ell^3$.
Note, that every element of $\mathbb{Z}[\cos\frac{\pi}{8}]$ can be written as $a_0 + a_1 \ell + a_2 \ell^2 + z \ell^3$
for $a_k \in \{0,1\}$ and $z$ from $\mathbb{Z}[\cos\frac{\pi}{8}]$. 
For this reason, there are only $2^{3\cdot4} = 4096$ options to consider when building the lookup table. 
Using lookup table reduces the number of matrix multiplications needed in the exact synthesis algorithm by factor of nine.

\section{Applications}
\subsection{Shorter quantum circuits for single qubit unitaries (Numerical Results)}\label{sec:numerical-results}

We have implemented algorithms described in the paper in Magma. 
For the numerical results we focus on four approximation protocols for diagonal unitaries~(diagonal, mixed diagonal, fallback and mixed fallback) and two gate sets~(Clifford+$T$ and Clifford+$\sqrt{T}$). 
We target rotations by random angles and by Fourier angles $\pi/2^k$.
The data-sets of angles for which we computed solutions numerically are summarized in \cref{tab:data-sets}.

\change{We have made available the circuits for the approximations we have found in a publicly accessible data-set \cite{CircuitsDataset}. Our results are reproducible using a Python notebook provided on GitHub \cite{CircuitsNotebook}, which uses the circuits as a starting point. The data-set structure is thoroughly documented in the Python notebook, making it easy to compare our results with any future work.
}

\begin{table}[h]
    \caption[Data sets]{
    \label{tab:data-sets}
    Sets of angles for which we computed solutions numerically.
    We approximate diagonal rotations using diagonal, mixed diagonal, fallback and mixed fallback protocols.}
\begin{center}
{
\scriptsize
\setlength{\tabcolsep}{0.5em}
\begin{tabular}{|c|c|c|c||c|c|}
\hline 
Gate set & Cost & \multicolumn{4}{c|}{Data sets and corresponding figures}\\
\cline{3-6} \cline{4-6} \cline{5-6} \cline{6-6} 
 & function & Random angles & Figure & Fourier angles & Figure\\
\hline 
\hline 
\multirow{2}{*}{Clifford+$T$} & \multirow{2}{*}{T-count} & $1358$ uniformly random angles & \multirow{2}{*}{\cref{fig:clifford-t-random}} & $\pi/2^{n}$ & \multirow{2}{*}{\cref{fig:clifford-t-fourier}}\\
 &  & from interval $[0,2\pi]$ &  & $n\in\{3,\ldots,36\}$ & \\
\hline 
\multirow{3}{*}{Clifford+$\sqrt{T}$} & Power  & $1221$ uniformly random angles &  \cref{fig:clifford-root-t-random-power} & $\pi/2^{n}$ & \cref{fig:clifford-root-t-fourier-power} \\
\cline{2-2} \cline{4-4} \cline{6-6} 
 & Gate count & from interval $[0,2\pi]$ & \cref{fig:clifford-root-t-random-rcount} & $n\in\{3,\ldots,45\}$ & \cref{fig:clifford-root-t-fourier-rcount} \\
\cline{2-2} \cline{4-4} \cline{6-6} 
 & T-count &  & \cref{fig:clifford-root-t-random-tcount} &  & \cref{fig:clifford-root-t-fourier-tcount} \\
\hline 
\end{tabular}
}
\end{center}
\end{table}

Numerical results for random angles agree with the heuristic cost scaling derived in~\cref{sec:cost-scaling} as we can see from~\cref{tab:approximation-cost-scaling} and 
from~\cref{fig:clifford-t-random,fig:clifford-root-t-random}.
Results for Fourier angles are a bit more complex.  For many choices of angle and target accuracy, the Identity is a sufficient approximation.
This is evident in the wide gap between minimum and maximum cost illustrated by the shaded regions of~\cref{fig:clifford-t-fourier} and~\cref{fig:clifford-root-t-fourier}.
These low-cost Identity approximations have the effect of pulling down the overall mean cost as compared to random angles.
However, once the Identity is no longer a viable option at high accuracy, the cost scaling is roughly the same as that of random angles.

\begin{figure}
\caption[Cost of approximating Fourier angles with Clifford+T gates]{
\label{fig:clifford-t-fourier}
Cost of approximating a set of Fourier angles rotations~(see~\cref{tab:data-sets}) with Clifford+$T$ gates using four approximation protocols.
We fix a set of approximation accuracy values. For each value in the set we compute mean cost over 
all target angles. 
Shaded regions indicate range of costs from min to max over all angles for given accuracy value. 
In all reported fallback protocols the probability of the fallback step $1-q$ is at most $0.01$.
}
\includegraphics{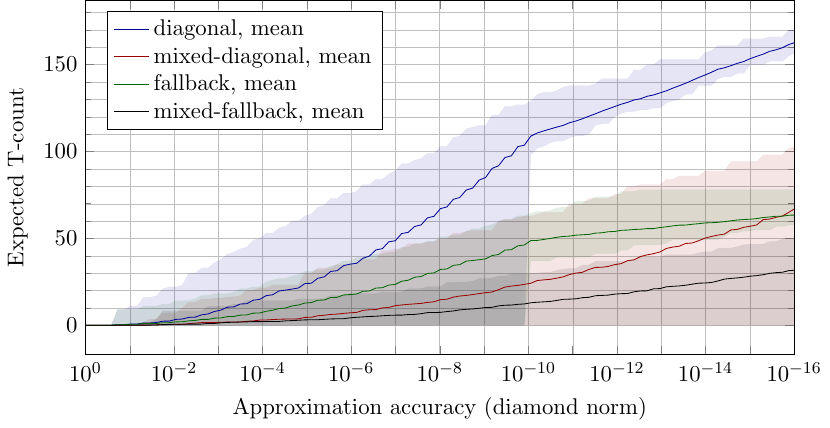}
\end{figure}

\begin{figure}
\caption[Cost of approximating random angles with Clifford+Root(T) gates]{
\label{fig:clifford-root-t-random}
Cost of approximating a set of random rotations~(see~\cref{tab:data-sets}) with Clifford+$\sqrt T$ gates using four approximation protocols.
We fix a set of approximation accuracy values. For each value in the set we compute mean cost over 
all target angles. Vertical bars show the cost standard deviation for given accuracy value. 
Shaded regions indicate range of costs from min to max over all angles for given accuracy value. 
In all reported fallback protocols the probability of the fallback step $1-q$ is at most $0.01$.
The linear fit results are in \cref{tab:approximation-cost-scaling}.
\change{When running approximation algorithms we limited the maximum cost of found approximations 
separately for each protocol. This is why the lines end at different accuracy.}
}
\begin{subfigure}[b]{\textwidth}
 \centering
 \caption{\label{fig:clifford-root-t-random-power} Scaling of denominator power with approximation accuracy. Denominator power of $\sqrt{T}$ and $T$ is $3$ and $2$.}

 \includegraphics[scale=0.8]{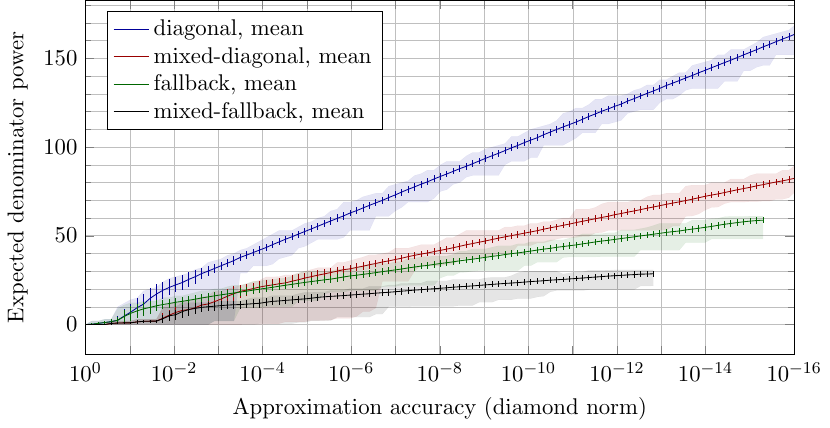}
\end{subfigure}

\vspace{1em}

\begin{subfigure}[b]{\textwidth}
 \centering
 \caption{\label{fig:clifford-root-t-random-rcount} Scaling of gate count with approximation accuracy.  $\sqrt{T}$ gates and $T$ gates contribute $1$ to the gate count.}

 \includegraphics[scale=0.8]{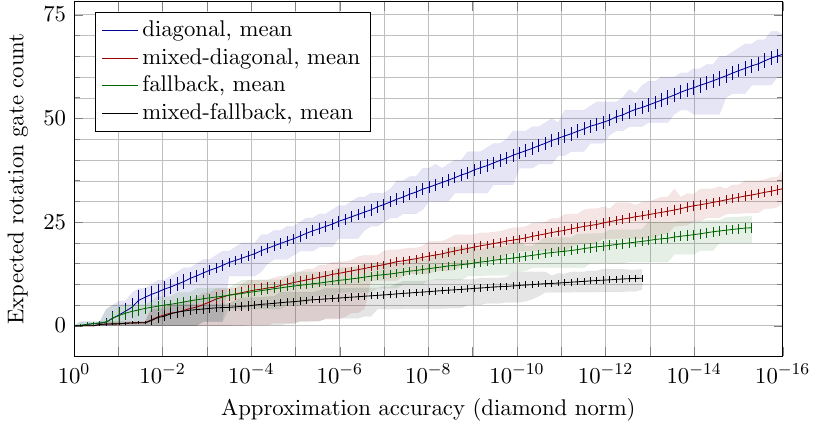}
\end{subfigure}

\vspace{1em}

\begin{subfigure}[b]{\textwidth}
 \centering
 \caption{\label{fig:clifford-root-t-random-tcount} Scaling of T-count with approximation accuracy. T-count of $\sqrt{T}$ gates is four.}
 \includegraphics[scale=0.8]{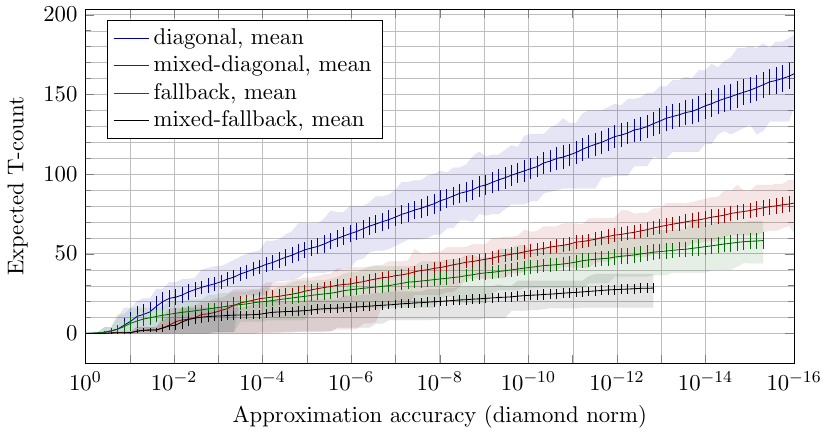}
\end{subfigure}
\end{figure}

\begin{figure}
\caption[Cost of approximating Fourier angles with Clifford+Root(T) gates]{
\label{fig:clifford-root-t-fourier}
Cost of approximating a set of Fourier angles rotations~(see~\cref{tab:data-sets}) with Clifford+$\sqrt{T}$ gates using four approximation protocols.
We fix a set of approximation accuracy values. For each value in the set we compute mean cost over 
all target angles.
Shaded regions indicate range of costs from min to max over all angles for given accuracy value. 
In all reported fallback protocols the probability of the fallback step $1-q$ is at most $0.01$.
\change{When running approximation algorithms we limited the maximum cost of found approximations 
separately for each protocol. This is why the lines end at different accuracy.}
}
\begin{subfigure}[b]{\textwidth}
 \centering
 \caption{Scaling of denominator power with approximation accuracy. Denominator power of $\sqrt{T}$ and $T$ is $3$ and $2$.}
 \label{fig:clifford-root-t-fourier-power}
 \includegraphics[scale=0.8]{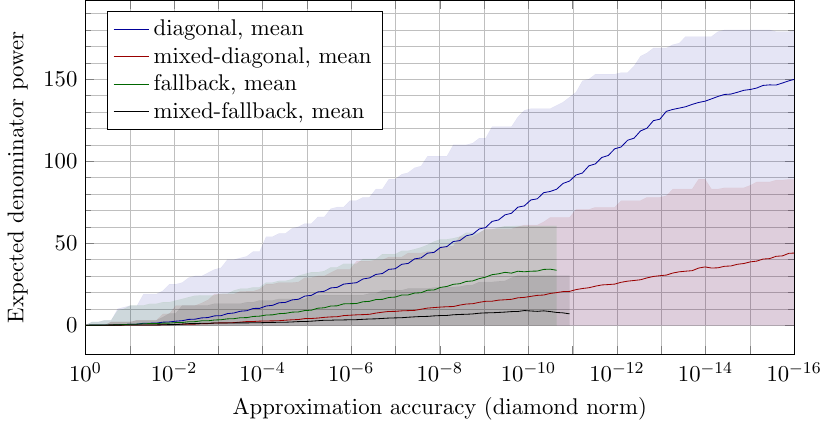}
\end{subfigure}

\vspace{1em}

\begin{subfigure}[b]{\textwidth}
 \centering
 \caption{Scaling of gate count with approximation accuracy.  $\sqrt{T}$ gates and $T$ gates contribute $1$ to the gate count.}
 \label{fig:clifford-root-t-fourier-rcount}
 \includegraphics[scale=0.8]{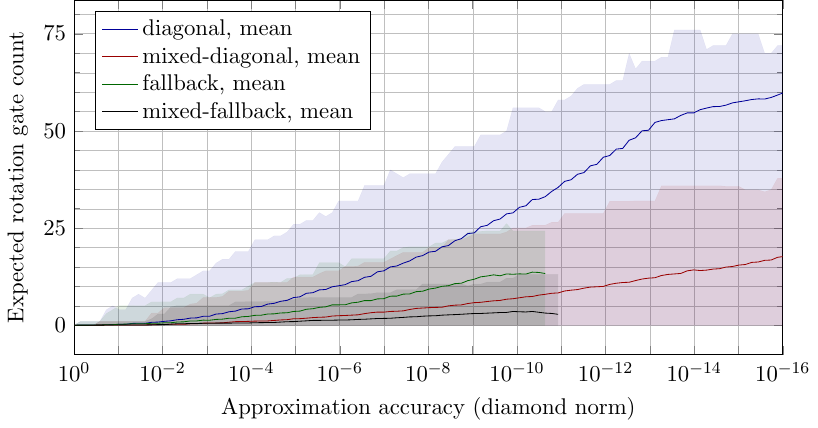}
\end{subfigure}

\vspace{1em}

\begin{subfigure}[b]{\textwidth}
 \centering
 \caption{Scaling of T-count with approximation accuracy. T-count of $\sqrt{T}$ gates is four.}
 \label{fig:clifford-root-t-fourier-tcount}
 \includegraphics[scale=0.8]{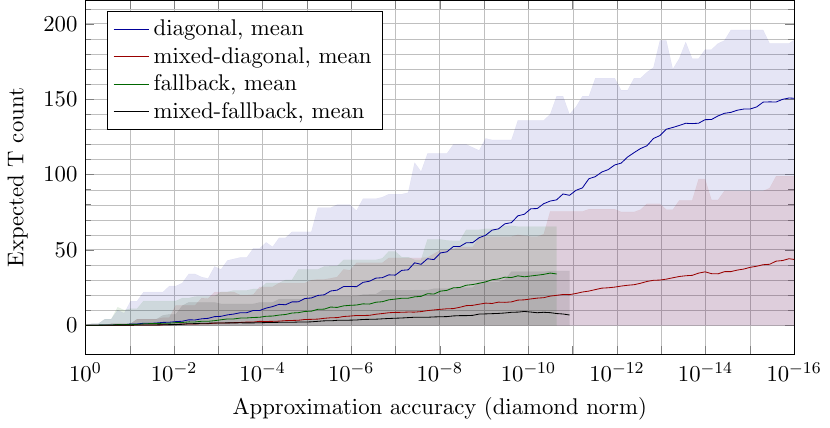}
\end{subfigure}
\end{figure}

There is more variation in gate count and T-count when approximating using the Clifford+$\sqrt{T}$ gate set than when approximating using Clifford+$T$,
 even for random rotations (\cref{fig:clifford-t-random} and \cref{fig:clifford-root-t-random}). This is because our algorithm finds optimal solutions to the sub-problems only with respect to 
denominator power cost function.
For Clifford+$T$  gate set, the power cost function coincides with T-count and non-Clifford gate count. 
For Clifford+$\sqrt{T}$  gate set, the power cost function can be related to T-count and non-Clifford gate count using an additional assumption 
that the number of $\sqrt{T}, \sqrt{T}^3$ gates is the sequence roughly the same as the number of $T$ gates. 
For this reason, we see that the variations in power cost function in \cref{fig:clifford-t-random} and \cref{fig:clifford-root-t-random-power} are similar. 
However, there is more variations in T-count and non-Clifford gate count in \cref{fig:clifford-root-t-random-tcount} and \cref{fig:clifford-root-t-random-power}.

\change{We conducted most of our numerical experiments using Intel Xeon Gold 6136. We did not collect detailed data on the runtime of the algorithms, nor did we attempt a high-performance implementation. Our goal was to create a large enough data-set to help us understand the advantages provided by new approximation protocols and the benefits of using the Clifford+$\sqrt{\text{T}}$ gate-set instead of the Clifford+${\text{T}}$ gate-set.}

\change{To achieve an accuracy of $10^{-15}$ with the Clifford+${\text{T}}$ diagonal approximation, our implementation required 11 seconds. However, achieving the same accuracy with Clifford+$\sqrt{\text{T}}$ required 259 seconds. In both cases, the bottleneck was the integer point enumeration sub-routine. The integer point enumeration sub-routine specially designed for Clifford+${\text{T}}$ in \cite{RossSelinger2014} is much faster, and all protocols discussed here can benefit from it when targeting the Clifford+${\text{T}}$ gate-set.}

\change{In our implementation of integer point enumeration we relied on rational arithmetic instead of floating-point arithmetic to avoid numerical stability issues. However, using floating-point arithmetic, combined with careful treatment of numerical stability, might significantly improve performance.
Integer point enumeration problems that arise require multi-precision arithmetic
which limits the applicability of the existing Integer Programming libraries.
}

\subsection{Further applications}
\label{sec:magnitude-approx-applications}

We show that magnitude approximation problem can provide resource savings not only when approximating 
general $\mathrm{SU}(2)$ unitaries, as discussed in \cref{sec:magnitude-approximation} and \cref{sec:mixed-magnitude-approximation},
but also for qubit state preparation and approximating general $\mathrm{SU}(4)$
unitaries. We believe that idea of approximating $X$ rotations up to $Z$ rotations will be fruitful beyond provided examples.
\newcommand{\appscale}{0.7}

\includegraphics[scale=\appscale]{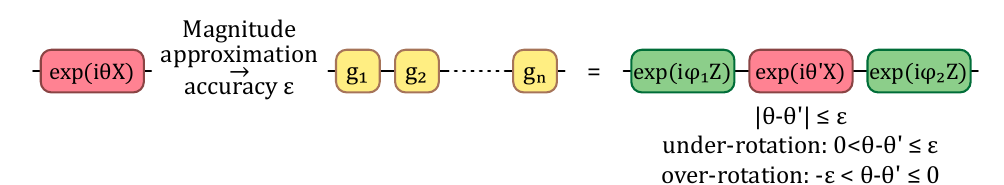}

The main idea is that when approximating X rotation within a quantum circuit, extra Z exponents can be absorbed into surrounding gates.
Similarly, if we are approximating Z rotations, extra X exponents can be absorbed into surrounding gates.

\includegraphics[scale=\appscale]{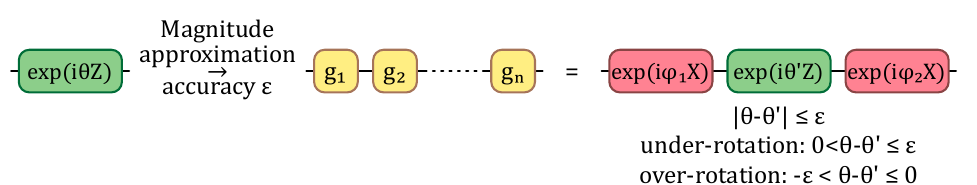}

It is easy to exchange X and Z in our circuits by using Hadamard gates and following circuit identities:

\includegraphics[scale=\appscale]{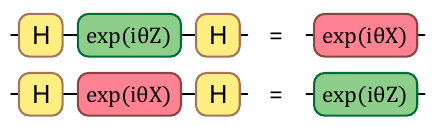}

For the above to hold we require that the gate set is fixed by Hadamard conjugation, that is for any gate $g$ from the gate set, $HgH$ is also in the gate set.
Luckily Clifford+$T$, Clifford+$\sqrt{T}$ and V basis all have this property.
The use of the magnitude approximation for approximating $\mathrm{SU}(2)$ unitaries discussed in \cref{sec:magnitude-approximation}
is summarized using circuit diagrams as follows:

\includegraphics[scale=\appscale]{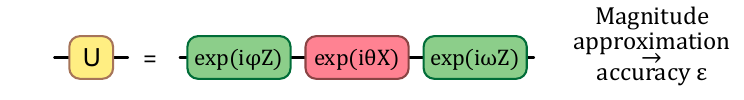}

\includegraphics[scale=\appscale]{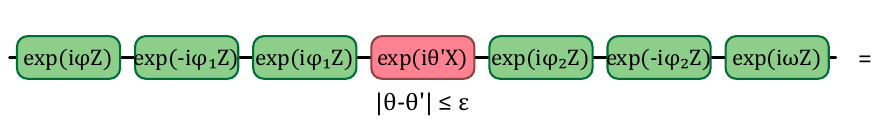}

\includegraphics[scale=\appscale]{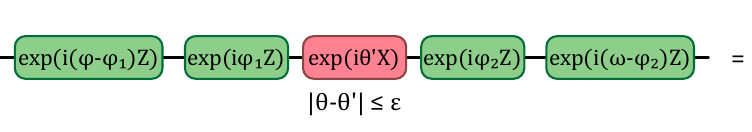}

\includegraphics[scale=\appscale]{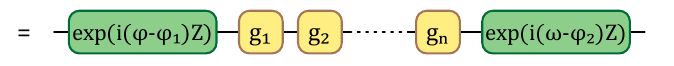}

Approximating $\mathrm{SU}(2)$ requires solving one magnitude approximation and two diagonal approximation problems.
Similarly we improve the preparation of an arbitrary one qubit state.

\includegraphics[scale=\appscale]{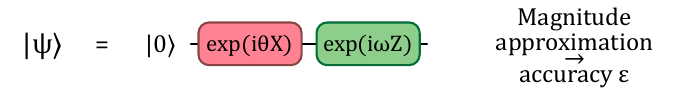}

\includegraphics[scale=\appscale]{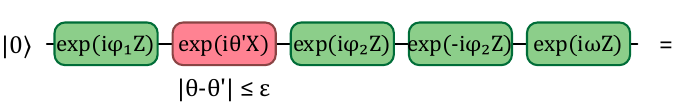}

\includegraphics[scale=\appscale]{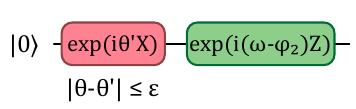}
 
Above we use the fact that $|0\rangle$ is an eigenstate of any Z rotation.
Approximating  qubit state requires solving one magnitude approximation and one diagonal approximation problem. 

Finally, we show that magnitude approximation can be used to find shorter approximations of two qubit unitaries.
We use a rotation and CNOT optimal circuit from \href{https://arxiv.org/pdf/quant-ph/0308033.pdf\#page=5}{arxiv:quant-ph/0308033}.

\includegraphics[scale=\appscale]{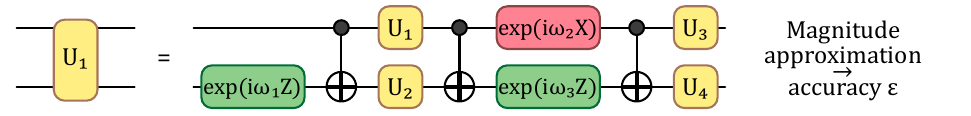}

\includegraphics[scale=\appscale]{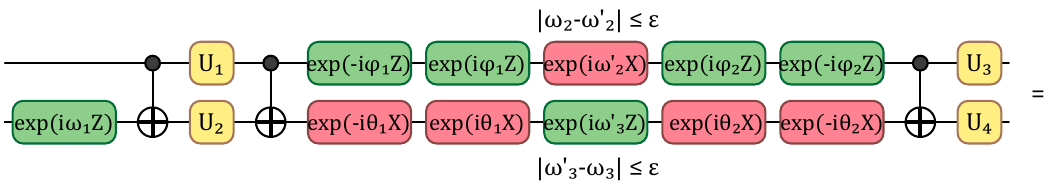}

\includegraphics[scale=\appscale]{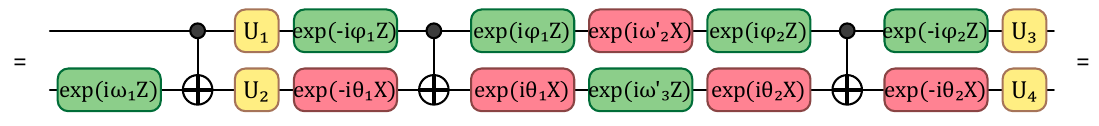}

\includegraphics[scale=\appscale]{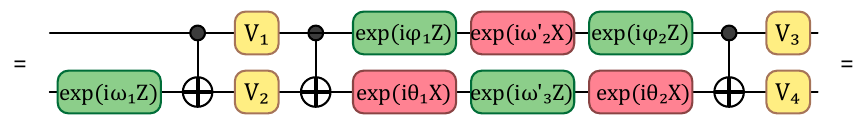}

\includegraphics[scale=\appscale]{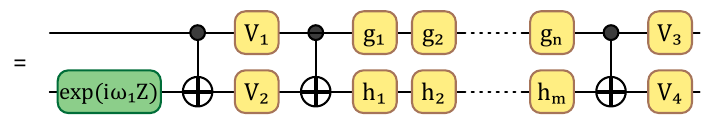}

Above we used the following circuit identities: 

\includegraphics[scale=\appscale]{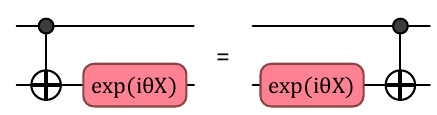}
\hspace{3em}
\includegraphics[scale=\appscale]{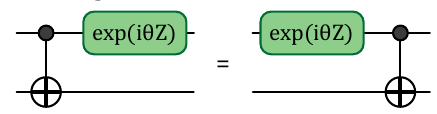}

which follow from representing CNOT matrix as: 

$$
 \mathrm{CNOT} = |0\rangle \langle 0| \otimes I + |1\rangle \langle 1| \otimes X = ((I+Z)\otimes I + (I-Z)\otimes X)/2
$$

We then apply the result for approximating arbitrary qubit unitaries to $V_1, V_2, V_3, V_4$. 
Approximating $\mathrm{SU}(4)$ requires solving six magnitude approximation and nine diagonal approximation problems.



Mixing of under-rotated
and over-rotated magnitude approximations applies in all of the above cases similarly to the general $\mathrm{SU}(2)$ case
considered in \cref{sec:magnitude-mixing}. Our improvements to approximating general SU(2), SU(4) unitaries and to 
qubit state preparation are summarized in~\cref{tab:magnitude-approx-applications}, along with comparison to the previous state of the art.

\section{Related problems and algorithms}
\label{sec:crypto-connection}
  In this section we will recall some definitions and results about cryptographic hash functions. In particular, we explain the connection between the Charles, Goren and Lauter hash construction \cite{charles2009cryptographic},  built from LPS graphs, to unitary synthesis problems.

A \emph{hash function} $h:\{0,1\}^\ast\rightarrow\{0,1\}^m$ is a function which takes bitstrings of arbitrary length as inputs,  and outputs bitstrings of fixed length. A hash function is required to be \textit{preimage resistant}; that is, given a value $y\in\{0,1\}^m$ in the image of $h$, it must be computationally infeasible to find a bitstring $x$ which hashes to that value. This is formalized in \problem{preimage}.

\begin{prob}[Preimage Finding Problem]\label{prob:preimage}
Given a hash function $h$ and a value $y\in\mathrm{Im}(h)$, find $x$ such that $h(x) = y$.
\end{prob}

  There are several constructions of hash functions built on Cayley graphs. Given a group $\mathcal{G}$ with generating set $S = \{s_0,\dots,s_k\}$, the corresponding Cayley graph has vertices associated with elements $g$ in $\mathcal{G}$ and directed edges $(g,h)$ if and only if $gh^{-1}\in S$. Writing a message $m=m_1m_2\dots m_k$ with $m_i\in\{0,\dots,k\}$, the hash function is defined by $H(m) = s_{m_1}s_{m_2}\dots s_{m_n}$. For such constructions, called Cayley hash functions, \problem{preimage} can be reformulated as the group theoretic problem below.

\begin{prob}[Constructive Membership Problem]\label{prob:factor}
Let $\mathcal{G}$ be a group with generating set $S = \{s_1,\dots,s_k\}$ and let $N\in\z$ be small. Given an element $g\in\mathcal{G}$, find a sequence $m_1,\dots,m_N$ such that $g = \prod_is_{m_i}.$
\end{prob}

In \cite{charles2009cryptographic}, Charles, Goren and Lauter (CGL) proposed a Cayley hash function based on LPS graphs. LPS graphs were introduced by Lubotsky, Phillips and Sarnak in \cite{lubotzky1988ramanujan}. Let $p,\ell$ be distinct primes congruent to $1\bmod 4$, \change{where $\ell$ is a quadratic residue modulo $p$}. Let $\f_p$ denote the finite field with $p$ elements and let $\iota$ such that $\iota^2 = -1\mod p$. An LPS graph $X_{p,\ell}$ is the Cayley graph with $\mathcal{G}=PSL(2,\f_p)$, the projective special linear group of $2\times2$ matrices over $\f_p$, and generating set $S=\left\{\left(\begin{smallmatrix}
    a+\iota b & c + \iota d\\ -c + \iota d & a-\iota b
\end{smallmatrix}\right): a^2 + b^2 + c^2 + d^2 = \ell\right\}$, where $a>0$ and $b,c,d$ even. We can write $g\in\mathcal{G}$ as $\left(\begin{smallmatrix}
    a+\iota b & c + \iota d\\-c+\iota d & a-\iota b
\end{smallmatrix}\right)$ with $a,b,c,d\in\f_p$ and define the norm function $n(g)=a^2 + b^2 + c^2 + d^2.$  The preimage problem for the CGL hash function amounts to path finding on an LPS graph. Since these are Cayley graphs, the preimage problem is equivalent to \problem{factor}.

Recall that the unitary synthesis problem is the search for a circuit, or sequence, of unitaries from a specified gate set that is equivalent to some target unitary. Clearly, this problem is highly related to \problem{factor}. Unitaries are represented by matrices over $\c$ and vertices in $X_{p,\ell}$ correspond to matrices over $\f_p$, and both problems look for `short' sequences in a subset of matrices. Remarkably, algorithms developed independently to solve the constructive membership problem for LPS graphs \cite{petit2008full,sardari2017complexity} and the quantum unitary synthesis problem \cite{Ross2015, BlassEtAl2015} have many similarities. Petit, Lauter and Quisquater \cite{petit2008full} proposed an algorithm for finding short paths in LPS graphs in which a matrix from the group $\mathcal{G}$ is decomposed into the product of four diagonal matrices with square determinant and graph generators, up to multiplication by a unit. Each diagonal matrix is factorized into elements from $S$ using an extension of the Tillich-Z\'emor algorithm \cite{tillich2008collisions} for collision finding in an LPS graph. This diagonal decomposition method is reminiscent of the Euler decomposition method for unitary synthesis, described in greater detail in \sec{approximation-problems}, in which the target unitary is decomposed into the product of $Z$-axis rotations. Notably, $Z$-axis rotations can be expressed as diagonal matrices. Carvalho Pinto and Petit \cite{pinto2018better} later improved upon the algorithm in \cite{petit2008full}, by decomposing the target matrix into the product of two diagonal matrices and a third non-diagonal, easily-factorizable matrix, resulting in path lengths of $7\log_\ell(p)$. In \sec{approximation-problems} we translate the algorithm to the continuous setting of general unitary approximation, achieving a similar improvement in sequence length. We obtain an additional constant factor improvement by implementing approximation via quantum channel mixing. 
 
We deal with the problem of approximating unitaries to some chosen accuracy $\varepsilon$. The algorithms described in \cite{petit2008full} and \cite{pinto2018better} both involve `lifting' a matrix $M\in PSL(2,\f_p)$ to a matrix $M'\in GL(\z[i])$, such that the corresponding entries of each matrix are congruent modulo $p$. In other words, for some well-defined p-adic norm the distance between $M$ and $M'$ is $O(p^{-1})$. The matrix $M'$ is then factorized over $GL(\z[i])$, with some conditions regarding the determinant size, with each factor mapped back to $PSL(2,\f_p)$  via a group homomorphism. The lifting step is analogous to approximation in the quantum setting, using $p^{-1}$ as a measure of accuracy. Clearly, $p^{-1}$ is analogous to $\varepsilon.$ Of course, in the LPS hash setting $p$ is fixed, whereas in the quantum setting we have some control over the value $\varepsilon$. The length of a sequence indicates the cost of approximating the target unitary in the context of gate synthesis. For the CGL hash function, the sequence length will equal the length of the corresponding path in the LPS graph, and is similarly used as measure of performance for path-finding algorithms. The length of a sequence is determined by taking the norm of the target matrix. For matrices over $\c$, we can use some some complex matrix norm, while matrices in $PSL(2,\f_p)$ use some $p$-adic norm. For instance, the six unitary approximation problems defined in \sec{approximation-problems} use the diamond norm to measure accuracy. Note, however, that despite the similarities just described, not all of these approximation problems have natural analogues in cryptography. In particular, those problems that utilize fall-back and channel mixing techniques do not translate to the classical setting. Moreover, the other properties required of cryptographic hash functions - collision resistance and second preimage resistance - do not yet have quantum approximation analogues, either.

\CP{One problem relevant for the hash function is a non-trivial factorization of the identity; I suppose this one will have absolutely zero interest for gate synthesis? (unless there is a good reason to introduce dummy computation at places?)}
\VK{We have studied approximation of identity, that is not identity for the sake of lower-bounds in 
other papers. For example, see \href{https://arxiv.org/pdf/1904.01124.pdf\#thm.5.2}{Theorem 5.2 in arXiv:1904.01124}.}

The gate sets considered in this paper are quaternion gate sets, so-called due to their relationship to quaternion algebras, \change{which we explain in greater detail in \sec{approximation-problems:general}.} Sarnak first observed the connection between LPS graphs, quaternion orders and quantum gate sets in his letter to Aaronson and  Pollington on the Solvay-Kitaev Theorem and golden gates \cite{sarnak2637letter}. For instance, synthesis over the $V$ basis is analogous to path finding in an LPS graph $X_{p, \ell}$, where $p\equiv 1\bmod 4$ and $\ell=5$. We return to the V basis in \sec{approx-solutions}, as an example of a quaternion gate set, along with the Clifford $+T$ basis and the Clifford$+\sqrt{\text{T}}$ basis.

\printbibliography

\appendix

\section{An example of V basis diagonal approximation of \texorpdfstring{$e^{i\frac{\pi}{4}Z}$}{eipi4Z}}\label{sec:V-basis-example}

We use the notation $I,X,Y,Z$ for Pauli matrices:
\[
I = \at{\begin{array}{cc}1 & 0 \\ 0 & 1\end{array}},
\,
X = \at{\begin{array}{cc}0 & 1 \\ 1 & 0\end{array}},
\,
Y = \at{\begin{array}{cc}0 & -i \\ i & 0\end{array}},
\,
Z = \at{\begin{array}{cc}1 & 0 \\ 0 & -1\end{array}}.
\]
Recall that the V basis consists of the following six  matrices:
\begin{align*}
V_{\pm Z} &= \frac{1}{\sqrt{\ell}}\left(I\pm 2iZ\right),
&V_{\pm Y} &= \frac{1}{\sqrt{\ell}}\left(I\pm 2iY\right),
&V_{\pm X} &= \frac{1}{\sqrt{\ell}}\left(I\pm 2iX\right),
\end{align*}
where $\ell= 5$. 
Let $\theta = \frac{\pi}{4}$ and suppose we want to approximate $U= e^{i\theta Z} = \left(\begin{smallmatrix}
e^{i\pi/4} & 0\\
0 & e^{-i\pi/4}
\end{smallmatrix}\right)$ using the V basis within accuracy $\varepsilon=0.1$ with respect to the diamond norm. In other words, we look for $W$, a product of unitaries from the V basis, which satisfies $\nrm{\mathcal{Z}_\theta - \mathcal{W}}_\diamond\le \varepsilon,$ where $\mathcal{Z}_\theta$ and $\mathcal{W}$ are the channels\footnote{The channel induced by a unitary $U$ is an action of $U$ on a density matrix $\rho$: $\mathcal{U}(\rho)=U\rho U^\dagger$. Channels and density matrices are defined fully in \sec{approximation-problems}.} induced by $e^{i\theta Z}$ and $W$, respectively. 

Writing $W$ as $\left(\begin{smallmatrix}
u&-v^\ast\\v&u^\ast
\end{smallmatrix}\right)$, with $u,v\in\c$, we obtain the following:
\begin{equation}\label{eq:example-constraint}
\abs{\mathrm{Re}(ue^{-i\pi/4})} \ge 1 - \varepsilon^2/8 \implies \nrm{\mathcal{Z}_{\pi/4} - \mathcal{W}}_\diamond\le \varepsilon.
\end{equation}
\change{The full derivation of this constraint is provided in \sec{approximation-problems}.} The constraint on $u$ is represented geometrically by the region in \fig{region0}. 

\begin{figure}[!h]
    \centering
    
   \includegraphics[width=0.4\textwidth]{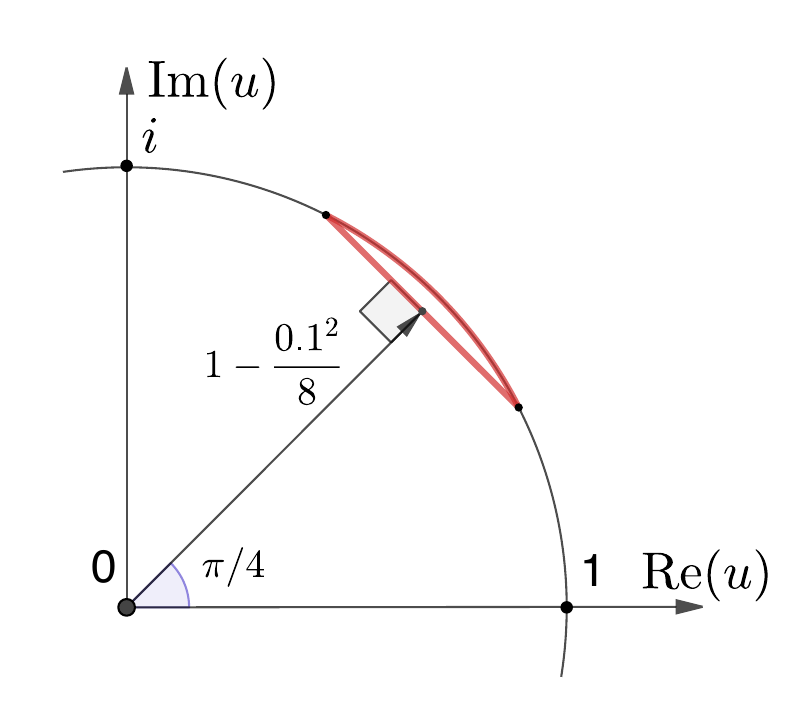}
    \caption{Geometric interpretation of constraint on complex number $u$ in Equation \eq{example-constraint}. The region with the red boundary contains candidate points $(a,b)\in\z^2$, such that $u=a+ib$ and $\abs{\mathrm{Re}(ue^{-i\pi/4})} \ge 1 - (0.1)^2/8$.
     \label{fig:region0}}
\end{figure}

Since $W$ is a product of V basis matrices, there exists $N\in\n$ such that $V=\frac{1}{\sqrt{5^N}}\left(\begin{smallmatrix}
u^\prime & -(v^\prime)^\ast\\v^\prime&(u^\prime)^\ast
\end{smallmatrix}\right)$, with $u^\prime,v^\prime\in\z[i]$. It follows that $u = u^\prime/\sqrt{5^N}$ and $v=v^\prime/\sqrt{5^N}$. Hence, we scale the region in \fig{region0} by $\sqrt{5^N}$ and look for integer points $(a,b)\in\z^2$, each corresponding to a candidate $u'=a+ib.$ We initialize $N:=1$, and iterate over $N$ until a solution is found.

 We find that there are no integer solutions for $N=1,2,3,4$. At $N=5$, there are four candidates for $u'$, namely $\{38+41i,39+40i,40+39i,41+38i\}$, shown in \fig{region5}. Since $V$ is unitary, we require $\det(V) = uu^\ast+vv^\ast = 1$ or, equivalently, $u^\prime(u^\prime)^\ast + v^\prime(v^\prime)^\ast = 5^5=3125.$ So we must have $0\le v^\prime(v^\prime)^\ast = 3125 - u^\prime(u^\prime)^\ast.$ Then,
\begin{eqnarray}
u^\prime=38+41i&\implies& u^\prime(u^\prime)^\ast = 38^2+41^2 = 3125\label{eq:candidate1}\\
u^\prime=39+40i&\implies& u^\prime(u^\prime)^\ast = 39^2 + 40^2 = 3121\label{eq:candidate2}\\
u^\prime=40+39i&\implies& u^\prime(u^\prime)^\ast = 3121\label{eq:candidate3}\\
u^\prime=41+38i&\implies& u^\prime(u^\prime)^\ast = 3125\label{eq:candidate4}.
\end{eqnarray}

  \begin{figure}[!h]
    \centering
    \includegraphics[width=0.7\textwidth]{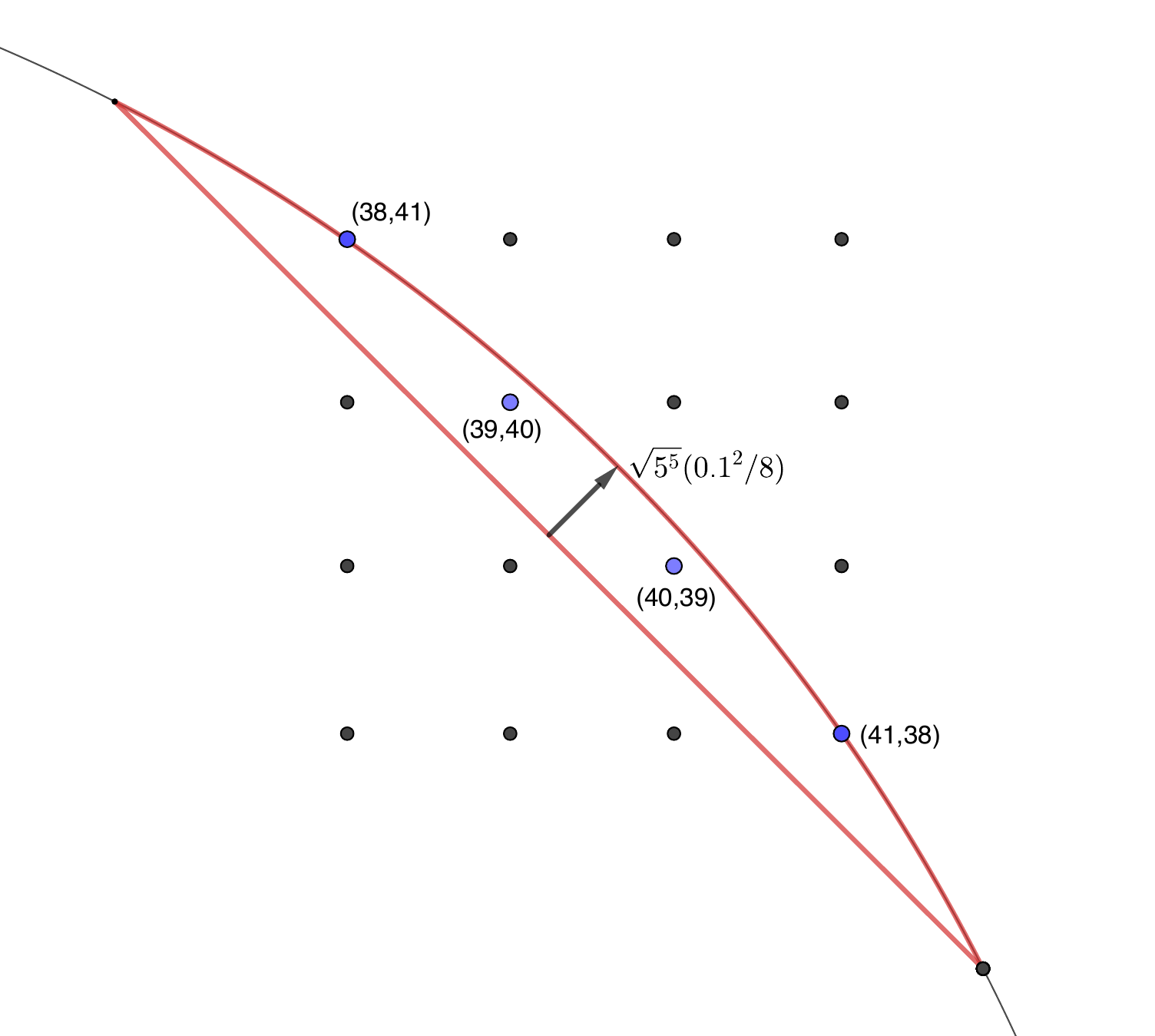}
    \caption{Geometric interpretation of the constraint on complex number $u'$, such that $W = \frac{1}{\sqrt{5^5}}\left(\begin{smallmatrix}u' & -(v')^\ast\\v' & (u')^\ast \end{smallmatrix}\right)$ approximates $e^{i\frac{\pi}{4}Z}$ to accuracy $\varepsilon=0.1$. The region with the red boundary contains four candidate complex numbers satisfying $\abs{\mathrm{Re}(u'e^{-i\pi/4})} \ge \sqrt{5^5}(1 - (0.1)^2/8)$. \label{fig:region5}}
    \end{figure}

Let $v^\prime = c + id,$ so \begin{equation}\label{eq:example-norm}v^\prime(v^\prime)^\ast = c^2 + d^2 = 5^5 - (a^2 + b^2).\end{equation} For Equations \eq{candidate1} and \eq{candidate4}, we have $v^\prime(v^\prime)^\ast = 0$ so $v=0$ is the only solution. Equations \eq{candidate2} and \eq{candidate3} yield $v^\prime(v^\prime)^\ast =4,$ so $c^2 + d^2 = 4 = 2^2$ then either $c=\pm2,d=0$  or $c=0,d=\pm2$. The two corresponding values for $v^\prime$ are $\pm2$ and $\pm2i$. In general, Equation \eq{example-norm} admits a solution for $v\in\z[i]$ if and only if all terms $p^k$ in the prime factorization of $5^5 - (a^2 + b^2)$, with $p\equiv 3\bmod{4}$, have even exponent $k$.
Each candidate pair $(u^\prime,v^\prime)$ defines an approximation unitary $W = \frac{1}{\sqrt{3125}}\left(\begin{smallmatrix}
u^\prime & -(v^\prime)^\ast\\v^\prime&(u^\prime)^\ast
\end{smallmatrix}\right)$, which is factorized over the V basis. \change{We have used the method for efficient factorization outlined in \sec{exact-synthesis}.} These factorizations are given in Table \ref{tab:V-basis-solutions}.

\begin{table}[!h]
    \centering
    \begin{tabular}{c|c|c}
        $u^\prime$ & $v^\prime$ & V basis factorization \\
        \hline
        $41+38i$&0 & $(V_{-Z})^5$\\
        &&\\
         $38+41i$&0& $iZ\cdot(V_{+Z})^5$\\
        &&\\
        $39+40i$&$2i$&$e^{i\pi}\cdot V_{-X}V_{-Y}V_{+X}V_{+Y}V_{-X}$\\
                 &$2$&$e^{i\pi}\cdot V_{+Y}V_{-X}V_{-Y}V_{+X}V_{+Y}$\\
       &$-2i$&$e^{i\pi}\cdot V_{+X}V_{+Y}V_{-X}V_{-Y}V_{+X}$\\
        
          &$-2$&$e^{i\pi}\cdot V_{-Y}V_{+X}V_{+Y}V_{-X}V_{-Y}$\\

        &&\\
        $40+39i$&$2i$&$-iZ\cdot V_{-Y}V_{-X}V_{+Y}V_{+X}V_{-Y}$\\
                &$2$&$-iZ\cdot V_{+X}V_{-Y}V_{-X}V_{+Y}V_{+X}$\\
          &$-2i$&$-iZ\cdot V_{+Y}V_{+X}V_{-Y}V_{-X}V_{+Y}$\\
        
        &$-2$&$-iZ\cdot V_{-X}V_{+Y}V_{+X}V_{-Y}V_{-X}$\\

    \end{tabular}
    \caption{V basis factorizations of unitaries $W=\frac{1}{\sqrt{5^5}}\left(\begin{smallmatrix}u'&-(v')^\ast\\v'&(u')^\ast\end{smallmatrix}\right)$, satisfying \change{$\nrm{\mathcal{Z}_{\pi/4}-\mathcal{W}}_\diamond\leq\varepsilon=0.1$.}}
    \label{tab:V-basis-solutions}
\end{table}

\section{Properties of the diamond norm}\label{app:diamond-norm-properties}

We use the diamond norm as the accuracy metric for all approximation problem definitions.
Let us recall why we can replace various parts of a quantum algorithm with their approximations and still get useful results.
The diamond norm is the key mathematical tool for understanding this.
The result of running any quantum algorithm is a sample from a probability distribution.
Let us call this distribution the answer distribution.
We then process the answer distribution to get the final answer.
This processing of the answer distribution is robust, that is if we are given a sample from 
a distribution that is within total variational distance $\varepsilon$ from the answer 
distribution we can still recover the final answer. 
The value $\varepsilon$ is different for different quantum algorithms.
Every quantum algorithm corresponds to a quantum channel.
Suppose that the diamond norm distance between the channels corresponding to 
the ideal quantum algorithm and its approximation is $\varepsilon$.
Then, the total variational distance between the ideal algorithm's answer distribution 
and the answer distribution produced by the approximation is also $\varepsilon$.
The proof of this fact follows from the definition of the diamond norm that is discussed below.

The diamond norm also has two key properties that let us estimate the distance 
between two algorithms, given the diamond norm distances between their parts.
The first property is the chain rule for the composition of channels, 
that is for channels $\Phi_1$, $\Phi_2$, $\Psi_1$, $\Psi_2$ we have 
$$
 \nrm{ \Phi_1 \Psi_1 - \Phi_2 \Phi_2 }_\diamond \le \nrm{\Phi_1-\Phi_2}_\diamond +  \nrm{\Psi_1-\Psi_2}_\diamond
$$
That is if we replaced $N$ parts of a quantum algorithm with their $\varepsilon$ approximations, 
the distance between the quantum algorithm and its approximation is at most $ N \cdot \varepsilon$.
The second property is stability with respect to the tensor product. 
When we write a quantum algorithm as a composition of channels $\Phi_1 \ldots \Phi_N$
each acting on $n$-qubits, it is frequently the case that each $\Phi_k$
acts non-trivially on one qubit, that is
$$
\Phi_k = \Phi'_k \otimes \mathcal{I} \text{ where } \mathcal{I}\text{ is the identity channel on } n-1 \text{ qubits.}
$$
When we replace $\Phi'_k$ with its approximation $\Psi'_k$ we can argue that $\Psi'_k \otimes \mathcal{I}$
is close to $\Phi_k$ by using the stability with respect to the tensor product:
$$
\nrm{ \Phi'_k \otimes  \mathcal{I} -  \Psi'_k \otimes  \mathcal{I} }_\diamond = \nrm{ (\Phi'_k  -  \Psi'_k) \otimes  \mathcal{I} }_\diamond = \nrm{ \Phi'_k -  \Psi'_k  }_\diamond
$$ 
For an explanation of this and many other useful properties of the diamond norm we refer the reader to \href{https://cs.uwaterloo.ca/\~watrous/TQI/TQI.pdf\#page=174}{Chapter~3.3.2} of~\cite{watrous2018}. 
For completeness we provide the basic definition and other important properties of the diamond norm below.

Let $\c^{d\times d}$ be the linear space over $\c$ of $d$ by $d$ matrices with entries in $\c$.
For an arbitrary element of $\c^{d\times d}$ the Schatten one norm (also known as trace norm) is defined as $\nrm{A}_1 = \tr \sqrt{A^\dagger A}$.
For an arbitrary linear transformation $\Phi$ from $\c^{d\times d}$ into $\c^{d\times d}$, the \emph{induced} one norm is defined as
\begin{equation}\label{eq:induced-trace-norm}
    \nrm{\Phi}_1 = \max \set{ \nrm{\Phi\at{X}}_1 : X \in \c^{d\times d}, \nrm{X}_1 \le 1}.
\end{equation}
The diamond norm (also known as the completely-bounded trace norm) is a ``stable'' version of the induced one norm
\begin{equation}
    \nrm{\Phi}_{\diamond} = \nrm{\Phi \otimes \mathcal{I}_d }_1,
\end{equation}
where $\mathcal{I}_d$ is the identity map from $\c^{d\times d}$ into $\c^{d\times d}$. 
We call diamond norm stable, because in general 
$
\nrm{ \Phi }_1 \le \nrm{ \Phi \otimes \mathcal{I}_k }_1.
$
There are examples where the inequality is strict.
However, for any $k \ge d$ we have equality $\nrm{ \Phi \otimes \mathcal{I}_k }_1 = \nrm{ \Phi \otimes \mathcal{I}_d }_1$.

Direct calculation of the diamond norm is tedious, in general.  
However, there are two cases useful for this paper when there is a simple way to calculate the diamond distance. 
The first case is the diamond distance between two channels $\mathcal{U}$ and $\mathcal{V}$ induced by unitaries $U$ and $V$.
Below is a re-statement of \href{https://arxiv.org/pdf/0711.3636.pdf#page=13}{Theorem 26} in \cite{ComputingStabilizedNorms}:
\begin{thm}[Diamond distance between unitary channels]
\label{thm:diamond-distance-between-unitary-channels}
For any two unitary operators $U,V$, the diamond norm of the difference of the unitary channels 
$\mathcal{U}$, $\mathcal{V}$ induced by $U,V$ is equal to the diameter of the smallest disc (not necessarily centered at the origin)
containing all the eigenvalues of $U^\dagger V$.
\end{thm}

We use above result when approximating unitaries by unitaries.
The second case is the diamond distance between two Pauli channels.
Below we restate the result \href{https://arxiv.org/pdf/1109.6887.pdf#page=14}{Section V.A} in \cite{Magesan2012}
and definition of a Pauli channel:
\begin{thm}[Diamond norm distance between Pauli channels]
\label{thm:app-diamond-distance-beween-pauli-channels}
Suppose $\mathcal{E}_1$,  $\mathcal{E}_2$ are $n$-qubit Pauli channels, that is 
$$
  \mathcal{E}_1(\rho) = \sum_{P \in \{I,X,Y,Z\}^{\otimes n}} q_P P \rho P^\dagger,\,  \mathcal{E}_2(\rho) = \sum_{P \in \{I,X,Y,Z\}^{\otimes n}} r_P P \rho P^\dagger, 
$$
then $\nrm{\mathcal{E}_1-\mathcal{E}_2}_\diamond = \sum_{P \in \{I,X,Y,Z\}^{\otimes n}} |q_P - r_P|$.
\end{thm}
We use above result when approximating unitaries by probabilistic mixtures of unitaries. 
To take advantage of the above property we frequently use the unitary invariance of the diamond norm.
\begin{prop}[Unitary invariance of the diamond norm]
\label{prop:diamond-norm-unitary-invariance}
Let $\Phi$ be channel and $\mathcal{U},\mathcal{U}^\dagger$ be unitary channels induced by 
unitaries $U$,$U^\dagger$, then the following holds:
$$
\nrm{ \Phi - \mathcal{U}}_\diamond  = \nrm{ \mathcal{U}^\dagger \Phi - \mathcal{I} }_\diamond =  \nrm{ \Phi \mathcal{U}^\dagger - \mathcal{I} }_\diamond
$$
\end{prop}
\begin{proof}
The first equality follows from the unitary left and right invariance of the Schatten one norm ( also know as trace norm ).
That is for any matrix $A$ we have $\nrm{AU}_1 = \nrm{UA}_1 = \nrm{A}_1$.
For any density matrix $\rho$ we have
$$
\nrm{ (\Phi\otimes \mathcal{I})(\rho) - (U \otimes I) \rho (U^\dagger \otimes I) }_1 = \nrm{ (U^\dagger \otimes I) (\Phi\otimes \mathcal{I})(\rho) (U \otimes I) - \rho }_1,
$$
which show the first equality by the definition of the diamond norm.
The second equality follows from the fact that taking minimum over all matrices in \cref{eq:induced-trace-norm} with trace norm 
at most one is the same as taking minimum over all matrices $(U^\dagger \otimes I) X  (U \otimes I)$
with trace norm at most one.
\end{proof}

When approximating unitary $U$ with probabilistic mixtures of unitaries described by a channel $\Phi$, we typically find that 
$\mathcal{U}^\dagger \Phi $ is a Pauli channel. For qubit unitaries the following corollaries of \cref{thm:diamond-distance-between-unitary-channels}
are useful:

\begin{cor}[Diamond distance between qubit diagonal unitaries]
\label{cor:diamond-distance-between-diaognal}
For any real numbers $\phi_1,\phi_2$, the diamond norm distance between channels $\mathcal{Z}_{\phi_1}$, $\mathcal{Z}_{\phi_2}$ induced by unitaries $e^{i\phi_1 Z}$, $e^{i\phi_2 Z}$
is
$$
\nrm{\mathcal{Z}_{\phi_1} - \mathcal{Z}_{\phi_2}}_\diamond = 2|\sin(\phi_1 - \phi_2)|\le 2|\phi_1-\phi_2|.
$$
\end{cor}
\begin{proof}
Use \cref{thm:diamond-distance-between-unitary-channels} and note that eigenvalues of $e^{i\phi_1 Z} e^{-i\phi_2 Z}$ are $e^{\pm i \delta}$ where $\delta = \phi_1 - \phi_2$.
The diameter of the disc enclosing both eigenvalues is $|e^{i\delta} - e^{-i\delta}| = 2|\sin(\delta)|$.
Finally we use inequality $|\sin(\delta)| \le |\delta|$.
\end{proof}

There is a closed form expression for diamond norm distance when approximating diagonal qubit unitary by any qubit unitary:

\begin{cor}[Diamond distance between qubit diagonal and general qubit unitaries]
\label{cor:diamond-distance-between-general-and-diaognal}
For any real numbers $\phi$ and special qubit unitary $U  =  \at{\begin{smallmatrix} u & -v^\ast \\ v & u^\ast \end{smallmatrix}}$, the diamond norm distance between channels $\mathcal{Z}_{\phi}$, $\mathcal{U}$ induced by unitaries $e^{i\phi Z}$, $U$
is
$$
\nrm{\mathcal{Z}_{\phi} - \mathcal{U}}_\diamond = 2 \sqrt{1 - \left(\mathrm{Re}(ue^{-i\phi})\right)^2} \le  2 \sqrt{2 - 2|\mathrm{Re}(ue^{-i\phi})|}
$$
\end{cor}
\begin{proof}
Use \cref{thm:diamond-distance-between-unitary-channels} and let $e^{\pm i\delta}$ be eigenvalues of 
$U e^{-i\phi Z}$. The diameter of the disc enclosing both eigenvalues is $|e^{i\delta} - e^{-i\delta}| = 2|\sin(\delta)|$.
Recall that $e^{i \delta} + e^{-i \delta} = \mathrm{Tr}(U e^{-i\phi Z}) = 2 \mathrm{Re}(ue^{-i\phi})$.
That is $2\cos(\delta) = 2 \mathrm{Re}(ue^{-i\phi})$. 
Now we use 
$$
|\sin(\delta)|=\sqrt{1-\cos^2(\delta)} = \sqrt{1 - \left(\mathrm{Re}(ue^{-i\phi})\right)^2}.
$$ 
To show the remaining inequality, write $\varepsilon = 1 - |\mathrm{Re}(ue^{-i\phi})|$. 
Note that $\varepsilon$ is always positive and:
$$
\sqrt{1 - \left(\mathrm{Re}(ue^{-i\phi})\right)^2}  = \sqrt{ 1 - (1-\varepsilon)^2} = \sqrt{2\varepsilon - \varepsilon^2} \le \sqrt{2\varepsilon}
$$
Note that the inequality is tight when $\varepsilon$ goes to zero.
\end{proof}
Finally we note that for qubit unitaries the spectral norm is related to the diamond norm distance between 
the corresponding channels.
\begin{cor}[Diamond distance between qubit unitaries]
\label{cor:diamond-distance-between-general}
For an $U,V$ from $\mathrm{SU}(2)$ the diamond distance between the unitary channels $\mathcal{U}$, $\mathcal{V}$ induced by $U,V$
$$
\nrm{\mathcal{U} - \mathcal{V} }_\diamond \le 2\min\at{\nrm{U-V},\nrm{U+V}}
$$
\end{cor}
\begin{proof}
There exist unitary $V_0$ such that $V = V_0 e^{i\phi Z} V_0^\dagger$. 
By using unitary invariance of the diamond norm and spectral norm, it is sufficient to show
inequality for $U' = V_0^\dagger U V_0 =  \at{\begin{smallmatrix} u & -v^\ast \\ v & u^\ast \end{smallmatrix}}$ and diagonal unitary $e^{i\phi Z}$.
Let $e^{\pm i\delta}$ be eigenvalues of $U' e^{-i\phi Z}$, then the spectral distance between $U'$ and $e^{i\phi Z}$
is equal to $\max\{ |e^{\pm i\delta} - 1| \}  = \sqrt{2-2\cos(\delta)}$. 
Similar to the argument in \cref{cor:diamond-distance-between-general-and-diaognal}, $\cos(\delta) = \mathrm{Re}(ue^{-i\phi})$
and therefore $\nrm{U' \pm e^{i\phi Z}} = \sqrt{2\pm2\mathrm{Re}(ue^{-i\phi})}$.
Finally we use \cref{cor:diamond-distance-between-general-and-diaognal} and note that 
$$
\sqrt{2-2|\mathrm{Re}(ue^{-i\phi})|} = 2\min\at{\nrm{U-V},\nrm{U+V}}.
$$
\end{proof}
Note that bound in the above corollary is tight when $\nrm{U\pm V}$ goes to zero.

\section{Additional solutions for unitary and fallback mixing}
\label{appendix:optimal-mixing-solutions}

\subsection*{Additional solutions for unitary mixing}

In \cref{sec:unitary-mixing} we discussed a mixing strategy in which the approximation error $\varepsilon$ was evenly divided between an "under rotation"
\begin{equation}
  U_1 =
  \left(
  \begin{array}{cc}
      r_1e^{i(\theta + \delta_1)} & v_1^*                        \\
      v_1                         & r_1e^{-i(\theta + \delta_1)}
    \end{array}
  \right)
\end{equation}
and "over rotation"
\begin{equation}
  U_2 =
  \left(
  \begin{array}{cc}
      r_2e^{i(\theta + \delta_2)} & v_2^*                        \\
      v_2                         & r_1e^{-i(\theta + \delta_2)}
    \end{array}
  \right).
\end{equation}
In this section, we discuss a strategy that might lead to lower expected and worst case gate cost.

The inefficiency in partitioning $\varepsilon$ ahead of time is that it precludes solutions in which the approximation accuracy of $U_1$ and $U_2$ are significantly different. 
If, for example, $U_1$ is a close approximation to $e^{i\theta Z}$ then we would like to consider low-cost solutions for $U_2$ that are correspondingly loose.

Recall, that according to \cref{thm:unitary-mixture-diamond-distance} we randomly choose $\{S,Z\}$ twirls of 
$U_1, U_2$ with probability $p,1-p$ where 
$$
    p = \frac{r_2^2\sin(2\delta_2)}{r_2^2\sin(2\delta_2) - r_1^2\sin(2\delta_1)}.
$$
and then the diamond norm distance is given by 
$$
    ||p\mathcal{T}_{U_1} + (1-p)\mathcal{T}_{U_2} - \mathcal{Z}_\theta||_\diamond = 2\left(p(1-r_1^2\cos^2(\delta_1)) + (1-p)(1 - r_2^2\cos^2(\delta_2))\right) \le \varepsilon
$$
Suppose that we have already found a very cheap under-rotation and $r_1, \delta_1$ are fixed.
For example, when target angle $\theta$ is close to zero
but
$$\mathcal{D}_\diamond(e^{i\theta Z},I) = 2|\sin(\theta)| > \varepsilon $$
the identity gate is a very cheap under-rotated approximation, however it is not sufficiently close to $e^{i\theta Z}$
such that identity gate is a solution to the diagonal approximation problem.
Next we derive the 2D region of all possible values of $r_2e^{i(\theta + \delta_2)}$, such that
the diamond norm distance is bounded by $\varepsilon$. We focus on the rotated region for values $r_2 e^{i\delta_2}$

Let us assume that $\delta_1 > -\pi/2$, $\delta_2 < \pi/2$ and introduce the following notation: 
$$
x_1 = r_1 \cos(\delta_1),\, y_1 = -r_1 \sin(\delta_1), x = r_2\cos(\delta_2), y = r_2 \sin(\delta_2) 
$$
By definition $x_1,y_1$ are non-negative and the following constraints on $x,y$ hold:
$$
x > 0,\, y \ge 0,\, x^2 + y^2 \le 1
$$
Using this notation probabilities $p,1-p$ become: 
$$
p = xy/(xy+x_1 y_1),\, 1-p = x_1 y_1 / (xy + x_1 y_1)
$$
And the diamond norm condition becomes
$$
xy\cdot(1 - x_1^2) + x_1 y_1\cdot(1 - x^2) \le \varepsilon(xy + x_1 y_1)/2
$$
Next we collect coefficients in front of $x^2$ on the right and in front of $xy$ on the left:
$$
xy\cdot(1 - \varepsilon/2 - x_1^2 ) \le x^2\cdot(x_1 y_1) + (x_1 y_1) (\varepsilon/2 - 1)
$$
Because $x$ is positive we get: 
\begin{equation}\label{eq:mixing-hyperbola}
y\cdot(1  - \varepsilon/2 - x_1^2) \le x\cdot(x_1 y_1) + \frac{1}{x}\cdot(x_1 y_1)(\varepsilon/2 - 1).
\end{equation}

Note that when $x_1 = \sqrt{ 1 - \varepsilon/2}$, above conditions becomes $x \ge \sqrt{ 1 - \varepsilon/2}$.
This is exactly the case of evenly distributed errors between under and over rotations.
When $x_1 > \sqrt{ 1 - \varepsilon/2}$ the under-rotation is looser than evenly-distributed and over-rotation needs to be more accurate.
When $x_1 < \sqrt{ 1 - \varepsilon/2}$ under-rotation is more accurate than evenly-distributed and over-rotation can be looser.
In the latter two cases the 2D region of possible values of the top-left entry is bounded by line $y = 0$, circle $x^2 + y^2$
and the hyperbola given by~\eq{mixing-hyperbola}.
The hyperbola crosses line $y = 0$ at $x = \sqrt{1-\varepsilon/2}$. In all cases we can use our general approach to 
find sequences with top-left entry in the 2D region. See~\fig{additional-unitary-mixing}.

\begin{figure}
\centering
\includegraphics[scale=0.4]{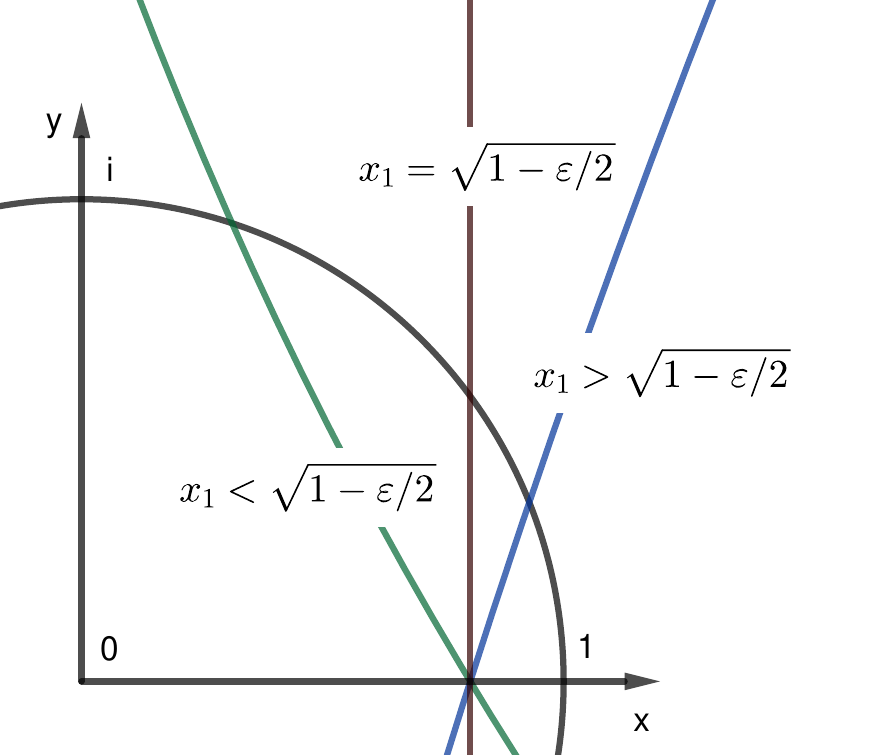}
\caption[Additional solutions for unitary mixing]{
\label{fig:additional-unitary-mixing}
Curves bounding 2D region for the top-left entry $u$ of over-rotated approximating unitary, when under-rotated approximating unitary is fixed.
All curves must be rotated by angle $\theta$ around the origin $(0,0)$ for target rotation angle $\theta$;
$x_1 = r_1 \cos(\delta_1)$, where $r_1 e^{i(\delta_1 + \theta)}$ is top-left entry of the under-rotated approximating unitary.
The curves intersection point has coordinates $(\sqrt{1-\varepsilon/2},0)$.
}
\end{figure}

\subsection*{Additional solutions for fallback mixing}

Similar analysis applies to \cref{sec:fallback-mixing}. We again consider under-rotated approximation fixed and use notation from the previous subsection.
Recall, that according to \cref{thm:fallback-mixture-diamond-distance} we mix under-rotated and over-rotated approximations with the same 
probabilities as in unitary mixing case. However, the expression for the diamond norm distance is slightly different: 

$$
\nrm{  p r_1^2 \left( \mathcal{Z}_{\theta + \delta_1} - \mathcal{Z}_\theta \right) + (1-p) r^2_2 \left( \mathcal{Z}_{\theta + \delta_2} - \mathcal{Z}_\theta \right) }_\diamond = 
2\left(p r^2_1\sin^2(\delta_1) + (1-p) r^2_2\sin^2(\delta_2)\right) \le \varepsilon
$$
Using the notation from the previous subsection we get 
$$
xy \cdot y_1^2 + y^2 \cdot x_1 y_1 \le \varepsilon(xy + x_1 y_1)/2
$$
Collecting coefficients near $xy$ on the left and coefficients near $y^2$ on the right we get: 
$$
xy\cdot(y_1^2 - \varepsilon/2) \le -y^2 \cdot x_1 y_1 + \varepsilon x_1 y_1 / 2 
$$
Using the fact that $y$ is positive and dividing both sides by $y$ we get inequality:
\begin{equation}\label{eq:fallback-mixing-hyperbola}
x\cdot(y_1^2 - \varepsilon/2) \le -y \cdot x_1 y_1 + \frac{1}{y} \cdot \varepsilon x_1 y_1 / 2.
\end{equation}
Similarly to the previous subsection, the 2D region for complex values $r_2 e^{i\delta_2}$ is bounded by a unit circle, line $x=0$
and the hyperbola~\eq{fallback-mixing-hyperbola}.

The condition $y_1 = \sqrt{\varepsilon/2}$ corresponds to equally distributing errors.
In this case we have $y \le \sqrt{\varepsilon/2}$. Note that this condition is different from the slightly weaker condition
that we use in \cref{sec:fallback-mixing}. One can modify proofs in \cref{sec:fallback-mixing}
to use condition $y \le \sqrt{\varepsilon/2}$ instead.

\section{Diamond difference of a twirled mixture}
\label{sec:diamond-difference-twirled-mixture}

In this section we prove \cref{lem:pauli-channel-error} and \cref{thm:unitary-mixture-diamond-distance} which provide expressions for the accuracy of twirled mixtures.
We begin by finding a convenient form for the $SZ$ twirl of a qubit unitary.
The following Proposition states that the twirl 
\begin{equation}
  \label{eq:sz-twirl-definition}
  \mathcal{T}_U(\rho)  
  = \frac{1}{4}\sum_{\sigma \in \{I,Z,S,S^\dagger\}} \left(\sigma U \sigma^\dagger\right) \rho \left(\sigma U^\dagger \sigma^\dagger\right)
\end{equation}
of $U$ is a Pauli channel, except for two terms: $\rho Z$ and $Z\rho$.
\begin{prop}[$SZ$ twirling of a qubit unitary]
  \label{prop:sz-twirl}
  Given a qubit unitary
  \begin{equation}
    U =
    \left(
    \begin{array}{cc}
        re^{i\theta} & -v^*          \\
        v             & re^{-i\theta}
      \end{array}
    \right),
  \end{equation}
  the twirl of $U$ over $\{S,Z\}$ can be expressed as
  \begin{equation}
    \label{eq:sz-twirl-process-matrix}
    \mathcal{T}_{U}(\rho) = \mathcal{P}_{r,\theta}(\rho) - \frac{ir^2\sin(2\theta)}{2}(\rho Z - Z\rho)
  \end{equation}
  where Pauli channel $\mathcal{P}_{r,\theta}$ is defined by
  \begin{equation}
    \mathcal{P}_{r,\theta}(\rho) = r^2\cos^2(\theta)\rho + \frac{1-r^2}{2}(X\rho X + Y\rho Y) + r^2\sin^2(\theta)Z\rho Z.
  \end{equation}
\end{prop}

\begin{proof}[Proof (sketch)]
  Recall that any qubit unitary can be written as:
  \begin{equation}
    U = t\cdot I + x\cdot iX + y\cdot iY + z\cdot iZ
  \end{equation}
  where
  \begin{equation}\begin{aligned}
      t & = \mathrm{Re}(r e^{i\theta}) = r\cos(\theta), \\
      x & =\mathrm{Im}(v),                               \\
      y & =-\mathrm{Re}(v),                            \\
      z & =\mathrm{Im}(r e^{i\theta}) = r\sin(\theta).
    \end{aligned}\end{equation}

  Using $SXS^\dagger = Y$, $SYS^\dagger = -X$ and $SZS^\dagger = Z$ we have
  \begin{align}
    ZUZ &= t\cdot I - x\cdot iX - y\cdot iY + z\cdot iZ \\
    SUS^\dagger &= t\cdot I - y\cdot iX + x\cdot iY + z\cdot iZ \\
    S^\dagger US &= t\cdot I + y\cdot iX - x\cdot iY + z\cdot iZ 
  \end{align}

  Any qubit channel can be written in term of its process matrix $\chi$ as:
  \begin{equation}
    \Phi(\rho)  = \sum_{P,Q \in \{I,X,Y,Z\}} \chi_{P,Q} P \rho Q.
  \end{equation}
  The proof then proceeds by deriving the process matrix with respect to $t, x, y, z$ for each of the four terms of \cref{eq:sz-twirl-definition} and taking the sum.  The resulting sum contains some $x^2 + y^2$ terms which can be rewritten as
  \begin{equation}
    x^2 + y^2 =
    1 - t^2 - z^2 =
    1-r^2
  \end{equation}
  in order to remove the dependence on $v$.
\end{proof}

We now proceed with proving \lemm{pauli-channel-error}, that the mixture $p\mathcal{T}_{U_1} + (1-p)\mathcal{T}_{U_2}$ yields a Pauli channel error.  
The main idea is to show that mixing $U_1$, $U_2$ eliminates the terms $\rho Z$ and $Z\rho$ from \cref{eq:sz-twirl-process-matrix}.
\begin{proof}[Proof of \cref{lem:pauli-channel-error}]
  By substituting $\mathcal{Z}_{-\theta}(\rho) = e^{-i Z\theta}\rho e^{i\theta Z}$ for $\rho$ in \cref{eq:pauli-channel-error-claim}, proving \cref{lem:pauli-channel-error} is equivalent to proving
  \begin{equation}
    p\mathcal{T}_{U_1}(\mathcal{Z}_{-\theta}(\rho)) + (1-p)\mathcal{T}_{U_2}(\mathcal{Z}_{-\theta}(\rho)) = \mathcal{E}(\rho).
  \end{equation}
  Define rotated versions of $U_1$ and $U_2$,
  \begin{equation}
    V_1 := U_1e^{-i\theta Z}, V_2 := U_2e^{-i\theta Z}.
  \end{equation}
  Then the twirl of $V_1$ is equivalent to $e^{-i\theta}$ followed by the twirl of $U_1$,
  \begin{equation}\begin{aligned}
      \mathcal{T}_{V_1}(\rho) & = \frac{1}{4} \sum_{W\in\{I,Z,S,S^\dagger\}} W V_1\rho V_1^\dagger W^\dagger                         \\
                              & = \frac{1}{4} \sum_{W\in\{I,Z,S,S^\dagger\}} W U_1 (e^{-i\theta Z} \rho e^{i\theta Z}) U_1^\dagger W^\dagger \\
                              & = \mathcal{T}_{U_1}(\mathcal{Z}_{-\theta}(\rho)).
    \end{aligned}\end{equation}
  Similarly for $V_2$,
  \begin{equation}
    \mathcal{T}_{V_2}(\rho) = \mathcal{T}_{U_2}(\mathcal{Z}_{-\theta}(\rho)).
  \end{equation}
  We therefore seek to show that
  \begin{equation}
    \label{eq:rotated-mixture-is-pauli-channel}
    p\mathcal{T}_{V_1}(\rho) + (1-p)\mathcal{T}_{V_2}(\rho) = \mathcal{E}(\rho) = p\mathcal{P}_{r_1,\delta_1} + (1-p)\mathcal{P}_{r_2,\delta_2}.
  \end{equation}

  First, expand $V_1$
  \begin{equation}
    V_1 = U_1 e^{-i\theta Z} =
    \left(
    \begin{array}{cc}
        r_1 e^{i \delta_1} & -e^{i \theta } v_1^* \\
        e^{-i \theta } v_1  & r_1 e^{-i \delta_1 }
      \end{array}
    \right)
  \end{equation}
  Next substitute into~\eq{sz-twirl-process-matrix} to obtain
  \begin{equation}
    \mathcal{T}_{V_1}(\rho) = \mathcal{P}_{r_1, \delta_1}(\rho) - \frac{ir_1^2\sin(2\delta_1)}{2}(\rho Z - Z\rho).
  \end{equation}
  A similar result is obtained for $\mathcal{T}_{V_2}$,
  \begin{equation}
    \mathcal{T}_{V_2}(\rho) = \mathcal{P}_{r_2, \delta_2}(\rho) - \frac{ir_2^2\sin(2\delta_2)}{2}(\rho Z - Z\rho).
  \end{equation}

  In order to obtain~\eq{rotated-mixture-is-pauli-channel} the $\rho Z$ and $Z\rho$ terms of $p\mathcal{T}_{V_1} + (1-p)\mathcal{T}_{V_2}$ must cancel.  Define $\alpha_k = r_k^2\sin(2\delta_k)$.  Then indeed
  \begin{equation}\begin{aligned}
      p\frac{ir_1^2\sin(2\delta_1)}{2} + (1-p)\frac{ir_2^2\sin(2\delta_2)}{2} & =
      \frac{\alpha_2}{\alpha_2 - \alpha_1}\cdot \frac{i\alpha_1}{2} + \left(1-\frac{\alpha_2}{\alpha_2 - \alpha_1}\right)\frac{i\alpha_2}{2}                                                                    \\
                                                                              & = \frac{\alpha_2}{\alpha_2 - \alpha_1}\cdot \frac{i\alpha_1}{2} - \frac{\alpha_1}{\alpha_2 - \alpha_1}\cdot \frac{i\alpha_2}{2} \\
                                                                              & = \frac{i\alpha_2\alpha_1}{2(\alpha_2 - \alpha_1)} - \frac{i\alpha_1\alpha_2}{2(\alpha_2 - \alpha_1)}                           \\
                                                                              & = 0.
    \end{aligned}\end{equation}
  Finally, note that $0 < p < 1$ since $r_1^2\sin(2\delta_1) < 0 < r_2^2\sin(2\delta_2)$ due to constraints on $\delta_1, \delta_2$.
\end{proof}

We are now ready to prove~\theo{unitary-mixture-diamond-distance} that
\begin{equation}
  ||p \mathcal{T}_{U_1} + (1-p)\mathcal{T}_{U_2} - \mathcal{Z}_\theta||_\diamond 
  = 2\left(1 - p r_1^2\cos^2(\delta_1) - (1-p)r_2^2\cos^2(\delta_2)\right).
\end{equation}
\begin{proof}[Proof of~\theo{unitary-mixture-diamond-distance}]
  First note that
  \begin{equation}
    e^{-i\theta} U_1 =
    \left(
    \begin{array}{cc}
        r_1e^{i\delta_1} & e^{i\theta} v_1^* \\
        e^{-i\theta} v_1 & r_1e^{-i\delta_1}
      \end{array}
    \right),
  \end{equation}
  \begin{equation}
    e^{-i\theta} U_2 =
    \left(
    \begin{array}{cc}
        r_2e^{i\delta_2} & e^{i\theta} v_1^* \\
        e^{-i\theta} v_1 & r_2e^{-i\delta_2}
      \end{array}
    \right).
  \end{equation}
  By~\lemm{pauli-channel-error} we have
  \begin{equation}
    p\mathcal{T}_{e^{-i\theta} U_1}(\rho) + (1-p)\mathcal{T}_{e^{-i\theta} U_2}(\rho)
    = \mathcal{E}(\mathcal{Z}_0(\rho))
    = \mathcal{E}(\rho).
  \end{equation}
  Also note that $\mathcal{Z}_\theta(\mathcal{T}_U(\rho)) = \mathcal{T}_{e^{i\theta Z}U}(\rho)$ since $e^{i\theta Z}$ commutes with $S$ and $Z$.
  Therefore by unitary invariance of the diamond norm
  \begin{equation}\begin{aligned}
      ||p\mathcal{T}_{U_1} + (1-p)\mathcal{T}_{U_2} - e^{i\theta Z}||_\diamond
       & = ||p\mathcal{Z}_{-\theta}(\mathcal{T}_{U_1}) + (1-p)\mathcal{Z}_{-\theta}(\mathcal{T}_{U_2}) - I||_\diamond \\
       & = ||p\mathcal{T}_{e^{-i\theta} U_1} + (1-p)\mathcal{T}_{e^{-i\theta} U_2} - I||_\diamond                         \\
       & = ||\mathcal{E} - I||_\diamond.
    \end{aligned}\end{equation}
  The quantity $\mathcal{E} - I$ is a Pauli channel. The diamond norm of a Pauli channel is 
  given by the sum of the absolute values of the terms in the process matrix~(see~\cref{thm:diamond-distance-beween-pauli-channels}).
  We have
  \begin{equation}\begin{aligned}
      ||\mathcal{E} - I||_\diamond =\,
       & |p r_1^2\cos^2(\delta_1) + (1-p) r_2^2\cos^2(\delta_2) - 1| \\
       & + 2|p (1-r_1^2)/2 + (1-p) (1-r_2^2)/2|                      \\
       & +|p r_1^2\sin^2(\delta_1) + (1-p) r_2^2\sin^2(\delta_2)|.
    \end{aligned}\end{equation}
  The right hand side can then be simplified to $2(1 - p r_1^2\cos^2(\delta_1) - (1-p) r_2^2\cos^2(\delta_2))$.
\end{proof}

\section{Diamond distance of a fallback mixture}
\label{sec:diamond-distance-fallback-mixture}

In this section we prove the bound of~\theo{fallback-mixture-diamond-distance} on the diamond difference between a mixture of fallback protocols and a target diagonal unitary.  The main idea is to coerce the mixture of projective rotations into the form required by~\theo{unitary-mixture-diamond-distance}.  The bound then follows by simple substitution.

\begin{proof}[Proof of \theo{fallback-mixture-diamond-distance}]
  The mixture of the two fallback channels $F_1$ and $F_2$ is given by
  \begin{equation}
    pF_1(\rho) + p'F_2(\rho) =
    p q_1\mathcal{Z}_{\theta_1}(\rho) + pq'_1\mathcal{B}_1(\rho) +
    p' q_2\mathcal{Z}_{\theta_2}(\rho) + p'q'_2\mathcal{B}_{2}(\rho),
  \end{equation}
  where we have used $p' = 1-p$ and similarly for $q_1',q_2'$.
  Using the triangle inequality we obtain
  \begin{equation}\begin{aligned}
      ||pF_1 + p'F_2 - \mathcal{Z}_{\theta}||_\diamond
      \leq & \, ||p q_1\mathcal{Z}_{\theta_1} + p'q_2\mathcal{Z}_{\theta_2} - (pq_1 + p'q_2)\mathcal{Z}_{\theta}||_\diamond \\
           & + pq'_1||\mathcal{B}_1 - \mathcal{Z}_{\theta}||_\diamond                                                       \\
           & + p' q_2'||\mathcal{B}_{\theta'_2}-\mathcal{Z}_{\theta}||_\diamond.
    \end{aligned}\end{equation}
  The second and third terms match those of~\cref{eq:fallback-mixture-bound}.
  We now examine the first term
  \begin{equation}
    \label{eq:projective-mixture-rescaled}
    ||p q_1\mathcal{Z}_{\theta_1} + p'q_2\mathcal{Z}_{\theta_2} - (pq_1 + p'q_2)\mathcal{Z}_{\theta}||_\diamond  =
    (pq_1 + p'q_2) ||s\mathcal{Z}_{\theta_1} + (1-s)\mathcal{Z}_{\theta_2} - \mathcal{Z}_\theta||_\diamond
  \end{equation}
  where
  \begin{equation}
    s = pq_1 / (pq_1 + p'q_2) = \frac{\sin(2\delta_2)}{\sin(2\delta_2) - \sin(2\delta_1)}.
  \end{equation}
  Now, since $\mathcal{Z}_\theta = \mathcal{T}_{e^{i\theta Z}}$ and by~\theo{unitary-mixture-diamond-distance} we have
  \begin{equation}\begin{aligned}
      ||s\mathcal{Z}_{\theta+\delta_1} + (1-s)\mathcal{Z}_{\theta + \delta_2} - \mathcal{Z}_\theta||_\diamond
       & = ||s\mathcal{T}_{R(\theta+\delta_1)} + (1-s)\mathcal{T}_{R(\theta + \delta_2)} - \mathcal{Z}_\theta||_\diamond \\
       & = 2(1-s\cos^2(\delta_1) - (1-s)(\cos^2(\delta_2))                                                               \\
       & = 2(s\sin^2(\delta_1) - (1-s)(\sin^2(\delta_2)).
    \end{aligned}\end{equation}
  Finally, substituting back into~\cref{eq:projective-mixture-rescaled} we obtain
  \begin{equation}\begin{aligned}
      ||p q_1\mathcal{Z}_{\theta_1} + p'q_2\mathcal{Z}_{\theta_2} - (pq_1 + p'q_2)\mathcal{Z}_{\theta}||_\diamond
       & = (pq_1 + p'q_2) \cdot 2(s\sin^2(\delta_1) - (1-s)(\sin^2(\delta_2))\\
       & = 2\left(pq_1\sin^2(\delta_1) + (1-p)q_2\sin^2(\delta_2)\right).
    \end{aligned}\end{equation}
\end{proof}

\end{document}